\documentclass[a4paper,oneside]{book}
\usepackage[british]{babel}
\usepackage{amsmath,amssymb,graphicx,stmaryrd,bbm,latexsym,calc,
  capt-of,xcolor,alltt,ifthen,caption,subcaption,cancel,hyperref,colortbl}
\usepackage[tikz]{mdframed}
\usepackage{setspace}
\usepackage[left=3.5cm, right=2.5cm, top=2.5cm, bottom=2.5cm]{geometry} 
\usepackage{cite}
\usepackage{multirow}
\newlength\mylen

\newcommand{\TeXmacs}{T\kern-.1667em\lower.5ex\hbox{E}\kern-.125emX\kern-.1em\lower.5ex\hbox{\textsc{m\kern-.05ema\kern-.125emc\kern-.05ems}}}
\newcommand{\longequal}{{=\!\!=}}
\newcommand{\mathd}{\mathrm{d}}
\newcommand{\nin}{\not\in}
\newcommand{\nobracket}{}
\newcommand{\nocomma}{}
\let\oldbullet\bullet
\let\oldcirc\circ
\let\oldodot\odot
\let\olddownarrow\downarrow
\let\olduparrow\uparrow
\let\oldcircledast\circledast
\renewcommand{\bullet}{{\oldbullet}}
\renewcommand{\circ}{{\oldcirc}}
\renewcommand{\odot}{{\oldodot}}
\renewcommand{\downarrow}{{\olddownarrow}}
\renewcommand{\uparrow}{{\olduparrow}}
\renewcommand{\circledast}{{\oldcircledast}}
\newcommand{\nospace}{}

\newcommand{\tmcolor}[2]{{\color{#1}{#2}}}
\newcommand{\tmmathbf}[1]{\ensuremath{\boldsymbol{#1}}}
\newcommand{\tmname}[1]{\textsc{#1}}
\newcommand{\tmop}[1]{\ensuremath{\operatorname{#1}}}
\newcommand{\tmrsub}[1]{\ensuremath{_{\textrm{#1}}}}
\newcommand{\tmscript}[1]{\text{\scriptsize{$#1$}}}
\newcommand{\tmtextbf}[1]{{\bfseries{#1}}}
\newcommand{\tmtextit}[1]{{\itshape{#1}}}
\newcommand{\tmtexttt}[1]{{\ttfamily{#1}}}
\newcommand{\tmverbatim}[1]{{\ttfamily{#1}}}
\newcommand{\um}{-}
\newcommand{\upl}{+}
\newenvironment{proof}{\noindent\textbf{Proof\ }}{\hspace*{\fill}$\Box$\medskip}
\newenvironment{tmcode}[1][]{\begin{alltt} }{\end{alltt}}
\newenvironment{tmindent}{\begin{tmparmod}{1.5em}{0pt}{0pt} }{\end{tmparmod}}
\newenvironment{tmparmod}[3]{\begin{list}{}{\setlength{\topsep}{0pt}\setlength{\leftmargin}{#1}\setlength{\rightmargin}{#2}\setlength{\parindent}{#3}\setlength{\listparindent}{\parindent}\setlength{\itemindent}{\parindent}\setlength{\parsep}{\parskip}} \item[]}{\end{list}}
\newtheorem{theorem}{Theorem}
\newcommand{\tmfloatcontents}{}
\newlength{\tmfloatwidth}
\newcommand{\tmfloat}[5]{
  \renewcommand{\tmfloatcontents}{#4}
  \setlength{\tmfloatwidth}{\widthof{\tmfloatcontents}+1in}
  \ifthenelse{\equal{#2}{small}}
    {\setlength{\tmfloatwidth}{0.45\linewidth}}
    {\setlength{\tmfloatwidth}{\linewidth}}
  \begin{minipage}[#1]{\tmfloatwidth}
    \begin{center}
      \tmfloatcontents
      \captionof{#3}{#5}
    \end{center}
  \end{minipage}}
\newmdenv[hidealllines=false,innertopmargin=1ex,innerbottommargin=1ex,innerleftmargin=1ex,innerrightmargin=1ex]{tmornamented}

\let\originaltable\table
\let\endoriginaltable\endtable
\renewenvironment{table}[1][h]{%
  \originaltable[#1]
  \centering}%
  {\endoriginaltable}
  
\let\originalfigure\figure
\let\endoriginalfigure\endfigure
\renewenvironment{figure}[1][h]{%
  \originalfigure[#1]
  \centering}%
  {\endoriginalfigure}

\newcommand\bigforall{\mbox{\LARGE $\mathsurround0pt\forall$}} 
\definecolor{wordblue}{rgb}{0.69,0.851,1.0}
\definecolor{wordgreen}{rgb}{0.8745,1.0,0.8745}
\definecolor{wordyellow}{rgb}{1.0,1.0,0.8745}
\newcolumntype{C}{>{\hfil$}p{\mylen}<{$\hfil}}

\addto\captionsbritish{
  \renewcommand{\contentsname}%
    {Table of Contents}%
}

\begin{document}

\title{
  {\Huge \textbf{Low-dimensional}\\
  \textbf{quantum systems}
  }
}

\author{
  \bigskip \Large \tmname{Marcin Szyniszewski}\bigskip\\
  \textit{The Department of Physics, Lancaster University}\\
  \textit{The School of Chemistry, the University of Manchester}
  \\
  \\
  \\
  A dissertation submitted for the degree of
  Doctor~of~Philosophy\\
  at Lancaster University and the University of Manchester\\
  \\
  \\
  \\
  \date{August 2016}
}

\maketitle

\newpage

\pagenumbering{gobble}

\begin{center}
  {\vspace*{10cm}\normalsize{\tmtextit{``Better to illuminate than merely to 
  shine,\\
  to deliver to others contemplated truths than merely to contemplate.''\\
  {\vspace{1cm}\hspace{21.5em}}--- Thomas Aquinas}}}
\end{center}

\chapter*{Foreword}

\pagenumbering{arabic}
\setcounter{page}{3}

This thesis describes work carried out between September 2012 and August 2016
in the Condensed Matter Theory group at the Department of Physics, Lancaster
University, and at the School of~Chemistry in the University of Manchester,
under the supervision of Dr~N.D.~Drummond. The following sections of this
thesis are included in work that has been published, is submitted or to be
submitted:
\begin{description}
  \item[Chapter \ref{ch:1-criticality}:] M.~Szyniszewski, and E.~Burovski,
      ``The generalized $t$-$V$ model in one dimension'', J.~Phys.: Conf. Ser.
      592, 012057 (2015).
  
  \item[Chapter \ref{ch:1-cdw}:] M.~Szyniszewski, ``Charge-density-wave phases
      of the generalized $t$-$V$ model'', available online on arXiv:1511.07043
      [cond-mat.str-el].
      
  \item[Chapter \ref{ch:2-results}:] M.~Szyniszewski, E.~Mostaani,
      N.~D.~Drummond, and V.~I.~Fal'ko, ``Trions and biexcitons in
      two\mbox{-}dimensional transitional metal dichalcogenide semiconductors'',
      to be submitted.\\
      E.~Mostaani, M.~Szyniszewski, C.~Price, R.~Maezono, N.~D.~Drummond, 
      and V.~I.~Fal'ko, ``Complete study of charge carrier complexes in
      two\mbox{-}dimensional semiconductors'', to be submitted.
\end{description}
This thesis is my own work and contains nothing which is the outcome of work
done in collaboration with others, except as specified in the text and
Acknowledgements. This thesis has not been submitted in substantially the same
form for the award of a higher degree elsewhere. This thesis does not exceed
the word limit of 80~000 words.

\

\

\begin{flushright}
  {\tmname{Marcin Szyniszewski}}

  \tmtextit{Lancaster, August 2016}
\end{flushright}

\chapter*{Abstract}

  We study low-dimensional quantum systems with analytical and computational
  methods. Firstly, the one-dimensional extended $t$-$V$ model of fermions
  with interactions of finite range is investigated. The model exhibits a
  phase transition between liquid and insulating regimes. We use various
  analytical approaches to generalise previous theoretical studies. We devise
  a strong coupling expansion to go beyond first-order perturbation theory.
  The method is insensitive to the presence or the lack of integrability of
  the system. We extract the ground state energy and critical parameters of
  the model near the Mott insulating commensurate density. A summary of the
  methods used is provided to give a broader view of their advantages and
  disadvantages.
  
  We also study the possible charge-density-wave phases that exist when the
  model is at the critical density. A complete description of phase diagrams
  of the model is provided: at low critical densities the phases are defined
  analytically, and at higher critical densities we tackle this problem
  computationally. We also provide a future outlook for determining the phases
  that occur at~non\mbox{-}zero temperature.
  
  Secondly, we investigate Mott-Wannier complexes of two (excitons), three
  (trions) and four (biexcitons) charge carriers in two-dimensional
  semiconductors. The fermions interact through an effective interaction of a
  form introduced by Keldysh. Our study also includes impurity-bound
  complexes. We provide a~classification of trions and biexcitons in
  transition\mbox{-}metal dichalcogenides, which incorporates the difference
  of spin polarisation between molybdenum- and tungsten-based materials. Using
  the diffusion Monte Carlo method, which is statistically exact for these
  systems, we extract binding energies of the complexes for a complete set of
  parameters of the model. Our results are compared with theoretical and
  experimental work on transition-metal dichalcogenides. Agreement is found
  for excitonic and trionic results, but we also observe a large discrepancy
  in the theoretical biexcitonic binding energies as compared to the
  experimental values. Possible reasons for this are outlined. Simple
  interpolation formulas for binding energies are provided, that can be used
  to easily determine the values within the accuracy of 5\% for any
  two\mbox{-}dimensional semiconductor. We also calculate contact pair
  densities, which in the future can be used in the determination of the
  contact interaction.

\chapter*{Acknowledgements}

This study would not have been possible without the support of many people.
Firstly, I am grateful to Evgeni Burovski for his supervision during the first
half of this project. I would also like to thank my supervisor throughout the
second part, Neil Drummond, who was kind enough to share his enormous
knowledge and skills with me. His guidance, supervision and incredible
patience were essential for the success of this work. I am also grateful to
Volodya Fal'ko for many useful discussions.

I would like to thank my fellow colleagues, Elaheh Mostaani, Cameron Price,
and Mark Danovich, for keeping me company on this scientific endeavour, and
with whom I have spent countless hours of discussions; and Viktor Z{\'o}lyomi
for his useful advices. Thanks to colleagues from the condensed matter theory
group in Physics Department of Lancaster University, especially Jake
Arkinstall, Simon Malzard, and Matthew Malcomson, for their big and small
suggestions. Very~spe{\nobreak}cial thanks to my friends who share the same
passion for physics and mathematics, Stephen Flood and Nathan Woollett, for
always being determined to make me smile.

This work has been financially supported by the Engineering and Physical
Sciences Research Council (EPSRC), NOWNANO DTC grant number EP/G03737X/1.
Computational facilities have been provided by the Lancaster High End
Computing cluster (HEC). This document was produced with GNU {\TeXmacs}.

Finally, I wish to express my love and gratitude to my dear family and
friends, for their endless support and love.

{\tableofcontents}

\

\chapter*{Introduction}

In the last century, quantum physics has mystified and perplexed the world's
greatest minds. Some of them refused to believe that a theory that has such
counter-intuitive predictions could ever be true. However, nowadays it is a
cornerstone of modern physics and technology. We now know that when we
approach nanoscale systems, they may exhibit unusual behaviour that could
never be explained by classical theory. Additionally, in the last couple of
decades, we have made serious progress in harnessing the extraordinary
properties that quantum systems exhibit to our technological advantage.
Nanotechnology is a new prominent direction, which not only gives us promise
of future advancement, but already delivers materials and devices that we can
use today.

After the experimental discovery of graphene in 2004, a two-dimensional carbon
allotrope with atomic thickness, scientists have realised that low-dimensional
quantum systems may not only be used as toy models, but can be manufactured in
real life, together with all their interesting properties. Materials with low
dimensionality may have possible applications in all areas of our lives, such
as pushing the technological limit of Moore's law into beyond\mbox{-}silicon
electronics.

Theoretical understanding of low-dimensional materials is crucial in
determining their future use and discovering their properties. With today's
advancement of computational power, we can not only tackle this problem
analytically, but also solve quantum systems using computer
simula{\nobreak}tions.

In this work, we have chosen to work with two quantum systems with low
dimensionality that exhibit very interesting properties. Firstly, we study a
system of fermions on a~one-dimensional lattice in which fermions have
long-range interactions and which displays an~insulator--conductor phase
transition. Solving this very general model can give us insight into a full
range of one-dimensional systems and can show us which methods have advantages
and disadvantages in a system with interactions that go beyond nearest
neighbours. Secondly, we investigate a few-particle bound system of fermions
in a~two\mbox{-}dimensional semiconductor. This will advance us towards
understanding and utilising the opto-electronic properties of two-dimensional
materials.

The content of this work is as follows. The first part of this thesis deals
with a~one\mbox{-}dimensional quantum system that exhibits both insulating and
conducting regimes. Chapter \ref{ch:1-intro} introduces the concept of
Luttinger liquids and Mott insulators and talks about the generalised $t$-$V$
model, its known properties and previous results. Chapter
\ref{ch:1-criticality} shows attempts to solve the model under various
conditions. Here we use analytical and numerical methods in order to provide a
successful description of the critical behaviour near the transition between
insulating and conducting phases. In Chapter \ref{ch:1-cdw}, we try to assess
the properties of the Mott insulating phases that can occur in the generalised
$t$-$V$ model and show phase diagrams of the system.

In Part II, we investigate a model of charge carrier complexes in
two-dimensional semiconductors, in particular in transition-metal
dichalcogenides. Chapter \ref{ch:2-theory} contains the theoretical
background. Chapter \ref{ch:2-method} gives a brief overview of the quantum
Monte Carlo framework -- the method we use to simulate the system. Finally,
Chapter \ref{ch:2-results} presents binding energies of charge carrier
complexes and compares our results with other experimental and theoretical
work.

\

\begin{center}
  \part{Luttinger liquids and~Mott~insulators in~one~dimension}
\end{center}

\chapter{Theoretical background}\label{ch:1-intro}

\section{Luttinger liquids and criticality}

The usual description of interacting fermions in metals at low temperatures is
done using the Fermi liquid theory {\cite{Landau1956}}. The ground state of
such a system is composed of fermions occupying all momentum states up to the
Fermi momentum (assuming isotropy), and excitations are quasi-particles, which
carry both charge and spin and obey Fermi statistics. However, Fermi liquid
theory breaks down in one dimension and another theory is needed
{\cite{Voit1995}}.

The correct theory describing interacting electrons in a one-dimensional
conductor is the Tomonaga-Luttinger liquid {\cite{Luttinger1963}}. Here, the
elementary excitations are bosonic fluctuations of two kinds: charge density
waves (plasmons) that carry charge, and spin density waves that carry spin and
propagate independently from the former. This spin-charge separation is the
most prominent difference from Fermi liquid theory {\cite{Mattis1965}}.

For a spinless case, the Hamiltonian of the diagonalised Luttinger liquid
model is {\cite{Mattis1965,Haldane1981}}:
\begin{equation}
  H_{\tmop{LL}} = v_S \sum_k | k | b_k^{\dag} b_k + \frac{\pi^{}}{2 L}  (v_N
  N^2 + v_J J^2), \label{eq:LLHamiltonian}
\end{equation}
where $b_k$ are bosonic charge density excitations with momentum $k = \frac{2
\pi i}{L}, i = \pm 1, \pm 2, \ldots$, $L$ is the system size, $N$ is the
particle number operator (or total charge) and $J$ is the current number
operator. Any bosonic excitation can be thus labelled by quantum numbers $N$
and $J$. For simplicity $\hbar$ in the equation above is set to unity.

There are three parameters present in Equation (\ref{eq:LLHamiltonian}) which
have dimensions of velocity: $v_S$ is sound velocity, similar to Fermi
velocity, which is related to bosonic excitations; $v_N$ is charge velocity,
which measures the changes in chemical potential; and $v_J$ is the current
velocity, which is a measure of the energy needed to create a charge current
throughout the chain. The charge velocity $v_N$ can be defined as
\begin{equation}
  v_N = \frac{L}{\pi}  \frac{\partial^2 E_0}{\partial N^2}, \label{eq:1-vN}
\end{equation}
where $E_0$ is the ground state energy. By introducing a flux of particles
$\phi$ going throughout the system, we can calculate the current velocity,
\begin{equation}
  \left. v_J = \frac{\pi}{L}  \frac{\partial^2 E_0}{\partial \phi^2}
  \right|_{\phi = 0} . \label{eq:1-vJ}
\end{equation}
The sound velocity can be determined using the scaling relation of the
Luttinger liquid {\cite{Haldane1981}},
\begin{equation}
  v_S = \sqrt{v_N v_J} . \label{eq:1-vS}
\end{equation}
In case of a non-interacting model, all the velocities are the same, $v_S =
v_N = v_J$, and equal to the Fermi velocity.

In a spinful case of the model, due to spin-charge separation, the Hamiltonian
will consist of Eq. (\ref{eq:LLHamiltonian}) and a similar part corresponding
to spin density waves. That additional part will be similarly described by
three analogous velocities.

According to Refs.~{\cite{Haldane1980,Haldane1981,Haldane1981eff}}, every
spinless, gapless and interacting system of fermions in 1D is a Luttinger
liquid. Thus, its low-energy physics can be described by two parameters: the
sound velocity $v_S$ and a dimensionless parameter $K$,
\begin{equation}
  K = \frac{1}{2}  \frac{v_S}{v_N} = \frac{1}{2}  \frac{v_J}{v_S},
  \label{eq:1-K}
\end{equation}
which is usually called the Luttinger liquid parameter. The value of $K$
describes the effective strength of interactions in the chain and also fully
characterises all critical exponents of the system, \tmtextit{i.e.} one can
calculate the power-law decay of all local correlation functions.

The theory of Luttinger liquids has already proven to be applicable in
experiments dealing with electrons in carbon nanotubes {\cite{Bockrath1998}},
edge states in the fractional quantum Hall effect
{\cite{Chang1996,Chang2003}}, and crystals of trapped ions
{\cite{Schneider2012}}.

\section{Mott insulators}

Many low-dimensional systems show interesting behaviour, such as the presence
of phases that cannot be explained using classical theory
{\cite{Giamarchi2003}}. One example of such a phase is a Mott insulator. If
one considers only conventional band theory within the nearly free electron
approximation {\cite{KittelBook,AshcroftBook}}, then a material in a Mott
insulating phase should conduct electricity. In other words, there is
a~non\mbox{-}zero density of charge carriers in the system, however the system
behaves like an insulator. Among the first experimentally observed Mott
insulators were some transition metal oxides {\cite{Boer1937}}
(\tmtextit{e.g.} nickel oxide), which have odd number of electrons in a~unit
cell and should therefore be conductors. However, Mott
{\cite{Mott1937,Mott1949}} proposed a~theory in which those materials behave
like insulators due to electron--electron interactions that prevent the
electrons from moving.

Mott insulators have applications ranging from high-temperature
superconductors {\cite{Kohsaka2008}} to a new type of energy-efficient field
effect transistor with fast switching times {\cite{Newns1998}}. Research into
the subject of one- and two-dimensional Mott transistors is currently ongoing
\cite{Inoue2008,Son2011,Nakano2012}. However, to
make an efficient Mott insulating device we first need an accurate description
of the underlying physics of the system.

\section{The generalised $t$-$V$ model}

\subsection{Description of the model}

The generalised $t$-$V$ model of spinless fermions in one dimension was
introduced by G{\'o}mez\mbox{-}Santos {\cite{Gomez-Santos1993}} as an example
of a model exhibiting both Luttinger liquid and Mott insulating regimes. The
Hamiltonian of the model on a periodic chain of $L$ sites is
\begin{equation}
  H = - t \sum_{i = 1}^L \left( c^{\dag}_i c_{i + 1} + \text{h.c.} \right) +
  \sum_{i = 1}^L \sum_{m = 1}^p U_m n_i n_{i + m}, \label{eq:genHamiltonian}
\end{equation}
where $c_i$ and $c^{\dag}_i$ are fermionic annihilation and creation operators
on site $i$, $n_i = c^{\dag}_i c_i$ is the particle number operator on site
$i$, $t$ is the hopping amplitude describing the kinetic part, $U_m$ is the
potential energy between two fermions $m$ sites apart from each other, and $p$
is the maximum range of interactions ($\forall_{m > p} \hspace{0.2em} U_m =
0$). There are no on\mbox{-}site interactions, \tmtextit{i.e.} $U_0 = 0$, and
all the non\mbox{-}zero interactions are repulsive, \tmtextit{i.e.} $U_m > 0$.
The Hamiltonian of the $t$-$V$ model is easily recovered by setting $p = 1$.
In this case the model is integrable (solved by the Bethe ansatz
{\cite{Bethe1931}}{\nocite{Bethe1997}} in Refs.
{\cite{Orbach1958,Walker1959}}) and equivalent to the XXZ Heisenberg model
after a Jordan\mbox{-}Wigner transformation (see Chapter
\ref{ch:1-jordanwigner}). For $p = 2$, the model is sometimes called the
$t$-$V$-$V'$ model or the $t$-$U$-$V$ model.

The kinetic part is assumed to be significantly smaller than the potential:
\begin{equation}
  t \ll U_m,
\end{equation}
and can be treated as a perturbation.

G{\'o}mez-Santos {\cite{Gomez-Santos1993}} introduces one more important
assumption,
\begin{equation}
  \underset{m}{\bigforall}\ U_m < \frac{U_{m - 1} + U_{m + 1}}{2} .
  \label{eq:GSAssumption}
\end{equation}
If the fermion-fermion distance is required to be less than $p$ sites (due to
high density in the system), then the particles will want to be as spread out
as possible. One can for example consider two similar systems, both in Fock
states, which are different only by fermion chains: $(\bullet \circ \circ
\bullet \circ \circ \bullet)$ and $(\bullet \circ \circ \circ \bullet \circ
\bullet)$, where $\bullet$ and $\circ$ denote occupied and empty sites
respectively. Assumption (\ref{eq:GSAssumption}) tells us that the first
system will always have lower energy regardless of the maximum range of
interactions, if $p > 1$. By converting condition (\ref{eq:GSAssumption}) into
\begin{equation}
  \underset{m}{\bigforall}\ \frac{U_{m + 1} + U_{m - 1} - 2 U_m}{a^2} > 0,
\end{equation}
where $a$ is the lattice constant, we can immediately see that this assumption
is a discrete version of the (continuous) inequality
\begin{equation}
  \lim_{\Delta r \rightarrow 0} \frac{U (r + \Delta r) + U (r - \Delta r) - 2
  U (r)}{(\Delta r)^2} = \frac{\mathd^2 U (r)}{\mathd r^2} > 0,
  \label{eq:condUcont}
\end{equation}
or that the potential must always fall with a decreasing rate,
\tmtextit{i.e.} the potential $U (r)$ is strictly convex.
  
One can easily check that assumption (\ref{eq:condUcont}) holds for Coulomb
and dipole potentials, and all potentials of a form (see
Fig.~\ref{fig:1-coulombpot}):
\begin{equation}
  U (r) = \frac{C}{r^k}, \quad k \nin [- 1 ; 0] .
\end{equation}
However, in principle, a potential that does not satisfy such a condition
could also be considered (such as the P{\"o}schl-Teller potential
{\cite{Poschl1933}} used in the description of ultracold atomic gases).

\begin{figure}[h]
  \resizebox{5.6cm}{!}{\includegraphics{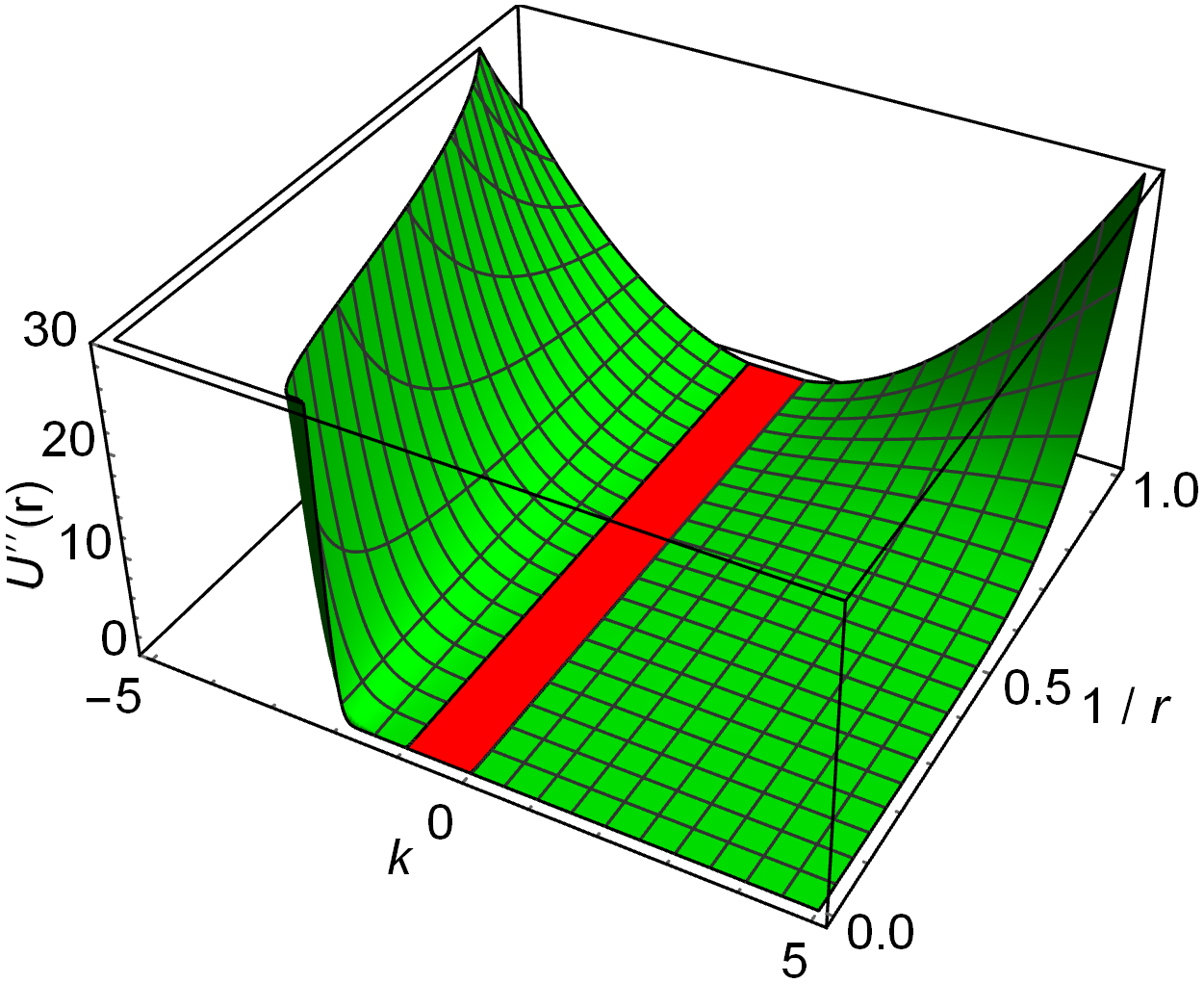}}
  \caption{
  Plot of $U'' (r)$ for $U (r) = r^{- k}$. Red region shows where the 
  assumption (\ref{eq:condUcont}) does not hold.\label{fig:1-coulombpot}
  }
\end{figure}

\subsection{Motivation}

The extended $t$-$V$ model of fermions has one prominent feature that makes it
of theoretical interest: its description is rather general, as the values of
the potentials are not set. Therefore, it could describe an experimentally
realisable one-dimensional system. With the recent nanotechnological
advancements, it has become possible to engineer condensed matter systems
using ultracold atoms in optical lattices {\cite{Lubasch2011,Haller2015}}.
Additionally, because the system exhibits the Luttinger liquid--Mott insulator
phase transition, it may be relevant to the experimental production and the
operation of a~one-dimensional Mott transistor. Finally, the investigation of
the generalised $t$\mbox{-}$V$~model may serve as a theoretical exercise that
helps in the development and classification of usefulness of analytical and
numerical methods that can deal with quantum systems with long-range
interactions.

\subsection{Formulation as a chain of spins}\label{ch:1-jordanwigner}

Every one-dimensional\footnote{\setstretch{1.66}Although the same procedure can 
be used in
higher-dimensional systems, the transformation will produce multiple non-local
terms that require keeping track of many non-local quantum numbers. In one
dimension all the non-local terms (see
Eqs.~(\ref{eq:1-JWnonlocal1}--\ref{eq:1-JWnonlocal2})) include the number of
particles $N$, which is usually fixed during the setup of the system.} system
of spinless fermions can be formulated as an equivalent system of spins
(spin-half). One way to do this is to use the Jordan-Wigner transformation
{\cite{Jordan1928}}, in which fermionic creation and annihilation operators
are redefined as chains of spin operators,
\begin{equation}
  \left\{ \begin{array}{lll}
    c_i & = & e^{- i \pi \sum_{k = 1}^{i - j} \sigma_k^+ \sigma_k^-}
    \sigma_i^-\\
    c_i^{\dag} & = & e^{i \pi \sum_{k = 1}^{i - j} \sigma_k^+ \sigma_k^-}
    \sigma_i^+
  \end{array} \right.,
\end{equation}
where $\sigma_i^{\pm} = (\sigma_i^x \pm i \sigma_i^y) / 2$ and $\sigma^x,
\sigma^y, \sigma^z$ are the Pauli spin matrices:
\begin{equation}
  \sigma^x = \left(\begin{array}{cc}
    0 & 1\\
    1 & 0
  \end{array}\right), \quad \sigma^y = \left(\begin{array}{cc}
    0 & - i\\
    i & 0
  \end{array}\right), \quad \sigma^z = \left(\begin{array}{cc}
    1 & 0\\
    0 & - 1
  \end{array}\right) .
\end{equation}
We can see that
\begin{equation}
  c_i = e^{- i \pi \sum_{k = 1}^{i - j} \sigma_k^+ \sigma_k^-} \sigma_i =
  \prod_{k = 1}^{i - 1} e^{\tmscript{- i \pi \left(\begin{array}{cc}
    1 & 0\\
    0 & 0
  \end{array}\right)_k}} \sigma_i^- = \prod_{k = 1}^{i - 1}
  \left(\begin{array}{cc}
    - 1 & 0\\
    0 & 1
  \end{array}\right)_k \sigma_i^- = \prod_{k = 1}^{i - 1} (- \sigma_k^z)
  \sigma_i^-,
\end{equation}
and the transformation is simply
\begin{equation}
  \left\{ \begin{array}{lll}
    c_i & = & \prod_{k < i} (- \sigma_k^z) \sigma_i^-\\
    c_i^{\dag} & = & \prod_{k < i} (- \sigma_k^z) \sigma_i^+
  \end{array} \right. .
\end{equation}
Therefore one can calculate the following relations:
\begin{eqnarray}
  c_i^{\dag} c_{i + 1} & = & \prod_{k < i} (- \sigma_k^z) \sigma_i^+ \prod_{l
  < i + 1} (- \sigma_l^z) \sigma_{i + 1}^- \\
  & = & (- \sigma_1^z) \cdots (- \sigma_{i - 1}^z) \sigma_i^+ (- \sigma_1^z)
  \cdots (- \sigma_{i - 1}^z) (- \sigma_i^z) \sigma_{i + 1}^- \nonumber\\
  & = & - \sigma_i^+ \sigma_i^z \sigma_{i + 1}^- = - \left(\begin{array}{cc}
    0 & 1\\
    0 & 0
  \end{array}\right)_i \left(\begin{array}{cc}
    1 & 0\\
    0 & - 1
  \end{array}\right)_i \sigma_{i + 1}^- \nonumber\\
  & = & \sigma_i^+ \sigma_{i + 1}^-, \nonumber
\end{eqnarray}
\begin{eqnarray}
  c_{i + 1}^{\dag} c_i & = & \prod_{k < i + 1} (- \sigma_k^z) \sigma_{i + 1}^+
  \prod_{l < i} (- \sigma_l^z) \sigma_i^- \\
  & = & (- \sigma_1^z) \cdots (- \sigma_{i - 1}^z) (- \sigma_i^z) \sigma_{i +
  1}^+ (- \sigma_1^z) \cdots (- \sigma_{i - 1}^z) \sigma_i^- \nonumber\\
  & = & - \sigma_i^z \sigma_i^- \sigma_{i + 1}^+ = - \left(\begin{array}{cc}
    1 & 0\\
    0 & - 1
  \end{array}\right)_i \left(\begin{array}{cc}
    0 & 0\\
    1 & 0
  \end{array}\right)_i \sigma_{i + 1}^+ \nonumber\\
  & = & \sigma_{i + 1}^+ \sigma_i^-, \nonumber
\end{eqnarray}
\begin{eqnarray}
  n_i = c_i^{\dag} c_i & = & \prod_{k < i} (- \sigma_k^z) \sigma_i^+ \prod_{l
  < i} (- \sigma_l^z) \sigma_i^- \\
  & = & (- \sigma_1^z) \cdots (- \sigma_{i - 1}^z) \sigma_i^+ (- \sigma_1^z)
  \cdots (- \sigma_{i - 1}^z) \sigma_i^- \nonumber\\
  & = & \sigma_i^+ \sigma_i^- = \left(\begin{array}{cc}
    0 & 1\\
    0 & 0
  \end{array}\right)_i \left(\begin{array}{cc}
    0 & 0\\
    1 & 0
  \end{array}\right)_i = \left(\begin{array}{cc}
    1 & 0\\
    0 & 0
  \end{array}\right)_i \nonumber\\
  & = & \mathbbm{P}_i^{\uparrow} = (1 + \sigma^z) / 2, \nonumber
\end{eqnarray}
where $\mathbbm{P}_i^{\uparrow}$ is the projector operator to the spin-up
subspace on site $i$. In a periodic chain the hoppings over the boundary must
be calculated independently,
\begin{eqnarray}
  c_L^{\dag} c_1 & = & \prod_{k < L} (- \sigma_k^z) \sigma_L^+ \sigma_1^- = (-
  \sigma_1^z) (- \sigma_2^z) \cdots (- \sigma_{L - 1}^z) \sigma_L^+ \sigma_1^-
  \\
  & = & - \left(\begin{array}{cc}
    1 & 0\\
    0 & - 1
  \end{array}\right)_1 \left(\begin{array}{cc}
    0 & 0\\
    1 & 0
  \end{array}\right)_1  \prod_{k = 2}^{L - 1} (- \sigma_k^z) \sigma_L^+ =
  \sigma_1^- \sigma_L^+ (- 1)^{\sum_{k = 2}^{L - 1} n_k} \nonumber\\
  & = & \sigma_1^- \sigma_L^+ (- 1)^{N - 1}, \label{eq:1-JWnonlocal1}
  \nonumber
\end{eqnarray}
\begin{eqnarray}
  c_1^{\dag} c_L & = & \sigma_1^+ \prod_{k < L} (- \sigma_k^z) \sigma_L^- =
  \sigma_1^+ (- \sigma_1^z) (- \sigma_2^z) \cdots (- \sigma_{L - 1}^z)
  \sigma_L^- \\
  & = & - \left(\begin{array}{cc}
    0 & 1\\
    0 & 0
  \end{array}\right)_1 \left(\begin{array}{cc}
    1 & 0\\
    0 & - 1
  \end{array}\right)_1  \prod_{k = 2}^{L - 1} (- \sigma_k^z) \sigma_L^-
  \nonumber\\
  & = & \sigma_1^+ \sigma_L^- (- 1)^{N - 1}, \label{eq:1-JWnonlocal2}
  \nonumber
\end{eqnarray}
where $N = \sum_k n_k$ is total number of particles. Terms describing hopping
across the boundary can be treated as an effective 
flux\footnote{\setstretch{1.66}Since $(-
1)^{N - 1} = e^{i \pi (N - 1)}$, the flux is $\phi = \pi (N - 1) / L$.
Therefore, in a system with an odd number of particles, there is no additional
phase shift acquired while hopping, and in a system with even $N$, there is a
$- 1 = e^{i \pi}$ phase factor that needs to be included in the bosonic
Hamiltonian.} going through the chain.

The Hamiltonian (\ref{eq:genHamiltonian}) of the generalised $t$-$V$ model
becomes:
\begin{equation}
  H = - t \sum_{i = 1}^{L - 1} (\sigma^+_i \sigma^-_{i + 1}) - t (- 1)^{N - 1}
  \sigma^+_L \sigma_1^- + \text{h.c.} + \sum_{i = 1}^L \sum_{m = 1}^p U_m
  \mathbbm{P}_i^{\uparrow} \mathbbm{P}^{\uparrow}_{i + m},
\end{equation}
or
\begin{equation}
  H = - t \sum_{i = 1}^{L - 1} (\sigma^+_i \sigma^-_{i + 1}) - t (- 1)^{N - 1}
  \sigma^+_L \sigma_1^- + \text{h.c.} + \frac{1}{4} \sum_{i = 1}^L \sum_{m =
  1}^p U_m (1 + \sigma^z_i) (1 + \sigma^z_{i + m}) .
\end{equation}
In case of $p = 1$, the model becomes equivalent to the XXZ Heisenberg model
with background effective field, pierced by the magnetic flux.

\subsection{Solution for the infinite potential}\label{ch:GSSolution}

Here we summarise the solution of the generalised $t$-$V$ model given by
G{\'o}mez-Santos in Ref.~{\cite{Gomez-Santos1993}}. This solution presents a
very simple effective picture of the system, which nevertheless gives a lot of
physical insight. Assumption (\ref{eq:GSAssumption}) will hold during this
consideration.

Firstly, let us consider the case of low energies, when any two fermions are
never close enough to incur a potential energy penalty. The system loses some
degrees of freedom, namely every particle effectively occupies $(p + 1)$
sites. A particle can therefore be thought of as a fermion with a hard core or
a hard rod occupying $(p + 1)$ sites (see Fig.~\ref{fig:1-hardrods1}a). Thus,
the system can be imagined as a chain of $N$ free fermions on
\begin{equation}
  \tilde{L} = L - Np = L (1 - Qp)
\end{equation}
sites, where $Q = N / L$ is the density of the original system. The energy can
be calculated to be simply:
\begin{equation}
  E (\{ \tilde{k} \}) = - 2 t \sum_{i = 1}^N \cos \tilde{k}_i, \quad
  \tilde{k}_i = \frac{2 \pi \tilde{n}_i}{\tilde{L}}, \quad \tilde{n}_i \in \{
  0, \ldots, \tilde{L} - 1 \},
\end{equation}
where $\{ \tilde{k} \}$ denotes a set of all particle momenta $\tilde{k}_i$ in
the system, and index $i$ labels the particles. The ground state energy
density can be calculated by occupying the lowest-energy momentum space and
taking the infinite volume limit:
\begin{equation}
  \frac{E}{L} = - \frac{2 t}{L} \sum_{\tilde{n}_i = - N / 2}^{N / 2} \cos
  \tilde{k}_i \quad \rightarrow \quad - 2 t \frac{1 - Qp}{\pi} \sin \frac{\pi
  Q}{1 - Qp} . \label{1-GSgsenergy}
\end{equation}
\begin{figure}[h]
  \resizebox{14cm}{!}{\includegraphics{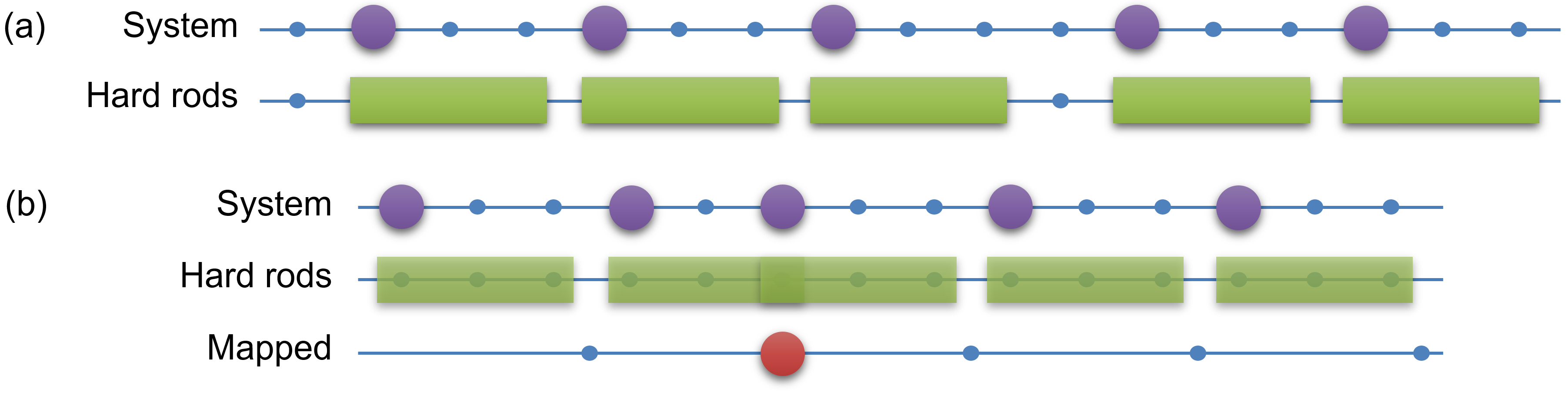}}
  \caption{(a) Low-energy subspace. Fermions are mapped to hard rods. (b)
  High-energy subspace. Overlapping rods (domain walls) are mapped to free
  fermions (red).\label{fig:1-hardrods1}}
\end{figure}

{\hspace{1.5em}}Let us now squeeze the system by removing empty sites one by
one. We notice that there is a critical density (also called commensurate
density), at which particles cannot move, otherwise it would cost them $U_p$
energy,
\begin{equation}
  Q = 1 / (p + 1) .
\end{equation}
At this density, the system is a Mott insulator and no fermion can move
without inducing a~huge energy penalty; hard rods fill the system completely
(see Fig.~\ref{fig:1-hardrods2}). The energy is zero and all Luttinger liquid
velocities also go to zero. By shaking this chain, one can create
a~charge-quasiparticle (an overlap of hard rods) moving through the system. A
coherent superposition of such quasiparticles is called a~charge-density-wave
(CDW).

\begin{figure}[h]
  \resizebox{10cm}{!}{\includegraphics{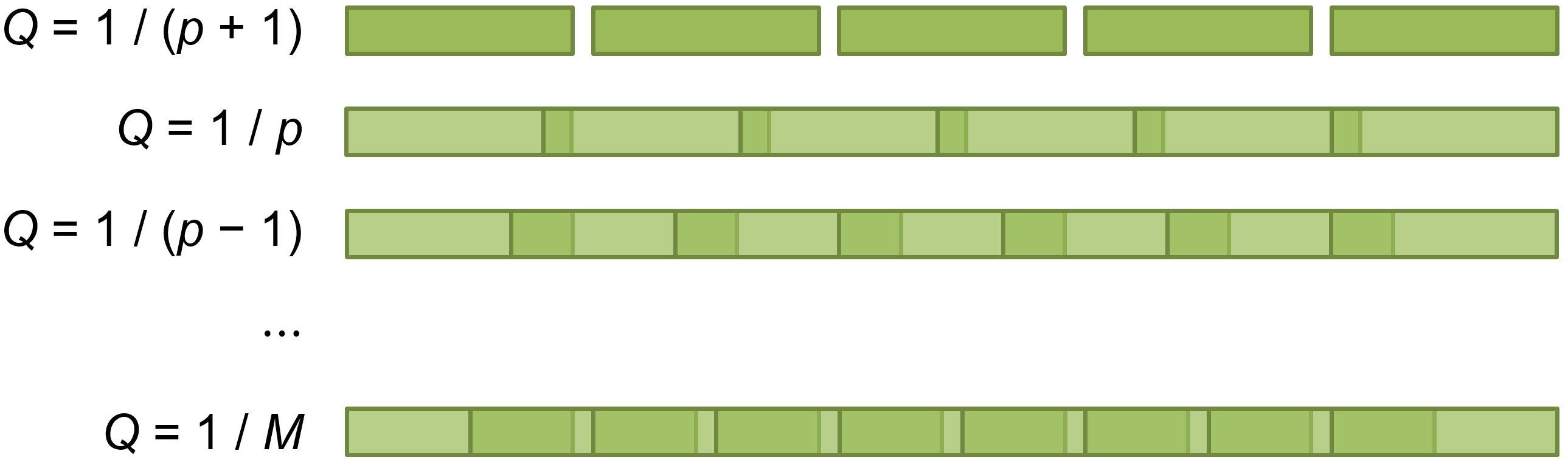}}
  \caption{Commensurate densities at which the system is a Mott
  insulator.\label{fig:1-hardrods2}}
\end{figure}

Squeezing this chain even more will necessarily add a potential energy $U_p$.
A~quasiparticle that is created in the squeezing process will behave as a free
particle in a tight-binding model (see Fig.~\ref{fig:1-hardrods1}b), with
particle hopping by $p$ sites. If we shorten the chain even more, then either
a new quasiparticle is created or the existing quasiparticle incorporates the
new instability. However, the latter is forbidden by Assumption
(\ref{eq:GSAssumption}), and therefore, by squeezing the chain, we will always
create new quasiparticles. Any number of those instabilities will behave like
particles in a tight-binding chain, until there are so many instabilities they
fill the whole system. We again reach an insulating density, this time at $Q =
1 / p$.

A similar trend continues until the chain becomes half-filled. The rest of the
insulating densities can easily be recovered using the particle-hole symmetry:
a system with density $1 - Q$, $Q \in [0, 0.5]$ exhibits the same physics as
the system with density $Q$. In short, the system is a Mott insulator at
critical (commensurate) densities
\begin{equation}
  Q = \frac{1}{m}, \quad m = p + 1, p, \ldots, M, \label{eq:1-MottDensities}
\end{equation}
and is a Luttinger liquid at other, incommensurate densities. $M = \max
(\lfloor p / 2 \rfloor + 1, 2)$ designates the last density $Q = 1 / M$, where
the potential $U_{\lfloor p / 2 \rfloor}$ does not contribute to the potential
energy of the system\footnote{\setstretch{1.66}Going beyond the density $Q = 1 
/ M$ may not be
an issue if the potential energy decreases rapidly. In such a~case, one can
expect a Mott insulating phase at densities $Q = 1 / m, m = p + 1, p, \ldots,
2$.}.

The energy of the system in the high-energy subspace can be calculated by
noticing that the system is now of the effective size $\tilde{L} = LQ$ and the
number of quasiparticles is $\tilde{N} = L (1 - Q \lfloor 1 / Q \rfloor)$.
Similarly to Eq.~(\ref{1-GSgsenergy}), the energy in the infinite volume limit
is:
\begin{equation}
  \frac{E}{L} \rightarrow - 2 t \frac{Q}{\pi} \sin \pi \left( \frac{1}{Q} -
  \left\lfloor \frac{1}{Q} \right\rfloor \right) . \label{1-GSgsenergy-high}
\end{equation}
{\hspace{1.5em}}Using Eqs.~(\ref{1-GSgsenergy}) and (\ref{1-GSgsenergy-high}),
the parameters of the Luttinger liquid were assessed to be:
\begin{equation}
  v_S = \left\{ \begin{array}{ll}
    \frac{2 t}{1 - Qp} \sin \frac{\pi Q}{1 - Qp} & \tmop{for} Q \leqslant
    \frac{1}{p + 1},\\
    \frac{2 t}{Q} \left| \sin \frac{\pi}{Q} \right| & \tmop{otherwise},
  \end{array} \right. \quad \tmop{and} \quad K = \left\{ \begin{array}{ll}
    \frac{1}{2} (1 - Qp)^2 & \tmop{for} Q \leqslant \frac{1}{p + 1},\\
    \frac{1}{2} Q^2 & \tmop{otherwise} .
  \end{array} \right. \label{eq:GSparams}
\end{equation}
This simple yet rich solution gives us the energy spectrum shown in
Fig.~\ref{fig:GSSpectrum}. The sound velocity $v_S$ is plotted in
Fig.~\ref{fig:1-vS}, which also shows the densities at which Luttinger liquid
(blue) and Mott insulator (red) phases are formed. At the insulating density,
all Luttinger liquid velocities go to zero.

\begin{figure}[h]
  \raisebox{-0.5\height}{\includegraphics{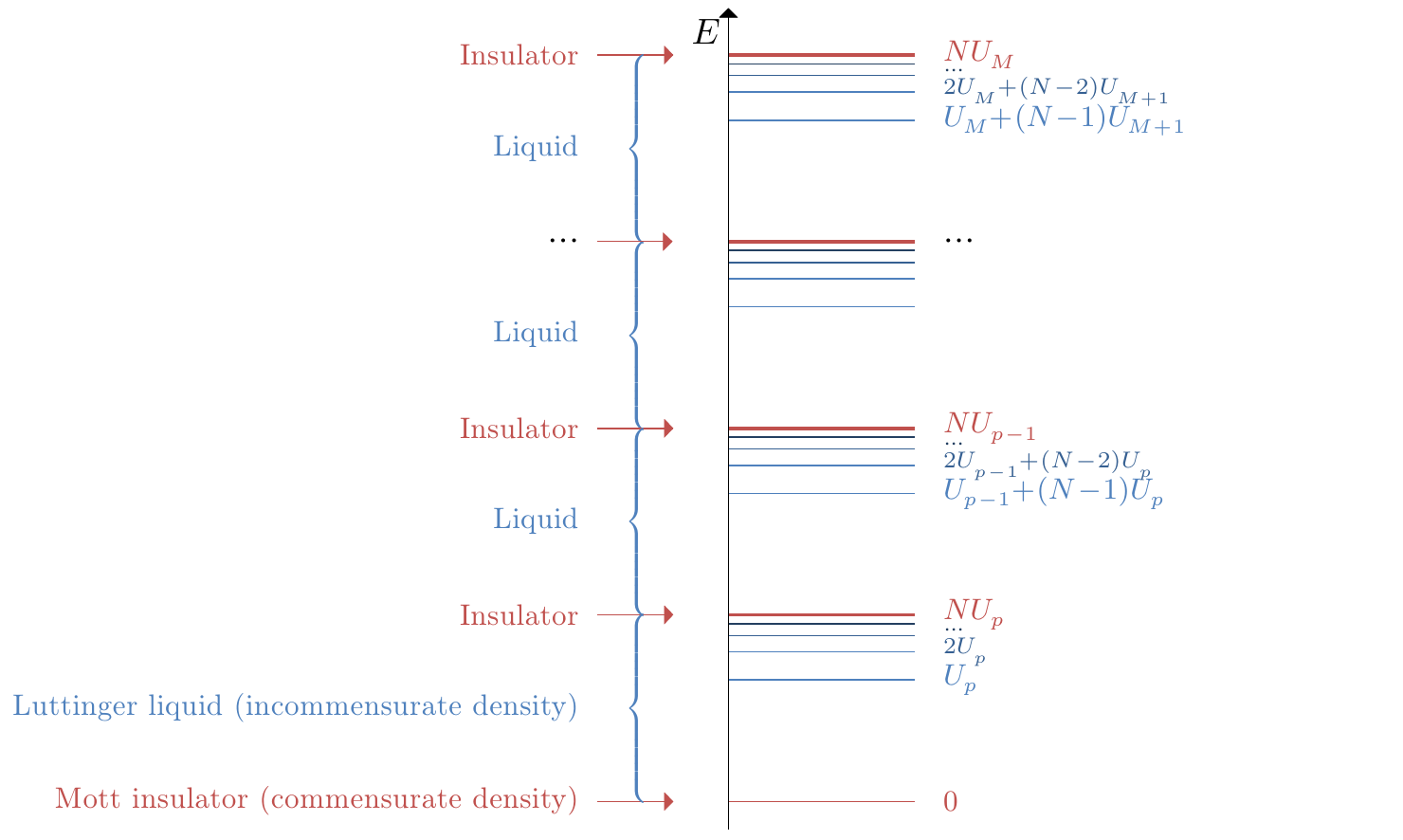}}
  \caption{Spectrum of the potential energy of the generalised $t$-$V$ model
  in G{\'o}mez-Santos's solution. The scale is logarithmic and it is assumed
  that $0 \ll U_p \ll U_{p - 1} \ll \cdots$ and that condition
  (\ref{eq:GSAssumption}) holds.\label{fig:GSSpectrum}}
\end{figure}

\begin{figure}[h]
  \resizebox{10cm}{!}{\includegraphics{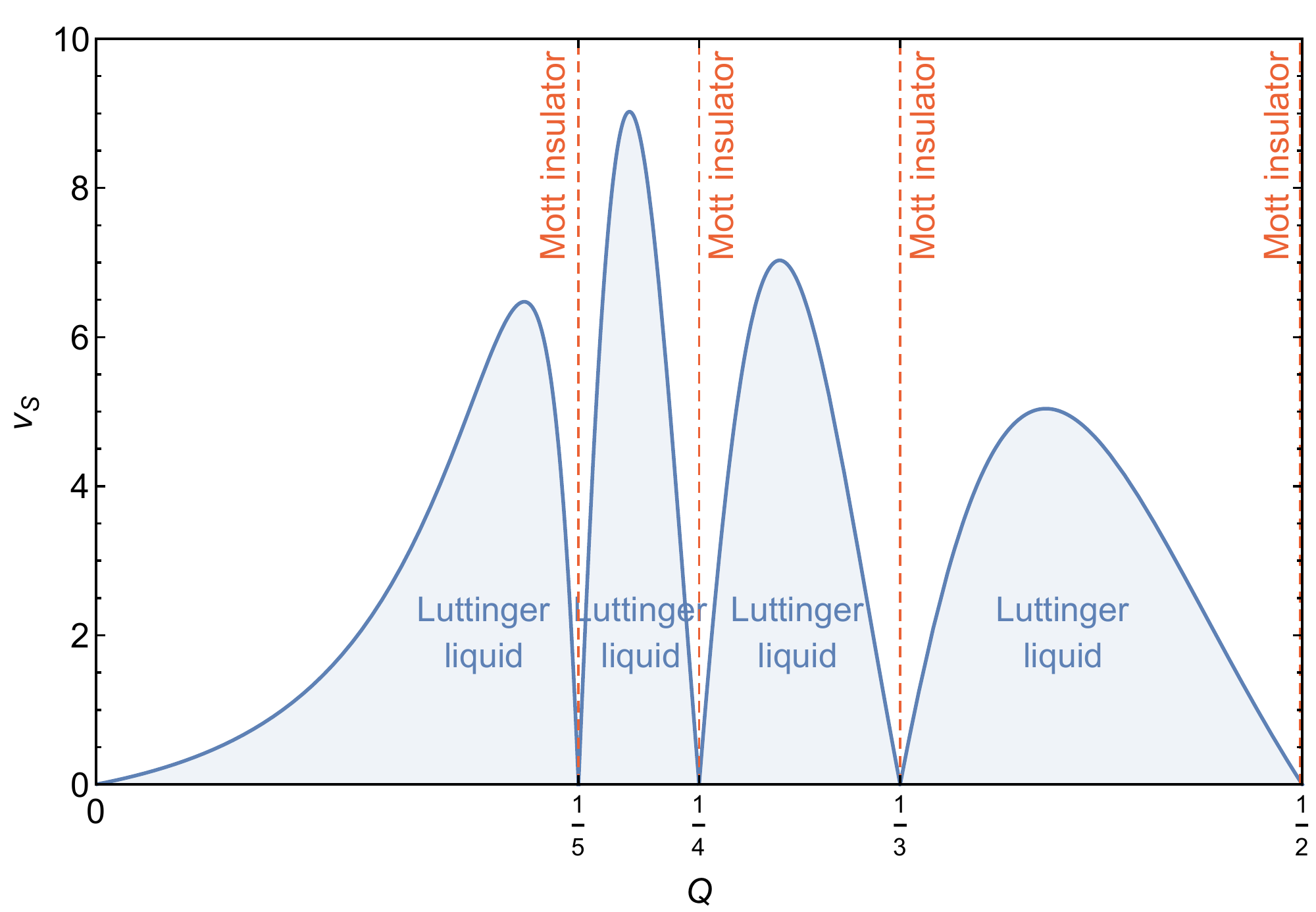}}
  \caption{Sound velocity $v_S$ as a function of density $Q$. Example for $p =
  4$.\label{fig:1-vS}}
\end{figure}

\subsection{Non-Bethe ansatz solution for the system with
nearest\mbox{-}neighbour interactions}\label{ch:DiasSolution}

Here we summarise method alternative to the Bethe ansatz of solving the
$t$-$V$ model ($p = 1$). The method was developed by Dias in
Ref.~{\cite{Dias2000}}. The main motivation behind developing a~different
method is that, although the Bethe ansatz is a~very powerful tool, it is still
based on an educated guess and therefore misses a lot of the underlying
physics.

Firstly, the kinetic part of the Hamiltonian (\ref{eq:genHamiltonian}) is
assumed to be much smaller than the potential ($t \ll U_1$). If $t = 0$, the
Hilbert space consists of degenerate subspaces with states of constant values
of $\sum_{i = 1}^L n_i n_{i + 1}$. For $t \ll U_1$ this degeneracy is lifted,
but in order to calculate the eigenvalues up to the first perturbation order,
we can diagonalise the Hamiltonian within each of the degenerate subspaces.
The projected Hamiltonian (\ref{eq:genHamiltonian}) is:
\begin{eqnarray}
  H & = & - t \sum_{i = 1}^L \left( (1 - n_{i - 1}) c^{\dag}_i c_{i + 1} (1 -
  n_{i + 2}) + \text{h.c.} \right) \\
  & - & t \sum_{i = 1}^L \left( n_{i - 1} c^{\dag}_i c_{i + 1} n_{i + 2} +
  \text{h.c.} \right) + U_1 \sum_{i = 1}^L n_i n_{i + 1} .
  \label{eq:pertHamiltonian} \nonumber
\end{eqnarray}
Now we proceed by observing which hoppings are allowed and which are forbidden
in the projected Hamiltonian. The operators $(1 - n_{i - 1}), (1 - n_{i + 2}),
n_{i - 1}$ and $n_{i + 2}$ ensure that the hopping between sites $i$ and $(i +
1)$ is only possible, if sites $(i - 1)$ and $(i + 2)$ are either both
occupied or both empty. Table \ref{tab:p1mappingearly} shows all possible
hoppings and chains in which hoppings should not be allowed.

\begin{table}[h]
  \begin{tabular}{ccc}
    \hline \hline
    Hopping & Decomposition into two-site chains & Mapping\\
    \hline
    $(\circ \circ \bullet \circ)$ & ${\color[HTML]{0070C0}(\circ \circ)}
    {\color[HTML]{76923C}(\circ \bullet)} (\bullet \circ)$ & $\downarrow
    \odot$\\
    $(\circ \bullet \circ \circ)$ & ${\color[HTML]{76923C}(\circ \bullet)}
    (\bullet \circ) {\color[HTML]{0070C0}(\circ \circ)}$ & $\odot
    \downarrow$\\
    \hline
    $(\bullet \circ \bullet \bullet)$ & $(\bullet \circ)
    {\color[HTML]{76923C}(\circ \bullet)} {\color[HTML]{C0504D}(\bullet
    \bullet)}$ & $\odot \uparrow$\\
    $(\bullet \bullet \circ \bullet)$ & ${\color[HTML]{C0504D}(\bullet
    \bullet)} (\bullet \circ) {\color[HTML]{76923C}(\circ \bullet)}$ &
    $\uparrow \odot$\\
    \hline
    \multicolumn{3}{c}{No hoppings should be allowed } \\
    \hline
    $(\circ \circ \bullet \bullet)$ & ${\color[HTML]{0070C0}(\circ \circ)}
    {\color[HTML]{76923C}(\circ \bullet)} {\color[HTML]{C0504D}(\bullet
    \bullet)}$ & $\downarrow \cancel{\odot} \uparrow$\\
    $(\circ \bullet \circ \bullet)$ & ${\color[HTML]{76923C}(\circ \bullet)}
    (\bullet \circ) {\color[HTML]{76923C}(\circ \bullet)}$ & $\odot \odot$\\
    $(\bullet \circ \bullet \circ)$ & $(\bullet \circ)
    {\color[HTML]{76923C}(\circ \bullet)} (\bullet \circ)$ & $\odot$\\
    $(\bullet \bullet \circ \circ)$ & ${\color[HTML]{C0504D}(\bullet \bullet)}
    (\bullet \circ) {\color[HTML]{0070C0}(\circ \circ)}$ & $\uparrow
    \downarrow$\\
    \hline \hline
  \end{tabular}
  \caption{All possible hoppings from the second to the third site in a
  four-site chain in a $p = 1$ system. Two-site chains are mapped according to
  Eq.~(\ref{eq:DiasMapping}).\label{tab:p1mappingearly}}
\end{table}

In order to simplify the problem, we now decompose each four-site chain into
two-site chains. The two-site chains are then named according to the following
mapping:
\begin{eqnarray}
  \hspace{6em} {\color[HTML]{C0504D}(\bullet \bullet)} & = & \uparrow \quad
  \text{(occupied site, spin-up)} \\
  {\color[HTML]{0070C0}(\circ \circ)} & = & \downarrow \quad \text{(occupied
  site, spin-down)} \nonumber\\
  {\color[HTML]{76923C}(\circ \bullet)} & = & \odot \hspace{0.75em}
  \text{(empty site).} \label{eq:DiasMapping} \nonumber
\end{eqnarray}
Notice that in order to make the mapping unambiguous, the $(\bullet \circ)$
chain is always discarded during the mapping. To maintain the properties of
the system, the mapping must also preserve which states are allowed to hop and
which are forbidden to hop. The spin-up occupied site $\uparrow$ is only
allowed to hop with the empty site $\odot$. Similarly, the spin-down occupied
site $\downarrow$ can only hop with the empty site $\odot$. However, no other
hoppings are allowed.

Using the rules above, we notice that there is one four-site chain, which
introduces a~problem: in $(\circ \circ \bullet \bullet)$ chain the hopping
should not be possible, but after the mapping the chain becomes $(\downarrow
\odot \uparrow)$, in which the empty site $\odot$ can either hop with the
$\downarrow$ site or the $\uparrow$ site. We notice that this problem is due
to the two-site chain $(\circ \bullet)$ acting as a domain wall between empty
sites $\cdots \circ \circ \cdots$ and occupied sites $\cdots \bullet \bullet
\cdots$. Therefore, in order to recover unambiguity of the mapping, we are
forced to remove any chain $(\circ \bullet)$, which acts as a domain wall
between sites $\downarrow$ and $\uparrow$.

\begin{tmornamented}[skipbelow=1em,skipabove=1em,roundcorner=1.7ex]
  \tmtextbf{Example.} A chain (with periodic boundary conditions)
  \begin{equation}
    \circ \circ \circ \circ \bullet \bullet \bullet \circ \circ \circ
  \end{equation}
  can be decomposed into two-site chains
  \begin{equation}
    {\color[HTML]{0070C0}(\circ \circ)} {\color[HTML]{0070C0}(\circ \circ)}
    {\color[HTML]{0070C0}(\circ \circ)} {\color[HTML]{76923C}(\circ \bullet)}
    {\color[HTML]{C0504D}(\bullet \bullet)} {\color[HTML]{C0504D}(\bullet
    \bullet)} (\bullet \circ) {\color[HTML]{0070C0}(\circ \circ)}
    {\color[HTML]{0070C0}(\circ \circ)} {\color[HTML]{0070C0}(\circ \circ)}
  \end{equation}
  and mapped into
  \begin{equation}
    \downarrow \downarrow \downarrow \nospace \cancel{\odot} \uparrow \uparrow
    \downarrow \downarrow \downarrow .
  \end{equation}
  Notice that because the empty site $\odot$ was acting as a domain wall, it
  must be removed during the mapping.
\end{tmornamented}

\begin{tmornamented}[skipbelow=1em,skipabove=1em,roundcorner=1.7ex]
  \tmtextbf{Example.} A chain
  \begin{equation}
    \circ \circ \circ \circ \bullet \circ \bullet \circ \bullet \bullet
    \bullet \circ \circ \circ
  \end{equation}
  can be decomposed into the following two-site chains
  \begin{equation}
    {\color[HTML]{0070C0}(\circ \circ)} {\color[HTML]{0070C0}(\circ \circ)}
    {\color[HTML]{0070C0}(\circ \circ)} {\color[HTML]{76923C}(\circ \bullet)}
    (\bullet \circ) {\color[HTML]{76923C}(\circ \bullet)} (\bullet \circ)
    {\color[HTML]{76923C}(\circ \bullet)} {\color[HTML]{C0504D}(\bullet
    \bullet)} {\color[HTML]{C0504D}(\bullet \bullet)} (\bullet \circ)
    {\color[HTML]{0070C0}(\circ \circ)} {\color[HTML]{0070C0}(\circ \circ)}
    {\color[HTML]{0070C0}(\circ \circ)}
  \end{equation}
  and mapped into
  \begin{equation}
    \downarrow \downarrow \downarrow \nospace \odot \odot \cancel{\odot} 
    \uparrow
    \uparrow \downarrow \downarrow \downarrow .
  \end{equation}
  Again, one of the empty sites $\odot$ was acting as a domain wall between
  $\downarrow$ and $\uparrow$, and therefore it must be removed during the
  mapping.
\end{tmornamented}

So, if in the system there is a spin down $\downarrow$ state followed by empty
site(s) $\odot$ and a~spin up $\uparrow$, one of the empty sites needs to be
removed, because it acts as a domain wall:
\begin{equation}
  \cdots \downarrow \nospace \odot \cdots \nospace \odot \cancel{\odot} \uparrow
  \cdots .
\end{equation}
{\hspace{1.5em}}Finally, let us consider a periodic system, which begins with
two-site chain that will be discarded during the mapping (it can be either a
$(\bullet \circ)$ chain or a $(\circ \bullet)$ that acts as a~domain wall).
Notice that the mapping is no longer unique: translating the system by one
site to the left will result in a system with the same mapping. To mend this
situation, we have to translate the system first, so that the first two-site
chain can be mapped.

\begin{tmornamented}[skipbelow=1em,skipabove=0.5em,roundcorner=1.7ex]
  \tmtextbf{Example.} The periodic system
  \begin{equation}
    \bullet \circ \bullet \circ \circ \bullet \bullet
  \end{equation}
  starts with a $(\bullet \circ)$ chain. Therefore, we use a translation
  operator $T$,
  \begin{equation}
    T (\circ \bullet \circ \circ \bullet \bullet \bullet),
  \end{equation}
  and now we can map the system:
  \begin{equation}
    T \left( \odot \downarrow \nospace \cancel{\odot} \uparrow \uparrow \right) 
    .
    \label{1-Diasmapexampletrans}
  \end{equation}
\end{tmornamented}

Notice that if before the mapping the system has $L$ sites, after the mapping
the system has length:
\begin{equation}
  \tilde{L} = L - N_{\odot} - 2 N_{\tmop{DW}},
\end{equation}
where $N_{\odot}$ is the number of mapped empty sites $\odot$ and
$N_{\tmop{DW}}$ is the number of empty sites acting as domain walls
$\cancel{\odot}$. The number of fermions after the mapping is
\begin{equation}
  \tilde{N} = N_{\uparrow} + N_{\downarrow} .
\end{equation}
{\hspace{1.5em}}If the initial states were designated by positions $a_i$ of
the fermions,
\begin{equation}
  | a_1, \ldots, a_N \rangle = \prod_{i = 1}^N c^{\dag}_{a_i} | 0 \rangle,
\end{equation}
then in the mapped system, the states will be designated by positions of
fermions, $\tilde{a}_i$, and their spins, $\sigma_i$,
\begin{equation}
  | \tilde{a}_1, \ldots, \tilde{a}_{\tilde{N}} ; \sigma_1, \ldots,
  \sigma_{\tilde{N}} \rangle = \prod_{i = 1}^{\tilde{N}}
  \tilde{c}^{\dag}_{\tilde{a}_i \sigma_i} | \tilde{0} \rangle .
\end{equation}
\begin{tmornamented}[skipbelow=1em,skipabove=0em,innerleftmargin=1em,innerrightmargin=1em,roundcorner=1.7ex]
  \tmtextbf{Example.} One can map
  \begin{equation}
    \bullet \bullet \circ \circ \bullet \circ \bullet \circ \bullet
  \end{equation}
  into
  \begin{equation}
    \uparrow \downarrow \nospace \odot \odot \cancel{\odot} \uparrow .
  \end{equation}
  A spinless chain of $N = 5$ fermions on $L = 9$ sites is therefore mapped
  into a spinful chain of $\tilde{N} = 3$ fermions on $\tilde{L} = 5$ sites.
  The initial state is
  \begin{equation}
    | 1, 2, 5, 7, 9 \rangle
  \end{equation}
  and after mapping it becomes
  \begin{equation}
    | \tilde{1}, \tilde{2}, \tilde{5} ; \uparrow, \downarrow, \uparrow \rangle
    .
  \end{equation}
\end{tmornamented}

\begin{tmornamented}[skipbelow=1em,roundcorner=1.7ex]
  \tmtextbf{Example.} The system from Eq.~(\ref{1-Diasmapexampletrans})
  corresponds to the following state:
  \begin{equation}
    | 1, 3, 6, 7 \rangle \longrightarrow T | \tilde{2}, \tilde{3}, \tilde{4} ;
    \downarrow, \uparrow, \uparrow \rangle . \label{eq:1-Diasmapexampletrans2}
  \end{equation}
\end{tmornamented}

Now, we create new states that are invariant by translation and have total
momentum $P$:
\begin{equation}
  | \{ \tilde{a} \}, \{ \sigma \}, P \rangle = | \{ a \}, P \rangle =
  \frac{1}{\sqrt{L}} \sum_{j = 1}^L e^{iPj} T^{j - 1} | \{ a \} \rangle .
\end{equation}
Due to the requirement that the first two-site chain must be possible to map
(see the example state from Eq.~(\ref{eq:1-Diasmapexampletrans2})), the
Hamiltonian elements that correspond to jumps $\tilde{1} \rightarrow
\tilde{L}$ and $\tilde{2} \rightarrow \tilde{1}$, will give additional $e^{\pm
iP}$ factors, that can be incorporated into the hopping constant $t$. A
detailed table of all these elements is shown in Ref.~{\cite{Dias2000}}
p.~7795. The Hamiltonian $H_1 = H - U_1 N_{\uparrow}$ becomes
\begin{equation}
  H_1 (P) = - \sum_{\tilde{i}, \sigma} t_{\tilde{i} \sigma} (1 - n_{\tilde{i}
  \sigma}) \tilde{c}^{\dag}_{\tilde{i} \sigma} \tilde{c}_{\tilde{i} + 1,
  \sigma} (1 - n_{\tilde{i} + 1, \sigma}) + \text{h.c.}, \label{eq:DiasMap}
\end{equation}
with $t_{\tilde{L} \uparrow} = t (- 1)^{L - N}, t_{\tilde{L} \downarrow} =
te^{iP} (- 1)^{L - N}, t_{\tilde{1} \uparrow} = te^{\sigma_{\tilde{N}} iP}$
and $t_{\tilde{i} \sigma} = t$ otherwise. This is the Hamiltonian of a $U
\rightarrow \infty$ Hubbard chain pierced by a magnetic 
flux\footnote{\setstretch{1.66}The
Hubbard model is a model of spinful fermions with the following Hamiltonian:
\begin{equation}
  H = - t \sum_{\langle i, j \rangle, \sigma} \left( c^{\dag}_{i, \sigma}
  c_{j, \sigma} + \text{h.c.} \right) + U \sum_{i = 1}^N n_{i \uparrow} n_{i
  \downarrow} .
\end{equation}
Notice that the sites can be doubly occupied, unless $U \rightarrow \infty$.}.

{\hspace{1.5em}}Now, we consider a subspace of states in the Hamiltonian, with
the same configuration of $\{ \sigma \}$:
\begin{equation}
  H_1 (P, \{ \sigma \}) = - t \sum_{\tilde{i} \neq \tilde{L}}
  \tilde{c}_{\tilde{i} + 1}^{\dag} \tilde{c}_{\tilde{i}} - t_{\tilde{1}
  \sigma_1} \tilde{c}^{\dag}_{\tilde{1}} \tilde{c}_{\tilde{2}} - t_{\tilde{L}
  \sigma_1} \tilde{c}^{\dag}_{\tilde{L}} \tilde{c}_{\tilde{1}} Q +
  \text{h.c.}, \label{eq:DiasSpin}
\end{equation}
where $Q$ is the cyclic spin permutation operator, that permutes $\{ \sigma
\}$ in a state. We will now do the following gauge transformation,
\begin{equation}
  \forall_{\tilde{a}_1 \geqslant 2} \enspace | \tilde{a}_1, \ldots,
  \tilde{a}_{\tilde{N}} ; \uparrow, \ldots, \sigma_{\tilde{N}} \rangle
  \rightarrow e^{\sigma_{\tilde{N}} iP} | \tilde{a}_1, \ldots,
  \tilde{a}_{\tilde{N}} ; \uparrow, \ldots, \sigma_{\tilde{N}} \rangle .
\end{equation}
The Hamiltonian (\ref{eq:DiasSpin}) in the subspace of the same spin
configuration will now have the same form as in Eq.~(\ref{eq:DiasMap}), but
with $t_{\tilde{L} \sigma} = (- 1)^{L - N} te^{\frac{1}{2} (1 + \sigma_1 \cdot
\sigma_{\tilde{N}}) iP}$ and $t_{\tilde{i} \sigma} = t$ otherwise. If we now
hop a fermion across the boundary, it will induce a cyclic permutation of $\{
\sigma \}$ with the same phase factor that is included in $t_{\tilde{L}
\sigma}$. Now, we want to define states that remain invariant under such
cyclic permutations:
\begin{equation}
  Q_{\{ \sigma \}} \left( \sum_{i = 1}^{r_{\alpha_c}} \tilde{a}_i Q^i | \{
  \sigma \} \rangle \right) = e^{i \phi' / r_{\alpha_c}} \left( \sum_{i =
  1}^{r_{\alpha_c}} \tilde{a}_i Q^i | \{ \sigma \} \rangle \right),
  \label{eq:DiasCyclic}
\end{equation}
where $r_{\alpha_c}$ is the periodicity of the spins $\{ \sigma \}$ (notice
that $r_{\alpha_c}$ must be a divisor of $\tilde{N}$), $\alpha_c$ is a number
designating a possible spin configuration, $Q_{\{ \sigma \}}$ is defined by
\begin{equation}
  Q_{\{ \sigma \}} | \sigma_1, \ldots, \sigma_{\tilde{N}} \rangle =
  \frac{t_{\tilde{L} \sigma}}{t} Q | \sigma_1, \ldots, \sigma_{\tilde{N}}
  \rangle,
\end{equation}
and $\phi'$ is the effective flux through the newly redefined system.

\begin{tmornamented}[skipbelow=1.0em,skipabove=0.5em,roundcorner=1.7ex]
  \tmtextbf{Example.} Spin periodicity of some example systems:
  \begin{eqnarray}
    (\odot \uparrow \nospace \odot \uparrow \uparrow) \rightarrow
    \hspace{6.2em} | \tilde{2}, \tilde{4}, \tilde{5} ; \uparrow, \uparrow,
    \uparrow \rangle & \rightarrow & r_{\alpha_c} = 1, \\
    \left( \downarrow \downarrow \nospace \odot \cancel{\odot} \uparrow \nospace
    \odot \downarrow \nospace \odot \odot \downarrow \nospace \cancel{\odot}
    \uparrow \right) \rightarrow | \tilde{1}, \tilde{2}, \tilde{4}, \tilde{6},
    \tilde{9}, \widetilde{10} ; \downarrow, \downarrow, \uparrow, \downarrow,
    \downarrow, \uparrow \rangle & \rightarrow & r_{\alpha_c} = 3, \nonumber\\
    T \left( \odot \downarrow \nospace \cancel{\odot} \uparrow \uparrow \right)
    \rightarrow \hspace{5.4em} T | \tilde{2}, \tilde{3}, \tilde{4} ;
    \downarrow, \uparrow, \uparrow \rangle & \rightarrow & r_{\alpha_c} = 3.
    \nonumber
  \end{eqnarray}
\end{tmornamented}

We want the system to become a~one\mbox{-}particle tight\mbox{-}binding model
with $r_{\alpha_c}$ sites and hopping constant $te^{\frac{1}{2} (1 + \sigma_1
\cdot \sigma_{\tilde{N}}) iP}$ and with one\mbox{-}particle states defined as
\begin{equation}
  Q^{i - 1} | \{ \sigma \} \rangle \equiv | i \rangle .
\end{equation}
To achieve all that, the following gauge transformation is needed,
\begin{equation}
  te^{\frac{1}{2} (1 + \sigma_1 \cdot \sigma_{\tilde{N}}) iP} \rightarrow
  te^{i \phi_1 / r_{\alpha_c}},
\end{equation}
where the total flux $\phi_1$ through this tight-binding chain is found to be
\begin{equation}
  \phi_1 = r_{\alpha_c}  \frac{\tilde{N} - 2 N_{\tmop{DW}}}{\tilde{N}} P,
\end{equation}
where $N_{\tmop{DW}}$ is the number of $(\circ \bullet)$ chains that are
treated as domain walls ($\cancel{\odot}$). The chain now is a~Hubbard model 
with
flux $\phi_1$ and its eigenstates are Bloch states $| \alpha_c, q_c \rangle$,
where $q_c$ is a momentum of the Bloch state in the cyclic permutations of the
$\{ \sigma \}$ configuration,
\begin{equation}
  q_c = \frac{2 \pi n_{\alpha_c}}{r_{\alpha_c}}, \quad n_{\alpha_c} = 0,
  \ldots, r_{\alpha_c} - 1.
\end{equation}
Following Ref.~{\cite{Peres2000}}, the energy (for odd $N$) is given by
\begin{equation}
  E (\{ \tilde{k} \}, q_c, P) = - 2 t \sum_{i = 1}^{\tilde{N}} \cos \left(
  \tilde{k}_i + \alpha \frac{P}{\tilde{L}} + \frac{q_c}{\tilde{L}} \right),
\end{equation}
with $\alpha = \frac{2 \tilde{L} - L}{\tilde{N}} = \frac{\tilde{N} - 2
N_{\tmop{DW}}}{\tilde{N}}$, while for even $N$ there is a $+
\frac{\pi}{\tilde{L}}$ correction under the cosine. There is also the
following condition on the total momentum that needs to be satisfied:
\begin{equation}
  P \frac{L}{\tilde{L}} = \sum_{i = 1}^{\tilde{N}} \tilde{k}_i + \tilde{N} 
  \frac{q_c}{\tilde{L}} \quad (\tmop{mod} 2 \pi) .
\end{equation}
To calculate the ground state energy, we assume that there are no $(\bullet
\bullet)$ chains, and thus no spin-down particles in the mapped system. Then,
the spin periodicity is $r_{\alpha_c} = 1$, and thus $q_c = 0$. Particles
occupy the lowest momentum states, and thus $\sum_i \tilde{k}_i = 0$. The
ground state energy can be therefore evaluated to be:
\begin{equation}
  E = \left\{ \begin{array}{ll}
    - 2 t \frac{\sin \frac{\pi N}{\tilde{L}}}{\sin \frac{\pi}{\tilde{L}}} &
    \tmop{for} \tmop{odd} N,\\
    - 2 t \frac{\sin \frac{\pi N}{\tilde{L}}}{\sin \frac{\pi}{\tilde{L}}} \cos
    \frac{\pi}{L} & \tmop{for} \tmop{even} N.
  \end{array} \right. \label{1-Diasgsenergy}
\end{equation}
{\hspace{1.5em}}In short, by using the mapping (\ref{eq:DiasMapping}), one can
change the spinless system into a spinful chain with magnetic flux. Then, by
alternating gauge transformations and redefinition of states, so that the new
states are invariant by translation (in a mapped position space and mapped
spin space), we obtain a simple Hubbard model with a known solution.

We would like to use a similar method to solve the generalised ($p \neq 1$)
$t$-$V$ model of fermions, since this method provides a more complete
description of the system than, for example, Bethe ansatz. In contrast to
G{\'o}mez-Santos' solution summarised in Chapter \ref{ch:GSSolution}, this
method also considers domains of high energy, that can be present in the
system. For example, the chain $(\circ \circ \bullet \bullet \bullet \circ
\circ \circ)$ is never considered in G{\'o}mez-Santos' solution, but is
present in Dias' solution. Thus, we could not only achieve more physical
understanding, but also include the full spectrum of possible states.

Secondly, the ground state behaviour described in the G{\'o}mez-Santos'
solution is only considered with the infinite volume limit: comparing Eqs.
(\ref{1-GSgsenergy}) and (\ref{1-Diasgsenergy}) shows a discrepancy, but one
can see that they match if $L \rightarrow \infty$:
\begin{equation}
  \lim_{L \rightarrow \infty} E^{\text{Dias}} = E^{\text{G{\'o}mez-Santos}}_{p
  = 1} .
\end{equation}
\chapter{Solving the generalised $t$-$V$ model}\label{ch:1-criticality}

\section{Finite volumes with any interaction range and with infinite
potential}\label{ch:DiasLongSolution}

\subsection{Low-energy subspace}\label{ch:DiasBasedLow}

Similarly to the solution presented in Chapter~\ref{ch:DiasSolution}, we can
introduce a mapping of the original states into a system of lower length. Let
us consider a small subchain of $p + 1$ consecutive sites. Of course there are
$2^{p + 1}$ possible configurations of such a chain. In the low\mbox{-}energy
subspace though, there are no $(p + 1)$-site subchains which include more than
one particle, because otherwise there would be energy penalty $\geqslant U_p$.
Thus we are left with the following possible (one-particle and no-particle)
chains:
\begin{equation}
  (\bullet \circ \circ \cdots \circ) ; (\circ \bullet \circ \cdots \circ) ;
  (\circ \circ \bullet \cdots \circ) ; \cdots ; (\circ \circ \circ \cdots
  \bullet) ; (\circ \circ \circ \cdots \circ) .
\end{equation}
However, to have a unique one-to-one mapping (similarly to mapping
(\ref{eq:DiasMapping})), we forget about all other subchains except $(\circ
\cdots \circ \circ \bullet)$ and $(\circ \circ \circ \cdots \circ)$.
Therefore, we are left with only two subchains, which we will name for the
sake of simplicity
\begin{eqnarray}
  (\circ \cdots \circ \circ \bullet) & = & \circledast \quad (\tmop{occupied}
  \tmop{site}),  \label{map}\\
  (\circ \cdots \circ \circ \circ) & = & \odot \quad (\tmop{empty}
  \tmop{site}) . \nonumber
\end{eqnarray}
Notice that in the whole system with $N$ particles, there are exactly
$\tilde{L} = L - Np$ subchains like these and thus the size of the mapped
system is $\tilde{L}$.

\subsubsection{One particle}

Of course, with only one particle the solution is quite simple and we expect
to have a~free particle that can propagate through the system with momentum $k
= \frac{2 \pi n}{L}, n = 0, \ldots, L - 1$. However, this case serves as an
example of how to proceed in the many-particle case.

Firstly, a state with a particle on site $i$ in the original chain can be
written as:
\begin{equation}
  | i \rangle = c_i^{\dag} | 0 \rangle .
\end{equation}
Using mapping (\ref{map}), we can rename those states to:
\begin{equation}
  | \tilde{i} \rangle = c_{i + p}^{\dag} | 0 \rangle =
  \tilde{c}_{\tilde{i}}^{\dag} | \tilde{0} \rangle,
\end{equation}
where $\tilde{c}_{\tilde{i}}^{\dag}$ is the creation operator of subchain
$\circledast$ and $| \tilde{0} \rangle$ is the ``empty space'' state $\odot
\odot \odot \cdots$ of length $\tilde{L} = L - p$. However, some states $| i
\rangle$ are not included in this new set, namely states for $i \leq p$. We can
use translation operator $T$ to define them:
\begin{equation}
  c_i^{\dag} | 0 \rangle = T^{- p + i - 1} c_{p + 1}^{\dag} | 0 \rangle = T^{-
  p + i - 1} \tilde{c}_{\tilde{1}}^{\dag} | \tilde{0} \rangle \quad \tmop{for}
  i \leq p.
\end{equation}
The Hamiltonian in the new notation is therefore
\begin{eqnarray}
  \frac{H}{- t} & = & \sum_{i = p + 1}^{L - 1} c_{\tilde{i}}^{\dag}
  c_{\tilde{i} + 1} + c_L^{\dag} c_1 + c_1^{\dag} c_2 + \cdots + c_p^{\dag}
  c_{p + 1} + \text{h.c.} \\
  & = & \sum_{\tilde{i} = \tilde{1}}^{\tilde{L} - 1}
  \tilde{c}_{\tilde{i}}^{\dag} \tilde{c}_{\tilde{i} + 1} +
  \tilde{c}_{\tilde{L}}^{\dag} \tilde{c}_{\tilde{1}} T^p + T^{- p}
  \tilde{c}_{\tilde{1}}^{\dag} \tilde{c}_{\tilde{1}} T^{p - 1} + \cdots + T^{-
  1} \tilde{c}_{\tilde{1}}^{\dag} \tilde{c}_{\tilde{1}} + \text{h.c.}
  \nonumber\\
  & = & \sum_{\tilde{i} = \tilde{1}}^{\tilde{L} - 1}
  \tilde{c}_{\tilde{i}}^{\dag} \tilde{c}_{\tilde{i} + 1} +
  \tilde{c}_{\tilde{L}}^{\dag} \tilde{c}_{\tilde{1}} T^p + \sum_{m = 1}^p T^{-
  p + m - 1} \tilde{c}_{\tilde{1}}^{\dag} \tilde{c}_{\tilde{1}} T^{p - m} +
  \text{h.c.} \nonumber
\end{eqnarray}
Following Dias {\cite{Dias2000}}, we introduce an over-complete set of states
invariant by translation and with momentum $k$:
\begin{equation}
  | \tilde{i}, k \rangle = \frac{1}{\sqrt{L}} \sum_{j = 1}^L e^{i \nospace k
  \nospace j} T^{j - 1} | \tilde{i} \rangle . \label{1-Diasovercomplete}
\end{equation}
Creation and annihilation operators of these states will be designated as
$\tilde{c}_{\tilde{i}, k}^{\dag}$ and $\tilde{c}_{\tilde{i}, k}$ respectively.
The Hamiltonian becomes:
\begin{eqnarray}
  \frac{H}{- t} & = & \sum_k \left( \sum_{\tilde{i} = \tilde{1}}^{\tilde{L} -
  1} \tilde{c}_{\tilde{i}, k}^{\dag} \tilde{c}_{\tilde{i} + 1, k} + e^{i
  \nospace k \nospace p} \tilde{c}_{\tilde{L}, k}^{\dag} \tilde{c}_{\tilde{1},
  k} + \sum_{m = 1}^p e^{- i k} \tilde{c}_{\tilde{1}, k}^{\dag}
  \tilde{c}_{\tilde{1}, k} + \text{h.c.} \right) \\
  & = & \sum_k \left( \sum_{\tilde{i} = \tilde{1}}^{\tilde{L} - 1}
  \tilde{c}_{\tilde{i}, k}^{\dag} \tilde{c}_{\tilde{i} + 1, k} + e^{i \nospace
  k \nospace p} \tilde{c}_{\tilde{L}, k}^{\dag} \tilde{c}_{\tilde{1}, k} +
  \text{h.c.} + (2 p \cos k)  \tilde{c}_{\tilde{1}, k}^{\dag}
  \tilde{c}_{\tilde{1}, k} \right) \nonumber\\
  & = & \sum_k \left( \sum_{\tilde{i} = \tilde{1}}^{\tilde{L} - 1}
  \tilde{c}_{\tilde{i}, k}^{\dag} \tilde{c}_{\tilde{i} + 1, k} + e^{i \nospace
  k \nospace p} \tilde{c}_{\tilde{L}, k}^{\dag} \tilde{c}_{\tilde{1}, k} +
  \text{h.c.} \right) . \label{eq:1-Dias1particleHam} \nonumber
\end{eqnarray}
This is the Hamiltonian of a tight-binding model with fictitious flux $e^{i
\nospace k \nospace p}$. Eigenvalues are given by
\begin{equation}
  E (\tilde{k}, k) = - 2 t \cos \left( \tilde{k} - \frac{kp}{\tilde{L}}
  \right), \label{eq:1-Dias1particleenergy}
\end{equation}
with condition
\begin{equation}
  \tilde{k} = k \frac{L}{\tilde{L}} \quad (\tmop{mod} 2 \pi) .
  \label{eq:1-Dias1partcondition}
\end{equation}
This final condition comes from the fact that eigenstates of Hamiltonian
(\ref{eq:1-Dias1particleHam}) are combinations of the over-complete set of
states from Eq.~(\ref{1-Diasovercomplete}) and are non-zero only if
Eq.~(\ref{eq:1-Dias1partcondition}) is fulfilled. Finally, combining condition
(\ref{eq:1-Dias1partcondition}) and equation (\ref{eq:1-Dias1particleenergy})
for the energy, we recover the expected value for the energy of a particle
with momentum $k$:
\begin{equation}
  E (k) = - 2 t \cos k.
\end{equation}
\subsubsection{Many particle case}\label{ch:1-manyparticles}

Let us now consider a general case of $N$ particles. We must remember that any
$(p + 1)$ subchains must contain at most one particle. A specific realisation
of our system can be written as:
\begin{equation}
  | a_1, \ldots, a_N \rangle = \prod_{i = 1}^N c^{\dag}_{a_i} | 0 \rangle
  \label{mpstates},
\end{equation}
where $\{ a_i \}$ are the consecutive positions of particles in our system,
$a_i \in \{ 1, \ldots, L \}$. There is an additional condition on the first
number, $a_1 > p$. This means that the first particle cannot occupy sites from
$1$ to $p$ and thus there are states not included in this naming. However, we
can represent them using the translation operator $T$:
\begin{eqnarray}
  c_1 \prod_{i = 2}^N c^{\dag}_{a_i} | 0 \rangle & = & T^{- p} | p + 1, a_2 +
  p, \ldots, a_N + p \rangle, \\
  c_2 \prod_{i = 2}^N c^{\dag}_{a_i} | 0 \rangle & = & T^{- p + 1} | p + 1,
  a_2 + p - 1, \ldots, a_N + p - 1 \rangle, \nonumber\\
  & \vdots &  \nonumber\\
  c_p \prod_{i = 2}^N c^{\dag}_{a_i} | 0 \rangle & = & T^{- 1} | p + 1, a_2 +
  1, \ldots, a_N + 1 \rangle . \nonumber
\end{eqnarray}
Now, we will use mapping (\ref{map}) and rename states (\ref{mpstates}) to
\begin{equation}
  | \tilde{a}_1, \ldots, \tilde{a}_N \rangle = \prod_{i = 1}^N
  \tilde{c}^{\dag}_{\tilde{a}_i} | \tilde{0} \rangle,
\end{equation}
where $\{ \tilde{a}_i \}$ are now consecutive positions of subchains
$\circledast$, $| \tilde{0} \rangle$ is the ``empty'' system $\odot \odot \odot
\cdots$ of size $\tilde{L} = L - Np$ and $\tilde{c}^{\dag}_{\tilde{a}_i}$
creates a subchain $\circledast$ at position $\tilde{a}_i$. We can now create a
state invariant by translation, with momentum $P = \frac{2 \pi n}{L} \nocomma, n
= 0, \ldots, L - 1$:
\begin{equation}
  | \tilde{a}_1, \ldots, \tilde{a}_N, P \rangle = \frac{1}{\sqrt{L}} \sum_{j =
  1}^L e^{i P j} T^{j - 1} | \tilde{a}_1, \ldots, \tilde{a}_N \rangle .
\end{equation}
The creation and annihilation operators of the new states will be designated
as $\tilde{c}_{\tilde{i}, P}^{\dag}$ and $\tilde{c}_{\tilde{i} + 1, P}$. The
Hamiltonian between the new states is thus:
\begin{equation}
  \frac{H (P)}{- t} = \sum_P \left( \sum_{\tilde{i} \neq \tilde{L}}
  \tilde{c}_{\tilde{i}, P}^{\dag} \tilde{c}_{\tilde{i} + 1, P} + e^{i p P}
  \tilde{c}_{\tilde{1}, P}^{\dag} \tilde{c}_{\tilde{L}, P} \right) +
  \text{h.c.}
\end{equation}
This is a Hamiltonian of a tight-binding chain with $\tilde{L}$ sites and
fictitious flux $p P$. The energy spectrum is given by:
\begin{equation}
  E (\{ \tilde{k}_i \}, P) = - 2 t \sum_{i = 1}^N \cos \left( \tilde{k}_i -
  \frac{p P}{\tilde{L}} \right),
\end{equation}
where $\tilde{k}_i = \frac{2 \pi \tilde{n}_i}{\tilde{L}}, \tilde{n}_i = 0,
\ldots, \tilde{L} - 1$ is the pseudochain momentum. However, not all
combinations of $\{ \tilde{k}_i \}$ and $P$ are possible, and thus we get the
following condition, similar to Eq.~(\ref{eq:1-Dias1partcondition}):
\begin{equation}
  \sum_{i = 1}^N \left( \tilde{k}_i - \frac{p P}{\tilde{L}} \right) - P = 0
  \quad (\tmop{mod} 2 \pi),
\end{equation}
or
\begin{equation}
  P \frac{L}{\tilde{L}} = \sum_{i = 1}^N \tilde{k}_i \quad (\tmop{mod} 2 \pi)
  .
\end{equation}
The ground state energy can be calculated by assuming that only the
lowest-momentum states are occupied and is given by an equations similar to
Eq.~(\ref{1-Diasgsenergy}), but with $\tilde{L} = L - Np$:
\begin{eqnarray}
  E_{\tmop{odd} N} & = & - 2 t \frac{\sin \left( \frac{\pi N}{\tilde{L}}
  \right)}{\sin \left( \frac{\pi}{\tilde{L}} \right)}, 
  \label{eq:1-DiasManyEnergy1}\\
  E_{\tmop{even} N} & = & - 2 t \frac{\sin \left( \frac{\pi N}{\tilde{L}}
  \right)}{\sin \left( \frac{\pi}{\tilde{L}} \right)} \cos \left(
  \frac{\pi}{L} \right) .  \label{eq:1-DiasManyEnergy2}
\end{eqnarray}

\subsubsection{Luttinger liquid parameters and comparison with
G{\'o}mez-Santos res{\nobreak}ults}

We can compare our results to the ones calculated by G{\'o}mez-Santos and
summarised in Chapter \ref{ch:GSSolution}. To do that we need to find the
infinite volume limits of the Luttinger liquid parameters. Using equations
(\ref{eq:1-DiasManyEnergy1}--\ref{eq:1-DiasManyEnergy2}) for the ground state
energy and the definitions (\ref{eq:1-vN}--\ref{eq:1-vJ}), we find the
following expansions for the Luttinger liquid velocities for $L \rightarrow
\infty$:
\begin{eqnarray}
  v_N (E_{\tmop{odd}}) & \approx & \frac{2 t}{(1 - p Q)^3} \sin \frac{\pi Q}{1
  - p Q} \\
  &  & + t \frac{(\pi^2 - 2 p^2 (1 - pQ)^2) \sin \frac{\pi Q}{1 - p Q} - 4
  \pi p (1 - p Q) \cos \frac{\pi Q}{1 - p Q}}{3 (1 - p Q)^5}  \frac{1}{L^2} +
  O \left( \frac{1}{L^4} \right), \nonumber\\
  v_N (E_{\tmop{even}}) & \approx & \frac{2 t}{(1 - p Q)^3} \sin \frac{\pi
  Q}{1 - p Q} \\
  &  & + t \frac{- (\pi^2 (2 - 6 p Q + 3 p^2 Q^2) + 2 p^2 (1 - p Q)^2) \sin
  \frac{\pi Q}{1 - p Q} - 4 \pi p (1 - p Q) \cos \frac{\pi Q}{1 - p Q}}{3 (1 -
  p Q)^5}  \frac{1}{L^2} \nonumber\\
  & & + O \left( \frac{1}{L^4} \right)
  , \nonumber
\end{eqnarray}
\begin{eqnarray}
  v_J (E_{\tmop{odd}}) & \approx & 2 t (1 - p Q) \sin \frac{\pi Q}{1 - p Q} +
  t \frac{\pi^2}{3 (1 - p Q)} \sin \left( \frac{\pi Q}{1 - p Q} \right) 
  \frac{1}{L^2} + O \left( \frac{1}{L^4} \right), \\
  v_J (E_{\tmop{even}}) & \approx & 2 t (1 - p Q) \sin \frac{\pi Q}{1 - p Q}
  \\
  &  & + t \frac{\pi^2 (- 2 + 6 p Q - 3 p^2 Q^2)}{3 (1 - p Q)} \sin \left(
  \frac{\pi Q}{1 - p Q} \right)  \frac{1}{L^2} + O \left( \frac{1}{L^4}
  \right) . \nonumber
\end{eqnarray}
The Luttinger liquid parameters are calculated using Eqs.~(\ref{eq:1-vS}) and
(\ref{eq:1-K}):
\begin{eqnarray}
  v_S (E_{\tmop{odd}}) & \approx & \frac{2 t}{1 - p Q} \sin \frac{\pi Q}{1 - p
  Q}  \label{eq:1-MyCalcvS}\\
  &  & + t \frac{(\pi^2 - p^2 (1 - p Q)^2) \sin \frac{\pi Q}{1 - p Q} - 2 \pi
  p (1 - p Q) \cos \frac{\pi Q}{1 - p Q}}{3 (1 - p Q)^3}  \frac{1}{L^2} + O
  \left( \frac{1}{L^4} \right), \nonumber\\
  v_S (E_{\tmop{even}}) & \approx & \frac{2 t}{1 - p Q} \sin \frac{\pi Q}{1 -
  p Q} \\
  &  & + t \frac{- (\pi^2 (2 - 6 p Q + 3 p^2 Q^2) + p^2 (1 - p Q)^2) \sin
  \frac{\pi Q}{1 - p Q} - 2 \pi p (1 - p Q) \cos \frac{\pi Q}{1 - p Q}}{3 (1 -
  p Q)^3}  \frac{1}{L^2} \nonumber\\
  &  & + O \left( \frac{1}{L^4} \right), \nonumber
\end{eqnarray}
\begin{equation}
  K \approx \frac{1}{2} (1 - p Q)^2 + \frac{1}{12} p (1 - p Q) \left( p (1 - p
  Q) + 2 \pi \cot \frac{\pi Q}{1 - p Q} \right)  \frac{1}{L^2} + O \left(
  \frac{1}{L^4} \right) . \label{eq:1-MyCalcK}
\end{equation}
Notice that critical parameter $K$ has the same form for both even and odd
$N$.

All the results agree with G{\'o}mez-Santos results up to the first order
expansion in $L$.

\subsection{Domain walls -- high energy subspace}\label{ch:DiasBasedHigh}

A similar method can be used for the densities beyond the first critical
density of the system, in which charge density waves will be necessarily
present and can be treated as particles moving throughout the chain. We will
use a picture similar to that which G{\'o}mez-Santos introduced in
Ref.~{\cite{Gomez-Santos1993}}: charge quasiparticles with higher potential
energy, effectively occupying a smaller chain, will be identified as occupied
sites and quasiparticles with lower potential energy, effectively occupying a
larger chain, will be treated as empty sites. Therefore, the mapping is:
\begin{eqnarray}
  (\underbrace{\circ \cdots \circ \circ}_{p - d + 1} \bullet) & = & \odot
  \quad (\tmop{empty}), \\
  (\underbrace{\circ \cdots \circ}_{p - d} \bullet) & = & \circledast \quad
  (\tmop{occupied}), \nonumber
\end{eqnarray}
where $d$ is the number of insulating phases that occurred during squashing of
the chain (see Fig.~\ref{fig:1-valofd}).

\begin{figure}[h]
  \resizebox{7cm}{!}{\includegraphics{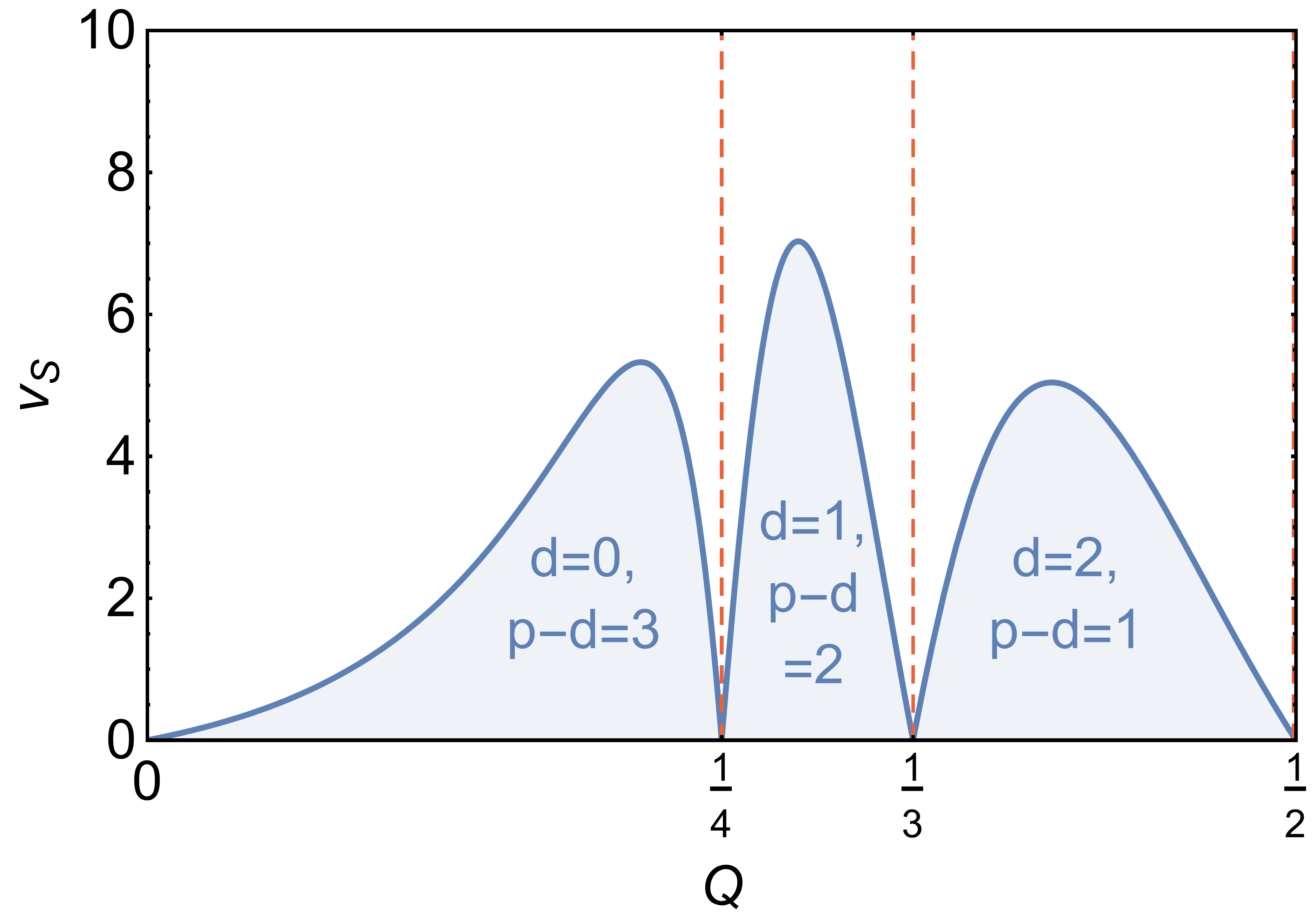}}\qquad\resizebox{7cm}{!}{\includegraphics{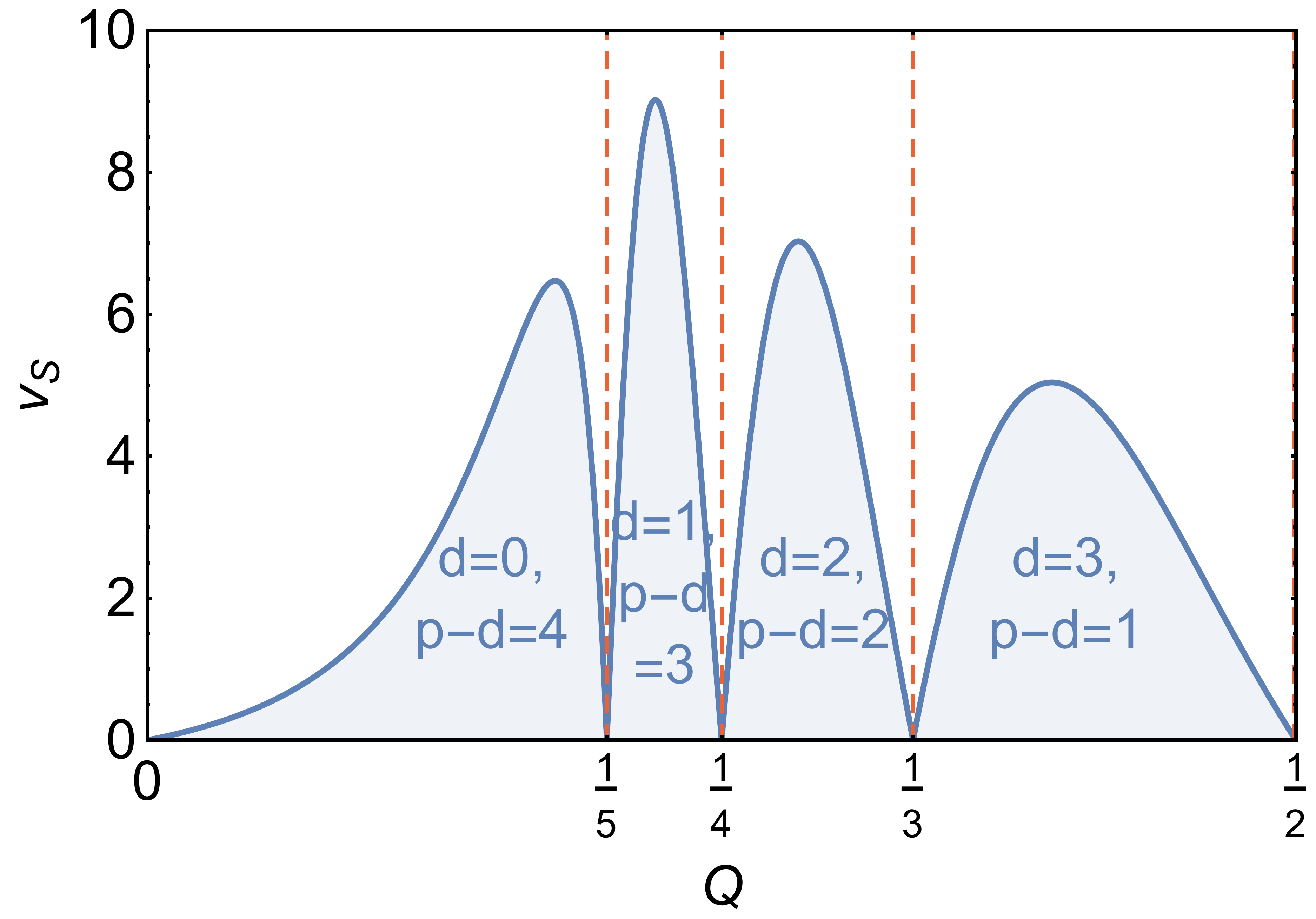}}
  \caption{Example of how to assign values of $d$ for $p = 3$ and $p =
  4$.\label{fig:1-valofd}}
\end{figure}

\begin{tmornamented}[skipbelow=0.5em,skipabove=1em,roundcorner=1.7ex]
  \tmtextbf{Example.} The following mapping can be done, if $p - d = 3$:
  \begin{equation}
    \circ \circ \circ \circ \bullet \circ \circ \circ \circ \bullet \circ
    \circ \circ \bullet \circ \circ \circ \bullet \circ \circ \circ \circ
    \bullet \cdots \xrightarrow[\tmop{mapping}]{} \odot \odot \circledast 
    \circledast \odot \cdots .
  \end{equation}
\end{tmornamented}

The new chain has $N_D = (p - d + 2) N - L$ fermions and a length of $L_D =
N$. Using the steps from Section \ref{ch:1-manyparticles} with the following
modifications,
\begin{equation}
  p \rightarrow (p - d), \qquad N \rightarrow N_D, \qquad \tilde{L}
  \rightarrow L_D,
\end{equation}
the ground state energy of the system is calculated to be:
\begin{eqnarray}
  E_{\tmop{odd} N} & = & - 2 t \frac{\sin \pi \left( p - d - \frac{1}{Q}
  \right)}{\sin \left( \frac{\pi}{N} \right)}, \\
  E_{\tmop{even} N} & = & - 2 t \frac{\sin \pi \left( p - d - \frac{1}{Q}
  \right)}{\sin \left( \frac{\pi}{N} \right)} \cos \left( \frac{\pi}{L}
  \right) . 
\end{eqnarray}
Using Eqs.~(\ref{eq:1-vN}--\ref{eq:1-K}), the Luttinger liquid parameters in
the limit $L \rightarrow \infty$ are found to be (for odd $N$):
\begin{eqnarray}
  v_N & \approx & - \frac{2 t}{Q^3} (- 1)^{p + d} \sin \frac{\pi}{Q} + \frac{t
  (- 1)^{p + d}}{3 Q^5} \left( 4 \pi Q \cos \frac{\pi}{Q} - (\pi^2 - 2 Q^2)
  \sin \frac{\pi}{Q} \right)  \frac{1}{L^2} + O \left( \frac{1}{L^4} \right),
  \\
  v_J & \approx & - 2 tQ (- 1)^{p + d} \sin \frac{\pi}{Q} - \frac{\pi^2 t}{3
  Q} (- 1)^{p + d} \sin \left( \frac{\pi}{Q} \right)  \frac{1}{L^2} + O \left(
  \frac{1}{L^4} \right), \\
  v_S & \approx & \frac{2 t}{Q} \left| \sin \frac{\pi}{Q} \right| - \frac{t}{3
  Q^3} \left( - \pi^2 + Q^2 + 2 \pi Q \cot \frac{\pi}{Q} \right) \left| \sin
  \frac{\pi}{Q} \right|  \frac{1}{L^2} + O \left( \frac{1}{L^4} \right), 
  \label{eq:1-vSDWMine}\\
  K & \approx & \frac{1}{2} Q^2 + \frac{1}{12} Q \left( Q + 2 \pi \cot
  \frac{\pi}{Q} \right)  \frac{1}{L^2} + O \left( \frac{1}{L^4} \right) . 
\end{eqnarray}
Parameters $v_S$ and $K$ are to be consistent with previous results
{\cite{Gomez-Santos1993}} up to the first 
order\footnote{\setstretch{1.66}Velocities $v_N$ and
$v_J$ were not given in Ref.~{\cite{Gomez-Santos1993}} for the high-energy
subspace.}. Figure~\ref{fig:1-vSMine} shows the comparison of the sound
velocity obtained in Eqs.~(\ref{eq:1-MyCalcvS}) and (\ref{eq:1-vSDWMine}), and
its infinite volume limit obtained in the G{\'o}mez-Santos picture.

\begin{figure}[h]
  \begin{tabular}{ll}
    \resizebox{7cm}{!}{\includegraphics{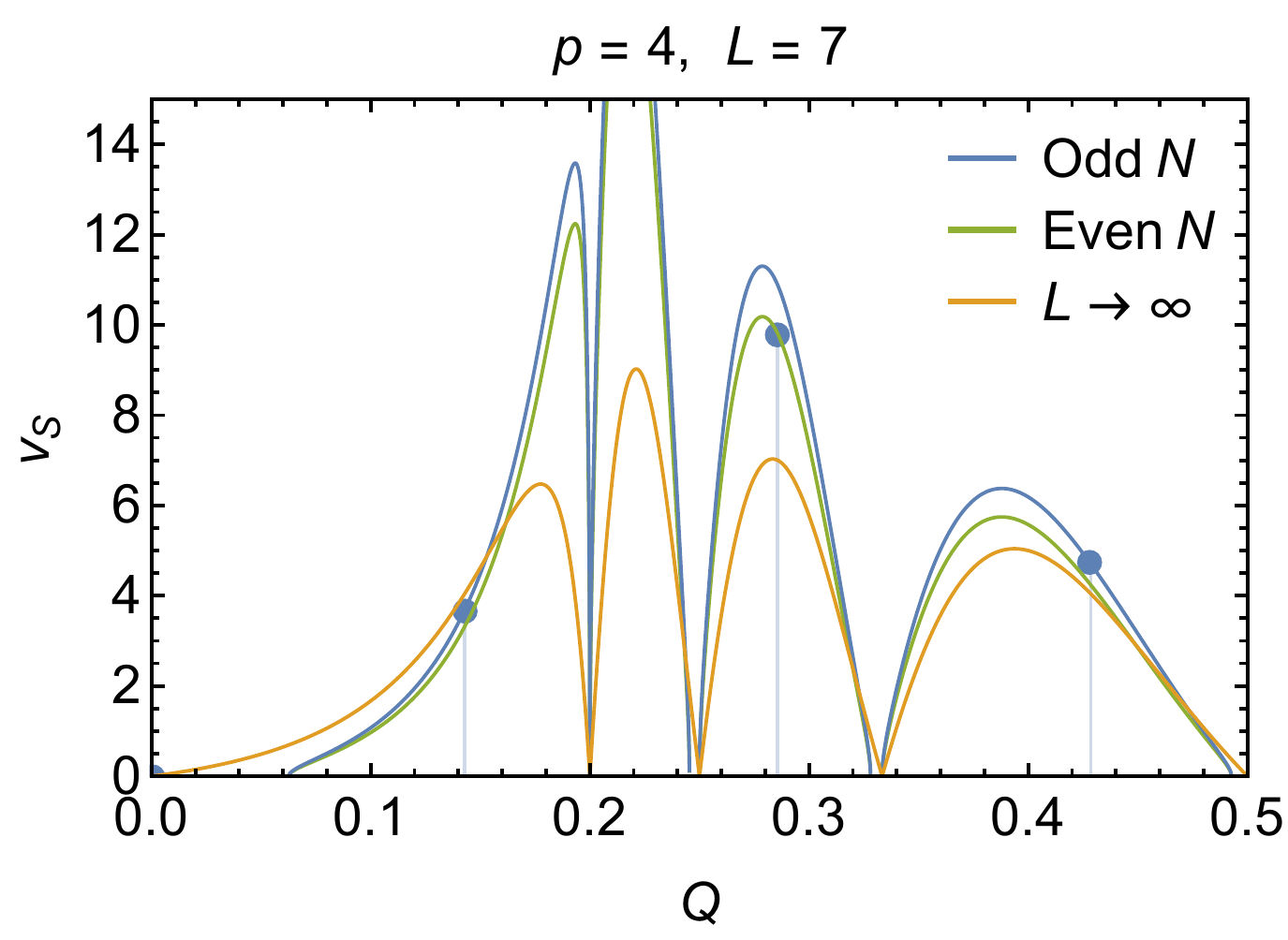}}
    &
    \resizebox{7cm}{!}{\includegraphics{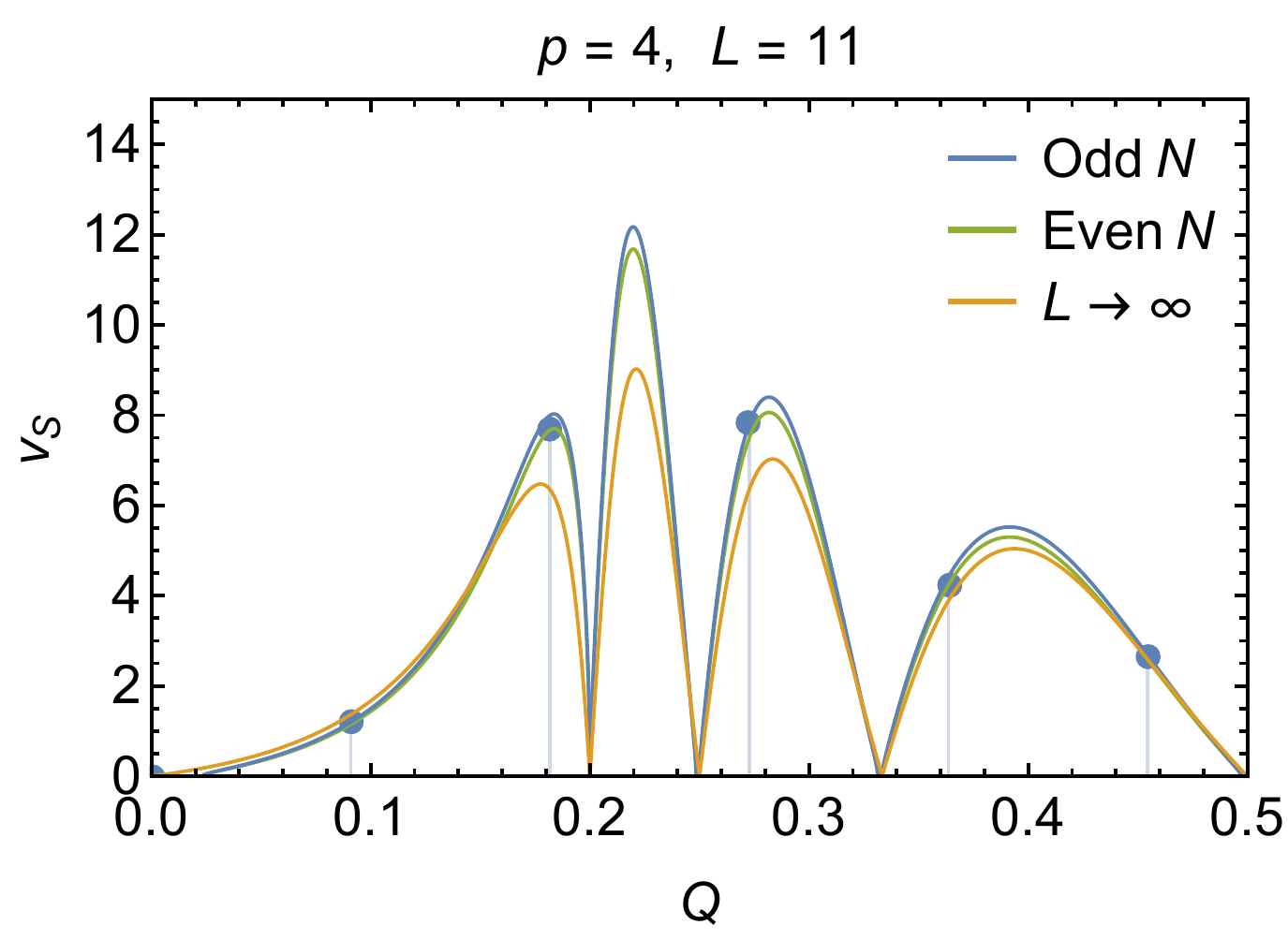}}\\
    \resizebox{7cm}{!}{\includegraphics{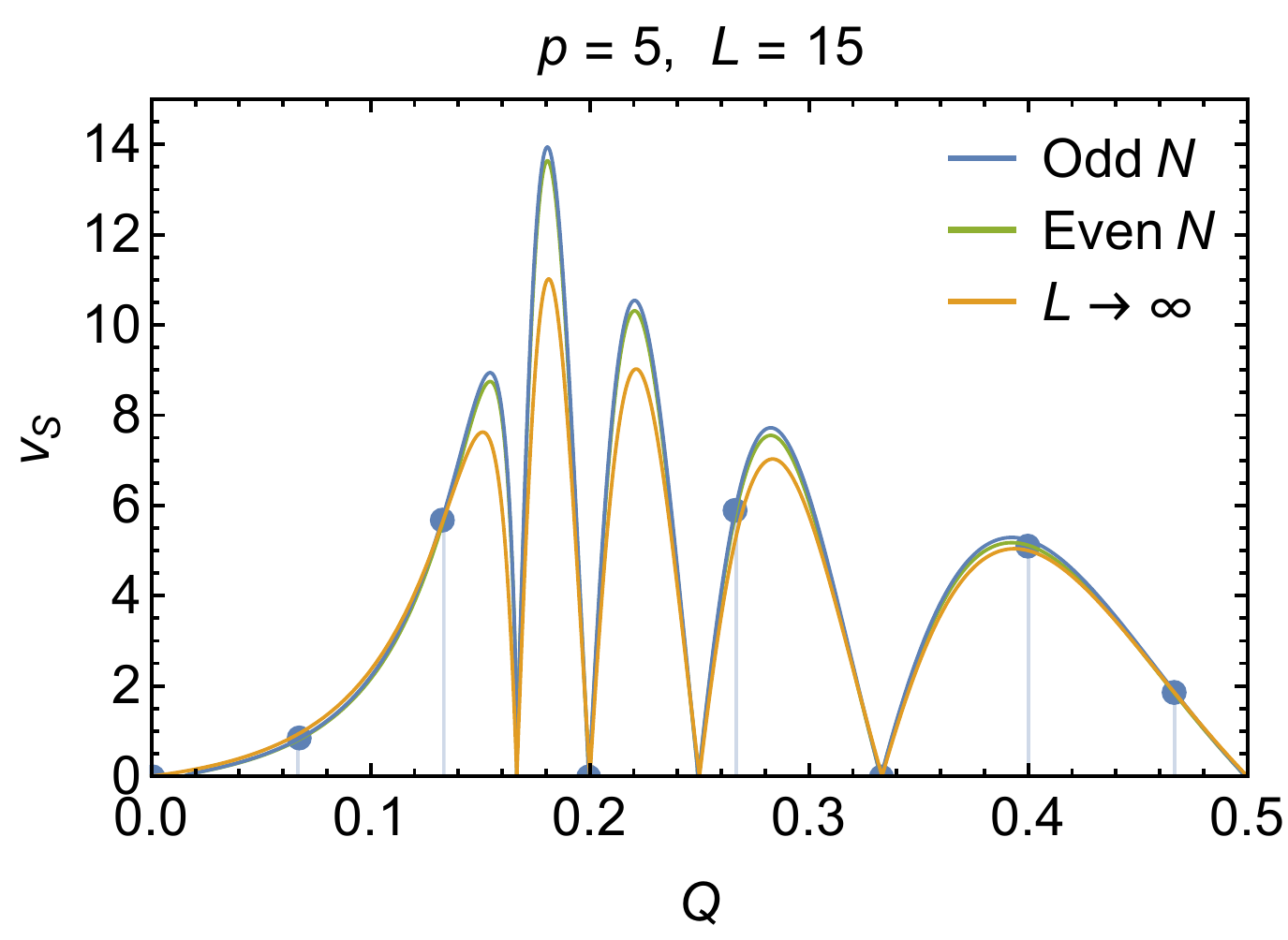}}
    &
    \resizebox{7cm}{!}{\includegraphics{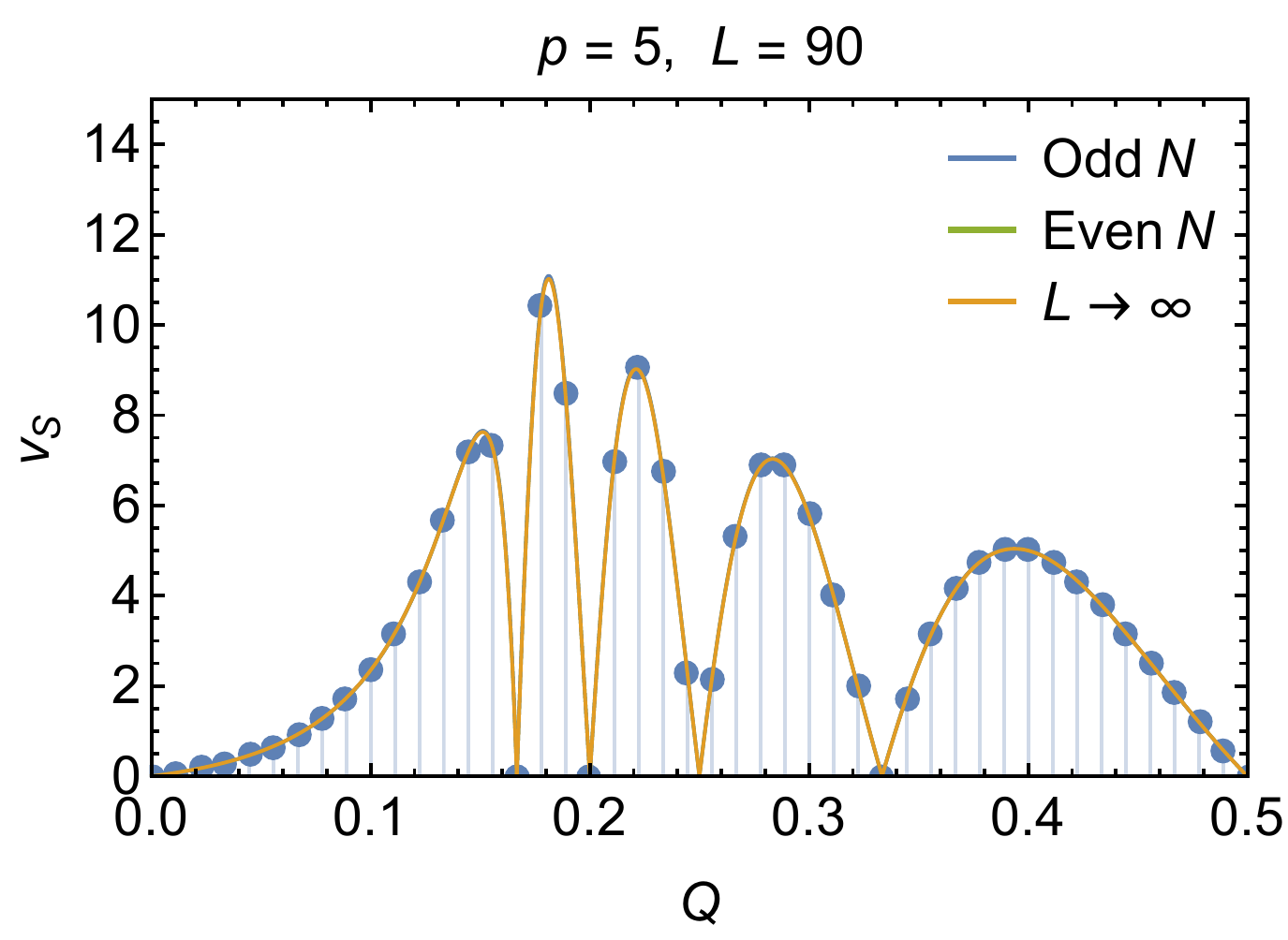}}
  \end{tabular}
  \caption{Sound velocity $v_S$ for different values of $p$ and $L$. Notice
  the difference between the infinite volume limit (yellow line) and finite
  $L$ values for odd (blue) and even (green) $N$.\label{fig:1-vSMine}}
\end{figure}

\subsection{General solution}

Similarly to Dias' solution of the $t$-$V$ model from Chapter
\ref{ch:DiasSolution}, we would like to develop a~solution for the extended
case $(p > 1)$. A full solution would include states that belong to the
high-energy subspace, and thus would give a more complete physical picture
than the results presented in Chapters \ref{ch:DiasBasedLow} and
\ref{ch:DiasBasedHigh}.

Firstly, one needs to know the projected Hamiltonian to first order in $t /
V$. Similarly to Eq.~(\ref{eq:pertHamiltonian}), the Hamiltonian of the
extended $t$-$V$ model is found to be:
\begin{equation}
  H = - t \sum_{i = 1}^L \sum_{\{ j_m \}} \left( \prod_{m = 1}^p
  \mathbbm{R}_{i - m}^{j_m} \right) c_i^{\dag} c_{i + 1} \left( \prod_{m =
  1}^p \mathbbm{R}_{i + 1 + m}^{j_m} \right) + \text{h.c.} + \sum_{i = 1}^L
  \sum_{m = 1}^p U_m n_i n_{i + m}, \label{eq:pertExtHam}
\end{equation}
where $\mathbbm{R}_i^j$ is an operator making sure the site $i$ is occupied
($j = 0$) or empty ($j = 1$), \tmtextit{i.e.}:
\begin{equation}
  \mathbbm{R}_i^j = j + (- 1)^j n_i = \left\{ \begin{array}{ll}
    n_i & \tmop{for} j = 0 \quad (\tmop{site} \tmop{must} \tmop{be}
    \tmop{occupied})\\
    1 - n_i & \tmop{for} j = 1 \quad (\tmop{site} \tmop{must} \tmop{be}
    \tmop{empty})
  \end{array} \right.,
\end{equation}
where $n_i$ is the particle operator. $\sum_{\{ j_m \}}$ is a sum over all
possible sets of $j_m$, \tmtextit{e.g.} $\{ j_1 = 0, j_2 = 1, j_3 = 1, \ldots
\}$ is one possible set. As an example, for $p = 3$ one of the elements in the
Hamiltonian is
\begin{equation}
  - t \sum_{i = 1}^L (1 - \underbrace{\nobracket n_{i - 3}) (1 -
  \underbrace{\nobracket n_{i - 2}) n_i \underbrace{_{- 1} c_i^{\dag} c_{i +
  1} n_i}\hphantom{}_{+ 2} (1 - \nobracket} n_{i + 3}) (1 - \nobracket} n_{i + 
  4}) .
\end{equation}
So, in short, if we want to hop the particle across a pair of sites, it is
only possible if the preceding $p$ sites are exactly a mirror image of the $p$
sites following the pair.

Now, we would like to create a mapping similar to Eq.~(\ref{eq:DiasMapping}),
where we renamed two\mbox{-}site chains. To do this, let us analyse the
requirements that need to be fulfilled in order to have an unambiguous
mapping. Let us look at Table \ref{tab:p1mapping} (adapted from Table
\ref{tab:p1mappingearly} for the reader's convenience), which shows two
possible hoppings that are allowed by the projected Hamiltonian from
Eq.~(\ref{eq:pertHamiltonian}), and four chains in which hoppings should not
be allowed.

\begin{table}[h]
  \begin{tabular}{cccc}
    \hline \hline
    Hopping & Decomposition into two-site chains & Mapping & Final mapping\\
    \hline
    $(\circ \circ \bullet \circ)$ & ${\color[HTML]{0070C0}(\circ \circ)}
    {\color[HTML]{76923C}(\circ \bullet)} (\bullet \circ)$ &
    $\mathbbm{A}\mathbbm{C}$ & $\downarrow \nospace \odot$\\
    $(\circ \bullet \circ \circ)$ & ${\color[HTML]{76923C}(\circ \bullet)}
    (\bullet \circ) {\color[HTML]{0070C0}(\circ \circ)}$ &
    $\mathbbm{C}\mathbbm{A}$ & $\odot \nospace \downarrow$\\
    \hline
    $(\bullet \circ \bullet \bullet)$ & $(\bullet \circ)
    {\color[HTML]{76923C}(\circ \bullet)} {\color[HTML]{C0504D}(\bullet
    \bullet)}$ & $\mathbbm{C}\mathbbm{B}$ & $\odot \nospace \uparrow$\\
    $(\bullet \bullet \circ \bullet)$ & ${\color[HTML]{C0504D}(\bullet
    \bullet)} (\bullet \circ) {\color[HTML]{76923C}(\circ \bullet)}$ &
    $\mathbbm{B}\mathbbm{C}$ & $\uparrow \nospace \odot$\\
    \hline
    \multicolumn{3}{c}{No hoppings should be allowed }  & \\
    \hline
    $(\circ \circ \bullet \bullet)$ & ${\color[HTML]{0070C0}(\circ \circ)}
    (\circ \bullet) {\color[HTML]{C0504D}(\bullet \bullet)}$ & $\mathbbm{A}
    \cancel{\mathbbm{C}} \mathbbm{B}$ & $\downarrow \nospace \cancel{\odot}
    \uparrow$\\
    $(\circ \bullet \circ \bullet)$ & ${\color[HTML]{76923C}(\circ \bullet)}
    (\bullet \circ) {\color[HTML]{76923C}(\circ \bullet)}$ &
    $\mathbbm{C}\mathbbm{C}$ & $\odot \odot$\\
    $(\bullet \circ \bullet \circ)$ & $(\bullet \circ)
    {\color[HTML]{76923C}(\circ \bullet)} (\bullet \circ)$ & $\mathbbm{C}$ &
    $\odot$\\
    $(\bullet \bullet \circ \circ)$ & ${\color[HTML]{C0504D}(\bullet \bullet)}
    (\bullet \circ) {\color[HTML]{0070C0}(\circ \circ)}$ &
    $\mathbbm{B}\mathbbm{A}$ & $\uparrow \downarrow$\\
    \hline \hline
  \end{tabular}
  \caption{All possible hoppings from the second to the third site in a
  four-site chain in a $p = 1$ system. We show how to produce a unique
  mapping, equivalent to the mapping from
  Eq.~(\ref{eq:DiasMapping}).\label{tab:p1mapping}}
\end{table}

From the decomposition of both possible hoppings, we can immediately see that
${\color[HTML]{0070C0}(\circ \circ)}$ and ${\color[HTML]{C0504D}(\bullet
\bullet)}$ switch places with ${\color[HTML]{76923C}(\circ \bullet)} (\bullet
\circ)$ and $(\bullet \circ) {\color[HTML]{76923C}(\circ \bullet)}$
respectively, and therefore should be uniquely mapped in a new system:
${\color[HTML]{0070C0}(\circ \circ)} =\mathbbm{A},
{\color[HTML]{C0504D}(\bullet \bullet)} =\mathbbm{B}$. Now we arrive at a
choice: we can either map {\color[HTML]{76923C}$\left({\circ}{\bullet}\right)$}
or $(\bullet \circ)$ into the third mapped one-particle operator. Either
choice is equally valid, but for the sake of consistency with
Eq.~(\ref{eq:DiasMapping}), we map ${\color[HTML]{76923C}(\circ \bullet)}
=\mathbbm{C}$. By considering which hoppings are allowed, we can see that
$\mathbbm{A}$ can hop with $\mathbbm{C}$ only, and so can $\mathbbm{B}$.
However chain $(\circ \circ \bullet \bullet)$, in which a hopping should not
be allowed, is mapped into $\mathbbm{A}\mathbbm{C}\mathbbm{B}$ and to make the
mapping uniquely defined, we treat $\mathbbm{C}$ in this case as a ``domain
wall''. Natural naming for $\mathbbm{A}, \mathbbm{B}$ and $\mathbbm{C}$ are
$\downarrow, \uparrow$ and $\odot$ (empty site), and therefore we arrive
exactly at the mapping (\ref{eq:DiasMapping}).

The situation is much more complicated, however, if one considers $p > 1$. The
possible number of hoppings is $2^p$ due to Eq.~(\ref{eq:pertExtHam}) and thus
grows exponentially with $p$. Additionally, the number of chains in which we
should not be allowed to hop is $2^{2 p + 1} - 2^{p + 1}$. Let us first
consider $p = 2$ and see if one can systematically devise a mapping for a
general case of any $p$.

\begin{table}[h]
  \begin{tabular}{cccc}
    \hline \hline
    No. & Hopping & Decomposition into three-site chains & Example mapping\\
    \hline
    1 & $(\circ \circ \circ \bullet \circ \circ)$ &
    ${\color[HTML]{0070C0}(\circ \circ \circ)} (\circ \circ \bullet) (\circ
    \bullet \circ) (\bullet \circ \circ)$ & $\mathbbm{A}\mathbbm{C}$\\
    & $(\circ \circ \bullet \circ \circ \circ)$ & $(\circ \circ \bullet)
    (\circ \bullet \circ) (\bullet \circ \circ) {\color[HTML]{0070C0}(\circ
    \circ \circ)}$ & $\mathbbm{C}\mathbbm{A}$\\
    \hline
    2 & $(\circ \bullet \circ \bullet \bullet \circ)$ & $(\circ \bullet \circ)
    (\bullet \circ \bullet) (\circ \bullet \bullet) (\bullet \bullet \circ)$ &
    $\mathbbm{E}\mathbbm{D}$\\
    & $(\circ \bullet \bullet \circ \bullet \circ)$ & $(\circ \bullet
    \bullet) (\bullet \bullet \circ) (\bullet \circ \bullet) (\circ \bullet
    \circ)$ & $\mathbbm{D}\mathbbm{E}$\\
    \hline
    3 & $(\bullet \circ \circ \bullet \circ \bullet)$ & $(\bullet \circ \circ)
    (\circ \circ \bullet) (\circ \bullet \circ) (\bullet \circ \bullet)$ &
    $\mathbbm{C}\mathbbm{E}$\\
    & $(\bullet \circ \bullet \circ \circ \bullet)$ & $(\bullet \circ
    \bullet) (\circ \bullet \circ) (\bullet \circ \circ) (\circ \circ
    \bullet)$ & $\mathbbm{E}\mathbbm{C}$\\
    \hline
    4 & $(\bullet \bullet \circ \bullet \bullet \bullet)$ & $(\bullet \bullet
    \circ) (\bullet \circ \bullet) (\circ \bullet \bullet)
    {\color[HTML]{C0504D}(\bullet \bullet \bullet)}$ &
    $\mathbbm{D}\mathbbm{B}$\\
    & $(\bullet \bullet \bullet \circ \bullet \bullet)$ &
    ${\color[HTML]{C0504D}(\bullet \bullet \bullet)} (\bullet \bullet \circ)
    (\bullet \circ \bullet) (\circ \bullet \bullet)$ &
    $\mathbbm{B}\mathbbm{D}$\\
    \hline \hline
  \end{tabular}
  \caption{Possible hoppings for $p = 2$.\label{tab:p2mapping}}
\end{table}

The decomposition of possible hoppings into three-site chains (see
Table~\ref{tab:p2mapping}) shows that ${\color[HTML]{0070C0}(\circ \circ
\circ)}$ and ${\color[HTML]{C0504D}(\bullet \bullet \bullet)}$ must
necessarily be mapped, ${\color[HTML]{0070C0}(\circ \circ \circ)}
=\mathbbm{A}, {\color[HTML]{C0504D}(\bullet \bullet \bullet)} =\mathbbm{B}$.
However, no map from three-site chains to one-site operators exists, so that
the map defines a one-to-one correspondence between hops in the original chain
and the mapped chain. We could, in principle, additionally use four-site
chains (example mapping $(\circ \circ \bullet) =\mathbbm{C}, (\circ \bullet
\bullet) =\mathbbm{D}, (\circ \bullet \circ \bullet) \tmop{or} (\bullet \circ
\bullet \circ) =\mathbbm{E}$ is shown in Table~\ref{tab:p2mapping}), but this
does not ensure that the map is a bijection, or that it preserves the
inability to hop.

We conclude that if a mapping similar to (\ref{eq:DiasMapping}) exists for the
extended model, it is too complicated to define it for any $p$. The number of
possible mappings and rules that would need to be applied to have a one-to-one
correspondence increases exponentially. Therefore, in the next section, we
present a different approach that can be used to determine the critical
parameters of the model.

\section{Strong coupling expansion -- going beyond the first order
perturbation}

\subsection{Formulation of the strong coupling expansion}

In the article by Hamer {\cite{Hamer1979}}, he introduced a method to truncate
the basis according to how states were connected to the unperturbed initial
subspace. The method is to reorder the basis (usually this is the Fock basis
-- particle number basis), firstly writing the 
desired\footnote{\setstretch{1.66}Notice that we
can choose what states we are interested in, be it the ground state or any of
the excited states.} subspace of unperturbed states that we want to
approximate (0th step), then states connected to them via the Hamiltonian (1st
step), then states connected to the 1st step states (2nd step) and so on. It
is easy to see that this results in a tri-block-diagonal Hamiltonian. We
truncate the basis to the step of our choice, resulting in a smaller,
truncated Hamiltonian, which will describe the full system up to a specific
perturbation order. However, the truncated basis is still usually quite big,
thus we will use an altered version of this method, called the strong coupling
expansion (SCE), and commonly used in investigations of the Schwinger model
{\cite{Crewther1980,Hamer1997,Cichy2013}}, a one-dimensional analogue of
quantum electrodynamics. This method is also equivalent to the block Lanczos
method {\cite{Montgomery1995}}.

We start by writing the Hamiltonian of the system as an unperturbed
Hamiltonian $\hat{H}_0$ and a~perturbation $\hat{V}$, with $\lambda$ being a
small parameter:
\begin{equation}
  \hat{H} = \hat{H}_0 + \lambda \hat{V} .
\end{equation}
For example, for the generalised $t$-$V$ model given by the Hamiltonian
(\ref{eq:genHamiltonian}), we identify:
\begin{equation}
  \hat{H}_0 = \sum_{i = 1}^L \sum_{m = 1}^p U_m n_i n_{i + m}, \quad \hat{V} =
  - \sum_{i = 1}^L \left( c^{\dag}_i c_{i + 1} + \text{h.c.} \right), \quad
  \lambda = t.
\end{equation}
{\hspace{1.5em}}The method proceeds as follows. Firstly, let us select the
desired initial subspace of unperturbed states that we want to approximate.
Usually that will be the ground state, but if one is interested in the
temperature dependance, it could be first excited states, second excited
states, etc. We will designate states in this subspace by $| 0^i \rangle, i =
1, 2, \ldots, N_{\tmop{init}}$, where $N_{\tmop{init}}$ is the number of
initial unperturbed states. Number 0 designates the 0th step of the expansion.

Secondly, to create states of the next step in the SCE, we will act with
perturbation operator $\hat{V}$ on the states from the previous step, $\hat{V}
| n^i \rangle$. $\hat{V} | n^i \rangle$ will be, in general, a~linear
combination of states from orders $n \um 1$, $n$ and $n \upl 1$. It will not
include lower orders, because $\hat{V} | n^i \rangle$ by definition does not
include orders higher than $n \upl 1$, which means $\forall_{m > n \upl 1}
\langle m^j | \hat{V} | n^i \rangle = 0$ and $\forall_{n < m \um 1} \langle
n^i | \hat{V} | m^j \rangle = 0$. This shows that the Hamiltonian in such a
basis is tri-block-diagonal, as it was in the original Hamer method. To
properly define states in the order $n \upl 1$, we have to separate states in
$\hat{V} | n^i \rangle$ according to their unperturbed energy; the states must
be eigenstates of $H_0$. Thus, the new states are:
\begin{equation}
  \hat{V} | n^i \rangle = \sum_j C_j | n \um 1^j \rangle + \sum_k C_k | n^k
  \rangle + \sum_l | \widetilde{n \upl 1}^l \rangle \label{separation},
\end{equation}
where $C_j, C_k$ are normalisation constants. The new states $| \widetilde{n
\upl 1}^l \rangle$ are not yet orthonormal to each other and to the previous
states (which is why we use a tilde to distinguish them from states that are
included in the new basis). After Gram-Schmidt orthonormalisation
{\cite{Arfken2005}} they become:
\begin{equation}
  | n^j \rangle = C_{\tilde{n}, j} | \tilde{n}^j \rangle - \sum_{m = n - 2}^{n
  - 1} \sum_{k = 1}^{k_{\max} (m)} C_{\tilde{n}, j ; m, k} | m^k \rangle -
  \sum_{k = 1}^{j - 1} C_{\tilde{n}, j ; n, k} | n^k \rangle
  \label{orthonormalization},
\end{equation}
where the coefficient $C_{\tilde{n}, j}$ is a normalisation constant and other
coefficients include normalisation and projection: $C_{\tilde{n}, j ; m, k} =
C_{\tilde{n}, j} \langle m^k | \tilde{n}^j \rangle$.

If we continue this procedure infinitely long, we may not however produce the
full basis spanning the whole Hilbert space. Thus, there may be states not
producible by this procedure, which we will call $| \alpha \rangle$, and which
will form together with states $| n^i \rangle$ an orthonormal non-truncated
basis of the system. However, we can easily see that using  (\ref{separation})
and then  (\ref{orthonormalization}) :
\begin{eqnarray}
  \langle \alpha | \hat{V} | n^i \rangle & = & \langle \alpha | \left( \sum_j
  C_j | n \um 1^j \rangle + \sum_k C_k | n^k \rangle + \sum_l | \widetilde{n
  \upl 1}^l \rangle \right) \\
  & = & \langle \alpha | \sum_j C_j | n \um 1^j \rangle + \langle \alpha |
  \sum_k C_k | n^k \rangle + \langle \alpha | \sum_l \frac{1}{C_{\widetilde{n
  \upl 1}, l}} \nonumber\\
  &  & \times \left( | n \upl 1^l \rangle + \sum_{r = n - 1}^n \sum_{k =
  1}^{k_{\max} (r)} C_{\widetilde{n + 1}, l ; r, k} | r^k \rangle + \sum_{k =
  1}^{l - 1} C_{\widetilde{n \upl 1}, l ; n \upl 1, k} | n \upl 1^k \rangle
  \right) \nonumber\\
  & = & \sum_j C_j \langle \alpha | n \um 1^j \rangle + \sum_k C_k \langle
  \alpha | n^k \rangle + \sum_l \frac{1}{C_{\widetilde{n \upl 1}, l}}
  \nonumber\\
  &  & \times \left( \langle \alpha | n \upl 1^l \rangle + \sum_{r = 1}^n
  \sum_{k = 1}^{k_{\max} (m)} C_{\widetilde{n \upl 1}, l ; r, k} \langle
  \alpha | r^k \rangle + \sum_{k = 1}^{l - 1} C_{\widetilde{n \upl 1}, l ; n
  \upl 1, k} \langle \alpha | n \upl 1^k \rangle \right) \nonumber\\
  & = & 0. \nonumber
\end{eqnarray}
This proves that states $| \alpha \rangle$ are in fact part of a completely
different subspace of the Hamiltonian than states $| n^i \rangle$. Therefore,
eigenvalues of the desired subspace that we will be approximating will not
depend on $| \alpha \rangle$ and neither will any averages over states from
the desired subspace.

The Hamiltonian is now in the tri-block-diagonal form:
\begin{equation}
  \settowidth\mylen{300000000}
  \hat{H} = \left(\begin{array}{CCCCC|c}
    \vphantom{\begin{array}{c}a \\ b\end{array}}\cellcolor{wordblue} \hat{E}_0 + \lambda \hat{V}_{00} & \cellcolor{wordblue} \lambda \hat{V}_{01} & 0 & 0 & \cdots &
    \\
    \vphantom{\begin{array}{c}a \\ b\end{array}}\cellcolor{wordblue} \lambda \hat{V}_{01}^{\dag} & \cellcolor{wordblue} \hat{E}_1 + \lambda \hat{V}_{11} & \cellcolor{wordblue} \lambda
    \hat{V}_{12} & 0 & \cdots & \\
    \vphantom{\begin{array}{c}a \\ b\end{array}}0 & \cellcolor{wordblue} \lambda \hat{V}_{12}^{\dag} & \cellcolor{wordblue} \hat{E}_2 + \lambda \hat{V}_{22} &
    \cellcolor{wordblue} \lambda \hat{V}_{23} & \ddots & 0\\
    \vphantom{\begin{array}{c}a \\ b\end{array}}0 & 0 & \cellcolor{wordblue} \lambda \hat{V}_{23}^{\dag} & \cellcolor{wordblue} \hat{E}_3 + \lambda \hat{V}_{33} &
    \cellcolor{wordblue}\ddots & \\
    \vphantom{\begin{array}{c}a \\ b\end{array}}\vdots & \vdots & \ddots & \cellcolor{wordblue}\ddots & \cellcolor{wordblue}\ddots & \\
    \hline
    &  & 0 &  &  & \begin{array}{c}
      \text{Hamiltonian}\\
      \text{elements between}\\
      \text{states } \{ | \alpha \rangle \}
    \end{array}
  \end{array}\right) \label{H}
\end{equation}
where $\hat{V}_{n, m}$ are projections of $\hat{V}$ between states $| n^i
\rangle$ and $| m^j \rangle$ and $\hat{E}_n$ are projections of $\hat{H}_0$
between states $| n^i \rangle$, \tmtextit{i.e.}
\begin{equation}
  \hat{V}_{n, m} = \sum_i \sum_j | n^i \rangle \langle n^i | \hat{V} | m^j
  \rangle \langle m^j |, \qquad \hat{E}_n = \sum_i | n^i \rangle \langle n^i |
  \hat{H}_0 | n^i \rangle \langle n^i | .
\end{equation}
{\hspace{1.5em}}Finally, we use a Hamiltonian truncated to a specific SCE step
to calculate the energy and behaviour of the desired subspace of states.
Notice that the method itself has no limitation on the size of the original
basis of the system, so one can consider systems with finite or infinite basis
sizes (the Schwinger model falls into the latter case, for example).

\

\subsection{Direct relation to the perturbation theory}

We can now use the standard degenerate perturbation theory (see
\tmtextit{e.g.} Ref. {\cite{Zettili2009}}) to show which Hamiltonian elements
contribute to the $m$-th order correction of the desired subspace. For a small
perturbation $\lambda$ the projected Hamiltonian can be written as:
\begin{equation}
  \hat{H} = \hat{H}_0 + \lambda \sum_n \mathbbm{P}_n \hat{V} \mathbbm{P}_n +
  \lambda^2 \sum_n \sum_{k \neq n} \frac{\mathbbm{P}_n \hat{V} \mathbbm{P}_k
  \hat{V} \mathbbm{P}_n}{E_n - E_k} + \cdots,
\end{equation}
where $\mathbbm{P}_n$ is a projection operator into the subspace unperturbed
states with energy $E_n$. In general, the $m$-th order correction in $\lambda$
will include matrices of the form:
\begin{equation}
  \mathbbm{P}_n \hat{V} \mathbbm{P}_{k_1} \hat{V} \mathbbm{P}_{k_2} \cdots
  \hat{V} \mathbbm{P}_{k_{m - 1}} \hat{V} \mathbbm{P}_n .
\end{equation}
In the Hamiltonian (\ref{H}), we can identify the following:
\begin{equation}
  \mathbbm{P}_n \hat{V} \mathbbm{P}_k = \left\{ \begin{array}{ll}
    \hat{V}_{n, n} & \tmop{if} k = n\\
    \hat{V}_{n, n + 1} & \tmop{if} k = n + 1\\
    \hat{V}_{n - 1, n}^{\dag} & \tmop{if} k = n - 1\\
    0 & \tmop{otherwise}
  \end{array} \right.,
\end{equation}
and therefore:
\begin{equation}
  \mathbbm{P}_n \hat{V} \mathbbm{P}_{k_1} \hat{V} \mathbbm{P}_{k_2} \cdots
  \hat{V} \mathbbm{P}_{k_{m - 1}} \hat{V} \mathbbm{P}_n = \hat{V}_{n, k_1}
  \hat{V}_{k_1, k_2} \cdots \hat{V}_{k_{m - 1}, n},
\end{equation}
since $\hat{V}_{k, n} = \hat{V}^{\dag}_{n, k}$. To calculate the corrections
of order $m$ to the desired subspace of states, we set $n = 0$. The term
including the highest number of $\hat{V}_{n, k}$ matrices is
\begin{equation}
  \hat{V}_{01} \hat{V}_{12} \cdots \hat{V}_{p, p + 1} \hat{V}_{p + 1, p}
  \cdots \hat{V}_{21} \hat{V}_{10}, \qquad p = \left\lceil \frac{m}{2}
  \right\rceil - 1,
\end{equation}
and we can immediately conclude that for the perturbation correction of order
$m$ in the desired subspace of states, we need the following matrices:
\begin{equation}
  \underbrace{\hat{E}_0 + \lambda \hat{V}_{00} \nocomma, \quad \lambda
  \hat{V}_{01}, \quad \hat{E}_1 + \lambda \hat{V}_{11}, \quad \ldots, \quad
  \hat{E}_p + \lambda \hat{V}_{pp}, \quad (\lambda \hat{V}_{p, p + 1})}_{m
  \tmop{matrices}} .
\end{equation}
Last matrix is only included if $m$ is even.

Therefore, in every step of Hamer's procedure, by including more states in the
Hamiltonian matrix, we increase the accuracy of the desired subspace of states
by two perturbation orders. More strictly, results obtained in step $k$ of the
SCE results will be accurate to order $2 k + 1$.

Additionally, the method goes beyond the usual perturbation theory: in SCE
step $k$, we not only include all the terms up to order $2 k + 1$, but also
terms such as $\hat{V}_{01} \hat{V}_{10} \hat{V}_{01} \hat{V}_{10} \cdots$,
which belong to higher orders of perturbation.

\section{Strong coupling expansion on the extended $t$\mbox{-}$V$~model near
critical dens{\nobreak}ities}\label{ch:SCESolution}

\subsection{Selected systems}

Although the strong coupling expansion has many advantages, such as
insensitivity to integrability (or the lack of it), it has also one prominent
drawback: one needs to know the unperturbed subspace of states that will be
approximated. However, in the generalised $t$-$V$ model, we know that near the
Mott insulating densities, the degeneracy of the ground state is very small
{\cite{Gomez-Santos1993}} and the Hamiltonian can be diagonalised
analytically. Thus, we will firstly select three critical densities that will
be studied: $Q = 1 / (p + 1)$, $p = 1$ (integrable), $p = 2$ (nonintegrable),
and $p = 3$ (nonintegrable). Then, we will try to generalise our results to
any $p$. Secondly, we will select near-critical densities, in order to assess
critical parameters of the model. Our results will be compared to known
values.

Due to finite computer resources, we choose to use SCE step 3 for the critical
densities and SCE step 1 for near-critical densities. While for critical
densities our results are analytical, the near-critical densities require
numerical estimates.

\subsection{Results at critical densities $Q = 1 / (p +
1)$}\label{ch:1-SCEcritical}

Interestingly, at the Mott insulating density, the size of truncated
Hamiltonian was found to be constant for a given $p$ and $Q$ and completely
independent of the system size $L$. The system size is incorporated into the
Hamiltonian. Additionally, due to the translational symmetry of the system,
the truncated Hamiltonian contains $(p + 1)$ equal subspaces, if the system
size is bigger than $L > (2 \times \tmop{SCE} \tmop{step}) (p + 1)$.

\subsubsection{$Q = 1 / 2$ (half-filling), $p = 1$ (integrable), SCE step 3}

The truncated Hamiltonian for this case is of dimension $16 \times 16$, but for
a very large system size \tmtextit{L }it can be separated into two equal
subspaces of dimension $8\times8$, which can easily be diagonalised,
\begin{equation}
  \resizebox{\textwidth}{!}{$
  H = \left(\begin{array}{cccccccc}
    \cdot & \sqrt{L} & \cdot & \cdot & \cdot & \cdot & \cdot & \cdot\\
    \sqrt{L} & U_1 & 2 & \sqrt{2 L - 10} & \cdot & \cdot & \cdot & \cdot\\
    \cdot & 2 & U_1 & \cdot & 2 & \sqrt{L - 6} & \cdot & \cdot\\
    \cdot & \sqrt{2 L - 10} & \cdot & 2 U_1 & \sqrt{\frac{2}{L - 5}} & (L - 7)
    \sqrt{\frac{8}{(L - 6) (L - 5)}} & \sqrt{\frac{8 (L - 7)}{(L - 6) (L -
    5)}} & \sqrt{\frac{3 (L - 7) (L - 8)}{L - 5}}\\
    \cdot & \cdot & 2 & \sqrt{\frac{2}{L - 5}} & U_1 & \cdot & \cdot & \cdot\\
    \cdot & \cdot & \sqrt{L - 6} & (L - 7) \sqrt{\frac{8}{(L - 6) (L - 5)}} &
    \cdot & 2 U_1 & \cdot & \cdot\\
    \cdot & \cdot & \cdot & \sqrt{\frac{8 (L - 7)}{(L - 6) (L - 5)}} & \cdot &
    \cdot & 2 U_1 & \cdot\\
    \cdot & \cdot & \cdot & \sqrt{\frac{3 (L - 7) (L - 8)}{L - 5}} & \cdot &
    \cdot & \cdot & 3 U_1
  \end{array}\right),
  $}
\end{equation}
where for the sake of clarity zeros are represented as dots and every
off-diagonal element should be multiplied by $(- t)$. The ground
state\footnote{\setstretch{1.66}For the formulas of ground states and truncated 
Hamiltonians
see Appendix \ref{ch:appendixsceham}--\ref{ch:appendixscecorr}.} is therefore
2-fold degenerate and the ground state energy was calculated assessed every
step:
\begin{equation}
  \begin{array}{llllllllll}
    \text{1st step} & \quad & \cellcolor{wordblue}- & \cellcolor{wordblue}\frac{L t^2}{U_1} & + & \frac{L^2
    t^4}{U_1^3} & - & \frac{2 L^3 t^6}{U_1^5} & + & \frac{5 L^4 t^8}{U^7_1},\\
    \text{2nd step} & \quad & \cellcolor{wordblue}- & \cellcolor{wordblue}\frac{L t^2}{U_1} & \cellcolor{wordgreen}+ & \cellcolor{wordgreen}\frac{L
    t^4}{U^3_1} & + & \frac{(- 2 L + L^2 + L^3) t^6}{2 U^5_1} & + & \frac{(L -
    3 L^2 - 19 L^3 - 3 L^4) t^8}{4 U^7_1},\\
    \text{3rd step} & \quad & \cellcolor{wordblue}- & \cellcolor{wordblue}\frac{L t^2}{U_1} & \cellcolor{wordgreen}+ & \cellcolor{wordgreen}\frac{L
    t^4}{U_1^3} & \cellcolor{wordyellow}+ & \cellcolor{wordyellow}0 & + & \frac{(- 30 L + 29 L^2 + 6 L^3 + 6 L^4) t^8}{6
    U^7_1} .
  \end{array}
\end{equation}

We can clearly see that indeed with every step we increase the accuracy of our
result by two orders in $t / U_1$. The ground state energy is thus:
\begin{equation}
  E_0 = - \frac{L t^2}{U_1} + \frac{L t^4}{U_1^3} + O (t^8) .
  \label{eq:1-gsep1}
\end{equation}
The density-density correlation functions $N_m = \left\langle \sum_{i = 1}^L
n_i n_{i + m} \right\rangle$ were found to be:
\begin{eqnarray}
  N_1 & = & L \frac{t^2}{U_1^2} - 3 L \frac{t^4}{U_1^4} + O (t^8), \\
  N_2 & = & \frac{L}{2} - 2 L \frac{t^2}{U_1^2} + 7 L \frac{t^4}{U_1^4} + O
  (t^6), \\
  N_3 & = & 2 L \frac{t^2}{U_1^2} - 5 L \frac{t^4}{U_1^4} + O (t^6), \\
  N_4 & = & \frac{L}{2} - 2 L \frac{t^2}{U_1^2} + 2 L \frac{t^4}{U_1^4} + O
  (t^6), \\
  N_5 & = & 2 L \frac{t^2}{U_1^2} - 2 L \frac{t^4}{U_1^4} + O (t^6) . 
\end{eqnarray}
This particular case of the generalized $t$-$V$ model can be mapped to the
Heisenberg XXZ spin model with background magnetic field, which is solved
analytically by Orbach {\cite{Orbach1958}} and Walker {\cite{Walker1959}}. On
closer inspection we can see that the analytical expressions for the ground
state energy and the density-density correlator $N_1$ (in the language of
spins it is the spin-spin correlator) presented in Ref.~{\cite{Walker1959}}
match our results.

Furthermore, the XXZ model for $t / U_1 \rightarrow 0$ is equivalent to the
Ising model {\cite{Takahashi2005}} for which the long-range density-density
correlators are:
\begin{equation}
  N_m = \left\{ \begin{array}{ll}
    0 & \tmop{for} m \tmop{odd}\\
    \frac{L}{2} & \tmop{for} m \tmop{even}
  \end{array} \right. \label{isingdensity},
\end{equation}
which is fully consistent with our results.

The current density is given by
\begin{equation}
  J = - it \left\langle \sum_{i = 1}^L c^{\dag}_i c_{i + 1} - \text{h.c.}
  \right\rangle,
\end{equation}
and was found to be zero up to order $O (t^8)$ for large systems. The $p = 1$
model was inspected thoroughly in the first-order approximation in Refs.
{\cite{Gomez-Santos1993,Dias2000}} and the ground state energy and current
density were found to vanish for the case of half filling; this also agrees
with our results.

\subsubsection{$Q = 1 / 3, p = 2$ (non-integrable), SCE step 3}

For $p > 1$ the model is non-integrable. In step 3 (7th order of
perturbation), the Hamiltonian is of dimension $36 \times 36$; however it can
be divided into three equivalent subspaces of dimension $12 \times 12$. The
ground state is therefore 3-fold degenerate and its energy was found to be:
\begin{equation}
  E_0 = - \frac{2 L}{3 U_2} t^2 + \left( \frac{2 L}{3 U_2^3} - \frac{2 L}{U_1
  U_2^2} \right) t^4 + \left( \frac{16 L}{3 U_1 U_2^4} - \frac{17 L}{3 U_1^2
  U_2^3} - \frac{10 L}{3 U_1^3 U_2^2} \right) t^6 + O (t^8) .
  \label{eq:1-gsep2}
\end{equation}
The density-density correlators are:
\begin{eqnarray}
  N_1 & = & \frac{2 L}{U_1^2 U_2^2} t^4 + \left( \frac{10 L}{U_1^4 U_2^2} +
  \frac{34 L}{3 U_1^3 U_2^3} - \frac{16 L}{3 U_1^2 U_2^4} \right) t^6 + O
  (t^8), \\
  N_2 & = & \frac{2 L}{3 U_2^2} t^2 + \left( \frac{4 L}{U_1 U_2^3} - \frac{2
  L}{U_2^4} \right) t^4 + \left( \frac{20 L}{3 U_1^3 U_2^3} + \frac{17
  L}{U_1^2 U_2^4} - \frac{64 L}{3 U_1 U_2^5} \right) t^6 + O (t^8), \\
  N_3 & = & \frac{L}{3} - \frac{4 L}{3 U_2^2} t^2 + \left( - \frac{16 L}{3
  U_1^2 U_2^2} - \frac{8 L}{U_1 U_2^3} + \frac{13 L}{3 U_2^4} \right) t^4 + O
  (t^6), \\
  N_4 & = & \frac{2 L}{3 U_2^2} t^2 + \left( \frac{10 L}{3 U_1^2 U_2^2} +
  \frac{4 L}{U_1 U_2^3} - \frac{7 L}{3 U_2^4} \right) t^4 + O (t^6), \\
  N_5 & = & \frac{2 L}{3 U_2^2} t^2 + \left( \frac{10 L}{3 U_1^2 U_2^2} +
  \frac{4 L}{U_1 U_2^3} - \frac{L}{3 U_2^4} \right) t^4 + O (t^6) . 
\end{eqnarray}
As in the case of Eq.~(\ref{isingdensity}), we expect the density-density
correlation functions to take the form
\begin{equation}
  N_m = \left\{ \begin{array}{ll}
    \frac{L}{p + 1} & \tmop{for} m \tmop{divisable} \tmop{by} (p + 1)\\
    0 & \tmop{otherwise}
  \end{array} \right. \label{correlationsgeneral},
\end{equation}
in the limit as $t / U_m \rightarrow 0$ when $Q = 1 / (p + 1)$. This is indeed
true for our results.

Again, the current density is zero up to order $O (t^8)$ for large systems.

\subsubsection{$Q = 1 / 4, p = 3$ (non-integrable), SCE step 3}

This is another non-integrable case. The Hamiltonian is of dimension $52
\times 52$, but it consists of four equal subspaces of dimension $13 \times
13$. The ground state is thus 4-fold degenerate and has energy:
\begin{eqnarray}
  E_0 & = & - \frac{L}{2 U_3} t^2 + \left( \frac{L}{2 U_3^3} - \frac{3 L}{2
  U_2 U_3^2} \right) t^4 \\
  &  & + \left( \frac{4 L}{U_2 U_3^4} - \frac{17 L}{4 U_2^2 U_3^3} - \frac{5
  L}{2 U_2^3 U_3^2} - \frac{5 L}{U_1 U_2^2 U_3^2} \right) t^6 + O (t^8) .
  \label{eq:1-gsep3} \nonumber
\end{eqnarray}
The density-density correlation functions are:
\begin{eqnarray}
  N_1 & = & \frac{5 L}{U_1^2 U_2^2 U_3^2} t^6 + O (t^8), \\
  N_2 & = & \frac{3 L}{2 U_2^2 U_3^2} t^4 + L \left( \frac{15}{2 U_2^4 U_3^2}
  + \frac{17}{2 U_2^3 U_3^3} - \frac{4}{U_2^2 U_3^4} + \frac{10}{U_1 U_2^3
  U_3^2} \right) t^6 + O (t^8), \\
  N_3 & = & \frac{L}{2 U_3^2} t^2 - L \left( \frac{3}{2 U_3^4} - \frac{3}{U_2
  U_3^3} \right) t^4 \\
  & & + L \left( \frac{5}{U_2^3 U_3^3} + \frac{51}{4 U_2^2
    U_3^4} - \frac{16}{U_2 U_3^5} + \frac{10}{U_1 U_2^2 U_3^3} \right) t^6 + 
    O(t^8),\nonumber\\
  N_4 & = & \frac{L}{4} - \frac{L}{U_3^2} t^2 + L \left( \frac{13}{4 U_3^4} -
  \frac{4}{U_2^2 U_3^2} - \frac{6}{U_2 U_3^3} \right) t^4 + O (t^6), \\
  N_5 & = & \frac{L}{2 U_3^2} t^2 + L \left( \frac{2}{U_2^2 U_3^2} +
  \frac{3}{U_2 U_3^3} - \frac{2}{U_3^4} \right) t^4 + O (t^6) . 
\end{eqnarray}
Again, our results for correlators are consistent with equation
(\ref{correlationsgeneral}).

For a large system size the current density was calculated to be zero up to
perturbation order $O (t^8)$.

\subsubsection{Collecting the results for $Q = 1 / (p + 1)$, any $p$}

Interestingly, if one exchanges the system size $L$ for corresponding values
of $N$ and $p$ using $L = N (p + 1)$, formulas in Eqs. (\ref{eq:1-gsep1}),
(\ref{eq:1-gsep2}) and (\ref{eq:1-gsep3}) for the ground state energies become
more systematic and we can easily write the following equation that
incorporates them all
\begin{eqnarray}
  E_0 & = & - \frac{2 N}{U_p} t^2 + \left( \frac{2 N}{U_p^3} - \frac{6 N}{U_{p
  - 1} U_p^2} \right) t^4 \\
  &  & + \left( \frac{16 N}{U_{p - 1} U_p^4} - \frac{17 N}{U_{p - 1}^2 U_p^3}
  - \frac{10 N}{U_{p - 1}^3 U_p^2} - \frac{20 N}{U_{p - 2} U_{p - 1}^2 U_p^2}
  \right) t^6 + O (t^8), \nonumber
\end{eqnarray}
where for $p = 1, 2$ we have to remove terms in which $U_{p - 1}$ or $U_{p -
2}$ is zero. This energy comes directly from the fact that the truncated
Hamiltonian does not change for $p \geqslant 3$ and it is always in the
following form\footnote{\setstretch{1.66}Again, the dots are zeros and all 
off-diagonal
elements should be multiplied by $(- t)$.}:\\
\\
$H=$\\
\\
\resizebox{\textwidth}{!}{
$\left( \begin{array}{ccccccccccccc}
  \cdot & \sqrt{2 N} & \cdot & \cdot & \cdot & \cdot & \cdot & \cdot & \cdot &
  \cdot & \cdot & \cdot & \cdot\\
  \sqrt{2 N} & U_p & \sqrt{3} & 2 & \sqrt{4 N - 10} & \cdot & \cdot & \cdot &
  \cdot & \cdot & \cdot & \cdot & \cdot\\
  \cdot & \sqrt{3} & U_{p - 1} & \cdot & \cdot & \sqrt{\frac{10}{3}} &
  \sqrt{\frac{1}{3}} & \sqrt{2 N - \frac{17}{3}} & \sqrt{\frac{5}{3}} & \cdot
  & \cdot & \cdot & \cdot\\
  \cdot & 2 & \cdot & U_p & \cdot & \cdot & 1 & \cdot & \sqrt{\frac{9}{5}} & 2
  & \sqrt{2 N - 7} & \sqrt{\frac{1}{5}} & \cdot\\
  \cdot & \sqrt{4 N - 10} & \cdot & \cdot & 2 U_p & \cdot & \cdot &
  \sqrt{\frac{2 (6 N - 17)}{2 N - 5}} & \sqrt{\frac{8}{5 (2 N - 5)}} &
  \sqrt{\frac{2}{2 N - 5}} & \sqrt{\frac{8 (2 N - 7)}{2 N - 5}} & -
  \sqrt{\frac{2}{5 (2 N - 5)}} & \sqrt{\frac{6 (N - 4) (2 N - 7)}{2 N - 5}}\\
  \cdot & \cdot & \sqrt{\frac{10}{3}} & \cdot & \cdot & U_{p - 2} & \cdot &
  \cdot & \cdot & \cdot & \cdot & \cdot & \cdot\\
  \cdot & \cdot & \sqrt{\frac{1}{3}} & 1 & \cdot & \cdot & 2 U_p & \cdot &
  \cdot & \cdot & \cdot & \cdot & \cdot\\
  \cdot & \cdot & \sqrt{2 N - \frac{17}{3} } & \cdot & \sqrt{\frac{2 (6 N -
  17)}{2 N - 5}} & \cdot & \cdot & U_{p - 1} + U_p & \cdot & \cdot & \cdot &
  \cdot & \cdot\\
  \cdot & \cdot & \sqrt{\frac{5}{3}} & \sqrt{\frac{9}{5}} & \sqrt{\frac{8}{5
  (2 N - 5)}} & \cdot & \cdot & \cdot & U_{p - 1} & \cdot & \cdot & \cdot &
  \cdot\\
  \cdot & \cdot & \cdot & 2 & \sqrt{\frac{2}{2 N - 5}} & \cdot & \cdot & \cdot
  & \cdot & U_p & \cdot & \cdot & \cdot\\
  \cdot & \cdot & \cdot & \sqrt{2 N - 7} & \sqrt{\frac{8 (2 N - 7)}{2 N - 5}}
  & \cdot & \cdot & \cdot & \cdot & \cdot & 2 U_p & \cdot & \cdot\\
  \cdot & \cdot & \cdot & \sqrt{\frac{1}{5}} & - \sqrt{\frac{2}{5 (2 N - 5)}}
  & \cdot & \cdot & \cdot & \cdot & \cdot & \cdot & U_{p - 1} & \cdot\\
  \cdot & \cdot & \cdot & \cdot & \sqrt{\frac{6 (N - 4) (2 N - 7)}{2 N - 5}} &
  \cdot & \cdot & \cdot & \cdot & \cdot & \cdot & \cdot & 3 U_p
\end{array} \right)$.}
\begin{equation}
  \ 
\end{equation}
We can see that the energy is proportional to the system size, and therefore
it is an~extensive quantity.

Lastly, the correlator $N_{p + 1}$ is found to be:
\begin{equation}
  N_{p + 1} = N - 4 N \frac{t^2}{U_p^2} + \left( \frac{13 N}{U_p^4} - \frac{16
  N}{U_{p - 1}^2 U_p^2} - \frac{24 N}{U_{p - 1} U_p^3} \right) t^4 + O (t^6),
\end{equation}
which agrees well with the prediction from Eq.~(\ref{correlationsgeneral}). We
also expect the other correlators to be zero up to the first perturbation
order. The behaviour of the correlators in the zeroth order reflects the
arrangement of fermions in the unperturbed ground state: fermions will be
spread evenly throughout the system, $p \text{+} 1$ sites away from each
other.

\subsection{Near-critical densities}

At the critical insulating density all Luttinger liquid velocities
{\cite{Gomez-Santos1993}} go to zero. However, we may evaluate near-critical
behaviour by choosing densities not exactly equal to the insulating one. To
determine the charge velocity $v_N$ from Eq.~(\ref{eq:1-vN}), we need a change
in the number $N$ of particles in the system. Let us assume that the system is
at the critical density $Q = 1 / (p + 1)$. By adding one particle to the
system, we increase the number of domain walls (which will now behave like
free particles) by more than one. For example, in a $p = 3$ system,
\begin{equation}
  (\bullet \circ \circ \bullet \circ \circ \bullet \circ \circ \bullet \circ
  \circ \cdots) \xrightarrow[N + 1]{} (\bullet \bullet \circ \bullet \circ
  \circ \bullet \circ \circ \bullet \circ \circ \cdots)
  \xrightarrow[\tmop{equilibrate}]{} (\underset{\smile}{\bullet \circ} 
  \underset{\smile}{\bullet \circ}  \underset{\smile}{\bullet \circ} \bullet \circ \circ
  \bullet \circ \circ \cdots),
\end{equation}
by adding one particle, we create three domain walls. This greatly increases
the degeneracy of the unperturbed ground state. Instead, we propose to change
the density of the system by changing the system size. By adding none, one,
and two empty sites to the insulating system, we can induce the following
densities:
\begin{equation}
  Q_{0 \circ} = \frac{L}{L (p + 1)}, \qquad Q_{1 \circ} = \frac{L - 1}{L (p +
  1)}, \qquad Q_{2 \circ} = \frac{L - 2}{L (p + 1)} .
\end{equation}
Notice that $Q_{2 \circ} < Q_{1 \circ} < Q_{0 \circ}$. The change in the
density is always
\begin{equation}
  \Delta Q = - \frac{1}{L (p + 1)} .
\end{equation}
One and two empty sites will behave like free fermions in the system and the
unperturbed ground state will be $L$-fold degenerate for one additional hole
and $L (L - 2 p - 1) / 2$-fold degenerate for two 
holes\footnote{\setstretch{1.66}The two holes
can be either next to each other, or will be separated by a fermion:
\begin{equation}
  \bullet \underbrace{\circ \cdots \circ \circ \circ}_{p + 2} \bullet
  \underbrace{\circ \cdots \circ}_p \bullet \underbrace{\circ \cdots \circ}_p
  \cdots \qquad \tmop{or} \qquad \bullet \underbrace{\circ \cdots \circ
  \circ}_{p + 1} \bullet \underbrace{\circ \cdots \circ}_p \bullet
  \underbrace{\circ \cdots \circ}_p \cdots \bullet \underbrace{\circ \cdots
  \circ \circ}_{p + 1} \bullet \underbrace{\circ \cdots \circ}_p \bullet
  \underbrace{\circ \cdots \circ}_p \cdots .
\end{equation}
Simple combinatorics can be used to calculate that the first case can be
realised in $L$ ways, and the second case can be done in $L (L - 2 p - 3) / 2$
ways.}. We can now rewrite Eq.~(\ref{eq:1-vN}) for the charge velocity:
\begin{equation}
  v_N = \frac{L}{\pi}  \frac{\partial^2 E}{\partial N^2} = \frac{1}{\pi L} 
  \frac{\partial^2 E}{\partial Q^2},
\end{equation}
or, in a discrete form:
\begin{equation}
  v_N = \frac{1}{\pi L}  \frac{E_{0 \circ} + E_{2 \circ} - 2 E_{1
  \circ}}{(\Delta Q)^2} = \frac{L (p + 1)^2}{\pi}  (E_{0 \circ} + E_{2 \circ}
  - 2 E_{1 \circ}), \label{1-vNdisc}
\end{equation}
where $E_{0 \circ}, E_{1 \circ}$, and $E_{2 \circ}$ are ground state energies
at densities $Q_{0 \circ}, Q_{1 \circ}$, and $Q_{2 \circ}$ respectively.

\subsubsection{Numerical calculation}

Due to higher degeneracy of the unperturbed ground state ($\sim L$ or $\sim
L^2$), we may not be able to obtain analytical results as easily as in the
calculation at the critical density in Chapter \ref{ch:1-SCEcritical}.
Sometimes, numerical estimates will need to be used, especially for higher
orders of perturbation.

For the first-order perturbation correction, $E^{(1)}$, we can use the results
from Eqs.~(\ref{eq:1-DiasManyEnergy1}) and (\ref{eq:1-DiasManyEnergy2})
obtained using the altered Dias' method. To extract the second- and
third-order corrections, we will use strong coupling expansion numerically,
obtain the ground state energies $E^{(\tmop{num})}$ for finite values of $t /
U_p \equiv \lambda$ and use the following formulas:
\begin{equation}
  E^{(2)} = \lim_{\lambda \rightarrow 0} E^{(2)} (\lambda) = \lim_{\lambda
  \rightarrow 0} \frac{E^{(\tmop{num})} - E^{(1)} \lambda}{\lambda^2},
\end{equation}
\begin{equation}
  E^{(3)} = \lim_{\lambda \rightarrow 0} E^{(3)} (\lambda) = \lim_{\lambda
  \rightarrow 0} \frac{E^{(\tmop{num})} - E^{(1)} \lambda - E^{(2)}
  \lambda^2}{\lambda^3} .
\end{equation}
One can therefore plot $E^{(2)} (\lambda)$ and $E^{(3)} (\lambda)$ as a
function of $\lambda$ and then extract the $\lambda = 0$ value by fitting a
polynomial.

\subsubsection{System with one additional empty site and its flux dependence}

In order to calculate the current velocity $v_J$ from Eq.~(\ref{eq:1-vJ}), we
need to assess the dependence of the energy on a small external flux. One can
introduce the flux in the Hamiltonian, by changing the kinetic energy to
include a phase factor $e^{i \phi}$:
\begin{equation}
  - t \sum_{i = 1}^L (c^{\dag}_i c_{i + 1} + c^{\dag}_{i + 1} c_i) \rightarrow
  - t \sum_{i = 1}^L (e^{i \phi} c^{\dag}_i c_{i + 1} + e^{- i \phi}
  c^{\dag}_{i + 1} c_i) .
\end{equation}
This is a small magnetic flux piercing our chain, with values $\phi = \left[ -
\frac{\pi}{L}, \frac{\pi}{L} \right]$.

The second order correction to the energy, $E^{(2)}_{1 \circ}$, was determined
numerically. Its dependencies on $N$ and $p$ are shown in
Figs.~\ref{fig:e1corr2p1} and \ref{fig:e1corr2p2} respectively. For odd $N$,
this dependence is only present in the additive term $- 2 N$. For even $N$,
there is an additional dependence that can be incorporated in the flux: $\phi
\rightarrow \pi / L - | \phi |$.

\begin{figure}[h]
  \resizebox{13cm}{!}{\includegraphics{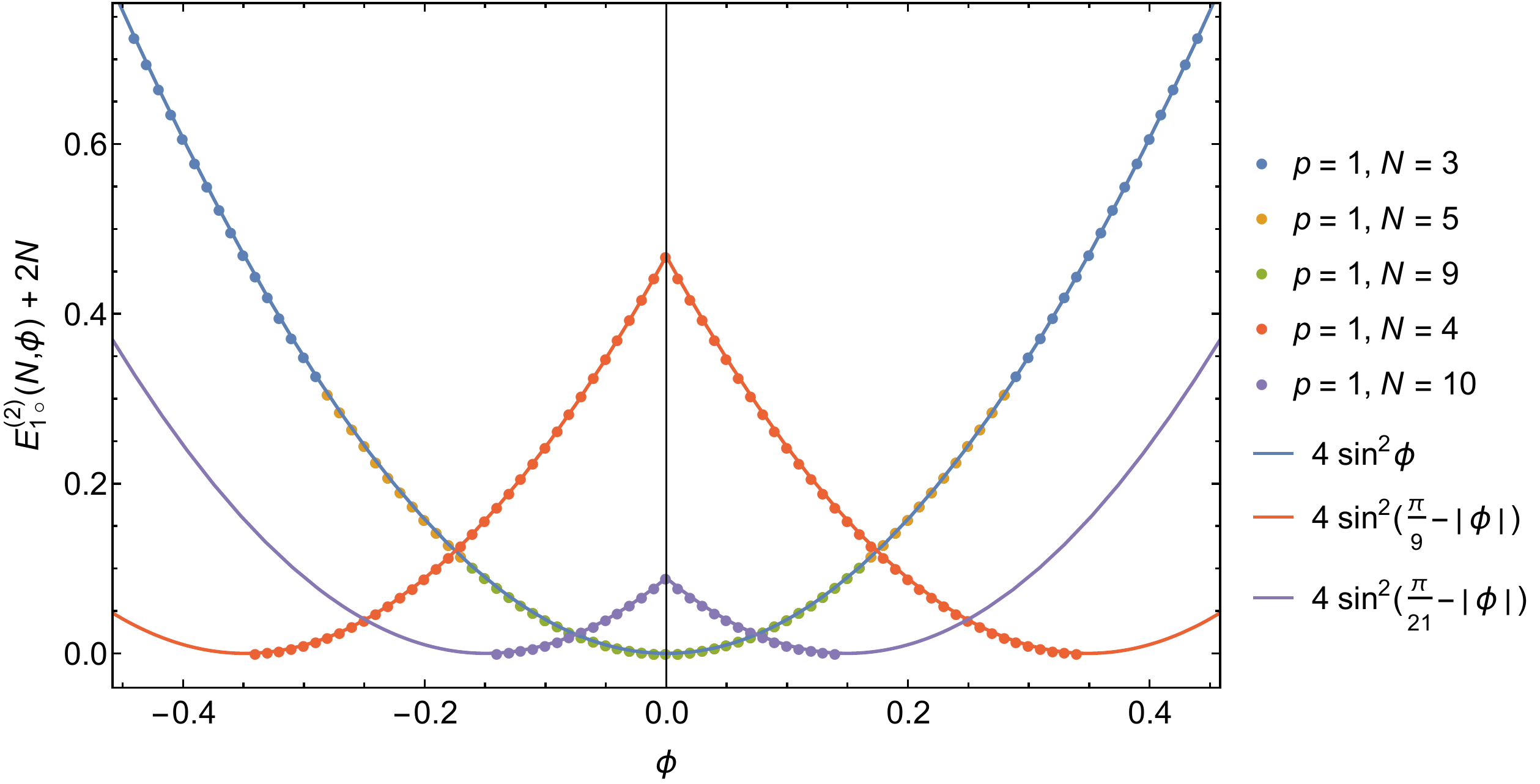}}
  \caption{Second-order correction to the energy, $E^{(2)}_{1 \circ}$, for
  different values of the number of particles $N$.\label{fig:e1corr2p1}}
\end{figure}

\begin{figure}[h]
  \resizebox{13cm}{!}{\includegraphics{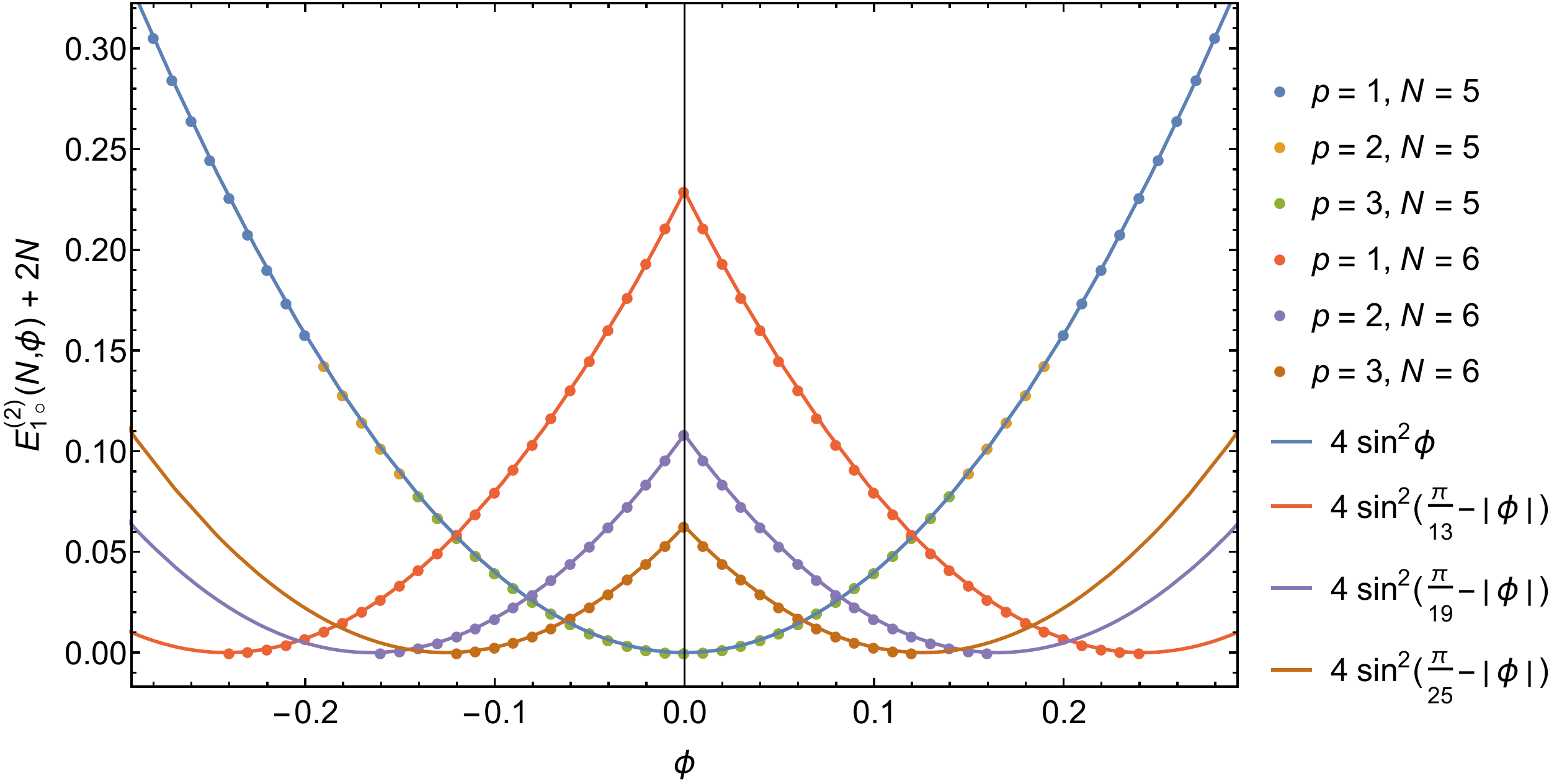}}
  \caption{Second-order correction to the energy, $E^{(2)}_{1 \circ}$, for
  different values of the interaction range $p$.\label{fig:e1corr2p2}}
\end{figure}

The final expression for the correction can be determined to be:
\begin{equation}
  E_{1 \circ}^{(2)} = \left\{ \begin{array}{ll}
    - 2 (N - 2 \sin^2 \phi) & \tmop{for} \tmop{odd} N\\
    - 2 \left( N - 2 \sin^2 \left( \frac{\pi}{L} - | \phi | \right) \right) &
    \tmop{for} \tmop{even} N
  \end{array} \right. .
\end{equation}
This expression is correct within the error bars ($\sim$8 digits of accuracy).

The third-order correction has also been evaluated numerically -- see
Fig.~\ref{fig:e1corr3}. The correction was found to be independent of
interaction range $p$ for a given value of $N$, within the error bars.

\begin{figure}[h]
  \resizebox{12.5cm}{!}{\includegraphics{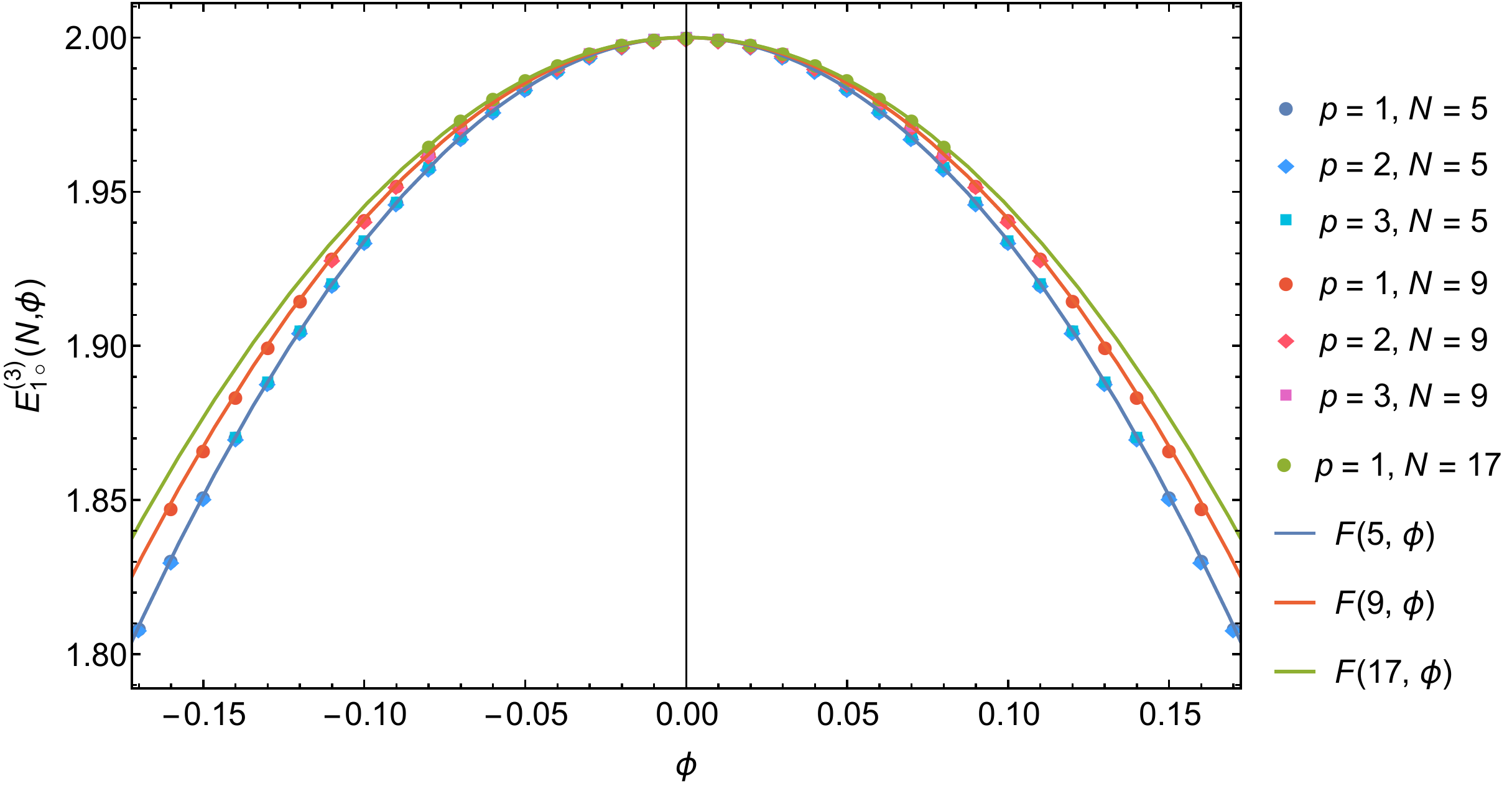}}
  \caption{Third-order correction to the energy, $E^{(3)}_{1 \circ}$, for odd
  $N$.\label{fig:e1corr3}}
\end{figure}

The approximate expression for the correction for odd $N$ can be written
as\footnote{\setstretch{1.66}Since corrections for even $N$ have an additional 
dependence on
$p$, these results have been presented in Appendix~\ref{ch:appendixsceeven}.}:
\begin{equation}
  E^{(3)}_{1 \circ} (N, \phi) \approx 2 - \left( 5 + \frac{8}{N} \right)
  \phi^2 + O (\phi^4) \equiv F (N, \phi) \label{eq:1-corr1} .
\end{equation}
This expression gives less than 0.02\% relative error near the $\phi
\rightarrow 0$ region.

\subsubsection{System with two additional empty sites}

We have also numerically calculated corrections for the system with two
additional empty sites. Both second- and third-order corrections, $E^{(2)}_{2
\circ}$ and $E^{(3)}_{2 \circ}$, are shown in Figure \ref{fig:e2corr}.

\begin{figure}[h]
  \resizebox{13cm}{!}{\includegraphics{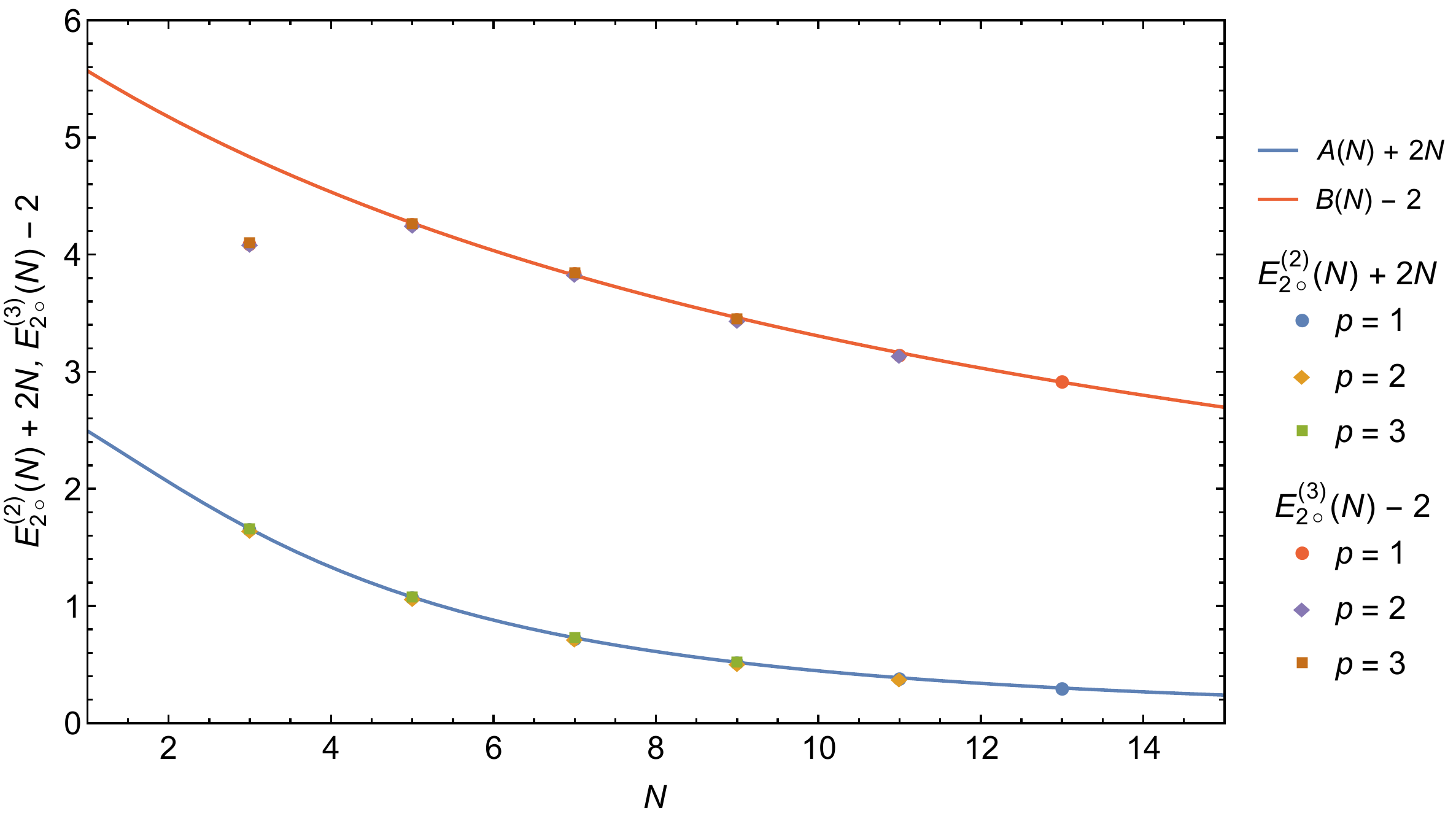}}
  \caption{Second- and third-order energy correction, $E_{2 \circ}^{(2)}$ and
  $E_{2 \circ}^{(3)}$ as a function of the number of particles
  $N$.\label{fig:e2corr}}
\end{figure}

For odd numbers of particles $N$, both second- and third-order corrections
were found to be independent of the range $p$ of the interactions (within the
fitting errors which were on the 8th and 4th digit respectively) and there is
only the dependence on $N$. Using Pad{\'e} approximants, we have devised the
following fitting formulas:
\begin{equation}
  E^{(2)}_{2 \circ} \approx - 2 N + \frac{0.651 N + 56.72}{N^2 + 2.19 N +
  19.81} \equiv A (N), \qquad E^{(3)}_{2 \circ} \approx 2 + \frac{73.17}{N +
  12.14} \equiv B (N) \label{eq:1-corr2},
\end{equation}
which give a relative error of less than 0.003\% for $E^{(2)}_{2 \circ}$, and
less than 0.3\% for $E^{(3)}_{2 \circ}$. Notice that the third-order
correction for $N = 3$ does not follow the trend, which is most likely an
artifact of a~too small system size.

\subsubsection{Calculating the critical parameter $K$}

One can collect the energies for near-critical densities for odd $N$:
\begin{eqnarray}
  E_{0 \circ} & = & - 2 N \frac{t^2}{U_p} + O (t^4), \\
  E_{1 \circ} & = & - 2 (\cos \phi) t - 2 (N - 2 \sin^2 \phi)  \frac{t^2}{U_p}
  + F (N, \phi)  \frac{t^3}{U_p^2} + O (t^4), \\
  E_{2 \circ} & = & - 4 \left( \cos \frac{\pi}{N + 2} \right) t + A (N) 
  \frac{t^2}{U_p} + B (N)  \frac{t^3}{U_p^2} + O (t^4) . 
\end{eqnarray}
with functions given by Eqs.~(\ref{eq:1-corr1}) and (\ref{eq:1-corr2}). Then,
the sound velocity $v_S$ and the critical parameter $K$ of the Luttinger
liquid can be calculated using Eqs.~(\ref{eq:1-vS}) and (\ref{eq:1-K}):
\begin{eqnarray}
  v_S & = & 4 (p + 1) t \sin \frac{\pi (p + 1)}{2 (L + 2 p)} \\
  &  & + \frac{t^2}{4 U_p} \left( (p + 1) A (N) + 2 (L - 2) + 16 (p + 1) 
  \left( 1 - \cos \frac{\pi (p + 1)}{L + 2 p} \right) \right) \sin^{- 1} 
  \frac{\pi (p + 1)}{2 (L + 2 p)} \nonumber\\
  &  & + O (t^3), \nonumber
\end{eqnarray}
\begin{eqnarray}
  K & = & \frac{\pi}{4 L (p + 1)} \sin^{- 1}  \frac{\pi (p + 1)}{2 (L + 2 p)}
  \\
  &  & - \frac{\pi \left( (p + 1) A (N) + 2 (L - 2) - 16 (p + 1) \left( 1 -
  \cos \frac{\pi (p + 1)}{L + 2 p} \right) \right) \sin^{- 1}  \frac{\pi (p +
  1)}{2 (L + 2 p)}}{32 L (p + 1)^2  \left( 1 - \cos \frac{\pi (p + 1)}{L + 2
  p} \right)}  \frac{t}{U_p} \nonumber\\
  &  & + O (t^2) . \nonumber
\end{eqnarray}
The formulas were double checked for consistency with previously calculated
first-order perturbation results, Eqs.~(\ref{eq:1-MyCalcvS}) and
(\ref{eq:1-MyCalcK}).

\section{Conclusions}

In summary, we have shown that one can use Dias' mapping and generalise it for
the $t$\mbox{-}$V$ model with any interaction range. Using G{\'o}mez-Santos's
picture together with Dias' mapping, one can also successfully reach the high
energy subspace. However, the complete mapping that would include states with
high-energy domains has been proven to be too complicated to devise, as its
complexity increases exponentially with the maximum interaction range.

Therefore, we have used another method, strong coupling expansion, that allows
us to reach high perturbation orders, rather than only the first-order
perturbation that both G{\'o}mez-Santos' and Dias' methods are limited to.
Although the method is mainly used for numerical investigations, in the case
of the generalised $t$-$V$ model, we have used it analytically. Indeed one can
prove that with every expansion step, the accuracy of the results is increased
by two orders of perturbation. Since the SCE works best for states with low
degeneracy, we used it for systems with density near the Mott insulating
phase. This allowed us to extract the critical parameters of the Luttinger
liquid, which were consistent with previously calculated values.

\begin{table}[h]
  \begin{tabular}{llll}
    \hline\hline
    Method & Advantages & Disadvantages & Chapter\\
    \hline
    G{\'o}mez-Santos & Simple picture & Only infinite volumes &
    \ref{ch:GSSolution}\\
    & Rich physics & First perturbation order & \\
    & High energy subspace & No high-energy domains & \\
    \hline
    Dias & Includes high-energy domains & Complicated solution &
    \ref{ch:DiasSolution}\\
    & Non-Bethe ansatz & Only for $p = 1$ & \\
    & More physical meaning extracted & First perturbation order & \\
    \hline
    Adapted Dias & Included $p > 1$ & No high-energy domains &
    \ref{ch:DiasLongSolution}\\
    for $p > 1$ & Results for finite $L$ & First perturbation order & \\
    & High energy subspace &  & \\
    \hline
    SCE & Higher perturbation orders & Only for low degeneracy &
    \ref{ch:SCESolution}\\
    & Goes beyond the perturbation theory & Only one setup at a time & \\
    \hline\hline
  \end{tabular}
  \caption{Comparison of methods presented in Chapters \ref{ch:1-intro} and
  \ref{ch:1-criticality}.\label{tab:1-methods}}
\end{table}

All methods that we have used in our investigations are presented in
Table~\ref{tab:1-methods}. Although G{\'o}mez-Santos' picture presents very
rich physics using simple tools, it is very hard to generalise for use in
other fermionic models. On the other hand, a more complicated method by Dias,
presents how one can map a fermion system into another fermion system with a
known solution. A~complete mapping, involving states with high-energy domains,
may be too complicated to use, but one can with little effort calculate the
low-energy behaviour of the system. This method could be in~principle adapted
for other models involving fermions, like in Chapter
\ref{ch:DiasLongSolution}, where it was adapted to include all interaction
ranges.

Finally, the strong coupling expansion is a very versatile method, that can be
used both numerically and analytically. Although analytical analysis may be
only possible for a very low degeneracy of the unperturbed target state, the
method enables one to reach high orders of perturbation relatively easily.
Additionally, the results will include terms going beyond the perturbation
theory, which means the method can be insensitive to divergences that arise
from including only finite number of perturbation orders. SCE has been already
shown to work for lattice field theories, spin systems, and fermion systems,
and can be used on models with both finite and infinite basis.

\chapter{Charge-density-wave{\hspace{0.4em}}phases
in~any~potential}\label{ch:1-cdw}

\section{Motivation}

In the picture presented in Chapter \ref{ch:GSSolution}, at Mott insulating
densities the extended $t$\mbox{-}$V$ model has a~very simple behaviour: if
the density is given by Eq.~(\ref{eq:1-MottDensities}) ($Q = 1 / m$, where $m
= p + 1, p, \ldots$), then the unperturbed $(t \rightarrow 0)$ ground state is
of the form
\begin{equation}
  \bullet \underbrace{\circ \circ \cdots \circ}_{\substack{1 / Q - 1\\
    \tmop{sites}}} \bullet \underbrace{\circ \circ \cdots \circ}_{1 / Q - 1}
  \bullet \underbrace{\circ \circ \cdots \circ}_{1 / Q - 1} \cdots,
\end{equation}
and the energy density is always $EQ / N = QU_{1 / Q}$, if $m > (p + 1) / 2$.
This is however only true, if the assumption from Eq.~(\ref{eq:GSAssumption})
holds.

By abandoning this assumption, Refs.~{\cite{Schmitteckert2004}} and
{\cite{Mishra2011}} have shown a non-trivial behaviour that cannot be
explained by the simple picture from Chapter \ref{ch:GSSolution}.
Investigation was done using models with $p = 2$ and densities $Q = 1 / 3$ and
$1 / 2$. Depending on the values of the potentials $\{ U_m \}$, one can have
different phases in the system: there can be multiple CDW insulating phases, a
long-range bond-order phase, and even a Luttinger liquid phase, despite the
existence of a critical (``insulating'') density.

In this Chapter, we would like to generalise this result for all interaction
ranges. However, to simplify the problem, we shall assume the atomic limit ($t
= 0$), in which the only phases that will be encountered are CDW insulators.

\section{Low critical densities in the atomic limit}\label{ch:1-LowCrit}

\subsection{Critical density $Q = 1 / (p + 1)$}

In the trivial case of $Q = 1 / (p + 1)$, the ground-state energy is always
equal to zero. The ground-state configuration is
\begin{equation}
  \bullet \underbrace{\circ \circ \cdots \circ}_p \bullet \underbrace{\circ
  \circ \cdots \circ}_p \bullet \underbrace{\circ \circ \cdots \circ}_p \cdots
  .
\end{equation}
Such a ground state is $(p + 1)$\mbox{-}fold degenerate since the energy is
invariant under translation.

\subsection{Critical density $Q = 1 / p$}

Let us now show how to construct the CDW phase for a system with any $p$ and
with critical density $Q = 1 / p$. Firstly, assume that $U_p$ is low enough
compared with the other $U_m$ to ensure that the preferable distance between
two fermions is always $p$ and thus we can say that $U_p$ orders the fermions
in the ground state; for example a chain $\bullet \underbrace{\circ \circ
\cdots \circ}_{p - 1} \bullet \underbrace{\circ \circ \cdots \circ}_{p - 1}
\bullet$ has lower energy than $\bullet \underbrace{\circ \circ \circ \cdots
\circ}_p \bullet \underbrace{\circ \cdots \circ}_{p - 2} \bullet$. The ground
state must have the simple form:
\begin{equation}
  \bullet \underbrace{\circ \circ \cdots \circ}_{p - 1} \bullet
  \underbrace{\circ \circ \cdots \circ}_{p - 1} \bullet \underbrace{\circ
  \circ \cdots \circ}_{p - 1} \cdots,
\end{equation}
and its energy is $E_1 = \frac{L}{p} U_p = N U_p$.

Now, let us assume that $U_{p - 1}$ is low enough to order the fermions. We
could use a~series of $\bullet \underbrace{\circ \circ \cdots \circ}_{p - 2}$
sections, but then we would not have the correct density $1 / p$. However, by
addition of sections $\bullet \underbrace{\circ \circ \cdots \circ}_p$ we can
tailor the density without changing the energy of the system. Thus, the
ground-state configuration is
\begin{equation}
  \begin{array}{|l|}
    \hline
    \bullet \underbrace{\circ \circ \cdots \circ}_{p - 2} \bullet
    \underbrace{\circ \circ \cdots \circ}_p\\
    \hline
  \end{array} \begin{array}{|l|}
    \hline
    \bullet \underbrace{\circ \circ \cdots \circ}_{p - 2} \bullet
    \underbrace{\circ \circ \cdots \circ}_p\\
    \hline
  \end{array} \begin{array}{|l|}
    \hline
    \bullet \underbrace{\circ \circ \cdots \circ}_{p - 2} \bullet
    \underbrace{\circ \circ \cdots \circ}_p\\
    \hline
  \end{array} \cdots,
\end{equation}
which gives us the correct density $Q = 1 / p$, and the energy is $E_2 = (L /
2 p) U_{p - 1} = (N / 2) U_{p - 1}$. The boxes are present to show that we
have correctly counted the energy and particle density. In general, however,
the whole subspace of the unperturbed ground states would include Fock states
in which sections with $(p - 2)$ holes could be beside each other, unless they
would change the energy of the system.

If one follows this prescription, in the $n$-th step the following ground
state is obtained:
\begin{equation}
  \begin{array}{|l|}
    \hline
    \bullet \underbrace{\circ \circ \cdots \circ}_{p - n} \overbrace{\bullet
    \underbrace{\circ \circ \cdots \circ}_p \bullet \underbrace{\circ \circ
    \cdots \circ}_p \bullet \underbrace{\circ \circ \cdots \circ}_p \cdots}^{n
    - 1 \tmop{times}}\\
    \hline
  \end{array} \begin{array}{|l|}
    \hline
    \bullet \underbrace{\circ \circ \cdots \circ}_{p - n} \overbrace{\bullet
    \underbrace{\circ \circ \cdots \circ}_p \bullet \underbrace{\circ \circ
    \cdots \circ}_p \bullet \underbrace{\circ \circ \cdots \circ}_p \cdots}^{n
    - 1 \tmop{times}}\\
    \hline
  \end{array} \cdots,
\end{equation}
with the energy
\begin{equation}
  E_n = \frac{L}{1 + p - n + (n - 1) (p + 1)} U_{p + 1 - n} = \frac{L}{n p}
  U_{p - n + 1} = \frac{N}{n} U_{p - n + 1} .
\end{equation}
We can now calculate the exact conditions in which an arbitrary phase
(designated by step $n$) will be dominant in the system:
\begin{equation}
  \underset{k \neq n}{\bigforall}\ E_n < E_k \quad \Rightarrow \quad
  \underset{k \neq n}{\bigforall}\ U_{p - n + 1} < \frac{n}{k} U_{p - k +
  1} .
\end{equation}
Renaming $\alpha = p - n + 1$ and $\beta = p - k + 1$,
\begin{equation}
  \underset{\beta \neq \alpha}{\bigforall}\ U_{\alpha} < \frac{p - \alpha +
  1}{p - \beta + 1} U_{\beta} .
\end{equation}
If this condition is fulfilled, then the phase with energy $E_{(\alpha)} =
\frac{N}{p - \alpha + 1} U_{\alpha}$ is dominant and the ground state consists
of $\frac{N}{p - \alpha + 1}$ blocks of $\bullet \underbrace{\circ \circ
\cdots \circ}_{\alpha - 1}$ and $N \frac{p - \alpha}{p - \alpha + 1}$ blocks
of $\bullet \underbrace{\circ \circ \cdots \circ}_p$ and is $f$-fold
degenerate, where:
\begin{equation}
  f = \left\{ \begin{array}{ll}
    \left( \begin{array}{c}
      N\\
      N / (p - \alpha + 1)
    \end{array} \right) \cdot p & \tmop{if} 2 \alpha > p\\
    \left( \begin{array}{c}
      N \frac{p - \alpha}{p - \alpha + 1}\\
      N / (p - \alpha + 1)
    \end{array} \right) \cdot \frac{p (p - \alpha + 1)}{p - \alpha} &
    \tmop{otherwise} .
  \end{array} \right.
\end{equation}
For $2 \alpha \leqslant p$, the problem with assessing the degeneracy is that
we need to exclude cases where blocks of $\bullet \underbrace{\circ \circ
\cdots \circ}_{\alpha - 1}$ are too close to each other and thus would
increase the energy by $U_{2 \alpha}$.

\section{Higher critical densities in the atomic limit}

For Mott insulators with $Q = 1 / m$, where $m = 1, \ldots, p - 1$, the number
of phases and their energies were found to be more difficult to obtain. Rather
than constructing the phases as done in Chapter \ref{ch:1-LowCrit}, we shall
use a brute-force analysis of the basis for systems of finite size.
Nevertheless, because we are interested in the thermodynamic limit, a periodic
system of $L$ sites can be thought of as an infinite system with a unit cell
of $L$ sites.\footnote{\setstretch{1.66}This may not be true if one does not 
work in the atomic
limit, since the kinetic term may introduce a flux in the bosonic
interpretation of the model.}

\subsection{Properties of the system}

Sampling the full basis in systems with $Q > 1 / p$ is problematic, because
the dimension of the basis grows rapidly with system size. However, many of
the Fock states will have the same energy. In particular, if two states are
cyclic permutations of each other, or cyclic permutations with inversion, then
such states must have the same energy due to the periodicity of the system.
Checking the full basis for cyclic permutations would still be computationally
quite a difficult task: firstly, because generating the full basis would take
a lot of memory, and secondly, because comparing all the states to check if
they are cyclic permutations would require a large computational time (of
complexity $O (2^{2 L})$). An alternative approach to this problem is to
consider the spaces between the fermions in our chain and to develop rules to
generate a set of Fock states that will always contain ground states of the
system.

\begin{theorem}
  For any basis state, the largest space between consecutive fermions must not
  be less than $1 / Q - 1$ sites.\label{1-CDWth1}
\end{theorem}

\begin{proof}
  All spaces between consecutive fermions are equal only if all particles are
  $1 / Q - 1$ sites apart, \tmtextit{i.e.}, the configuration is
  \begin{equation}
    \bullet \underbrace{\circ \circ \cdots \circ}_{1 / Q - 1} \bullet
    \underbrace{\circ \circ \cdots \circ}_{1 / Q - 1} \bullet
    \underbrace{\circ \circ \cdots \circ}_{1 / Q - 1} .
  \end{equation}
  Any attempt to move a fermion would make the largest space bigger than $1 /
  Q - 1$.
\end{proof}

Thus, any state will have a space that is larger than or equal to $1 / Q - 1$.
Due to the system's periodicity, we can therefore fix the first $1 / Q$ sites
to be
\begin{equation}
  \bullet \underbrace{\circ \circ \cdots \circ}_{1 / Q - 1} .
  \label{eq:1-CDWth1fix}
\end{equation}
This leaves us with a smaller subspace of the full basis to generate: the
system with size $(N - 1) / Q$ and $(N - 1)$ particles.

\begin{theorem}
  For any ground state of the system, the largest space must not exceed $p$
  sites.\label{1-CDWth2}
\end{theorem}

\begin{proof}
  Assume that there exists a ground state unit cell with the largest space
  equal to $(p + 1)$ sites. We can write it as
  \begin{equation}
    \underbrace{\bullet ? ? \cdots ? \bullet}_{\tmop{Block} A} 
    \underbrace{\circ \circ \cdots \circ}_{p \tmop{sites}} \circ .
    \label{eq:1-CDWth2unitcell}
  \end{equation}
  Let $E_A$ be the energy of the block $A$, so that the energy density of this
  ground state is $\frac{E_A}{N / Q}$. Let us construct the following unit
  cell, which consists of $p$ consecutive ground-state unit cells
  (\ref{eq:1-CDWth2unitcell}):
  \begin{equation}
    \underbrace{\begin{array}{|l|}
      \hline
      \underbrace{\bullet ? ? \cdots ? \bullet}_{\tmop{Block} A} 
      \underbrace{\circ \circ \cdots \circ}_p \circ\\
      \hline
    \end{array} \begin{array}{|l|}
      \hline
      \underbrace{\bullet ? ? \cdots ? \bullet}_{\tmop{Block} A} 
      \underbrace{\circ \circ \cdots \circ}_p \circ\\
      \hline
    \end{array} \cdots \begin{array}{|l|}
      \hline
      \underbrace{\bullet ? ? \cdots ? \bullet}_{\tmop{Block} A} 
      \underbrace{\circ \circ \cdots \circ}_p \circ\\
      \hline
    \end{array}}_p .
  \end{equation}
  This unit cell has the same energy density $\frac{pE_A}{pN / Q} =
  \frac{E_A}{N / Q}$ as the ground-state unit cell
  (\ref{eq:1-CDWth2unitcell}). Let us now move the additional empty spaces to
  the end of this chain, which still does not change the energy density:
  \begin{equation}
    \begin{array}{|l|}
      \hline
      \underbrace{\bullet ? ? \cdots ? \bullet}_{\tmop{Block} A} 
      \underbrace{\circ \circ \cdots \circ}_p\\
      \hline
    \end{array} \begin{array}{|l|}
      \hline
      \underbrace{\bullet ? ? \cdots ? \bullet}_{\tmop{Block} A} 
      \underbrace{\circ \circ \cdots \circ}_p\\
      \hline
    \end{array} \cdots \begin{array}{|l|}
      \hline
      \underbrace{\bullet ? ? \cdots ? \bullet}_{\tmop{Block} A} 
      \underbrace{\circ \circ \cdots \circ}_p  \underbrace{\circ \circ \cdots
      \circ}_p\\
      \hline
    \end{array} . \label{eq:1-CDWth2unitcell2}
  \end{equation}
  Now, let us assume that the last fermion in block $A$ contributes to the
  potential energy of this block by amount $E_{\Delta}$. If $E_{\Delta} = 0$,
  we can always swap this last fermion with the rest of block $A$ and again
  consider the last fermion of a new block. If we take this last fermion out
  and replace it with a hole, then the energy of the block $A$ will decrease
  by $E_{\Delta}$. Let us now move the last fermion in unit cell
  (\ref{eq:1-CDWth2unitcell2}) by $p$ sites to the right:
  \begin{equation}
    \begin{array}{|l|}
      \hline
      \underbrace{\bullet ? ? \cdots ? \bullet}_{\tmop{Block} A} 
      \underbrace{\circ \circ \cdots \circ}_p\\
      \hline
    \end{array} \begin{array}{|l|}
      \hline
      \underbrace{\bullet ? ? \cdots ? \bullet}_{\tmop{Block} A} 
      \underbrace{\circ \circ \cdots \circ}_p\\
      \hline
    \end{array} \cdots \begin{array}{|l|}
      \hline
      \underbrace{\bullet ? ? \cdots ? \circ}_{\tmop{Block} A'} 
      \underbrace{\circ \cdots \circ}_{p - 1} \bullet \underbrace{\circ \circ
      \cdots \circ}_p\\
      \hline
    \end{array},
  \end{equation}
  where block $A'$ is block $A$ with the last fermion replaced by a hole.
  Block $A'$ has energy $E_A - E_{\Delta}$. The last fermion does not
  contribute now to the overall potential energy, because it is surrounded by
  $p$ sites on both sides. Such a unit cell now has energy density
  \begin{equation}
    \frac{pE_A - E_{\Delta}}{pN / Q} = \frac{E_A}{N / Q} -
    \frac{E_{\Delta}}{pN / Q},
  \end{equation}
  which is lower than the energy of the ground-state unit cell
  (\ref{eq:1-CDWth2unitcell}), and this leads to a~contradiction. A similar
  process can be used to show that a ground state cannot have a~space equal to
  $(p + 2)$ and more sites. Thus, we conclude that the largest space in any
  ground state must have at most $p$ sites.
\end{proof}

Using Theorems \ref{1-CDWth1} and \ref{1-CDWth2}, we can significantly
decrease the number of generated states.

\subsection{Details of the calculation}\label{ch:1-CDWdetailscalc}

Our calculations were performed using Mathematica {\cite{Mathematica}} (see
Appendix \ref{ch:appendixCDWMath}). Firstly a~partial basis for a specific
number of particles $N$, density $Q$ and interaction range $p$ was generated.
States of this partial basis had the first $1 / Q$ sites fixed to the
configuration shown in Eq. (\ref{eq:1-CDWth1fix}) by Theorem \ref{1-CDWth1},
and any states that were not in agreement with Theorem \ref{1-CDWth2} were
removed. Then, the energy density was calculated for every state and this list
of energies was simplified by removing duplicates. In order to discard the
energies that cannot describe the ground state, the expression $\forall_{\beta
\neq \alpha} E_{\alpha} < E_{\beta}$ was assessed. Some energies however could
not be compared without knowing the values of $\{ U_m \}$. The final list
contains the energies of all phases that have the lowest energy for some set
values of $\{ U_m \}$; these are the CDW phases of the system.

\subsection{Results for $Q = 1 / (p - 1)$}

Unit cells and energy densities for $p = 3$, 4, and 5 are presented in Table
\ref{tab:1-CDWphases1}. Due to the finite size of the systems studied, we can
only look for CDW unit cells up to a specific size (designated by $L_{\max}$).
Phase diagrams in Figures \ref{fig:p3q1-2}--\ref{fig:p5q1-4} show what phases
are expected to appear for different values of the potentials $\{ U_m \}$.

\begin{table}[h]
  \begin{tabular}{lcccc}
    \hline\hline
    System & GS unit cell & Energy density & $f$ & \\
    \hline
    $p = 3,$ & $\bullet \circ$ & $\frac{1}{2} U_2$ & 2 &
    {\color[HTML]{00AAFF}$\blacksquare$}\\
    $Q = 1 / 2,$ & $\bullet \bullet \circ \circ$ & $\frac{1}{4} (U_1 + U_3)$ &
    4 & {\color[HTML]{55FF00}$\blacksquare$}\\
    $L_{\max} = 28$ & $\bullet \bullet \bullet \circ \circ \circ$ &
    $\frac{1}{6} (2 U_1 + U_2)$ & 6 & {\color[HTML]{FFFF00}$\blacksquare$}\\
    \hline
    $p = 4,$ & $\bullet \circ \circ$ & $\frac{1}{3} U_3$ & 3 &
    {\color[HTML]{00AAFF}$\blacksquare$}\\
    $Q = 1 / 3,$ & $\bullet \bullet \circ \circ \circ \circ$ & $\frac{1}{6}
    U_1$ & 6 & {\color[HTML]{55FF00}$\blacksquare$}\\
    $L_{\max} = 36$ & $\bullet \circ \bullet \circ \circ \circ$ & $\frac{1}{6}
    (U_2 + U_4)$ & 6 & {\color[HTML]{FFFF00}$\blacksquare$}\\
    & $\bullet \bullet \circ \circ \circ \bullet \circ \circ \circ$ &
    $\frac{1}{9} (U_1 + 2 U_4)$ & 9 & {\color[HTML]{FF0000}$\blacksquare$}\\
    & $\bullet \circ \bullet \circ \bullet \circ \circ \circ \circ$ &
    $\frac{1}{9} (2 U_2 + U_4)$ & 9 & {\color[HTML]{800080}$\blacksquare$}\\
    & $\bullet \circ \bullet \circ \circ \bullet \circ \bullet \circ \circ
    \circ \circ$ & $\frac{1}{12} (2 U_2 + U_3)$ & 12 &
    {\color[HTML]{FF8000}$\blacksquare$}\\
    & $\bullet \bullet \bullet \circ \circ \circ \circ \bullet \circ \bullet
    \circ \circ \circ \circ \bullet \circ \bullet \circ \circ \circ \circ$ &
    $\frac{1}{21} (2 U_1 + 3 U_2)$ & $2 \times 21$ & $\blacksquare$\\
    \hline
    $p = 5,$ & $\bullet \circ \circ \circ$ & $\frac{1}{4} U_4$ & 4 &
    {\color[HTML]{00AAFF}$\blacksquare$}\\
    $Q = 1 / 4,$ & $\bullet \circ \circ \bullet \circ \circ \circ \circ$ &
    $\frac{1}{8} (U_3 + U_5)$ & 8 & {\color[HTML]{55FF00}$\blacksquare$}\\
    $L_{\max} = 32$ & $\bullet \circ \bullet \circ \circ \circ \circ \circ$ &
    $\frac{1}{8} U_2$ & 8 & {\color[HTML]{FFFF00}$\blacksquare$}\\
    & $\bullet \circ \bullet \circ \circ \circ \circ \bullet \circ \circ
    \circ \circ$ & $\frac{1}{12} (U_2 + 2 U_5)$ & 12 &
    {\color[HTML]{FF0000}$\blacksquare$}\\
    & $\bullet \circ \circ \bullet \circ \circ \bullet \circ \circ \circ
    \circ \circ$ & $\frac{1}{12} 2 U_3$ & 12 &
    {\color[HTML]{800080}$\blacksquare$}\\
    & $\bullet \bullet \circ \circ \circ \circ \bullet \circ \circ \circ
    \circ \bullet \circ \circ \circ \circ$ & $\frac{1}{16} (U_1 + 3 U_5)$ & 16
    & {\color[HTML]{FF8000}$\blacksquare$}\\
    & $\bullet \bullet \circ \circ \circ \circ \circ \bullet \bullet \circ
    \circ \circ \circ \circ \bullet \circ \circ \circ \circ \circ$ &
    $\frac{1}{20} 2 U_1$ & 20 & $\blacksquare$\\
    \hline\hline
  \end{tabular}
  \caption{Ground-state (GS) unit cells and their energies in a system with $Q
  = 1 / (p - 1)$. \tmtextit{f }is the degeneracy of the ground state.
  $L_{\max}$ is the maximum size of the unit cell that was analysed. Colours
  designate phases shown in
  Figs.~\ref{fig:p3q1-2}--\ref{fig:p5q1-4}.\label{tab:1-CDWphases1}}
\end{table}

\begin{figure}[h]
  \centering
  \begin{minipage}{.48\textwidth}
    \centering
    \includegraphics[width=6cm]{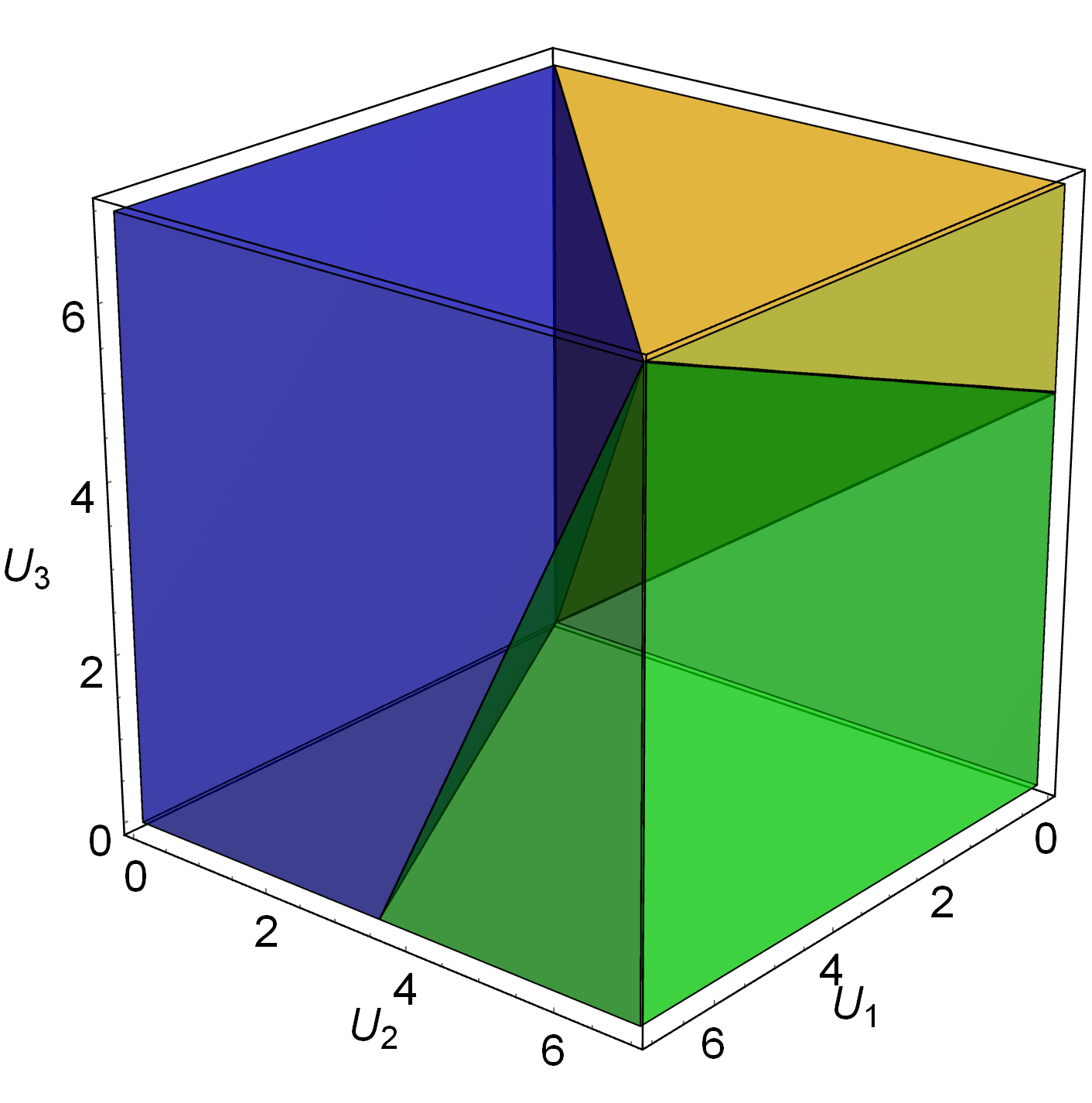}
    \captionof{figure}{
      Phase diagram of $p = 3, Q = 1 / 2$.}
    \label{fig:p3q1-2}
  \end{minipage}\hspace{0.03\linewidth}
  \begin{minipage}{.48\textwidth}
    \centering
    \includegraphics[width=6cm]{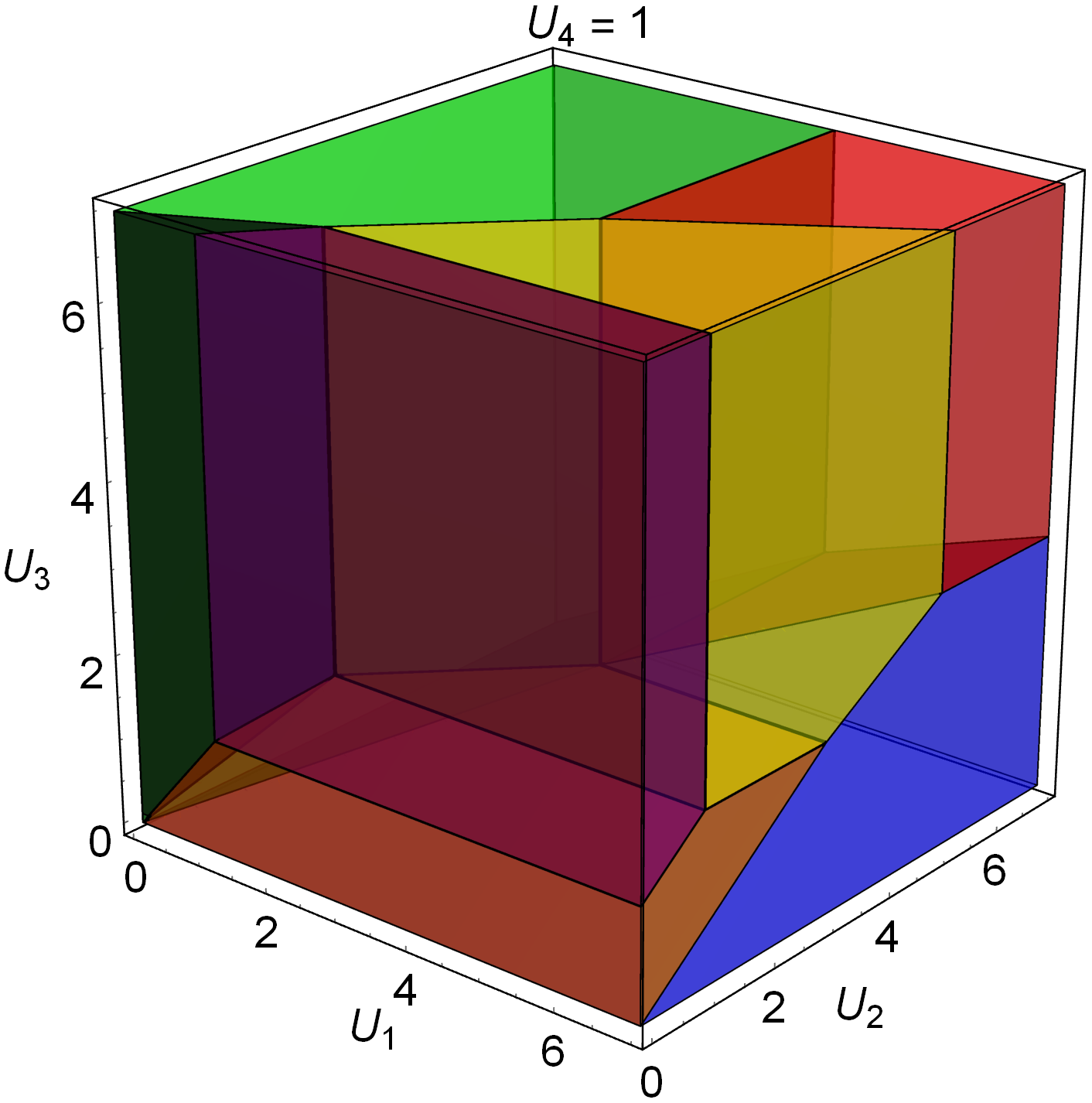}
    \captionof{figure}{
      Phase diagram of $p = 4, Q = 1 / 3$.}
    \label{fig:p4q1-3}
  \end{minipage}
\end{figure}

\begin{figure}[h]
  \begin{tabular}{cc}
    \resizebox{6cm}{!}{\includegraphics{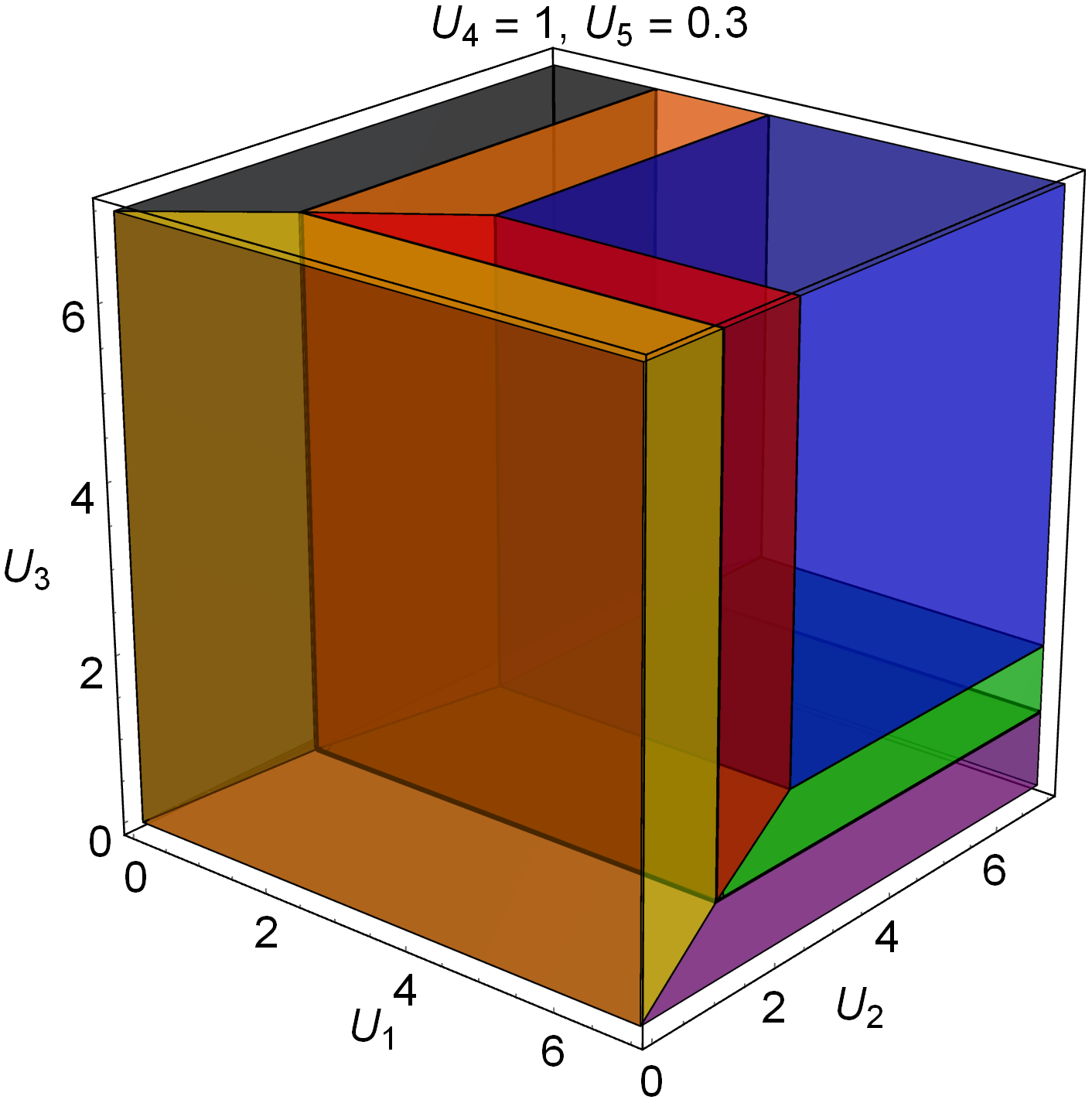}}
    &
    \resizebox{6cm}{!}{\includegraphics{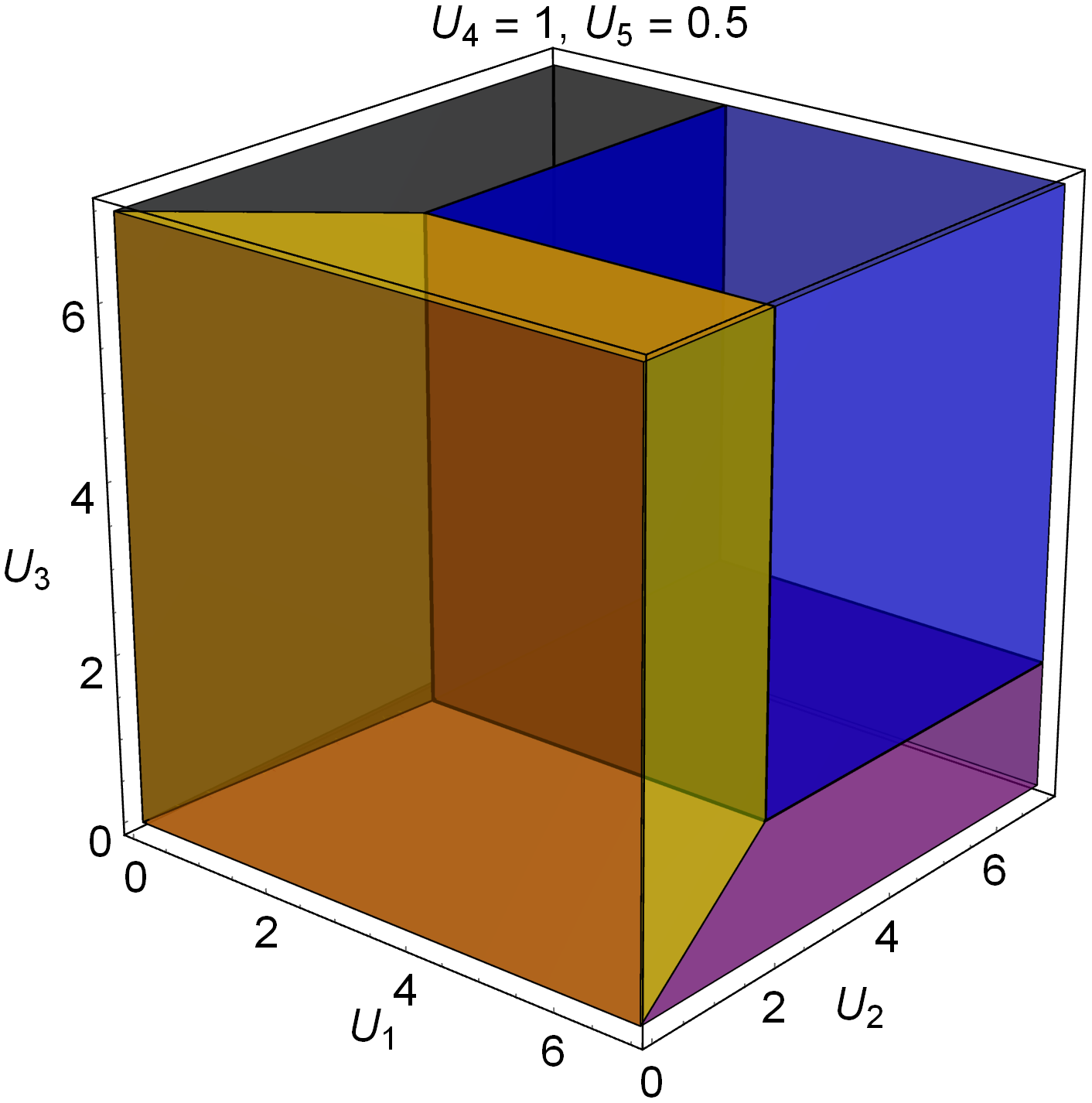}}
  \end{tabular}
  \caption{Phase diagrams of $p = 5, Q = 1 / 4$ with two values of $U_5 = 0.3$
  and $U_5 = 0.5$.\label{fig:p5q1-4}}
\end{figure}

\subsection{Results for $Q = 1 / (p - 2)$ and $Q = 1 / 2$}

Table \ref{tab:1-CDWphases2} presents the unit cells and energy densities for
$p = 4$ and 5 and $Q = 1 / (p - 2)$. Notice that for $p = 5$, we have found
ground-state unit cells up to $(L_{\max} - 3)$ and thus potentially there
could be a ground state containing an even larger unit cell that was not found
in our calculation. Phase diagram of the $p = 4, Q = 1 / 2$ system is shown in
Fig.~\ref{fig:p4q1-2}.

Similarly, results for $Q = 1 / 2, p = 5$ and $6$ are presented in Table
\ref{tab:1-CDWphases3}.

\begin{table}[h]
  \begin{tabular}{lcccc}
    \hline\hline
    System & GS unit cell & Energy density & $f$ & \\
    \hline
    $p = 4,$ & $\bullet \circ$ & $\frac{1}{2} (U_2 + U_4)$ & 2 &
    {\color[HTML]{00AAFF}$\blacksquare$}\\
    $Q = 1 / 2,$ & $\bullet \bullet \circ \circ$ & $\frac{1}{4} (U_1 + U_3 + 2
    U_4)$ & 4 & {\color[HTML]{55FF00}$\blacksquare$}\\
    $L_{\max} = 26$ & $\bullet \bullet \bullet \circ \circ \circ$ &
    $\frac{1}{6} (2 U_1 + U_2 + U_4)$ & 6 &
    {\color[HTML]{FFFF00}$\blacksquare$}\\
    & $\bullet \bullet \bullet \bullet \circ \circ \circ \circ$ &
    $\frac{1}{8} (3 U_1 + 2 U_2 + U_3)$ & 8 &
    {\color[HTML]{FF0000}$\blacksquare$}\\
    & $\bullet \bullet \circ \bullet \circ \circ \bullet \circ$ &
    $\frac{1}{8} (U_1 + 2 U_2 + 3 U_3)$ & 8 &
    {\color[HTML]{800080}$\blacksquare$}\\
    \hline
    $p = 5,$ & $\bullet \circ \circ$ & $\frac{1}{3} U_3$ & $3$ & \\
    $Q = 1 / 3,$ & $\bullet \circ \bullet \circ \circ \circ$ & $\frac{1}{6}
    (U_2 + U_4)$ & $6$ & \\
    $L_{\max} = 27$ & $\bullet \bullet \circ \circ \circ \circ$ & $\frac{1}{6}
    (U_1 + U_5)$ & $6$ & \\
    & $\bullet \bullet \circ \circ \circ \bullet \circ \circ \circ$ &
    $\frac{1}{9} (U_1 + 2 U_4 + 2 U_5)$ & $9$ & \\
    & $\bullet \circ \bullet \circ \bullet \circ \circ \circ \circ$ &
    $\frac{1}{9} (2 U_2 + U_4 + U_5)$ & $9$ & \\
    & $\bullet \circ \bullet \circ \circ \bullet \circ \bullet \circ \circ
    \circ \circ$ & $\frac{1}{12} (2 U_2 + U_3 + 3 U_5)$ & $12$ & \\
    & $\bullet \bullet \circ \circ \bullet \circ \circ \circ \bullet \circ
    \circ \circ \bullet \circ \circ$ & $\frac{1}{15} (U_1 + 2 U_3 + 4 U_4)$ &
    $15$ & \\
    & $\bullet \bullet \bullet \circ \circ \circ \circ \circ \bullet \bullet
    \circ \circ \circ \circ \circ$ & $\frac{1}{15} (3 U_1 + U_2)$ & 15 & \\
    & $\bullet \bullet \circ \circ \bullet \bullet \circ \circ \circ \circ
    \circ \bullet \bullet \circ \circ \circ \circ \circ$ & $\frac{1}{18} (3
    U_1 + U_3 + 2 U_4 + U_5)$ & $18$ & \\
    & $\bullet \bullet \circ \circ \bullet \circ \circ \bullet \bullet \circ
    \circ \circ \circ \circ \bullet \bullet \circ \circ \circ \circ \circ$ &
    $\frac{1}{21} (3 U_1 + 2 U_3 + 2 U_4)$ & $21$ & \\
    & $\bullet \bullet \bullet \circ \circ \circ \circ \bullet \circ \bullet
    \circ \circ \circ \circ \bullet \circ \bullet \circ \circ \circ \circ$ &
    $\frac{1}{21} (2 U_1 + 3 U_2 + 3 U_5)$ & $21$ & \\
    & $\bullet \circ \bullet \circ \circ \circ \circ \circ \bullet \bullet
    \bullet \circ \circ \circ \circ \circ \bullet \bullet \bullet \circ \circ
    \circ \circ \circ$ & $\frac{1}{24} (4 U_1 + 3 U_2)$ & 24 & \\
    \hline\hline
  \end{tabular}
  \caption{As Table \ref{tab:1-CDWphases1}, but for a system with density $Q =
  1 / (p - 2)$. Colours designate phases shown
  in~Fig.~\ref{fig:p4q1-2}.\label{tab:1-CDWphases2}}
\end{table}

\begin{figure}[h!]
  \resizebox{6cm}{!}{\includegraphics{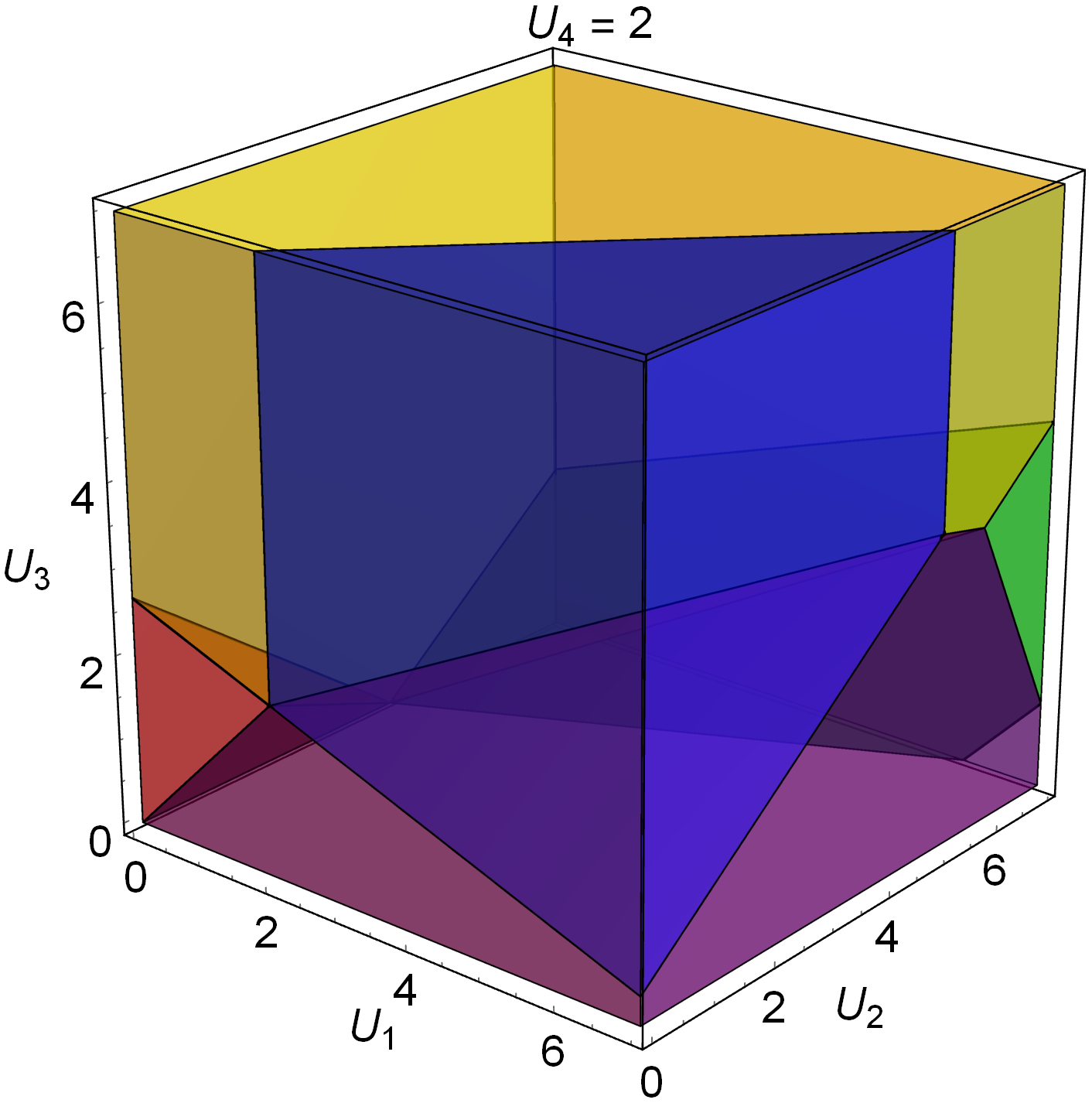}}
  \caption{Phase diagram of $p = 4, Q = 1 / 2$.\label{fig:p4q1-2}}
\end{figure}

\begin{table}[h]
  \begin{tabular}{lccc}
    \hline\hline
    System & GS unit cell & Energy density & $f$\\
    \hline
    $p = 5,$ & $\bullet \circ$ & $\frac{1}{2} (U_2 + U_4)$ & 2\\
    $Q = 1 / 2,$ & $\bullet \bullet \circ \circ$ & $\frac{1}{4} (U_1 + U_3 + 2
    U_4 + U_5)$ & 4\\
    $L_{\max} = 26$ & $\bullet \bullet \circ \bullet \circ \circ$ &
    $\frac{1}{6} (U_1 + U_2 + 2 U_3 + U_4 + U_5)$ & $2 \times 6$\\
    & $\bullet \bullet \bullet \circ \circ \circ$ & $\frac{1}{6} (2 U_1 + U_2
    + U_4 + 2 U_5)$ & 6\\
    & $\bullet \bullet \circ \bullet \circ \circ \bullet \circ$ &
    $\frac{1}{8} (U_1 + 2 U_2 + 3 U_3 + 3 U_5)$ & 8\\
    & $\bullet \bullet \bullet \bullet \circ \circ \circ \circ$ &
    $\frac{1}{8} (3 U_1 + 2 U_2 + U_3 + U_5)$ & 8\\
    & $\bullet \bullet \circ \bullet \bullet \circ \circ \bullet \circ \circ$
    & $\frac{1}{10} (2 U_1 + U_2 + 4 U_3 + 3 U_4)$ & 10\\
    & $\bullet \bullet \bullet \bullet \bullet \circ \circ \circ \circ \circ$
    & $\frac{1}{10} (4 U_1 + 3 U_2 + 2 U_3 + U_4)$ & 10\\
    \hline
    $p = 6,$ & $\bullet \circ$ & $\frac{1}{2} (U_2 + U_4 + U_6)$ & 2\\
    $Q = 1 / 2,$ & $\bullet \bullet \circ \circ$ & $\frac{1}{4} (U_1 + U_3 + 2
    U_4 + U_5)$ & 4\\
    $L_{\max} = 26$ & $\bullet \bullet \circ \bullet \circ \circ$ &
    $\frac{1}{6} (U_1 + U_2 + 2 U_3 + U_4 + U_5 + 3 U_6)$ & $2 \times 6$\\
    & $\bullet \bullet \bullet \circ \circ \circ$ & $\frac{1}{6} (2 U_1 + U_2
    + U_4 + 2 U_5 + 3 U_6)$ & 6\\
    & $\bullet \circ \bullet \bullet \circ \bullet \circ \circ$ &
    $\frac{1}{8} (U_1 + 2 U_2 + 3 U_3 + 3 U_5 + 2 U_6)$ & 8\\
    & $\bullet \bullet \bullet \bullet \circ \circ \circ \circ$ &
    $\frac{1}{8} (3 U_1 + 2 U_2 + U_3 + U_5 + 2 U_6)$ & 8\\
    & $\bullet \bullet \circ \bullet \bullet \circ \circ \bullet \circ \circ$
    & $\frac{1}{10} (2 U_1 + U_2 + 4 U_3 + 3 U_4 + 3 U_6)$ & 10\\
    & $\bullet \bullet \bullet \bullet \bullet \circ \circ \circ \circ \circ$
    & $\frac{1}{10} (4 U_1 + 3 U_2 + 2 U_3 + U_4 + U_6)$ & 10\\
    & $\bullet \circ \bullet \bullet \circ \bullet \circ \bullet \circ \circ
    \bullet \circ$ & $\frac{1}{12} (U_1 + 4 U_2 + 3 U_3 + 2 U_4 + 5 U_5)$ &
    12\\
    & $\bullet \bullet \bullet \bullet \bullet \bullet \circ \circ \circ
    \circ \circ \circ$ & $\frac{1}{12} (5 U_1 + 4 U_2 + 3 U_3 + 2 U_4 + U_5)$
    & 12\\
    & $\bullet \bullet \bullet \circ \bullet \circ \circ \circ \bullet \circ
    \bullet \bullet \bullet \circ \circ \bullet \circ \circ$ & $\frac{1}{18}
    (4 U_1 + 4 U_2 + 4 U_3 + 5 U_4 + 2 U_5 + 3 U_6)$ & 18\\
    & $\bullet \bullet \bullet \circ \bullet \circ \circ \circ \bullet
    \bullet \circ \bullet \bullet \circ \circ \bullet \circ \circ$ &
    $\frac{1}{18} (4 U_1 + 3 U_2 + 5 U_3 + 5 U_4 + 2 U_5 + 3 U_6)$ & $2 \times
    18$\\
    \hline\hline
  \end{tabular}
  \caption{As Table \ref{tab:1-CDWphases1}, but for a system with density $Q =
  1 / 2$.\label{tab:1-CDWphases3}}
\end{table}

\subsection{Discussion of the results}

Our results illustrate how highly nontrivial and unpredictable the
ground-state configurations are for critical densities higher than $Q = 1 /
p$. For example, for a half-filled system ($Q = 1 / 2$), judging only from the
$p = 3$ case, one would naively expect a similar trend to be present in all
other cases: for all units cells to consist of a chain of occupied sites,
followed by a chain of the same length, but with empty sites. However, Table
\ref{tab:1-CDWphases2} shows that for $p = 4$ there exists a ground state with
a~unit cell $(\bullet \bullet \circ \bullet \circ \circ \bullet \circ)$, that
does not follow this prediction. Therefore, it is very difficult to create
a~simple set of rules describing the ground-state properties of all the phases
in the system with high critical density.

We also conclude that the number of possible CDW phases in the system grows
with the maximum interaction range $p$ and the density $Q$. For example, in
the system $p = 5, Q = 1 / 3$ presented in Table \ref{tab:1-CDWphases2}, there
are at least 12 different CDW phases, and we expect $p = 6, Q = 1 / 4$ to
contain even more (preliminary results indicate 23 phases).

\begin{table}[h]
  \begin{tabular}{cc|cccccccc}
    \hline\hline
    &  & \multicolumn{8}{l}{$Q =${\hspace{7em}}$\longleftarrow$ } \\
    &  & 1/2 & 1/3 & 1/4 & 1/5 & 1/6 & 1/7 & 1/8 & $\cdots$\\
    \hline
    $p =$ & 1 & 1 & \cellcolor{gray!25} & \cellcolor{gray!25} & \cellcolor{gray!25} & \cellcolor{gray!25} & \cellcolor{gray!25} & \cellcolor{gray!25} & \cellcolor{gray!25}\\
    & 2 & 2 & 1 & \cellcolor{gray!25} & \cellcolor{gray!25} & \cellcolor{gray!25} & \cellcolor{gray!25} & \cellcolor{gray!25} & \cellcolor{gray!25}\\
    & 3 & 3 & 3 & 1 & \cellcolor{gray!25} & \cellcolor{gray!25} & \cellcolor{gray!25} & \cellcolor{gray!25} & \cellcolor{gray!25}\\
    {$\downarrow$} & 4 & 5 & 7 & 4 & 1 & \cellcolor{gray!25} & \cellcolor{gray!25} & \cellcolor{gray!25} & \cellcolor{gray!25}\\
    & 5 & 8 & 12 & 7 & 5 & 1 & \cellcolor{gray!25} & \cellcolor{gray!25} & \cellcolor{gray!25}\\
    & 6 & 12 & \tmtextit{63} & \tmtextit{23} & 9 & 6 & 1 & \cellcolor{gray!25} & \cellcolor{gray!25}\\
    & 7 & \tmtextit{26} & ? & ? & ? & \tmtextit{11} & 7 & 1 & \cellcolor{gray!25}\\
    & $\vdots$ &  &  &  &  &  &  &  & $\ddots$\\
    \hline\hline
  \end{tabular}
  \caption{Number of different possible CDW phases in the generalised $t$-$V$
  model as a function of interaction range $p$ and density $Q$. Results in
  italic are preliminary.}
\end{table}

For $t \neq 0$, we expect non-CDW phases to be present in the system. If one
considers the phase diagrams from Figs.~\ref{fig:p3q1-2}--\ref{fig:p4q1-2}, on
the interfaces between any two phases there are probably Luttinger liquid and
bond-order phases, similarly to the findings of Refs.
{\cite{Schmitteckert2004}} and {\cite{Mishra2011}}. Therefore, if our
assumption that the number of phases grows quickly with the maximum
interaction range is correct, then we can predict that for high $p$, the phase
diagram consists of mainly non-CDW phases, while CDW insulators are only
present when certain $U_m$ are very high. Thus, a large interaction range may
imply the loss of insulating properties of the material.

\section{Conclusions and outlook}

Mott insulating phases are currently of great interest, partly because of
their possible application in future transistor technology. However, as we
have shown in this chapter, the insulating properties of the material may be
altered depending on the effective interaction between electrons in the
system. We have shown how to construct the ground state of all the Mott
insulating phases at low critical densities in the generalised $t$-$V$ model,
and we have calculated the ground-state unit cells of a few example cases for
higher critical densities. Thus, we provide a description of possible CDW
phases of the system with any interaction range and any critical density in
the atomic limit. The number of possible CDW phases increases with the
interaction range and thus great care is needed in order to determine which
CDW phase will appear in a specific system if the interaction range is large.

For a smoothly varying potential, the main difficulty in obtaining an
insulating system may be due to the emergence of liquid phases in the system.
At a non-zero temperature (and thus high kinetic energy) non-insulating phases
may be prominent in the system, and therefore one would need an analysis of
the phase diagram of the model beyond the atomic limit (\tmtextit{i.e.} with
$t > 0$). For a~one-dimensional system, this can be achieved using novel
renormalisation-group methods based on examination of entanglement in the
system {\cite{Schollwock2005}}. Therefore, in the future, we propose to use
the matrix product states approach {\cite{PerezGarcia2007,Verstraete2008}},
which has proven useful in calculations of lattice models and requires
relatively low computational resources. In order to accurately describe the
long\mbox{-}range correlations in the system, one would need to use high bond
dimension in the matrix product state. Appendix \ref{ch:appendixMPS} shows the
initial steps needed to use this approach on the generalised $t$-$V$ model.
\begin{center}
  \part{Charge\mbox{-}carrier com{\nobreak}plexes
  in~two-dimensional semicon{\nobreak}ductors
  }
\end{center}
\chapter{Theoretical background}\label{ch:2-theory}

\section{Transition-metal dichalcogenides}

After the experimental discovery of graphene {\cite{Novoselov2004,Geim2007}},
two-dimensional (2D) materials have become a~major focus in physics. Due to
their wide potential applications, the properties of 2D materials are studied
intensely nowadays. Prominent examples of those novel materials are
transition\mbox{-}metal dichalcogenide (TMDC) monolayers, which are stable,
hexagonal 2D semiconductors, that have one advantage as compared to graphene,
namely they naturally possess a~band gap. Therefore, TMDCs can be used in
optoelectronics
\cite{Xiao2012,Zeng2012,Mak2012,Sallen2012,Cao2012,Conley2013,Xu2014},
in the production of transistors, photoemitters and photodetectors.

\begin{figure}
  \centering
  \begin{minipage}{.48\textwidth}
    \centering
    \includegraphics[width=.9\linewidth]{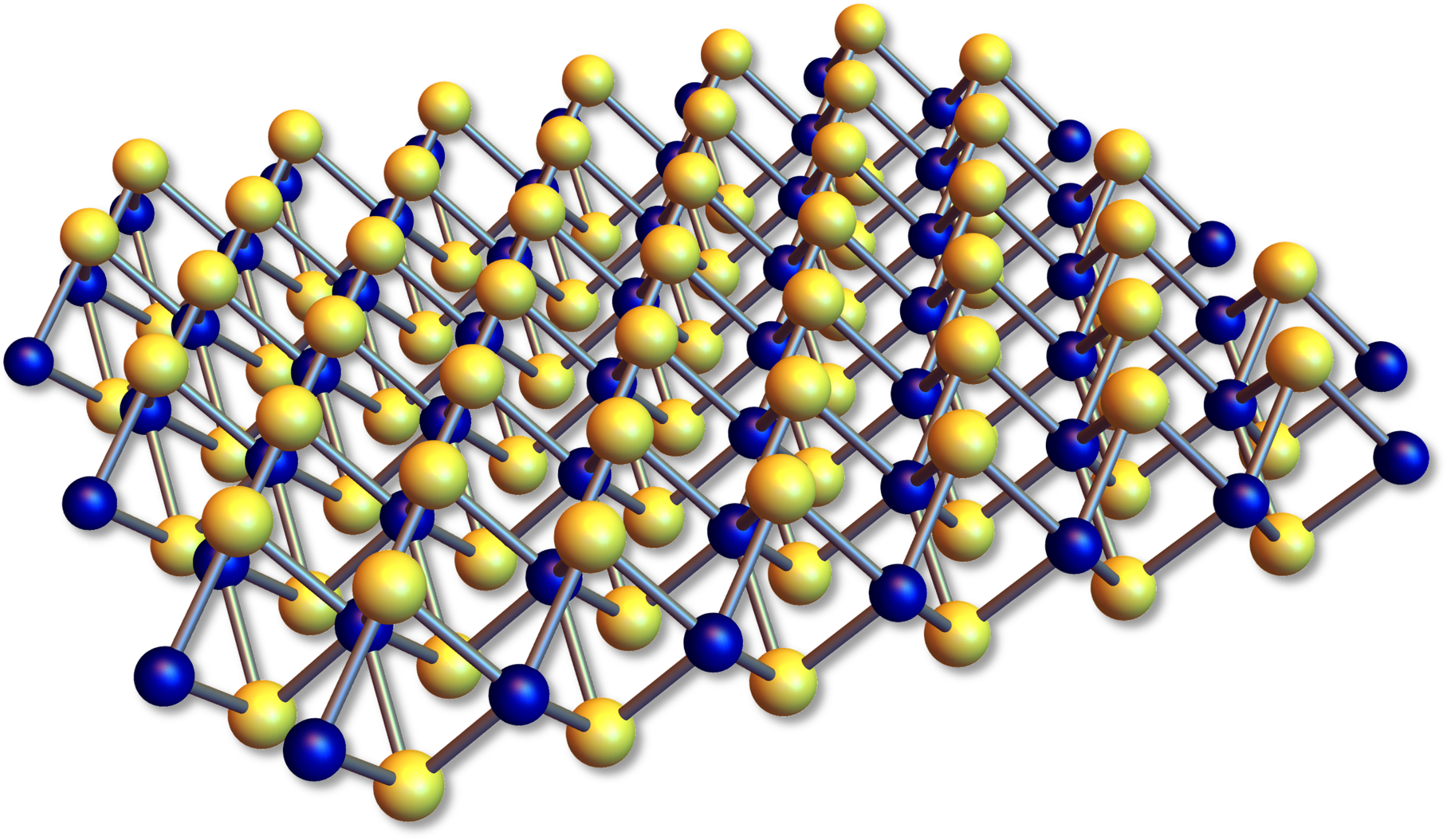}
    \captionof{figure}{
      Atomic structure of a transition-metal dichalcogenide MX$_2$, where 
      yellow atoms are of type X and blue atoms are of type M.}
    \label{fig:2-tmdc1}
  \end{minipage}\hspace{0.03\linewidth}
  \begin{minipage}{.48\textwidth}
    \centering
    \includegraphics[width=.7\linewidth]{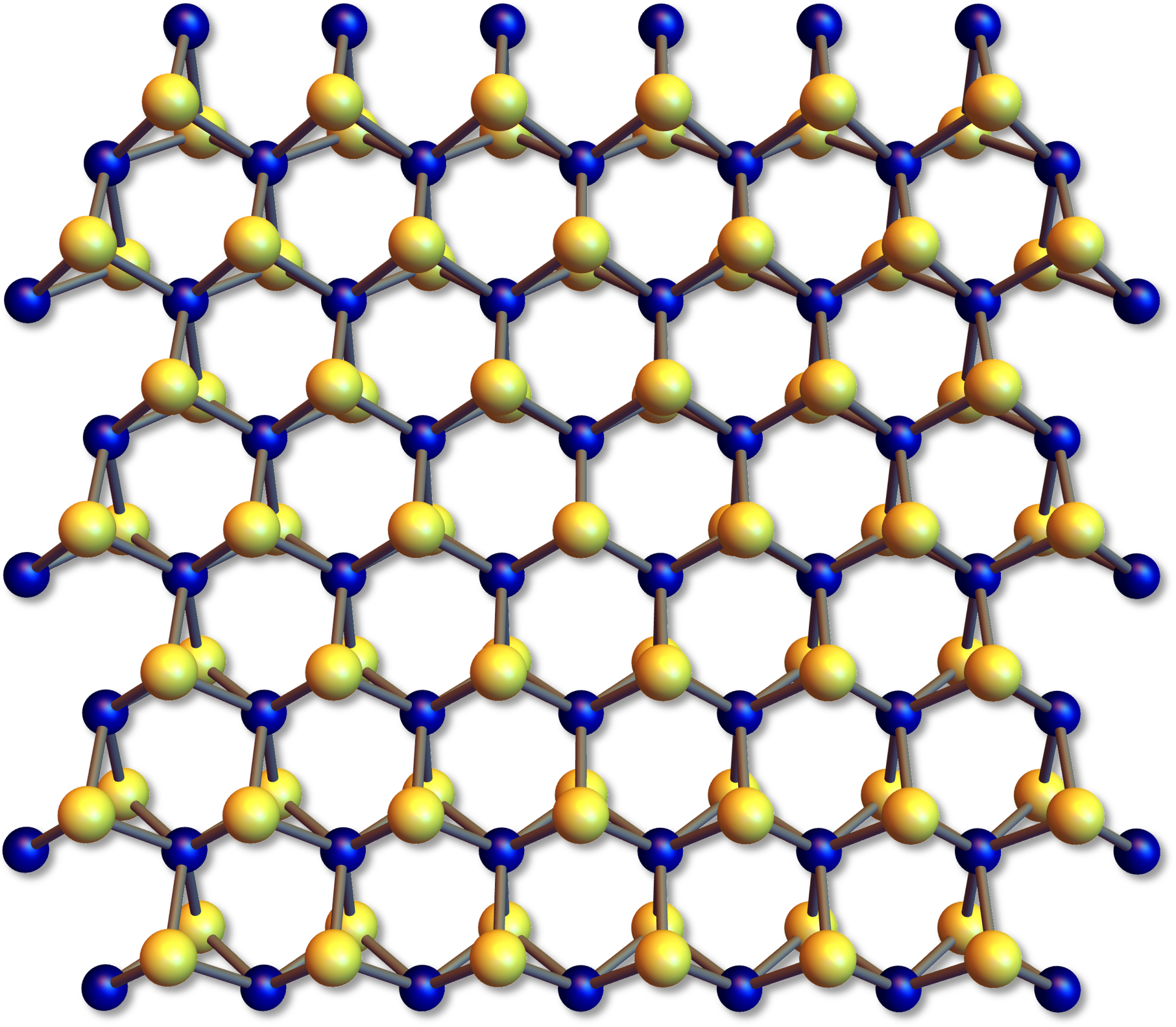}
    \captionof{figure}{
      Structure of a transition-metal dichalcogenide as viewed from the top.}
    \label{fig:2-tmdc2}
  \end{minipage}
\end{figure}

TMDC monolayers have chemical composition of MX$_2$, where M is a transition
metal atom (\tmtextit{e.g.}~molybdenum Mo, or tungsten W), and X is a
chalcogen atom (sulfur S, selenium Se, or tellurium Te). Figure
\ref{fig:2-tmdc1} shows the structure of one layer of a TMDC, where we can see
that the layer of transition metal atoms rests between two layers of chalcogen
atoms. The thickness of the TMDC monolayer is the distance between the two
chalcogen atom layers. Viewed from the top (see Fig.~\ref{fig:2-tmdc2}), the
lattice of the TMDC material is a~honeycomb, similar to graphene.

\begin{figure}
  \centering
  \begin{minipage}{.48\textwidth}
    \centering
    \includegraphics[width=.8\linewidth]{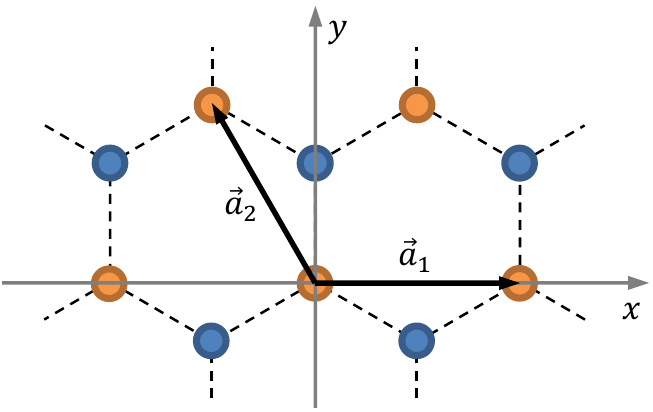}
    \captionof{figure}{
      Bravais lattice of TMDCs, viewed from the top. $\vec{a}_1$ and 
      $\vec{a}_2$ are primitive vectors of hexagonal lattice.}
    \label{fig:2-tmdc-lattice1}
  \end{minipage}\hspace{0.03\linewidth}
  \begin{minipage}{.48\textwidth}
    \centering
    \includegraphics[width=.8\linewidth]{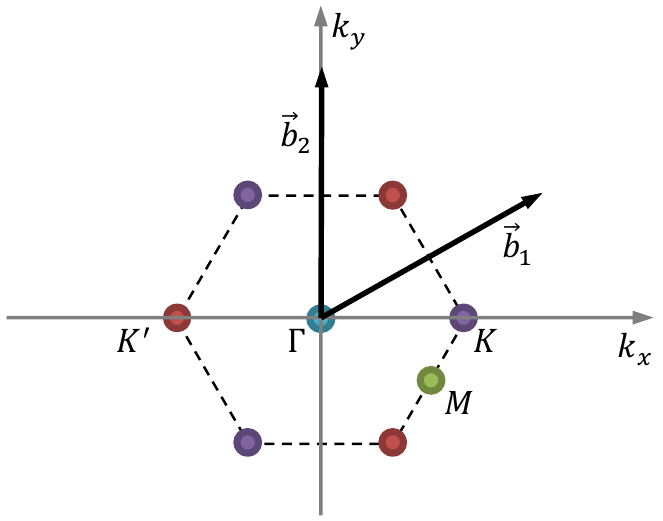}
    \captionof{figure}{
      First Brillouin zone of TMDCs with reciprocal lattice vectors $\vec{b}_1$ 
      and $\vec{b}_2$.}
    \label{fig:2-tmdc-lattice2}
  \end{minipage}
\end{figure}

The Bravais lattice of a TMDC is shown in Fig.~\ref{fig:2-tmdc-lattice1}, and
Fig.~\ref{fig:2-tmdc-lattice2} presents the corresponding reciprocal lattice.
Many TMDCs possess a direct band gap, \tmtextit{i.e.} the energy minimum of
the conduction band has the same momentum as the energy maximum of the valence
band. Those band edges occur at the corners ($K$ and $K'$ points) of the
hexagonal Brillouin zone, where the electron states carry angular momentum and
the bands exhibit spin splitting due to spin-orbit (SO) coupling
\cite{Kuc2011,Zhu2011,Molina2013,Jin2013,Kormanyos2013}.

\section{Motivation}

The SO splitting of TMDCs is large in the valence band and small in the
conduction band. This makes the TMDC photoluminescence sensitive to the valley
and spin polarisation of charge carriers. The optical properties of these
materials will be also influenced by the presence of excitons and similar
charge carrier complexes, and numerous observations of the luminescence
spectra of TMDCs show the presence of peaks ascribed to excitons
\cite{Komsa2012,Qiu2013,Glazov2014,Berghauser2014,Klots2014},
trions
\cite{Mak2013,Zhang2014,Srivastava2015,Zhang2015,Zhu2014,Jones2013}
and biexcitons \cite{Mai2014,Shang2015,You2015}.
Recent experiments performed on higher-quality monolayer TMDCs have shown
additional structure in their spectra
\cite{Srivastava2015nano,He2015,Koperski2015,Chakraborty2015},
that could potentially be explained by the SO splitting and detailed
classification of charge carrier complexes.

The goal of this study is therefore to accurately describe charge carrier
complexes in two\mbox{-}dimensional semiconductors, and to show their
classification that incorporates the SO splitting and spin/valley
polarisation. Using quantum Monte Carlo methods, we will also extract the
binding energies of various charge carrier complexes, including the ones
formed around impurities in the system.

\section{Charge carriers in 2D semiconductors}

\subsection{Effective interaction}

Instead of simulating the charge carriers \tmtextit{ab initio} as particles
scattered in the hexagonal lattice of the 2D semiconductor, one can consider
the carriers in an effective mass approximation being affected by an effective
potential. This is a~Mott\mbox{-}Wannier picture {\cite{Wannier1937}}, in
which we assume the resulting exciton to have a radius much larger than the
lattice spacing, and the effective masses of (quasi)electrons and (quasi)holes
include the effects of the underlying lattice.

To calculate the effective potential, let us consider a charge density $\rho
(x, y) \delta (z)$ that is placed in a 2D semiconductor at $z = 0$. The
electric displacement $\vec{D}$ due to this charge density is then
\begin{equation}
  \vec{D} = \varepsilon_0 \vec{E} + \vec{P} = - \varepsilon_0 \nabla \phi +
  \vec{P},
\end{equation}
where $\vec{E} = - \nabla \phi$ is the electric field, $\vec{P}$ is the
polarisation vector, $\phi$ is the electrostatic potential and $\varepsilon_0$
is the vacuum permittivity. The polarisation field can be expressed as
$\vec{P} (x, y, z) = \vec{P}_{x y} (x, y) \delta (z)$, with $\vec{P}_{x y}$
being the in-plane polarisation and $\delta (z)$ being the Dirac delta
function.

Using Gauss's law ($\nabla \cdot D = \rho \delta (z)$), we can write the
following equation:
\begin{equation}
  - \varepsilon_0 \nabla^2 \phi + \nabla \cdot \vec{P} = - \varepsilon_0
  \nabla^2 \phi + (\nabla \cdot \vec{P}_{x y}) \delta (z) = \rho \delta (z) .
  \label{eq:2-effinteraction1}
\end{equation}
However $\vec{P}_{x y} = \chi \varepsilon_0 \vec{E} (x, y, 0) = - \chi
\varepsilon_0 \nabla \phi (x, y, 0)$, where $\chi$ is the in-plane
susceptibility\footnote{\setstretch{1.66}2D susceptibility can be approximated 
using layer
separation $d$ (of the material in bulk) and the in-plane dielectric constant
$\varepsilon$ as $\chi = d (\varepsilon + 1)$.} of the material (the 2D
susceptibility has units of length). Equation (\ref{eq:2-effinteraction1})
becomes
\begin{equation}
  \varepsilon_0 \nabla^2 \phi = - \rho \delta (z) - \chi \varepsilon_0
  (\nabla^2 \phi (x, y, 0)) \delta (z) . \label{eq:2-effinteraction2}
\end{equation}
Now we take the Fourier transform of Eq.~(\ref{eq:2-effinteraction2}), using
$\vec{\kappa}$ for a~wave vector in the $(x, y)$ plane and $k$ for
a~wavenumber in the $z$ direction:
\begin{equation}
  \phi (\vec{\kappa}, k) = \frac{\frac{1}{\varepsilon_0} \rho (\vec{\kappa}) -
  \chi \kappa^2 \phi (\vec{\kappa}, z = 0)}{\kappa^2 + k^2} .
\end{equation}
However, we can use another Fourier transform to write an expression for $\phi
(\vec{\kappa}, z = 0)$:
\begin{equation}
  \phi (\vec{\kappa}, z = 0) = \frac{1}{2 \pi} \int \phi (\vec{\kappa}, k)
  \mathd k = \frac{1}{2 \kappa} \left( \frac{1}{\varepsilon_0} \rho
  (\vec{\kappa}) - \chi \kappa^2 \phi (\vec{\kappa}, z = 0) \right) .
\end{equation}
Rearranging for $\phi (\vec{\kappa}, z = 0)$ gives:
\begin{equation}
  \phi (\vec{\kappa}, z = 0) = \frac{\rho}{2 \varepsilon_0 \kappa \left( 1 +
  \frac{1}{2} \chi \kappa \right)} = \frac{2 \pi \rho}{4 \pi \varepsilon_0
  \kappa \left( 1 + \frac{1}{2} \chi \kappa \right)} .
\end{equation}
One can therefore find the effective potential between charges $q_i$ and
$q_j$:
\begin{equation}
  v (\kappa) = \frac{2 \pi q_i q_j}{4 \pi \varepsilon_0 \kappa \left( 1 +
  \frac{1}{2} \chi \kappa \right)} = \frac{2 \pi q_i q_j}{4 \pi \varepsilon_0
  \kappa (1 + r_{\ast} \kappa)},
\end{equation}
where $r_{\ast} = \chi / 2$ is a parameter with units of length, directly
related to the in-plane susceptibility of the material. Taking the Fourier
transform to real space,
\begin{eqnarray}
  \frac{4 \pi \varepsilon_0}{q_i q_j} v (r) & = & \frac{1}{(2 \pi)^2} \int_{-
  \infty}^{\infty} \int_{- \infty}^{\infty} \frac{2 \pi}{\kappa (1 + r_{\ast}
  \kappa)} e^{i \kappa_x x + i \kappa_y y} \mathd \kappa_x \mathd \kappa_y \\
  & = & \frac{1}{2 \pi} \int_0^{\infty} \int_{- \pi}^{\pi} \frac{1}{\kappa (1
  + r_{\ast} \kappa)} e^{i \kappa r \cos \phi \cos \theta + i \kappa r \sin
  \phi \sin \theta} \kappa \mathd \kappa \mathd \phi \nonumber\\
  & = & \frac{1}{2 \pi} \int_0^{\infty} \frac{1}{1 + r_{\ast} \kappa} \left(
  \int_{- \pi}^{\pi} e^{i \kappa r \cos (\phi - \theta)} \mathd \phi \right)
  \mathd \kappa \nonumber\\
  & = & \int_0^{\infty} \frac{J_0 (\kappa r)}{1 + r_{\ast} \kappa} \mathd
  \kappa, \nonumber
\end{eqnarray}
where $J_n (x)$ is Bessel function of the first kind. Evaluating the final
integral gives the following equation for the potential, first introduced by
Keldysh in Ref.~{\cite{Keldysh1979}}:
\begin{equation}
  v (r) = \frac{q_i q_j}{4 \pi \varepsilon_0 r_{\ast}} 
  \underbrace{\frac{\pi}{2}  \left( H_0 \left( \frac{r}{r_{\ast}} \right) -
  Y_0 \left( \frac{r}{r_{\ast}} \right) \right)}_{V (r / r_{\ast})},
  \label{eq:2-effinteraction}
\end{equation}
where $H_n (x)$ is a Struve function and $Y_n (x)$ is a Bessel function of the
second kind.

For long-range behaviour or small susceptibility $(r \gg r_{\ast})$, we
recover the usual Coulomb interaction:
\begin{equation}
  v (r) = \frac{q_i q_j}{4 \pi \varepsilon_0}  \frac{1}{r},
  \label{eq:2-Coulomb}
\end{equation}
while at short-range or large susceptibility $(r \ll r_{\ast})$, the potential
(\ref{eq:2-effinteraction}) has the logarithmic form:
\begin{equation}
  v (r) = \frac{q_i q_j}{4 \pi \varepsilon_0}  \frac{\log (2 r_{\ast} / r) -
  \gamma}{r_{\ast}}, \label{eq:2-logarithmic}
\end{equation}
where $\gamma$ is the Euler--Mascheroni constant. See
Fig.~\ref{fig:2-effinteraction} for a comparison of the Keldysh, Coulomb and
logarithmic potentials.

\begin{figure}[h]
  \resizebox{11cm}{!}{\includegraphics{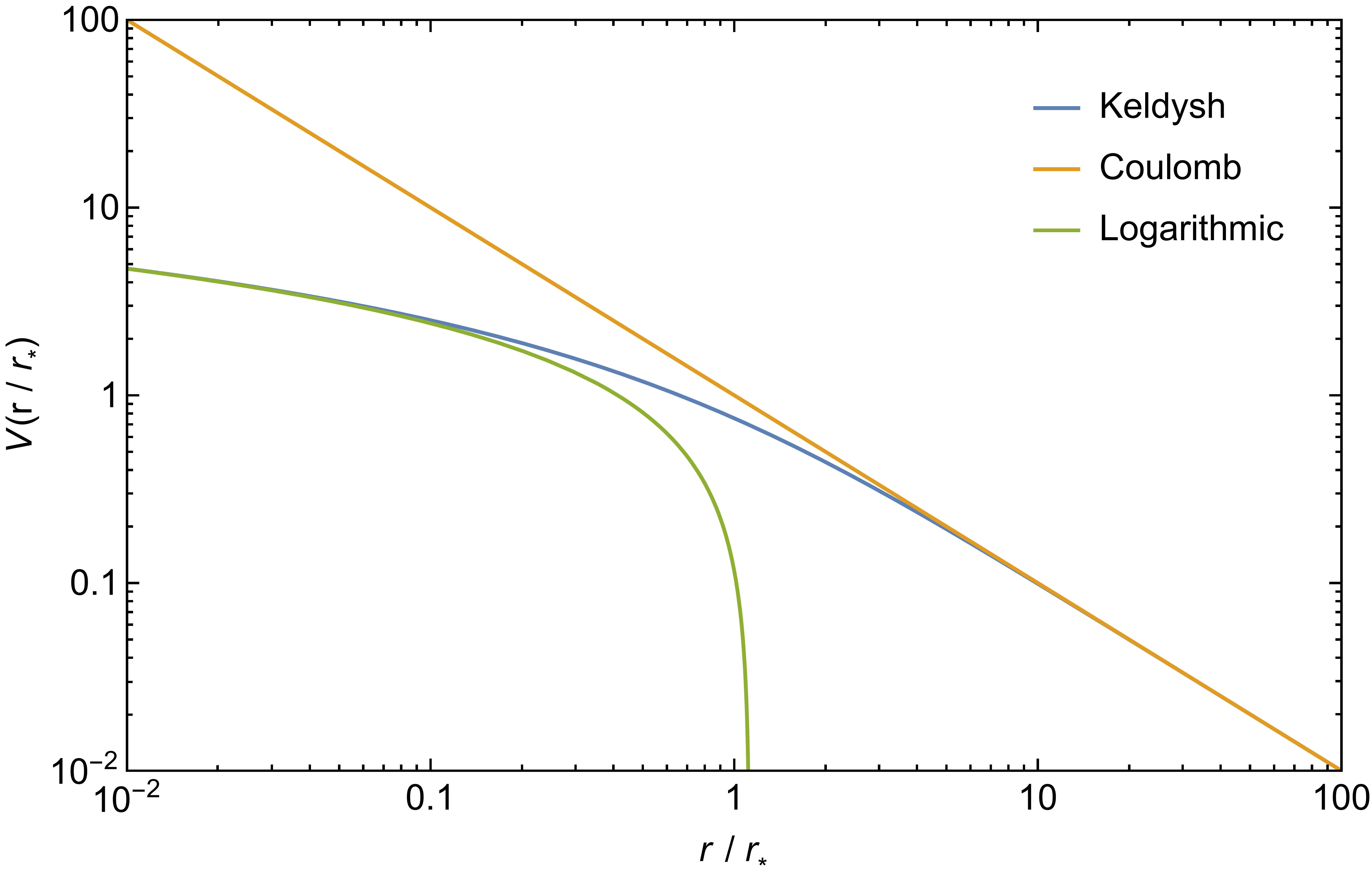}}
  \caption{Comparison of Keldysh, Coulomb and logarithmic
  potentials.\label{fig:2-effinteraction}}
\end{figure}

In short, due to the in-plane susceptibility of the material, the Coulomb
interaction between two charge carriers is modified to another form that
differs significantly from the usual Coulomb interaction, especially in its
short-range behaviour.

\subsection{Numerical evaluation of the effective interaction}

To evaluate the effective interaction, we use a Taylor expansion in $r /
r_{\ast}$ at small $r$ and an expansion in $r_{\ast} / r$ at large $r$. The
small-$r$ expansion is:
\begin{eqnarray}
  V (r) & = & \sum_{i = 0}^{\infty} \frac{(- 1)^{i + 1}}{\prod_{j = 1}^i (2
  j)^2} \left( \frac{r}{r_{\ast}} \right)^{2 i} \left( \log \frac{r}{2
  r_{\ast}} + \gamma - \sum_{k = 1}^i \frac{1}{k} \right) \\
  & + & \sum_{i = 0}^{\infty} \frac{(- 1)^i}{\prod_{j = 0}^i (2 j + 1)^2}
  \left( \frac{r}{r_{\ast}} \right)^{2 i + 1} . \nonumber
\end{eqnarray}
During the calculation, elements of the first sum must be paired with the
corresponding element of the second sum, in order to prevent numerical errors.
The most significant errors will arise due to cancellation of large elements
in the sums (which can happen due to the alternating sign in each element).
Therefore, the numerical error of the final result can be estimated by
checking the precision of the largest element in the sum. Additionally, the
term $\left( \gamma - \sum_{k = 1}^i \frac{1}{k} \right)$, rather than being
evaluated in each calculation, was tabulated for different values of $i$, for
efficiency purposes.

The large-$r$ expansion is:
\begin{equation}
  V (r) = \sum_{i = 0}^{\infty} (- 1)^i \left( \frac{r_{\ast}}{r} \right)^{2 i
  + 1} \prod_{j = 0}^{i - 1} (2 j + 1)^2 .
\end{equation}
This is an asymptotic expansion, and therefore the sum does not converge for
finite values of $r$ and only a finite number of terms should be
used\footnote{\setstretch{1.66}The error in this asymptotic series can be shown 
to be bounded
after }. The absolute value of the element after which the sum stops
converging can be used to estimate the numerical error in the final result.

Figure \ref{fig:2-effinterror} shows the difference between small-$r$ and
large-$r$ expansions that were calculated using the double-precision
arithmetic and higher precision. The switch between small-$r$ and large-$r$
expansion was chosen to be at $r = 18 r_{\ast}$, so that the numerical
precision of the calculated potential has always over 8 digits of accuracy.

\begin{figure}[h]
  \resizebox{10cm}{!}{\includegraphics{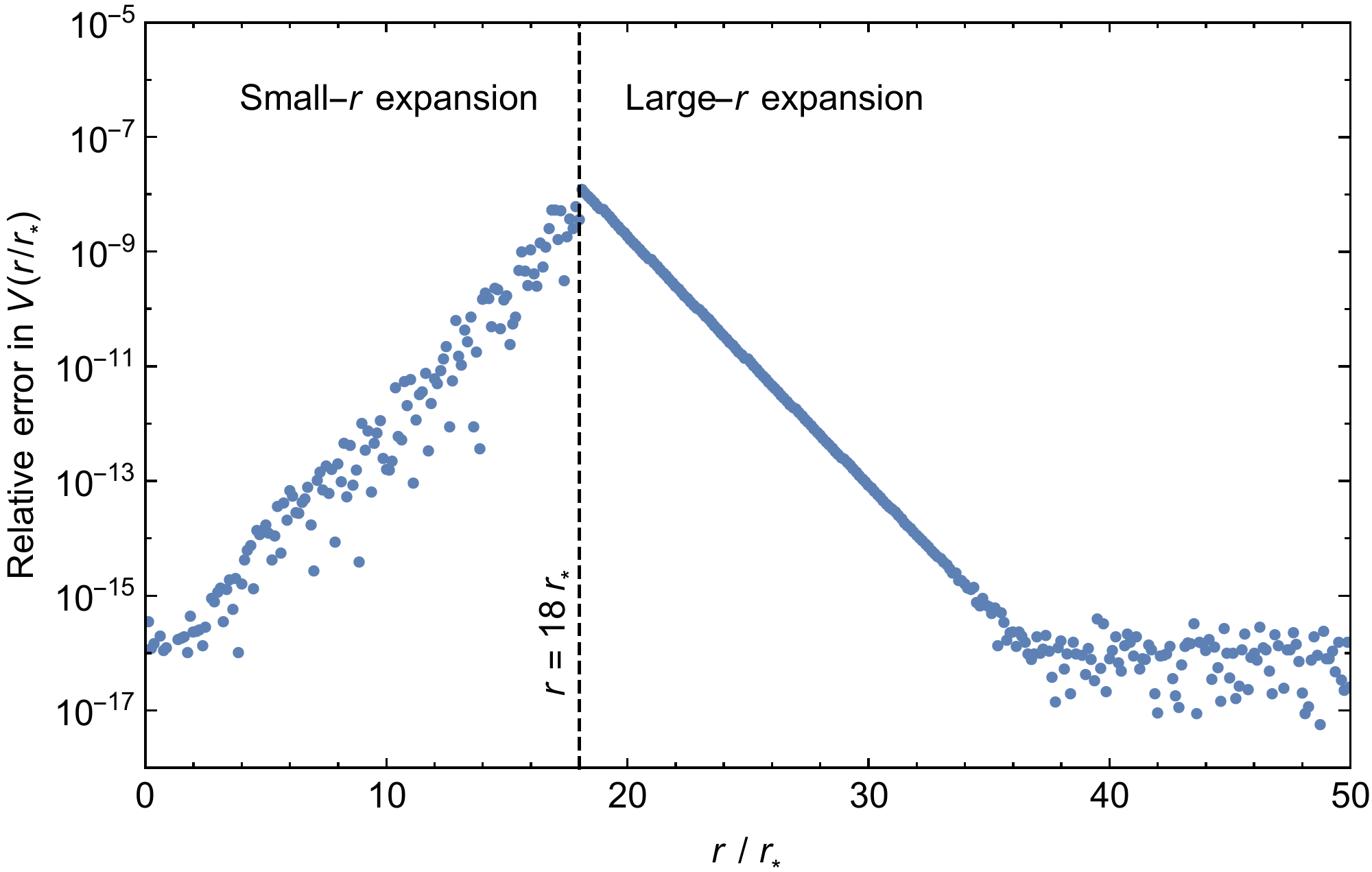}}
  \caption{Relative error in the numerical evaluation of the Keldysh
  interaction $V (r / r_{\ast})$ from
  Eq.~(\ref{eq:2-effinteraction}).\label{fig:2-effinterror}}
\end{figure}

\subsection{The Schr{\"o}dinger equation of a charge carrier complex}

Charge carriers in 2D semiconductors can be thought of as a set of quantum
particles in an effective potential given by Eq.~(\ref{eq:2-effinteraction}).
The Schr{\"o}dinger equation of this system is:
\begin{equation}
  \left[ - \sum_{i = 1}^N \frac{\hbar^2}{2 m_i} \nabla^2_i + \sum_{i = 1}^N
  \sum_{j = i + 1}^N \frac{q_i q_j}{4 \pi \varepsilon_0 r_{\ast}} V \left(
  \frac{r_{i j}}{r_{\ast}} \right) \right] \psi = E \psi, \label{eq:2-Schr}
\end{equation}
where the first sum is the kinetic energy and the second sum is the potential.
$N$ is the total number of charge carriers in the complex, $m_i \in \{ m_e,
m_h \}$ is the effective mass of an electron\footnote{\setstretch{1.66}This is 
not the bare
electron mass, which in this work will be designated $m_e^{\ast} \approx 9.1
\cdot 10^{- 31}$ kg. Rather, this is a~mass of a quasiparticle that arises as
the electron travels through the medium (we can call it a quasielectron) and
is given as a curvature of the conduction band minimum.} or a hole, $q_i$ is
the charge of the carrier ($| q_i | = e$), $V (x)$ is the effective
interaction given in Eq.~(\ref{eq:2-effinteraction}), $r_{i j}$ is the
distance between particle $i$ and particle $j$, $E$ is the energy of the
complex, and $\psi$ is the wave function of the system.

We adapt the following naming for the charge carriers and their complexes.
Constituent particles can be electrons e$^-$, holes h$^+$, and fixed classical
ions: positive donors D$^+$, or negative acceptors A$^-$. Excitons will be
designated as X, negative and positive trions as X$^-$ and X$^+$, and
biexcitons as XX. Donor- or acceptor-bound complexes will be written with
D$^+$ or A$^-$ in front of the corresponding complex symbol, \tmtextit{e.g.}
D$^+$XX is a donor-bound biexciton.

The binding energy of a complex is defined as the energy required to divide
the complex into two smaller complexes that are energetically the most
favourable, and are far away from each other. Table \ref{tab:2-decomposition}
shows the decomposition of various charge carrier complexes.

\begin{table}[h]
  \begin{tabular}{lccc}
    \hline\hline
    Complex & Symbol &  & Decomposition\\
    \hline
    Exciton & X & $\rightarrow$ & $\text{e}^- + \text{h}^+$\\
    Negative trion & X$^-$ & $\rightarrow$ & $\text{X} + \text{e}^-$\\
    Donor-bound exciton & D$^+$X & $\rightarrow$ & $\text{D}^+ \text{e}^- +
    \text{h}^+$\\
    Biexciton & XX & $\rightarrow$ & $\text{X} + \text{X}$\\
    Donor-bound negative trion & D$^+$X$^-$ & $\rightarrow$ & $\text{D}^+
    \text{e}^- + \text{X}$\\
    Donor-bound biexciton & D$^+$XX & $\rightarrow$ & $\text{D}^+ \text{X} +
    \text{X}$\\
    \hline\hline
  \end{tabular}
  \caption{Naming and decomposition of charge carrier
  complexes.\label{tab:2-decomposition}}
\end{table}

The binding energies are therefore defined as:
\begin{eqnarray}
  E_{\text{X}}^{\text{b}} & = & E_{\text{X}}, \\
  E_{\text{X}^-}^{\text{b}} & = & E_{\text{X}} - E_{\text{X}^-}, \\
  E_{\text{D}^+ \text{X}}^{\text{b}} & = & E_{\text{D}^+ \text{e}^-} -
  E_{\text{D}^+ \text{X}}, \\
  E_{\tmop{XX}}^{\text{b}} & = & 2 E_{\text{X}} - E_{\tmop{XX}}, \\
  E_{\text{D}^+ \text{X}^-}^{\text{b}} & = & E_{\text{D}^+ \text{e}^-} +
  E_{\text{X}} - E_{\text{D}^+ \text{X}^-}, \\
  E_{\text{D}^+ \text{XX}}^{\text{b}} & = & E_{\text{D}^+ \text{X}} +
  E_{\text{X}} - E_{\text{D}^+ \text{XX}} . 
\end{eqnarray}
The total energy of both an isolated electron and an isolated hole is zero.
Notice the convention\footnote{\setstretch{1.66}This convention is used in most 
of the
literature dealing with charge carrier complexes in 2D semiconductors.} used:
the binding energy of the exciton is a negative quantity, while other binding
energies are defined so that the binding energy is positive if the complex is
bound.

It will prove useful to do the following transformation of coordinates,
\begin{equation}
  \tilde{r} = r \sqrt{\frac{2 e^2 \mu}{\hbar^2 4 \pi \varepsilon_0 r_{\ast}}},
\end{equation}
where $\mu = m_e m_h / (m_e + m_h)$ is the reduced mass. The Schr{\"o}dinger
equation (\ref{eq:2-Schr}) becomes:
\begin{equation}
  \left[ - \sum_i \frac{\mu}{m_i} \tilde{\nabla}^2_i + \sum_{i, j > i}
  \bar{q}_i \bar{q}_j V \left( \tilde{r}_{i j} \sqrt{\frac{4 \pi \varepsilon_0
  \hbar^2}{2 e^2 \mu r_{\ast}}} \right) \right] \psi = \frac{4 \pi
  \varepsilon_0 r_{\ast}}{e^2} E \psi . \label{eq:2-Schr2}
\end{equation}
where $\bar{q}_i \equiv q_i / e$. We notice that $\mu / m_i$ is a
dimensionless quantity that is only dependent on the ratio of effective
masses, $m_e / m_h$. If $r_{\ast}$ is measured in the units of the excitonic
Bohr radius,
\begin{equation}
  a_{\text{B}}^{\ast} = \frac{4 \pi \varepsilon_0 \hbar^2}{\mu e^2},
  \label{eq:2-excBohrradius}
\end{equation}
then the Schr{\"o}dinger equation (\ref{eq:2-Schr2}) can be written as:
\begin{equation}
  \left[ - \sum_i \frac{\mu}{m_i} \tilde{\nabla}^2_i + \sum_{i, j > i}
  \bar{q}_i \bar{q}_j V \left( \frac{\tilde{r}_{i j}}{\sqrt{2 r_{\ast} /
  a_{\text{B}}^{\ast}}} \right) \right] \psi = \frac{4 \pi \varepsilon_0
  r_{\ast}}{e^2} E \psi . \label{eq:2-Schr3}
\end{equation}
We notice that the left-hand side of the equation is only dependent on two
dimensionless parameters: the mass ratio $m_e / m_h$ and the parameter related
to the susceptibility of the material, $r_{\ast} / a_{\text{B}}^{\ast}$. Both
the dimensionless parameters can range from zero to infinity, so we make the
following transformation, in order to present all the limits on one plot:
\begin{equation}
  \eta = \frac{m_e / m_h}{1 + m_e / m_h}, \qquad \nu = \frac{r_{\ast} /
  a_{\text{B}}^{\ast}}{1 + r_{\ast} / a_{\text{B}}^{\ast}},
\end{equation}
where $\eta$ is the rescaled mass ratio and $\nu$ is the rescaled
susceptibility. Both parameters take values in the $[0, 1]$ interval.

\subsubsection{Exciton complex}\label{ch:2-exciton}

For an exciton, the Schr{\"o}dinger equation (\ref{eq:2-Schr3}) will be
greatly simplified, if one uses the centre of mass as the origin of
coordinates, since then the kinetic part is
\begin{equation}
  - \sum_i \frac{\mu}{m_i} \tilde{\nabla}^2 = - \left( \frac{\mu}{m_e} +
  \frac{\mu}{m_h} \right) \tilde{\nabla}^2 = - \tilde{\nabla}^2,
  \label{eq:2-excitoncoord}
\end{equation}
where the Laplacian is taken with respect to the centre-of-mass coordinates.
The left-hand side of the Schr{\"o}dinger equation (\ref{eq:2-Schr3}) is then
only dependent on $r_{\ast} / a_{\text{B}}^{\ast}$ and independent of the mass
ratio, $m_e / m_h$.

The solution for the exciton complex in the Coulomb limit is known
{\cite{Kohn1955,Zaslow1967,Hassoun1981}} and the exciton energy is:
\begin{equation}
  E_{\text{X}} = - 4 \tmop{Ry}^{\ast} = - 2 \frac{e^2}{4 \pi \varepsilon_0
  a_{\text{B}}^{\ast}},
\end{equation}
where $\tmop{Ry}^{\ast}$ is the excitonic Rydberg energy, $\tmop{Ry}^{\ast} =
e^2 / \left( 2 (4 \pi \varepsilon_0) a_{\text{B}}^{\ast} \right)$. The
solution is identical to the two-dimensional hydrogen atom, and the excitonic
wave function can be evaluated as:
\begin{equation}
  \psi = \sqrt{\frac{2}{\pi}}  \frac{2}{a_{\text{B}}^{\ast}} e^{- 2 r /
  a_{\text{B}}^{\ast}} . \label{eq:2-Xwavefunction}
\end{equation}

\subsubsection{Logarithmic limit}

For the logarithmic limit, the potential is $V (x) = \log (2 / x) - \gamma$.
Thus, the Schr{\"o}dinger equation (\ref{eq:2-Schr3}) becomes:
\begin{equation}
  \left[ - \sum_i \frac{\mu}{m_i} \tilde{\nabla}^2_i + \sum_{i, j > i}
  \bar{q}_i \bar{q}_j \left( \log \left( \frac{2}{\tilde{r}_{i j}} \sqrt{2
  \frac{r_{\ast}}{a_{\text{B}}^{\ast}}} \right) - \gamma \right) \right] \psi
  = \frac{4 \pi \varepsilon_0 r_{\ast}}{e^2} E \psi,
\end{equation}
or
\begin{equation}
  \left[ - \sum_i \frac{\mu}{m_i} \tilde{\nabla}^2_i + \sum_{i, j > i}
  \bar{q}_i \bar{q}_j \left( \log \left( \frac{2}{\tilde{r}_{i j}} \right) +
  \frac{1}{2} \log \left( \frac{2 r_{\ast}}{a_{\text{B}}^{\ast}} \right) -
  \gamma \right) \right] \psi = \frac{4 \pi \varepsilon_0 r_{\ast}}{e^2} E
  \psi,
\end{equation}
or
\begin{equation}
  \left[ - \sum_i \frac{\mu}{m_i} \tilde{\nabla}^2_i + \sum_{i, j > i}
  \bar{q}_i \bar{q}_j \left( \log \frac{2}{\tilde{r}_{i j}} - \gamma \right)
  \right] \psi = \left( \frac{4 \pi \varepsilon_0 r_{\ast}}{e^2} E -
  \frac{1}{2}  \sum_{i, j > i} \bar{q}_i \bar{q}_j \log \frac{2
  r_{\ast}}{a_{\text{B}}^{\ast}} \right) \psi . \label{eq:2-logSchr}
\end{equation}
The left-hand side is now independent of $r_{\ast}$ and therefore the
right-hand side is a constant $C$, independent of $r_{\ast}$:
\begin{equation}
  C = \frac{4 \pi \varepsilon_0 r_{\ast}}{e^2} E - \frac{1}{2} \log \left(
  \frac{2 r_{\ast}}{a_{\text{B}}^{\ast}} \right) \sum_{i, j > i} \bar{q}_i
  \bar{q}_j .
\end{equation}
The energy of the complex can therefore be written as
\begin{equation}
  \frac{4 \pi \varepsilon_0}{e^2} E = \frac{C + \frac{1}{2} \left( \sum
  \bar{q}_i \bar{q}_j \right) \log \left( 2 r_{\ast} / a_{\text{B}}^{\ast}
  \right)}{r_{\ast}},
\end{equation}
which for $r_{\ast} \rightarrow \infty$ goes to zero. Additionally, one can
notice that
\[ \sum_{i, j > i} \bar{q}_i \bar{q}_j = \left\{ \begin{array}{ll}
     - 1 & \text{for neutral exciton},\\
     - 1 & \text{for positive or negative trion},\\
     - 2 & \text{for neutral biexciton},\\
     - 2 & \text{for donor- or acceptor-bound neutral biexciton},\\
     \frac{(n_1 - n_2)^2 - n_1 - n_2}{2} & \text{for complex with $n_1$
     positive and $n_2$ negative carriers} .
   \end{array} \right. \]
\begin{equation}
  \ 
\end{equation}
{\hspace{1.5em}}On the other hand, if one measures the energies in units
proportional to $1 / r_{\ast}$, the equation reads:
\begin{equation}
  \frac{4 \pi \varepsilon_0 r_{\ast}}{e^2} E = C + \frac{1}{2} \left( \sum
  \bar{q}_i \bar{q}_j \right) \log \frac{2 r_{\ast}}{a_{\text{B}}^{\ast}},
  \label{eq:2-logcontrib}
\end{equation}
which diverges as $r_{\ast} \rightarrow \infty$. However, for a binding energy
(other than an exciton), the logarithmic contribution to the energy will
cancel out. For example, in the case of a trion complex:
\[ \frac{4 \pi \varepsilon_0 r_{\ast}}{e^2}  \left( E_{\text{X}} -
   E_{\text{X}^-} \right) = \left( C_{\text{X}} + \frac{1}{2} (- 1) \log
   \frac{2 r_{\ast}}{a_{\text{B}}^{\ast}} \right) - \left( C_{\text{X}^-} +
   \frac{1}{2} (- 1) \log \frac{2 r_{\ast}}{a_{\text{B}}^{\ast}} \right) =
   C_{\text{X}} - C_{\text{X}^-} . \]
\begin{equation}
  \ 
\end{equation}
Therefore a natural unit to measure binding energy in the logarithmic limit is
\begin{equation}
  \frac{e^2}{4 \pi \varepsilon_0 r_{\ast}} .
\end{equation}
{\hspace{1.5em}}To summarise, in the logarithmic limit of large $r_{\ast}$,
the total and binding energies of complexes are zero, if one uses excitonic
units of energy similar to Hartree or Rydberg. On the other hand, in the units
of $e^2 / (4 \pi \varepsilon_0 r_{\ast})$, the total energy of the complexes
(and the binding energy of an exciton) will diverge, however the binding
energy of other complexes will be finite.

For an exciton, one can additionally notice that the left-side of
Eq.~(\ref{eq:2-logSchr}) is independent of the effective masses, due to
Eq.~(\ref{eq:2-excitoncoord}). Therefore, the constant $C_{\text{X}}$ is a
number, and the dependence of the excitonic energy on the mass in the
logarithmic limit is completely defined via the logarithmic contribution, $-
\frac{1}{2} \log \frac{r_{\ast}}{a_{\text{B}}^{\ast}}$.

\subsubsection{Units of the energy}

In order to present our results, so that all the energy values are finite, the
following choice of units is made. The excitonic binding energy will be
measured in the units of
\begin{equation}
  \frac{e^2}{4 \pi \varepsilon_0 a_{\text{B}}^{\ast}} = 2 \tmop{Ry}^{\ast} .
\end{equation}
Although the logarithmic limit is lost and will be zero in these units, if one
determines the constant $C_{\text{X}}$, the energy dependence on the mass
ratio will be then given by Eq.~(\ref{eq:2-logcontrib}), which has now the
following form:
\begin{equation}
  E_{\text{X}} = \frac{e^2}{4 \pi \varepsilon_0 r_{\ast}} \left( C_{\text{X}}
  - \frac{1}{2} \log \frac{2 r_{\ast}}{a_{\text{B}}^{\ast}} \right) .
  \label{eq:2-logexcitonbehaviour}
\end{equation}
{\hspace{1.5em}}In the case of any other complex, the binding energy will be
measured in the following unit
\begin{equation}
  \frac{e^2}{4 \pi \varepsilon_0  \left( r_{\ast} + a_{\text{B}}^{\ast}
  \right)}, \label{eq:2-unitE}
\end{equation}
which for the Coulomb limit $(r_{\ast} \rightarrow 0)$ is equal to $2
\tmop{Ry}^{\ast}$ and for the logarithmic limit is equal to the natural unit
of the logarithmic interaction, $e^2 / (4 \pi \varepsilon_0 r_{\ast})$.
Finally, the Schr{\"o}dinger equation (\ref{eq:2-Schr3}) attains the
dimensionless form:
\begin{equation}
  \left[ - \sum_{i = 1}^{N_e} (1 - \eta) \tilde{\nabla}^2_i - \sum_{i =
  1}^{N_h} \eta \tilde{\nabla}^2_i + \sum_{\langle i, j \rangle} \bar{q}_i
  \bar{q}_j V \left( \tilde{r}_{ij} \sqrt{\frac{1 - \nu}{2 \nu}} \right)
  \right] \psi = \nu \mathcal{E} \psi,
\end{equation}
where $\mathcal{E}$ is the energy measured in unit (\ref{eq:2-unitE}), $N_e$
and $N_h$ are the number of electrons and number of holes in the system
respectively.

\subsubsection{Electron-hole symmetry}

We also notice that the Schr{\"o}dinger equation (\ref{eq:2-Schr}) is
symmetric under electron--hole exchange, if we also switch the effective
masses,
\begin{equation}
  \tilde{m}_e = m_h, \hspace{4em} \tilde{m}_h = m_e, \label{eq:2-chconj}
\end{equation}
where $\tilde{m}_e$ and $\tilde{m}_h$ are the effective masses in the
conjugated system. Therefore, a negative trion system with effective masses
$m_e$ and $m_h$ is equivalent to the positive trion with effective masses
$\tilde{m}_e$ and $\tilde{m}_h$. Table \ref{tab:2-conj} presents a selection
of charge carrier complexes and their conjugate equivalents.

\begin{table}[h]
  \centering
  \begin{tabular}{llll}
    \hline\hline
    \multirow{ 2}{*}{Complex} & \multirow{ 2}{*}{Conjugated system} & 
    \multicolumn{2}{c}{Extreme 
    mass ratios } \\
    \cline{3-4}
    & & $m_e / m_h \rightarrow 0$ & $m_e / m_h \rightarrow
    \infty$\\
    \hline
    {\color[HTML]{4F81BD}X} & {\color[HTML]{4F81BD}X} &
    {\color[HTML]{C0504D}D$^+$e$^-$ $\boxdot$} &
    {\color[HTML]{C0504D}A$^-$h$^+$ $\boxdot$}\\
    X$^-$ & X$^+$ & D$^+$e$^-$e$^-$ & {\color[HTML]{76923C}A$^-$A$^-$h$^+$
    $\boxast$}\\
    D$^+$X & A$^-$X & {\color[HTML]{76923C}D$^+$D$^+$e$^-$ $\boxast$} &
    D$^+$A$^-$h$^+$ $\rightarrow$ h$^+$\\
    {\color[HTML]{4F81BD}XX} & {\color[HTML]{4F81BD}XX} &
    {\color[HTML]{8064A2}D$^+$D$^+$e$^-$e$^-$ $\boxbox$} &
    {\color[HTML]{8064A2}A$^-$A$^-$h$^+$h$^+$ $\boxbox$}\\
    D$^+$X$^-$ & A$^-$X$^+$ & {\color[HTML]{8064A2}D$^+$D$^+$e$^-$e$^-$
    $\boxbox$} & D$^+$A$^-$A$^-$h$^+$\\
    D$^+$XX & A$^-$XX & D$^+$D$^+$D$^+$e$^-$e$^-$ &
    D$^+$A$^-$A$^-$h$^+$h$^+$\\
    \hline\hline
  \end{tabular}
  \caption{Charge carrier complexes and their conjugates. Complexes in blue
  are their own conjugates. Extreme mass ratio limits are also shown, with
  coloured symbols indicating which complexes are
  equivalent.\label{tab:2-conj}}
\end{table}

In the dimensionless units, the rescaled mass ratio in the conjugated system
can be obtained:
\begin{equation}
  \tilde{\eta} = 1 - \eta .
\end{equation}
{\hspace{1.5em}}Therefore, we immediately see that the exciton and biexciton
complexes must be symmetric under electron--hole exchange, since they are
their own charge conjugates. It is thus sufficient to investigate mass ratios
$m_e / m_h \in [0, 1]$ or $\eta = [0, 0.5]$ in order to cover the whole
parameter space.

\subsubsection{Extreme mass ratios}\label{ch:2-extrememass}

If one of the effective masses of quasiparticles is much larger than the other
one, we may treat the massive particle as an ion, or a fixed particle. Such a
particle will not have any kinetic energy and will only interact via the
potential energy. For extreme mass ratios, we replace electrons with negative
acceptors or holes with positive donors.
\begin{eqnarray}
  \text{e}^- & \xrightarrow[m_e / m_h \rightarrow \infty]{} &
  \text{A}^- \\
  \text{h}^+ & \xrightarrow[m_e / m_h \rightarrow 0]{} & \text{D}^+  
\end{eqnarray}
Table \ref{tab:2-conj} shows how complexes will look at extreme mass ratios.

We notice that if we have multiple fixed ions of the same electric charge in
the complex, we will have to find the geometry of the system that minimises
the total energy. However, if the complex consists of fixed particles with
both negative and positive charges, more thought is needed, since a donor and
an acceptor will want to overlap, which would cause total energy to diverge.

Let us consider donor-bound complexes, where we take the mass ratio to
infinity, causing the electrons to be very heavy, namely D$^+$X, D$^+$X$^-$
and D$^+$XX. We need to separate three mass scales: infinite mass of a donor,
heavy mass of an electron and light mass of a hole.

We can write the binding energy of a donor-bound exciton as:
\begin{equation}
  E^{\text{b}}_{\text{D}^+ \text{X}} = E_{\text{D}^+ \text{e}^-} -
  E_{\text{$\text{D}^+ \text{X}$}}  \xrightarrow[m_e / m_h \rightarrow 
  \infty]{} E \left( \text{$\text{D}^+
  \text{e}$}_{\tmop{heavy}}^- \right) - E \left( \text{$\text{D}^+ \text{e}
  _{\tmop{heavy}}^-$h}_{\tmop{light}}^+ \right) .
\end{equation}
So this binding energy is that of a hole bound to a D$^+$A$^-$ complex.
However, the hole will see an effective potential of zero, since the donor and
the acceptor will coalesce.

The binding energy of a donor-bound negative trion is:
\begin{eqnarray}
  E^{\text{b}}_{\text{D}^+ \text{X}^-} & \xrightarrow[m_e / m_h \rightarrow 
  \infty]{} & E \left( \text{D}^+ \text{e}_{\tmop{heavy}}^-
  \right) + E \left( \text{e}_{\tmop{heavy}}^- \text{h}_{\tmop{light}}^+
  \right) - E \left( \text{D}^+ \text{e}_{\tmop{heavy}}^-
  \text{e}_{\tmop{heavy}}^- \text{h}_{\tmop{light}}^+ \right) \\
  & = & E \left( \text{D}^+ \text{e}_{\tmop{heavy}}^- \right) + E \left(
  \text{e}_{\tmop{heavy}}^- \text{h}_{\tmop{light}}^+ \right) \nonumber\\
  &  & - \left[ \underset{}{} \right. E \left( \text{D}^+
  \text{e}_{\tmop{heavy}}^- \right) + \underbrace{E \left( \text{D}^+
  \text{e}_{\tmop{heavy}}^- \text{e}_{\tmop{heavy}}^- \right) - E \left(
  \text{D}^+ \text{e}_{\tmop{heavy}}^- \right)}_{= - E^{\text{b}} \left(
  \text{D}^+ \text{e}^- \text{e}^- \right)} \nonumber\\
  &  & + \underbrace{E \left( \text{D}^+ \text{e}_{\tmop{heavy}}^-
  \text{e}_{\tmop{heavy}}^- \text{h}_{\tmop{light}}^+ \right) - E \left(
  \text{D}^+ \text{e}_{\tmop{heavy}}^- \text{e}_{\tmop{heavy}}^- \right)}_{= E
  \left( \text{A}^- \text{h}^+ \right) \text{, since hole will only see one
  negative charge}} \left. \underset{}{} \right] \nonumber\\
  & = & E^{\text{b}}_{\text{D}^+ \text{e}^- \text{e}^-} . \nonumber
\end{eqnarray}
However, we need to remember to convert between the units used for the
D$^+$e$^-$e$^-$ complex, $e^2 / {\nobreak} \left( 4 \pi \varepsilon_0 \hbar^2
\left( r_{\ast} + a_{\text{B}}^{(e)} \right) \right)$, where
$a_{\text{B}}^{(e)}$ is the quasielectron Bohr radius, and the units of energy
from Eq.~(\ref{eq:2-unitE}). We notice that:
\begin{equation}
  \frac{e^2}{4 \pi \varepsilon_0 \hbar^2 \left( r_{\ast} + a_{\text{B}}^{\ast}
  \right)}  \xrightarrow[m_e / m_h \rightarrow \infty]{} 
  \frac{e^2}{4 \pi \varepsilon_0 \hbar^2 \left( r_{\ast} + a_{\text{B}}^{(h)}
  \right)} = \frac{e^2}{4 \pi \varepsilon_0 \hbar^2 \left( r_{\ast} +
  a_{\text{B}}^{(e)} \right)}  \frac{\left( r_{\ast} + a_{\text{B}}^{(e)}
  \right)}{\left( r_{\ast} + a_{\text{B}}^{(h)} \right)},
\end{equation}
where $a_{\text{B}}^{(h)}$ is the quasihole Bohr radius. The quasielectron
Bohr radius is zero, since it is inversely proportional to the quasielectron
mass, and thus we need to multiply the binding energy of the D$^+$e$^-$e$^-$
complex by $r_{\ast} / \left( r_{\ast} + a_{\text{B}}^{(h)} \right)$ in order
to recover the binding energy of the D$^+$A$^-$A$^-$h$^+$ complex in the units
from Eq.~(\ref{eq:2-unitE}).

Similarly, for a donor-bound biexciton, the binding energy is:
\begin{eqnarray}
  E^{\text{b}}_{\text{D}^+ \tmop{XX}} & \xrightarrow[m_e / m_h \rightarrow 
  \infty]{} & E \left( \text{D}^+ \text{e}_{\tmop{heavy}}^-
  \text{h}_{\tmop{light}}^+ \right) + E \left( \text{e}_{\tmop{heavy}}^-
  \text{h}_{\tmop{light}}^+ \right) - E \left( \small{\text{D}^+
  \text{e}_{\tmop{heavy}}^- \text{e}_{\tmop{heavy}}^-
  \text{h}_{\tmop{light}}^+ \text{h}_{\tmop{light}}^+} \right) \\
  & = & E \left( \text{D}^+ \text{e}_{\tmop{heavy}}^-
  \text{h}_{\tmop{light}}^+ \right) + E \left( \text{e}_{\tmop{heavy}}^-
  \text{h}_{\tmop{light}}^+ \right) \nonumber\\
  &  & - \left[ \underset{}{} \right. E \left( \text{D}^+
  \text{e}_{\tmop{heavy}}^- \right) + \underbrace{E \left( \text{D}^+
  \text{e}_{\tmop{heavy}}^- \text{e}_{\tmop{heavy}}^- \right) - E \left(
  \text{D}^+ \text{e}_{\tmop{heavy}}^- \right)}_{= - E^{\text{b}} \left(
  \text{D}^+ \text{e}^- \text{e}^- \right)} \nonumber\\
  &  & + \underbrace{E \left( \text{D}^+ \text{e}_{\tmop{heavy}}^-
  \text{e}_{\tmop{heavy}}^- \text{h}_{\tmop{light}}^+ \right) - E \left(
  \text{D}^+ \text{e}_{\tmop{heavy}}^- \text{e}_{\tmop{heavy}}^- \right)}_{= E
  \left( \text{A}^- \text{h}^+ \right) \text{, since hole will only see one
  negative charge}} \nonumber\\
  &  & + \underbrace{E \left( \text{D}^+ \text{e}_{\tmop{heavy}}^-
  \text{e}_{\tmop{heavy}}^- \text{h}_{\tmop{light}}^+
  \text{h}_{\tmop{light}}^+ \right) - E \left( \text{D}^+
  \text{e}_{\tmop{heavy}}^- \text{e}_{\tmop{heavy}}^-
  \text{h}_{\tmop{light}}^+ \right)}_{= - E^{\text{b}} \left( \text{A}^-
  \text{h}^+ \text{h}^+ \right) \text{, since holes will again only see one
  negative charge}} \left. \underset{}{} \right] \nonumber\\
  & = & \underbrace{E \left( \text{D}^+ \text{e}_{\tmop{heavy}}^-
  \text{h}_{\tmop{light}}^+ \right) - E \left( \text{D}^+
  \text{e}_{\tmop{heavy}}^- \right)}_{= - E^{\text{b}}_{\text{D}^+ \text{A}^-
  \text{h}^+} = 0 \text{, as discussed above}} + E^{\text{b}} \left(
  \text{D}^+ \text{e}^- \text{e}^- \right) + E^{\text{b}} \left( \text{A}^-
  \text{h}^+ \text{h}^+ \right) \nonumber\\
  & = & E^{\text{b}} \left( \text{D}^+ \text{e}^- \text{e}^- \right) +
  E^{\text{b}} \left( \text{A}^- \text{h}^+ \text{h}^+ \right) . \nonumber
\end{eqnarray}
{\hspace{1.5em}}For a near extreme mass, we notice that if the complex
consists of two heavy particles with the same charge, we can use the
Born-Oppenheimer approximation {\cite{Born1927}}{\nocite{WilsonBook}} to
determine the dependence of the binding energy on the mass ratio. For example,
in the limit of $m_e / m_h \rightarrow \infty$ a negative trion will resemble
a~two-dimensional H$_2^+$ ion and a biexciton will resemble a~two-dimensional
H$_2$ molecule. The correction to the binding energy can be thus approximated
as the harmonic zero\mbox{-}point energy of the ion-ion vibrations. We expand
the potential energy near the minimum $r_0$,
\begin{equation}
  U (r) = U (r_0) + \frac{1}{2} U'' (r_0)  (r - r_0)^2 + O ((r - r_0)^3),
\end{equation}
and we identify the second-order correction as the vibrational energy of the
ions,
\begin{equation}
  \frac{1}{2} U'' (r_0)  (r - r_0)^2 = \frac{\nobracket p^2 |_{r = r_0}}{2
  \mu'} = \frac{\mu' \omega^2 (r - r_0)^2}{2},
\end{equation}
where $\mu'$ is the reduced mass of the pair of daughter complexes (see
Tab.~\ref{tab:2-decomposition}). The binding energy can be therefore written
as
\begin{equation}
  E = U (r_0) + \frac{\hbar \omega}{2} = U (r_0) + \frac{\hbar}{2}
  \sqrt{\frac{U'' (r_0)}{\mu'}} .
\end{equation}
Let us now assume that the hole mass is infinite, while the electron mass is
finite. We can write the potential in Rydberg units as $U (r) =\mathcal{U}
\left( r / a_{\text{B}} \right) \tmop{Ry}$, where $\mathcal{U}$ is a
dimensionless quantity. Therefore:
\begin{eqnarray}
  E & = & \mathcal{U} (r_0) \tmop{Ry} + \frac{\hbar}{2}
  \sqrt{\frac{\mathcal{U}'' \left( r_0 / a_{\text{B}} \right) \tmop{Ry}}{\mu'
  a_{\text{B}}^2}} \\
  & = & \left[ \mathcal{U} (r_0) \tmop{Ry} \frac{4 \pi \varepsilon_0 \left(
  r_{\ast} + a_{\text{B}}^{\ast} \right)}{e^2} + \sqrt{\frac{\hbar^2
  \mathcal{U}'' \left( r_0 / a_{\text{B}} \right) \tmop{Ry}}{4 \mu'
  a_{\text{B}}^2} \left( \frac{4 \pi \varepsilon_0 \left( r_{\ast} +
  a_{\text{B}}^{\ast} \right)}{e^2} \right)^2} \right]  \frac{e^2}{4 \pi
  \varepsilon_0 \left( r_{\ast} + a_{\text{B}}^{\ast} \right)} . \nonumber
\end{eqnarray}
Since Rydberg energy is proportional to electron mass, the electron (exciton)
Bohr radius is inversely proportional to the electron (exciton reduced) mass,
and $m_e / \mu' \approx m_e / m_h$ in the limit of infinite hole mass, then
the correction to the energy is proportional to
\begin{equation}
  E - U (r_0) \sim \sqrt{\frac{m_e}{m_h}}  \frac{e^2}{4 \pi \varepsilon_0
  \left( r_{\ast} + a_{\text{B}}^{\ast} \right)} .
\end{equation}
\chapter{Details of quantum Monte Carlo simulations}\label{ch:2-method}

\section{Introduction to Monte Carlo methods}

The problem of solving a many-body quantum system reliably and accurately is a
long standing problem of physics. Analytical solutions only exist in some
special cases, or in approximate conditions. With the advancement of the
world's computational power, numerical simulation seems to be the method of
choice for numerous studies. However, many computational methods, such as
exact diagonalisation, suffer from not being truly parallelisable and thus do
not use the full potential of now-common parallel machines and supercomputers.

Monte Carlo methods are algorithms based on repeated random updates of the
system {\cite{Balescu}}. The broad idea is that we can generate many
realisations of the same system using a~specific probability distribution. If
the ``randomness'' in the system is tailored 
correctly\footnote{\setstretch{1.66}\textit{I.e.} the
balance equation is satisfied.}, then by the law of large numbers, we can
calculate the physical observables by taking an average over all generated
copies of the system. The key aspect of most of the Monte Carlo methods is
that because the final results are calculated as averages, instead of
performing one large simulation, one can do separate uncorrelated simulations
at the same time to arrive at the answer. Thus, the methods are almost
perfectly parallelisable.

Here, we are interested in the Monte Carlo methods that deal with quantum
many-body systems and bear a general name of quantum Monte Carlo (QMC). These
methods will rely heavily on evaluations of multi-dimensional integrals, which
can be efficiently done using Monte Carlo 
integration\footnote{\setstretch{1.66}Actually, for
any multi-dimensional integral, Monte Carlo is always a method of choice,
unless one can simplify the integral to four or less dimensions.}. In this
chapter we briefly summarise the main ideas of two quantum Monte Carlo
methods, variational Monte Carlo and diffusion Monte Carlo. More detailed
explanation can be found, \tmtextit{e.g.}, in Ref.~{\cite{Hammond}}. All our
calculation were performed using the {\tmname{casino}} code
{\cite{Needs2010}}.

\section{Variational Monte Carlo}

The variational Monte Carlo (VMC) method relies mainly on the variational
principle of quantum mechanics {\cite{Zettili2009}}. If we define the
following functional,
\begin{equation}
  E [\psi (\vec{R})] = \frac{\int \psi^{\ast} (\vec{R}) H \psi (\vec{R})
  \mathd \vec{R}}{\int \psi^{\ast} (\vec{R}) \psi (\vec{R}) \mathd \vec{R}},
  \label{eq:2-varprinc}
\end{equation}
where $\psi (\vec{R})$ is a wave function of the system, and $H$ is its
Hamiltonian, then the variational principle states that $E [\psi (\vec{R})]$
reaches its global minimum only if $\psi (\vec{R})$ is exactly equal to the
ground state of the system, and then $E [\psi (\vec{R})]$ is equal to the
ground state energy. Otherwise, $E [\psi (\vec{R})]$ can only have values
higher than the ground state energy.

The total energy of the system can be evaluated using the equation
\begin{equation}
  E = \langle H \rangle = \frac{\int E_L (\vec{R}) | \psi_{\tmop{trial}}
  (\vec{R}) |^2 \mathd \vec{R}}{\int | \psi_{\tmop{trial}} (\vec{R}) |^2
  \mathd \vec{R}},
\end{equation}
where $\psi_{\tmop{trial}} (\vec{R})$ is a trial wave function -- a best guess
of the wave function of the system we want to simulate, and $E_L (\vec{R}) =
\frac{1}{\psi_{\tmop{trial}} (\vec{R})} H (\vec{R}) \psi_{\tmop{trial}}
(\vec{R})$ is the local energy. One configuration $\vec{R}$ is a~set of
positions of all particles: $\vec{R} = \{ \vec{r}_1, \ldots,
\vec{r}_{N_{\tmop{particles}}} \}$. In order to calculate this expression, we
must generate configurations, so that they are distributed according to $|
\psi_{\tmop{trial}} |^2$. Then, the total energy of the system is simply an
average of local energies of every configuration.

In order to generate configurations with a given distribution, Metropolis
updates {\cite{Metropolis1953}} of the configuration of the system are
performed. Every step in the simulation consists of proposing a move of one of
the particles and randomly accepting or rejecting this move. The probability
of acceptance of a move from $\vec{R}$ to $\vec{R}'$ is given by
\begin{equation}
  p (\vec{R}' \leftarrow \vec{R}) = \min \left\{ 1, \frac{|
  \psi_{\tmop{trial}} (\vec{R}') |^2}{| \psi_{\tmop{trial}} (\vec{R}) |^2}
  \right\} .
\end{equation}
Here we use a symmetric transition probability density. After the Metropolis
algorithm reaches an equilibrium, we can start accumulating the results. The
local energy does not need to be evaluated at every step. Additionally, the
configurations may be serially correlated with each other, so we need to wait
a number of moves between accumulation steps, so that the measured
configurations are independent.

The variance of the (Gaussian) Metropolis transition probability density is
referred to as the VMC time step $\tau_{\tmop{VMC}}$ and has to be carefully
chosen. If $\tau_{\tmop{VMC}}$ is too large, then many moves will be rejected
and thus the generated configurations will be greatly correlated. On the other
hand, if $\tau_{\tmop{VMC}}$ is too small, then the new configuration will not
be much different from the old one, which again leads to unwanted correlation
in configuration space. The time step therefore must be optimised, so that we
avoid any serial correlation: usually the value of $\tau_{\tmop{VMC}}$ is
chosen so that the acceptance probability is equal to 50\%.

\subsection{Wave function optimisation}

The accuracy of calculating the total energy relies on a good choice of
$\psi_{\tmop{trial}}$, the trial wave function. This trial wave function can
be either imported, for example, from an electronic structure calculation such
as density functional theory (DFT), or one can devise the form of
$\psi_{\tmop{trial}}$. Any knowledge of the system in question is crucial in
developing an applicable formula for the trial wave function. After having a
guess at $\psi_{\tmop{trial}}$, one needs then to optimise the unknown
parameters incorporated in the form of $\psi_{\tmop{trial}}$, in the spirit of
the variational principle from Eq.~(\ref{eq:2-varprinc}). Here, evaluation of
the energy expectation of the system using the Metropolis algorithm is
especially useful, due to its low computational requirements. One can quickly
generate configurations using the trial wave function, calculate their energy,
propose a change in the parameters in the trial wave function, generate new
configurations and calculate the new energy of the system for comparison with
the old one. This optimisation method is called \tmtextit{energy minimisation}
{\cite{Umrigar2007}}. To minimise the energy we diagonalise the Hamiltonian
matrix in the basis defined by the initial wave function and its derivatives
with respect to the parameter values. The matrix elements are calculated using
VMC.

In some cases however, we need a more robust method, as energy minimisation
may converge too slowly in our optimisation, or may have reached a local
minimum. \tmtextit{Variance minimisation} {\cite{Umrigar1988,Neil2005}} is
preferred in this case, where we minimise the variance of the energy,
\begin{equation}
  \sigma^2 = \frac{\int | \psi_{\tmop{trial}} (\vec{R}) |^2  | E_L (\vec{R}) -
  E |^2 \mathd \vec{R}}{\int | \psi_{\tmop{trial}} (\vec{R}) |^2 \mathd
  \vec{R}} .
\end{equation}
Here the idea is that if $\psi_{\tmop{trial}}$ was exact, then any
configuration would have the same local energy, and therefore the variance of
the energy would be zero.

Generally in all our calculations, we have firstly used variance minimisation
in order to make sure that we reach the vicinity of the global minimum of the
energy. Secondly, energy minimisation was used to further pinpoint the values
of optimisable parameters. In cases with a large number of particles in
a~complex such as a~donor-bound biexciton, variance minimisation was found to
be sometimes unreliable (the energy diverged during optimisation), and only
energy minimisation could be used.

\subsection{Numerical errors}\label{ch:2-reblock}

The VMC method suffers from two prominent sources of error. The statistical
error comes from the fact that we accumulate only a finite number of
configurations $N_{\tmop{config}}$, but can be lowered by simply increasing
$N_{\tmop{config}}$. One can easily see that this error goes as $\sim 1 /
\sqrt{N_{\tmop{config}}}$, just by using the error propagation formula for
averages. Due to configurations being correlated, there is also a problem with
estimating the value of the standard error for the energy mean. In order to
remove this issue, we perform a~reblocking analysis: we gather the results in
blocks and average them separately in every block. For too small blocks, the
standard error estimate is too small due to serial correlation. If the block
is big enough (the block length is bigger than the correlation length), then
the averages of different blocks are uncorrelated and thus the variances are
unbiased. Therefore, by plotting the error estimates against the block size,
we should see a plateau for high enough block sizes. Of course, if the block
size is too big, then the statistical error in the standard error estimate is
large since there are not enough blocks for accurate estimation. A good choice
of block size is obtained after reaching the plateau, but before the error in
the estimate of the standard error becomes too large.

Secondly, we have an error coming from the choice of the trial wave function,
which is usually the major issue. Our chosen trial wave function is presented
in Chapter~\ref{ch:2-trialWF}.

\section{Diffusion Monte Carlo}

Let us start by recalling the Schr{\"o}dinger equation for a general many-body
system, that includes a~constant energy shift, $E_0$,
\begin{equation}
  i \hbar \frac{\partial}{\partial t} \psi (\vec{R}, t) = - \sum_i
  \frac{\hbar^2}{2 m_i} \nabla^2 \psi (\vec{R}, t) + (V (\vec{R}) - E_0) \psi
  (\vec{R}, t) .
\end{equation}
We can rewrite it using the imaginary time, $\tau = it$. The Schr{\"o}dinger
equation now reads:
\begin{equation}
  \hbar \frac{\partial}{\partial t} \psi (\vec{R}, \tau) = \begin{array}{l}
    \sum^{}_i \frac{\hbar^2}{2 m_i} \nabla^2 \psi (\vec{R}, \tau)
  \end{array} + \begin{array}{l}
    (E_0 - V (\vec{R})) \psi (\vec{R}, \tau)
  \end{array} . \label{eq:2-imagtime}
\end{equation}
Now, if we omit the second term of the right-hand side of
Eq.~(\ref{eq:2-imagtime}), we get the following equation:
\begin{equation}
  \frac{\partial}{\partial t} \psi (\vec{R}, \tau) = \begin{array}{l}
    \sum^{}_i D_i \nabla^2 \psi (\vec{R}, \tau)
  \end{array}, \label{eq:2-diffusion}
\end{equation}
where $D_i = \hbar / 2 m_i$. This is a diffusion equation with diffusion
constant $D_i$ and $\psi (\vec{r}, \tau)$ the density of the diffusing
particles. On the other hand, if we omit the first term on the right-hand side
of Eq.~(\ref{eq:2-imagtime}), the equation becomes
\begin{equation}
  \frac{\partial}{\partial t} \psi (\vec{R}, \tau) = \begin{array}{l}
    k (\vec{R}) \psi (\vec{R}, \tau)
  \end{array} . \label{eq:2-branching}
\end{equation}
This is a first order rate equation or branching process with rate constant $k
(\vec{R})$ equal to $k (\vec{R}) = (E_0 - V (\vec{R})) / \hbar$. Notice that
if the reference energy, $E_0$, is exactly equal to the ground state energy of
the system, then $\psi (\vec{R}, \tau)$ is constant, or the density of the
diffusing particles remains the same. Additionally, if $\tau \rightarrow
\infty$, then the density of diffusing particles becomes independent of $\tau$
if $\psi = \psi_0$, where $\psi_0$ is the ground-state energy eigenfunction.

Therefore, in diffusion Monte Carlo (DMC) method, two processes are used (see
Fig.~\ref{fig:2-dmc}): a~diffusion process, in which we diffuse the
configurations according to Eq.~(\ref{eq:2-diffusion}), and a branching
process, in which we randomly duplicate or destroy configurations according to
the rate equation (\ref{eq:2-branching}). The main idea is that the average
over the distribution of configurations after the equilibration of the
diffusion process, will produce a distribution that will reflect exactly the
ground state wave function.

\begin{figure}[h]
  \resizebox{10cm}{!}{\includegraphics{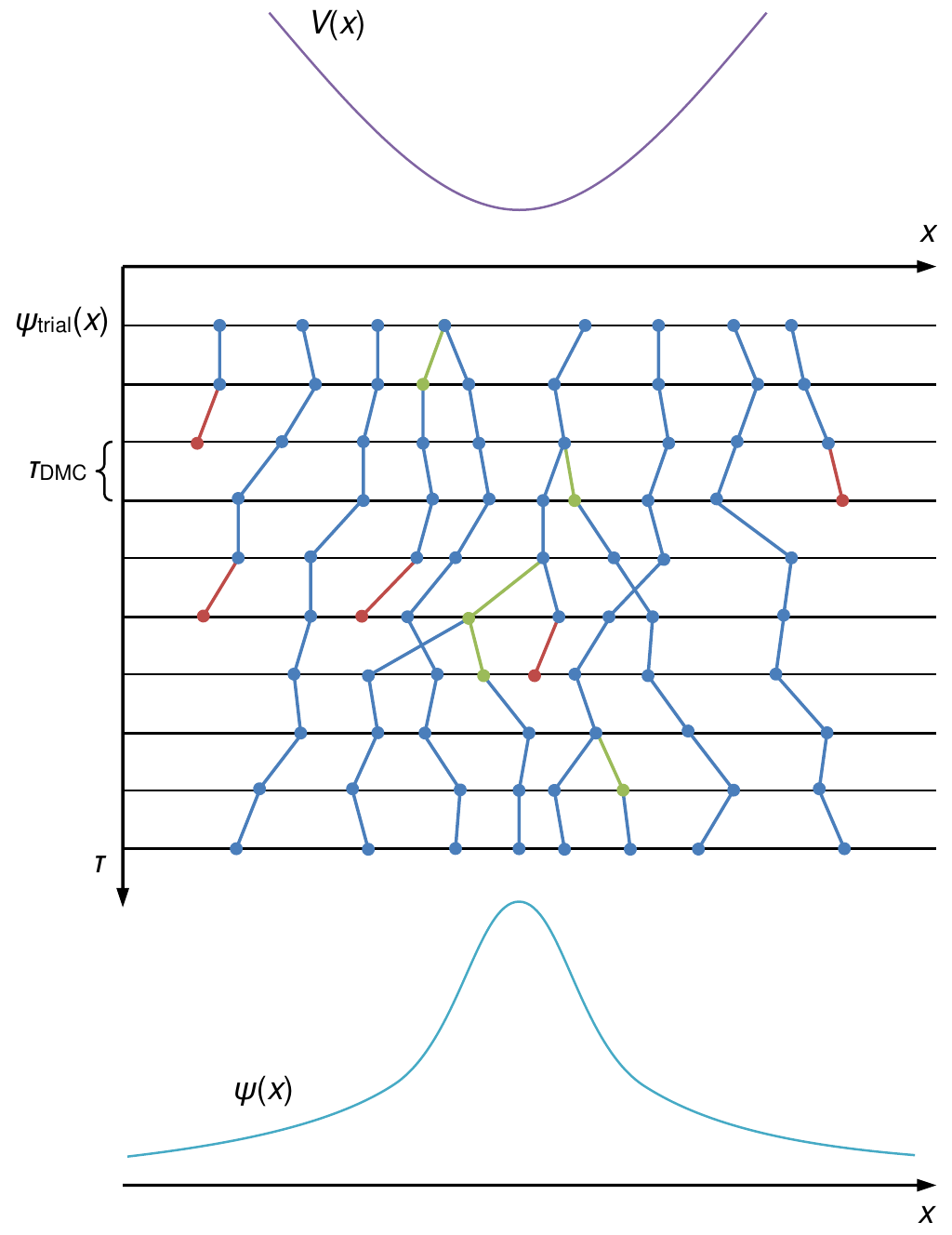}}{\hspace{2.5em}}
  \caption{Schematic view of the DMC method in a one-parameter system (adapted
  from Ref.~{\cite{Foulkes2001}}): the trial wave function
  $\psi_{\tmop{trial}}$ is used to generate the initial population, then the
  population is allowed to diffuse (blue) and the branching process will
  either create (green) or destroy (red) the configurations. Averaging over
  the distribution of configurations gives us the (statistically exact) ground
  state energy.\label{fig:2-dmc}}
\end{figure}

However, if the potential is very large and negative, the branching process
from Eq.~(\ref{eq:2-branching}) will induce a large fluctuation in the number
of configurations, and therefore a large uncertainty in $\psi (\vec{R}, \tau)$
and in estimating the energy. A solution to this problem is to introduce
importance sampling, in which we define an importance-sampled wave function
(or a mixed wave function),
\begin{equation}
  f (\vec{R}, \tau) = \psi (\vec{R}, \tau) \psi_{\tmop{trial}} (\vec{R}, \tau)
  .
\end{equation}
Then, the Schr{\"o}dinger equation (\ref{eq:2-imagtime}) becomes:
\begin{equation}
  \hbar \frac{\partial f}{\partial \tau} = \sum_i \frac{\hbar^2}{2 m_i}
  [\nabla^2_i f - \nabla_i \cdot (\vec{F}_i (\vec{R}) f)] + (E_0 - E_L
  (\vec{R})) f,
\end{equation}
where $\vec{F}_i (\vec{R}) = \psi_{\tmop{trial}}^{- 1} \nabla_i
\psi_{\tmop{trial}}$ is the drift velocity. Now, if the trial wave function is
close to the ground state, the local energy $E_L (\vec{R})$ is a smooth,
well-behaved function, and the branching term will not be too large to cause
huge fluctuations in the number of configurations {\cite{Foulkes2001}}.
Therefore, for the DMC method to work well, we need a good initial guess of
the trial wave function. This will be realised by optimising the trial wave
function using VMC.

In order to determine the probabilities of the diffusion and branching
processes, we write the Schr{\"o}dinger equation using the Green's function:
\begin{equation}
  f (\vec{R}', \tau + \tau_{\tmop{DMC}}) = \int f (\vec{R}, \tau) G (\vec{R}'
  \leftarrow \vec{R}, \tau_{\tmop{DMC}}) \mathd \vec{R},
\end{equation}
where $\tau_{\tmop{DMC}}$ is the DMC time step, and $G (\vec{R}' \leftarrow
\vec{R}, \tau_{\tmop{DMC}})$ is the Green's function. Its approximate form can
be chosen to be {\cite{Vrbik1986}}:
\begin{equation}
  G (\vec{R}' \leftarrow \vec{R}, \tau_{\tmop{DMC}}) = G_B (\vec{R}'
  \leftarrow \vec{R}, \tau_{\tmop{DMC}}) G_D (\vec{R}' \leftarrow \vec{R},
  \tau_{\tmop{DMC}}) + O (\tau_{\tmop{DMC}}^2), \label{eq:2-greensfunc}
\end{equation}
where $G_B$ and $G_D$ are Green's functions of the branching and diffusion
processes respectively,
\begin{eqnarray}
  G_B (\vec{R}' \leftarrow \vec{R}, \tau_{\tmop{DMC}}) & = & \exp \left( -
  \frac{\tau_{\tmop{DMC}}}{2 \hbar} (E_L (\vec{R}) + E_L (\vec{R}') - 2 E_0)
  \right), \\
  G_D (\vec{R}' \leftarrow \vec{R}, \tau_{\tmop{DMC}}) & = & \frac{1}{(2 \pi
  \tau_{\tmop{DMC}})^{dN_{\tmop{particles}} / 2}} \exp \left( - \frac{|
  \vec{R} - \vec{R}' - \tau_{\tmop{DMC}}  \vec{F} (\vec{R}') |^2}{2
  \tau_{\tmop{DMC}}} \right), 
\end{eqnarray}
where $d$ is the dimensionality of the system. $G_B$ and $G_D$ define the rate
of duplicating or destroying the configurations and the rate of the diffusion
process. The approximate form of $G$ must always be chosen so that $G$ becomes
exact as $\tau_{\tmop{DMC}} \rightarrow 0$ and the error in $G$ falls off
faster than linear in $\tau_{\tmop{DMC}}$.

\subsection{Statistical, time step and population
errors}\label{ch:2-dmcerrors}

Statistical error is introduced in every Monte Carlo method and comes from the
central limit theorem. It is proportional to $1 / \sqrt{N_{\tmop{accum}}}$,
where $N_{\tmop{accum}}$ is the number of accumulation steps that occur after
the equilibration. To accurately estimate the standard error in the mean of
the results, we use the reblocking method described in
Chapter~\ref{ch:2-reblock}.

In order to make sure that the population of configurations does not become
too small or too large, the reference energy $E_0$ will be varied during the
calculation -- this process is called population control. However, this will
introduce a small bias in the configurations, which will slightly promote
configurations with higher energies. This error is inversely proportional to
the population\footnote{\setstretch{1.66}The error in the energy, will be $1 /
\sqrt{N_{\tmop{population}}}$ due to central limit theorem, and the weights of
each configuration are inversely proportional to the local energy of the
configuration, and thus the total error in weighting will also have errors of
$1 / \sqrt{N_{\tmop{population}}}$.}, and thus is it necessary to extrapolate
the results to infinite population.

There is one additional error, introduced in Eq.~(\ref{eq:2-greensfunc}),
which is due to non-zero DMC time step. Similarly, we have to extract the
zero-time-step value of the energy, if we want to reduce this error. The
easiest way to extrapolate our results to zero time step and infinite
population at the same time is to perform two DMC calculations for two
different time steps with population varied in inverse proportion to time step
and then simultaneously extract the zero-time-step and
infinite\mbox{-}population value.

To assess the relevant time step $\tau_{\tmop{DMC}}$ and the number of
equilibration steps $N_{\tmop{equil}}$, we keep in mind two characteristic
lengths in the charge carrier complex system. One is the size of the exciton
in the logarithmic limit of $r_{\ast} \rightarrow \infty$, which is equal to:
\begin{equation}
  r_0 = \sqrt{\frac{\hbar^2 r_{\ast}}{2 e^2 \mu}},
\end{equation}
and the second one is the excitonic Bohr radius from
Eq.~(\ref{eq:2-excBohrradius}), which describes the size of the exciton in the
Coulomb limit.

On the other hand, during the DMC calculation, we can consider a root mean
squared diffusive distance $d_{\tmop{rms}}$, which is defined as:
\begin{equation}
  d_{\tmop{rms}} = \sqrt{\frac{2 N_{\tmop{DMC}} \tau_{\tmop{DMC}}}{m}},
\end{equation}
for a particle with mass $m$ that was diffused through $N_{\tmop{DMC}}$ steps
during the calculation. For the time step to be small enough, this distance
for one single step must be smaller than the smallest length scale in the
system. Similarly, the number of equilibration steps must be chosen so that
the diffusive distance after the equilibration procedure is always larger than
the longest length scale.

Therefore, one can use the following inequalities to determine the time step
and the number of equilibration steps:
\begin{eqnarray}
  \sqrt{\frac{2 \tau_{\tmop{DMC}}}{m_{\tmop{lightest}}}} & \ll &
  L_{\tmop{smallest}}, \label{eq:2-dtdmc} \\
  \sqrt{\frac{2 N_{\tmop{equil}} \tau_{\tmop{DMC}}}{m_{\tmop{heaviest}}}} &
  \gg & L_{\tmop{longest}}, \label{eq:2-Nequil} 
\end{eqnarray}
where $m_{\tmop{lightest}}$ and $m_{\tmop{heaviest}}$ are respectively the
effective masses of the lighter and heavier of the two species of particles in
the system (either $m_e$ or $m_h$), and $L_{\tmop{smallest}}$ and
$L_{\tmop{longest}}$ are the smallest and longest length scales in the system
(either $r_0$ or $a_{\text{B}}^{\ast}$).

The DMC method also experiences errors due to the fixed node approximation.
However, in our study the particles are distinguishable, which means there are
no nodes in the ground\mbox{-}state wave function of this system.

After removing both non-zero time-step and finite population errors, the only
errors that are left are the statistical ones. Therefore, we can say that our
DMC energy is statistically exact, \tmtextit{i.e.} one reaches the true ground
state energy of the system in the limit of infinite accumulation steps.

\subsection{Observables}\label{ch:2-extrapolatedestimation}

Finally, we show how to calculate the energy and other observables in the DMC
method. The ground state energy can be written as a mixed-estimator:
\begin{equation}
  E = \frac{E \langle \psi | \psi_{\tmop{trial}} \rangle}{\langle \psi |
  \psi_{\tmop{trial}} \rangle} = \frac{\langle \psi | H | \psi_{\tmop{trial}}
  \rangle}{\langle \psi | \psi_{\tmop{trial}} \rangle} = \frac{\int fE_L
  \mathd \vec{R}}{\int f \mathd \vec{R}},
\end{equation}
assuming the DMC wave function $\psi$ is exactly equal to the ground state.

For observables that do not commute with the Hamiltonian, we can use two
averages, the VMC and the DMC mixed average:
\begin{eqnarray}
  \langle A \rangle_{\tmop{VMC}} & = & \frac{\langle \psi_{\tmop{trial}} | A |
  \psi_{\tmop{trial}} \rangle}{\langle \psi_{\tmop{trial}} |
  \psi_{\tmop{trial}} \rangle} = \langle A \rangle + A' [\psi -
  \psi_{\tmop{trial}}] + O [(\psi - \psi_{\tmop{trial}})^2], \\
  \langle A \rangle_{\tmop{DMC}} & = & \frac{\langle \psi_{\tmop{trial}} | A |
  \psi \rangle}{\langle \psi_{\tmop{trial}} | \psi \rangle} = \langle A
  \rangle + 2 A' [\psi - \psi_{\tmop{trial}}] + O [(\psi -
  \psi_{\tmop{trial}})^2], 
\end{eqnarray}
where $A' [\psi]$ is a functional that occurs during the expansions in $(\psi
- \psi_{\tmop{trial}})$. Comparing both expansions, we extract the following
equation for the extrapolated estimator {\cite{Ceperley1979}},
\begin{equation}
  \langle A \rangle = 2 \langle A \rangle_{\tmop{DMC}} - \langle A
  \rangle_{\tmop{VMC}} + O [(\psi - \psi_{\tmop{trial}})^2],
\end{equation}
which is correct up to the second order in $(\psi - \psi_{\tmop{trial}})$.

\section{Trial wave function}\label{ch:2-trialWF}

In this study, the only complexes considered consist of distinguishable
particles, \tmtextit{i.e.} there is always a~quantum number that is different
for each particle. Therefore, the wave function of the complex will always be
symmetric under the exchange of two particles and will not have any
antisymmetric part or nodes. The trial wave function can therefore be written
in the Jastrow form:
\begin{equation}
  \psi_{\tmop{trial}} (\vec{R}) = \exp (J (\vec{R})),
\end{equation}
where $J (\vec{R})$ is a Jastrow factor, in a form proposed in
Ref.~{\cite{Neil2004}},
\begin{eqnarray}
  J (\{ \vec{r}_i \}, \{ \vec{r}_I \}) & = & \sum_{i = 1}^{N - 1} \sum_{j = i
  + 1}^N u (r_{i j}) + \sum_{i = 1}^N \sum_{I = 1}^{N_{\tmop{ions}}} \chi_I
  (r_{i I}) \\
  & + & \sum_{i = 1}^{N - 2} \sum_{j = i + 1}^{N - 1} \sum_{k = j + 1}^N h
  (r_{i j}, r_{i k}, r_{j k}) \nonumber\\
  & + & \sum_{i = 1}^{N - 1} \sum_{j = i + 1}^N \sum_{I =
  1}^{N_{\tmop{ions}}} f_I (r_{i I}, r_{j I}, r_{i j}) \nonumber\\
  & + & \sum_{i = 1}^{N + N_{\tmop{ions}} - 1} \sum_{j = i + 1}^{N +
  N_{\tmop{ions}}} u_{\text{EX2D}} (r_{i j}), \nonumber
\end{eqnarray}
with $N$ being the number of fermions in the system and $N_{\tmop{ions}}$
being the number of ions, or fixed particles (they only enter the
Schr{\"o}dinger equation through the potential term). Terms in the Jastrow
factor are:
\begin{itemize}
  \item $u$ term, describing correlation between two fermions,
  
  \item $h$ term, describing correlation between three fermions,
  
  \item $\chi$ term, describing correlation between an ion (fixed particle)
  and one fermion,
  
  \item $f$ term, describing correlation between an ion and two fermions,
  
  \item $u_{\text{EX2D}}$ term, that imposes cusp conditions relevant for the
  interaction in 2D semiconductors.
\end{itemize}
Terms $u, h, \chi$ and $f$ have the form {\cite{Neil2004,Lopez2012}} of a
general polynomial expansion in $r$, which goes to zero at a~cutoff length
specific to the term used. These truncated polynomials are continuous and have
continuous first and second derivatives even at the cutoff point, ensuring
that the gradient of the term and the local energy $E_L$ are both continuous.
Forms of the $u_{\text{EX2D}}$ term will be introduced in a subsection below.

Of course, for a given complex, only the relevant terms will be used. For
example, for a~trion, $u, h$ and $u_{\text{EX2D}}$ terms will be included in
$\psi_{\tmop{trial}}$. To additionally simplify the problem, we assume that
there is no spin dependence in the terms, \tmtextit{e.g.} the e$^{\uparrow}$-h
interaction in a trion will be the same as the e$^{\downarrow}$-h interaction;
the e$^{\uparrow}$-e$^{\downarrow}$-h$^{\uparrow}$ interaction in a biexciton
will be the same as the e$^{\uparrow}$-e$^{\downarrow}$-h$^{\downarrow}$
interaction.

\subsection{Kato cusp conditions}

The Kato cusp conditions {\cite{Kato1957,Pack1966}} are conditions that the
wave function must satisfy in order to make sure that the local energy is
nondivergent at zero distance, when two charges coalesce. The local energy can
be written as:
\begin{equation}
  E_L = \frac{H \psi_{\tmop{trial}}}{\psi_{\tmop{trial}}},
\end{equation}
where $H$ is the Hamiltonian of the system. For the system of charge
complexes, we simply require that
\[ \lim_{r \rightarrow 0} \left[ - \frac{\hbar^2 (m_1 + m_2)}{2 m_1 m_2
   \psi_{\tmop{trial}}} \left( \frac{\partial^2 \psi_{\tmop{trial}}}{\partial
   r^2} + \frac{1}{r} \frac{\partial \psi_{\tmop{trial}}}{\partial r} \right)
   + \frac{q_1 q_2}{4 \pi \varepsilon_0 r_{\ast}}  \frac{\pi}{2} \left( H_0
   \left( \frac{r}{r_{\ast}} \right) - Y_0 \left( \frac{r}{r_{\ast}} \right)
   \right) \right] = \tmop{const} . \]
\begin{equation}
  \ 
\end{equation}
Notice that in general, $\underset{r \rightarrow 0}{\lim} E_L$ may not exist.

\subsection{Devising the $u_{\text{EX2D}}$ term}

The $u_{\text{EX2D}}$ term needs to satisfy two conditions: firstly, we need
to satisfy the Kato cusp conditions for small $r$, and secondly, the wave
function must fall to zero at large $r$.

Initially the following form of the $u_{\text{EX2D}}$ term was used:
\begin{eqnarray}
  u_{\text{EX2D}}^{\tmop{eh}} (r) & = & \lambda_{\tmop{eh}} r^2 \log (r) e^{-
  c_1 r^2} - c_2 r (1 - e^{- c_1 r^2}),  \label{5-trialwf1-1}\\
  u_{\text{EX2D}}^{\tmop{ee}} (r) & = & \lambda_{\tmop{ee}} r^2 \log (r) e^{-
  c_3 r^2},  \label{5-trialwf1-2}
\end{eqnarray}
where $c_1, c_2$ and $c_3$ are optimisable parameters, and
$\lambda_{\tmop{eh}}$ and $\lambda_{\tmop{ee}}$ are fixed by the Kato cusp
conditions:
\begin{equation}
  \lambda_{\tmop{eh}} = \frac{e^2 \mu}{2 (4 \pi \varepsilon_0) \hbar^2
  r_{\ast}}, \qquad \lambda_{\tmop{ee}} = - \frac{e^2 m_e}{4 (4 \pi
  \varepsilon_0) \hbar^2 r_{\ast}} . \label{eq:2-katocuspparam}
\end{equation}
Also, in order for $u_{\text{EX2D}}^{\tmop{eh}} (r)$ not to diverge, we must
have: $c_1 > 0, c_2 > 0$ and $c_3 > 0$. This form was used for exciton, trion
and donor-bound exciton complexes.

A second form was devised:
\begin{equation}
  u_{\text{EX2D}} (r) = \frac{\lambda r^2 \log r + c_1 r^2 + c_2 r^3}{1 + c_3
  r^2}, \label{eq:5-trialwf2}
\end{equation}
where $c_1, c_2$ and $c_3$ are optimisable parameters, and $\lambda$ is fixed
by the Kato cusp conditions similarly to Eq.~(\ref{eq:2-katocuspparam}):
\begin{equation}
  \lambda = - \frac{q_1 q_2 m_1 m_2}{2 (4 \pi \varepsilon_0) \hbar^2 r_{\ast}
  (m_1 + m_2)} .
\end{equation}
In order to make sure that the term does not diverge as $r_{\ast} \rightarrow
\infty$, the following conditions must be applied: $c_2 < 0$ and $c_3 > 0$.
The second form of the term was found to be much easier to optimise,
especially near the limits of extreme effective mass ratio of the complex.

In case of the purely Coulomb interaction, the following term was used:
\begin{equation}
  u_{\text{EX2D}} = \frac{\lambda r + c_1 r^2}{1 + c_2 r},
\end{equation}
where $c_1$ and $c_2$ are optimisable parameters and $\lambda$ is fixed by the
Kato cusp condition:
\begin{equation}
  \lambda = \frac{2 q_1 q_2 m_1 m_2}{4 \pi \varepsilon_0 \hbar^2} .
\end{equation}
In order to make the term non-divergent for $r_{\ast} \rightarrow \infty$, we
use $c_1 < 0$ and $c_2 > 0$.

\subsection{Kimball cusp conditions}

The Kimball cusp conditions {\cite{Kimball1973}} are analogues of the Kato
cusp conditions, but for the pair correlation function. Because the pair
correlation function is proportional to the wave function squared,
\begin{equation}
  g (r) \sim \psi_{\tmop{trial}}^2,
\end{equation}
we may easily use the Kato cusp conditions to determine the behaviour of $g
(r)$. As in the case of the wave function, we express the pair correlation in
exponential form:
\begin{equation}
  g (r) = \exp [\tilde{g} (r)] .
\end{equation}
The expansion of $\tilde{g} (r)$ near $r \rightarrow 0$ using either the first
form of the wave function from Eqs.~(\ref{5-trialwf1-1}--\ref{5-trialwf1-2})
or the second form of the wave function from Eq.~(\ref{eq:5-trialwf2}) was
calculated to be:
\begin{equation}
  \tilde{g} (r) = a_0 + 2 \lambda r^2 \log r + a_2 r^2 + a_3 r^3 + \cdots .
  \label{eq:2-kimball1}
\end{equation}
In case of the Coulomb interaction, the behaviour of $\tilde{g} (r)$ is
\begin{equation}
  \tilde{g} (r) = a_0 + 2 \lambda r + a_2 r^2 + \cdots . \label{eq:2-kimball2}
\end{equation}
\chapter{Results}\label{ch:2-results}

\section{Classification of trions and biexcitons in transition metal 
  dichalcogenides}

In monolayer TMDCs the conduction-band minimum and valence-band maximum occur
at the $K$ and $K'$ points of the hexagonal Brillouin zone. In molybdenum
dichalcogenides (MoX\tmrsub{2}), within each valley the valence-band maximum
has the same spin as the conduction-band minimum, while in tungsten
dichalcogenides (WX\tmrsub{2}) such states have opposite spins
{\cite{Kormanyos2015}}. Figure~\ref{fig:2-trionclassification}a presents
examples of how negative trions can be formed in both MoX\tmrsub{2} and
WX\tmrsub{2}, while Fig.~\ref{classification}a presents similar examples for
biexcitons.

To classify possible trionic and biexcitonic complexes, we will use the
following notation: the symbol $\text{T}_{k_3 s_3}^{k_1 s_1 k_2 s_2}$
designates a negative trion consisting of conduction band electrons in valleys
$k_1$ and $k_2$ and with spins $s_1$ and $s_2$ respectively; and of a valence
band hole in valley $k_3$ and with spin $s_3$. Similarly, $\text{XX}_{k_3 s_3
k_4 s_4}^{k_1 s_1 k_2 s_2}$ denotes a biexciton. For example, both trions in
Fig.~\ref{fig:2-trionclassification}a can be written as $\text{T}_{K
\downarrow}^{K \downarrow K' \uparrow}$, while both biexcitons shown in
Fig.~\ref{classification}a can be designated as $\text{XX}_{K \downarrow K'
\uparrow}^{K \downarrow K' \uparrow}$.

Recombination of trions and biexcitons may be prevented if there is no
electron--hole pair with the same spin -- such a complex is called a
\tmtextit{dark} trion or biexciton. Otherwise, we are dealing with a
\tmtextit{bright} complex, for which after a finite amount of time, it will
recombine to a free electron (for a trion) or to an exciton (for a biexciton).
The dark complexes may recombine only in a higher order process, which will
have a much lower rate of recombination and will result in the emission of
multiple photons -- there will be no well-defined light frequency.

The complexes in which multiple charge carriers of the same spin occupy the
same band will not be considered in this work. Fermions in the system would
not be distinguishable (\tmtextit{i.e.} would not have unique quantum
numbers), and such complexes ($e.g.$ $\text{T}_{K \downarrow}^{K \downarrow K
\downarrow}$) would be very weakly bound or not bound at
all.\footnote{\setstretch{1.66}Biexcitons with indistinguishable particles were 
analysed by
Elaheh Mostaani and they were found to be unbound for most of the $m_e / m_h$
and $r_{\ast}$ values.}

\begin{figure}[h]
  {\includegraphics{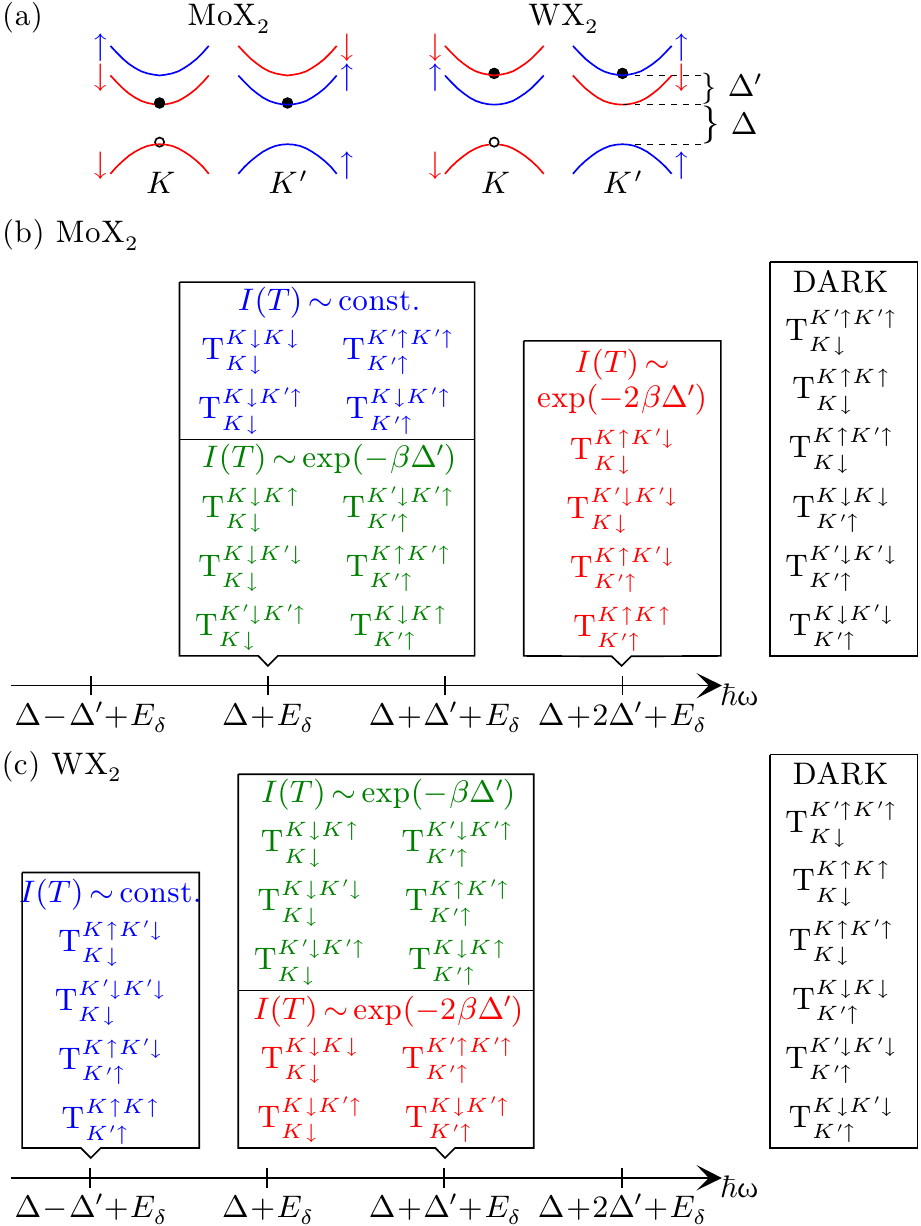}}
  \caption{Classification of trion recombination processes. Here $E_{\delta} =
  E_{\text{T}} = E_{\text{X}} - E_{\text{T}}^{\text{b}}$ is the difference
  between the exciton and trion binding
  energies.\label{fig:2-trionclassification}}
\end{figure}

Figures \ref{fig:2-trionclassification}b and \ref{fig:2-trionclassification}c
show the photon energies for bright trions in MoX\tmrsub{2} and WX\tmrsub{2},
respectively. Notice how the difference in spin polarisation of energy bands
in molybdenum- and tungsten-based materials changes the classification of
complexes. The precise photon energies depend on whether the electrons occupy
the higher- or lower-energy spin-split bands in the initial and final states.
Also, due to energy-momentum conservation, some recombination processes will
involve momentum exchange between the two electrons (\tmtextit{e.g.}, the
bright trion line for $\text{T}^{K' \downarrow K \uparrow}_{K \downarrow}$ in
MoX\tmrsub{$2$} corresponds to a final state in which there is a~single
spin-up electron in the $K'$ valley).

\begin{figure}[h!]
  {\includegraphics{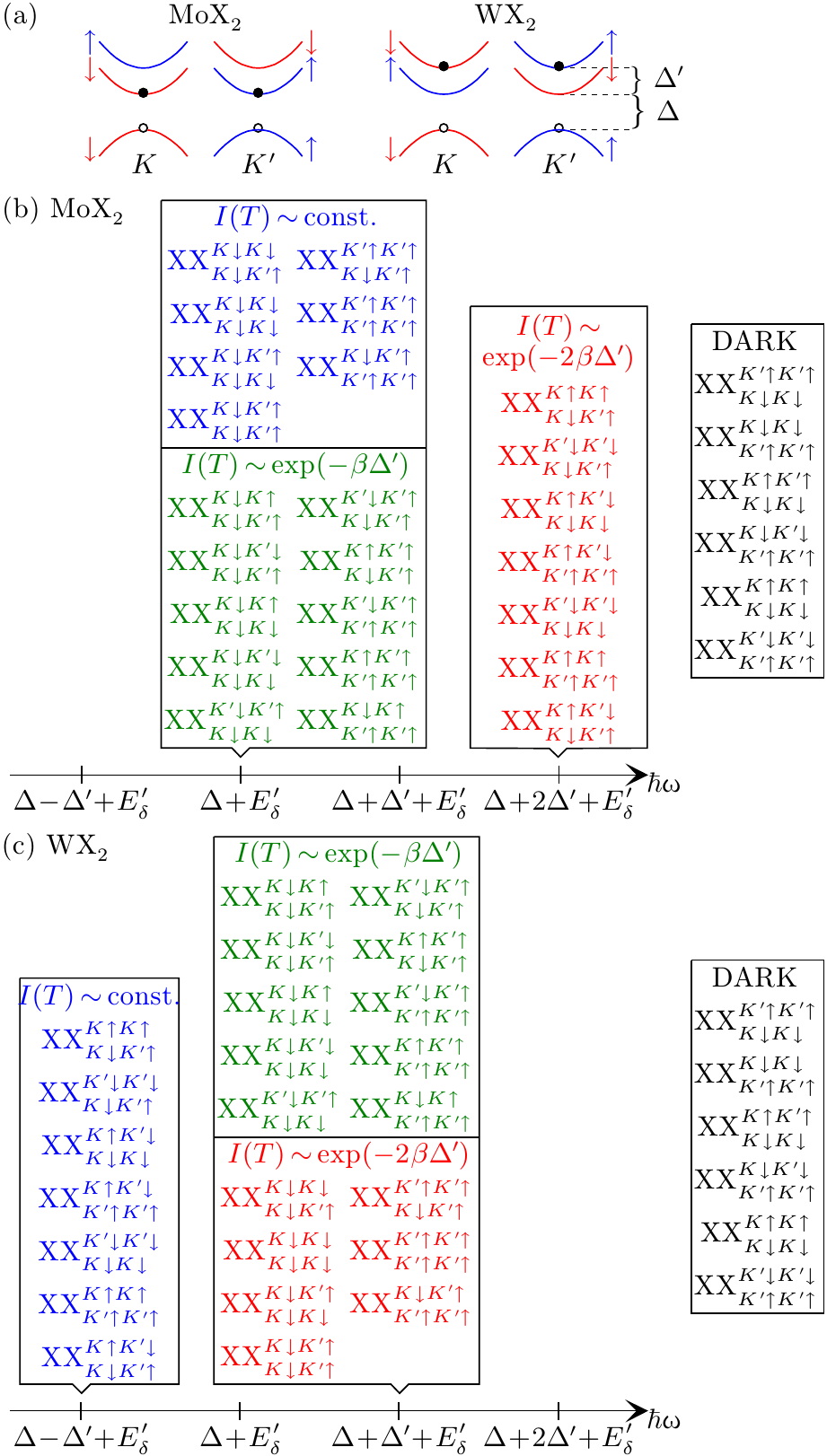}}
  \caption{(a) Difference in polarisation in molybdenum- and tungsten-based
  dichalcogenides in biexciton formation. (b,c) Classification of biexciton
  recombination processes in $\tmop{MoX}_2$ and $\tmop{WX}_2$. $E_{\delta}' =
  E_{\tmop{XX}} - E_{\text{X}} = E_{\text{X}} - E_{\text{XX}}^{\text{b}}$ is
  the difference between the binding energies of exciton and
  biexciton.\label{classification}}
\end{figure}

Furthermore, the intensity of a bright trion line depends on the thermal
occupancy of the initial state. This intensity has the following temperature
dependence,
\begin{equation}
  I (T) \sim \left\{ \begin{array}{ll}
    \tmcolor{blue}{\tmop{const} .} & \text{for no electrons in the upper
    spin-splitting conduction band,}\\
    {\color[HTML]{008000}e^{- \beta \Delta'}} & \text{for one electron in the
    upper spin-splitting conduction band},\\
    \tmcolor{red}{e^{- 2 \beta \Delta'}} & \text{for two electrons in the
    upper spin-splitting conduction band,}
  \end{array} \right.
\end{equation}
where $\beta = 1 / \left( k_{\text{B}} T \right)$, $k_{\text{B}}$ is the
Boltzmann's constant, $T$ is the temperature, and $\Delta'$ is the spin-orbit
induced splitting in the conduction band. For example, we expect the intensity
of the photoemission line for $\text{T}^{K' \downarrow K \uparrow}_{K
\downarrow}$ in MoS\tmrsub{2} at low temperature to be much lower than that of
$\text{T}^{K \downarrow K' \uparrow}_{K \downarrow}$, due to the thermal
suppression coming from the $e^{- 2 \beta \Delta'}$ factor.

Similarly, Figures~\ref{classification}b and \ref{classification}c present a
classification of biexcitons in $\tmop{MoX}_2$ and $\tmop{WX}_2$ with respect
to the recombination energy and the intensity of the emitted photon.

Finally, judging from this classification, one can predict possible lines and
their intensities on the absorption spectrum of a transition-metal
dichalcogenide monolayer. In a~photoabsorption or photoluminescence
experiment, we expect to see energies attributed to different kinds of trions
and biexcitons and emission lines of varying intensity, as presented in
Fig.~\ref{spectrum}.

\begin{figure}[h]
  {\includegraphics{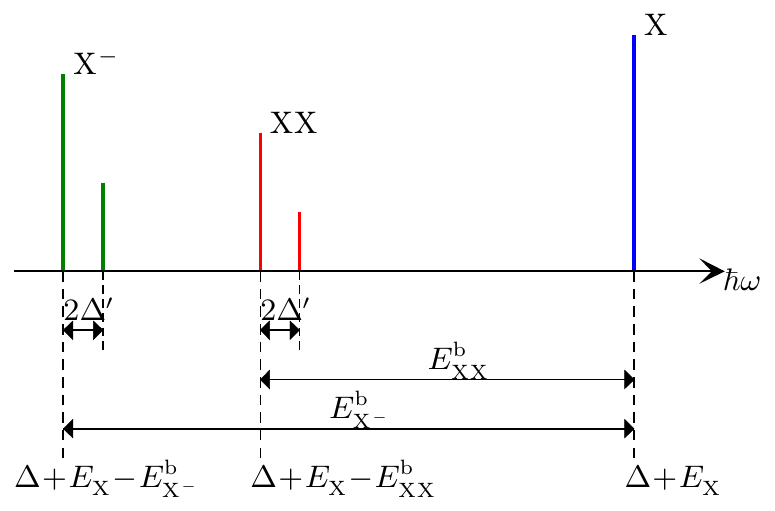}}
  \caption{Expected photoemission/photoabsorption spectrum showing lines for
  various complexes in $\tmop{MoX}_2$. For WX$_2$, the recombination energies
  will be slightly different ($\Delta \rightarrow \Delta -
  \Delta'$).\label{spectrum}}
\end{figure}

\section{Numerical setup}

Every charge carrier complex is defined by providing two parameters: the
effective mass ratio $m_e / m_h$, and the parameter $r_{\ast}$ related to the
in-plane susceptibility. We chose the following set of possible values of
$r_{\ast}$, which would cover different scales between the Coulomb limit
($r_{\ast} \rightarrow 0$) and the logarithmic limit ($r_{\ast} \rightarrow
\infty$):
\begin{equation}
  r_{\ast} / a_{\text{B}} \in \{ 0.1, 0.2, 0.5, 1.0, 2.0, 4.0, 6.0, 8.0 \},
\end{equation}
where $a_B = 4 \pi \varepsilon_0 \hbar^2 / (m_e^{\ast} e^2)$ is the electron
Bohr radius, with $m_e^{\ast}$ being the bare electron mass ($m_e^{\ast}
\approx 9.1 \cdot 10^{- 31}$ kg, as opposed to the effective electron
mass)\footnote{\setstretch{1.66}Our usage of atomic units here was due to the 
{\tmname{casino}}
implementation.}. On the other hand, the possible values of the effective
masses were chosen to be:
\begin{equation}
  m_e / m_e^{\ast} = 1, \qquad m_h / m_e^{\ast} \in \left\{ \frac{1}{0.1},
  \frac{1}{0.2}, \ldots, \frac{1}{0.9}, 1 \right\},
\end{equation}
so that the mass ratios are $m_e / m_h = \{ 0.1, \ldots, 1.0 \}$.

Additionally, we also simulate the conjugated system, with $m_e \rightarrow
m_h$ and $m_h \rightarrow m_e$. The limits of the pure Coulomb and pure
logarithmic interactions were treated separately by using either the pure
Coulomb or the pure logarithmic potential. The limits of extreme mass ratios
were also considered separately, by changing appropriate fermions to fixed
particles.

The final chosen grid of parameters is presented in
Fig.~\ref{fig:2-spaceparams-atomic}. However, because one would ideally want
to use excitonic units, instead of atomic units, this grid will be skewed, as
presented in Fig.~\ref{fig:2-spaceparams}.

\begin{figure}[h!]
  \centering
  \begin{minipage}{.48\textwidth}
    \centering
    \includegraphics[width=.9\linewidth]{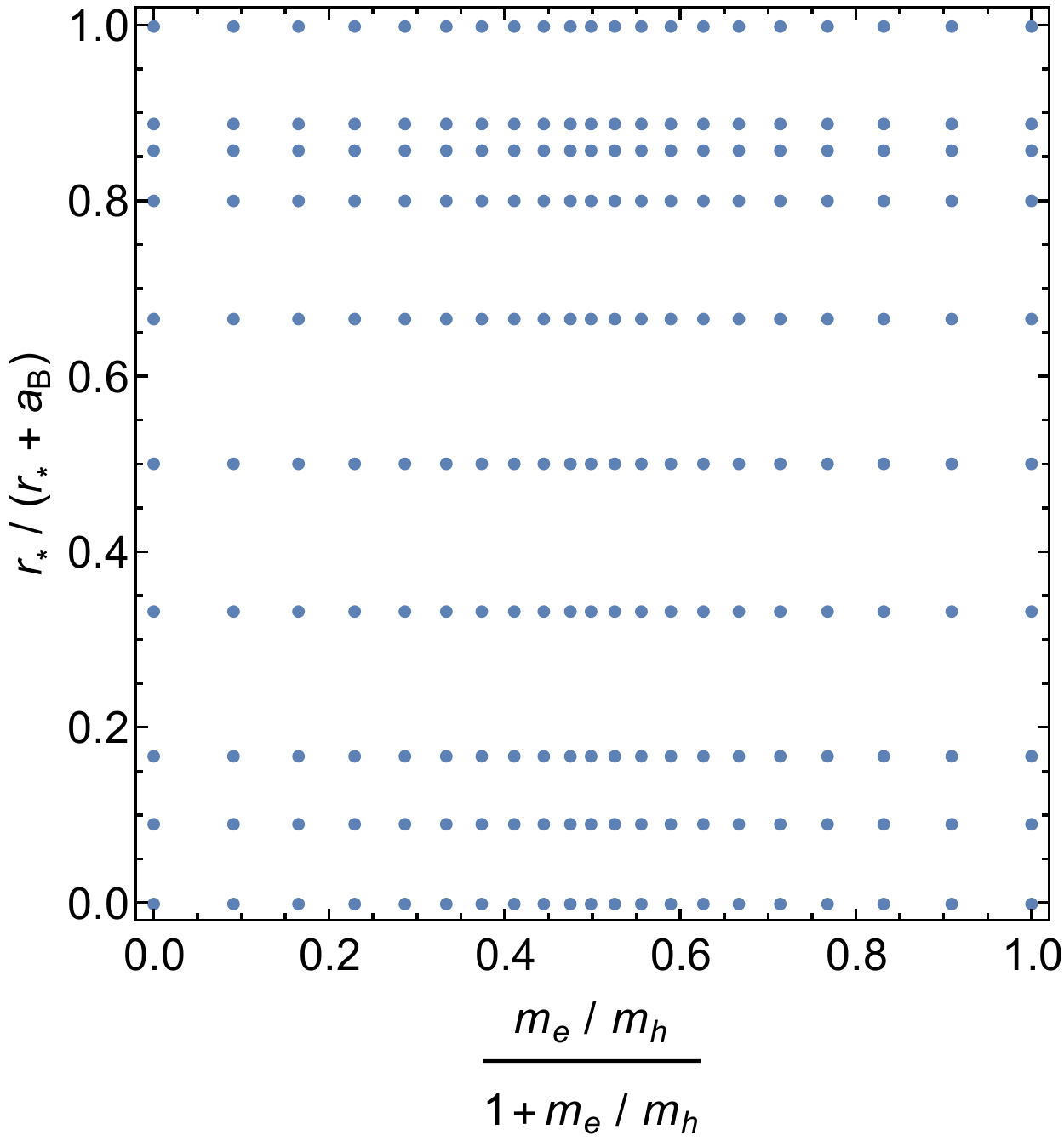}
    \captionof{figure}{
      Grid of chosen parameters in atomic units.
    }
    \label{fig:2-spaceparams-atomic}
  \end{minipage}\hspace{0.03\linewidth}
  \begin{minipage}{.48\textwidth}
    \centering
    \includegraphics[width=.9\linewidth]{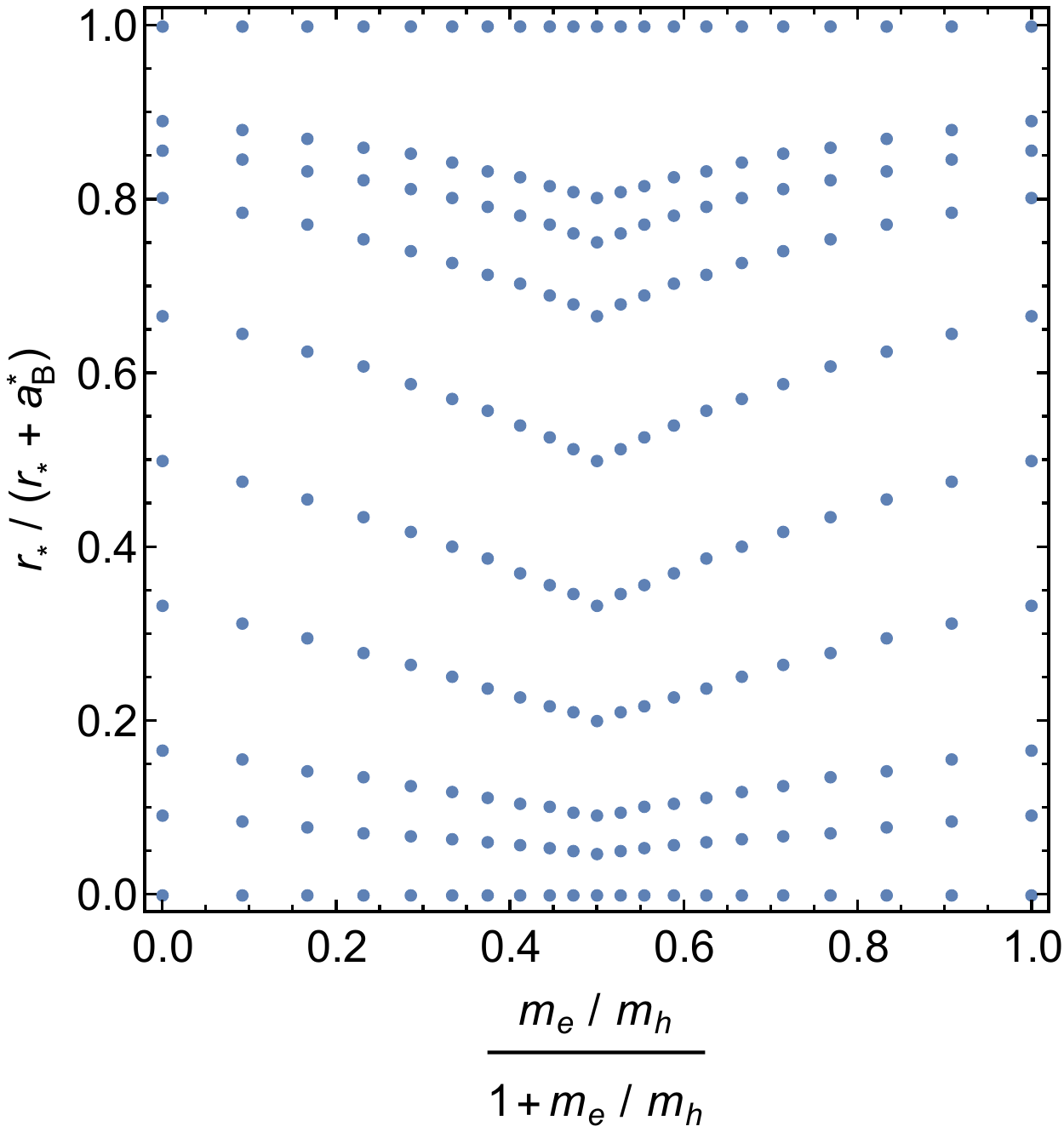}
    \captionof{figure}{
      Grid of chosen parameters in excitonic units.
    }
    \label{fig:2-spaceparams}
  \end{minipage}
\end{figure}

For some example systems, we have analysed what values of time step and
population we should use in order to be in the linear extrapolation regime to
zero time step (see Figs.~\ref{fig:2-dtdmcscaling} and
\ref{fig:2-dtdmcscaling2}). Using condition (\ref{eq:2-dtdmc}), we have chosen
the following setup: for $r_{\ast} / a_{\text{B}} \geqslant 0.5$, we will use
$\tau_{\tmop{DMC}} \in \{ 0.01, 0.04 \}$ with corresponding populations $\{
4096, 1024 \}$, while for $r_{\ast} / a_{\text{B}} < 0.5$, we will use
$\tau_{\tmop{DMC}} \in \{ 0.005, 0.01 \}$ with corresponding populations $\{
8192, 4096 \}$. The number of equilibration steps was chosen to be
$N_{\tmop{equil}} \geqslant 200000$, which agrees with condition
(\ref{eq:2-Nequil}).

\begin{figure}[h!]
  \centering
  \begin{minipage}{.48\textwidth}
    \centering
    \includegraphics[width=.9\linewidth]{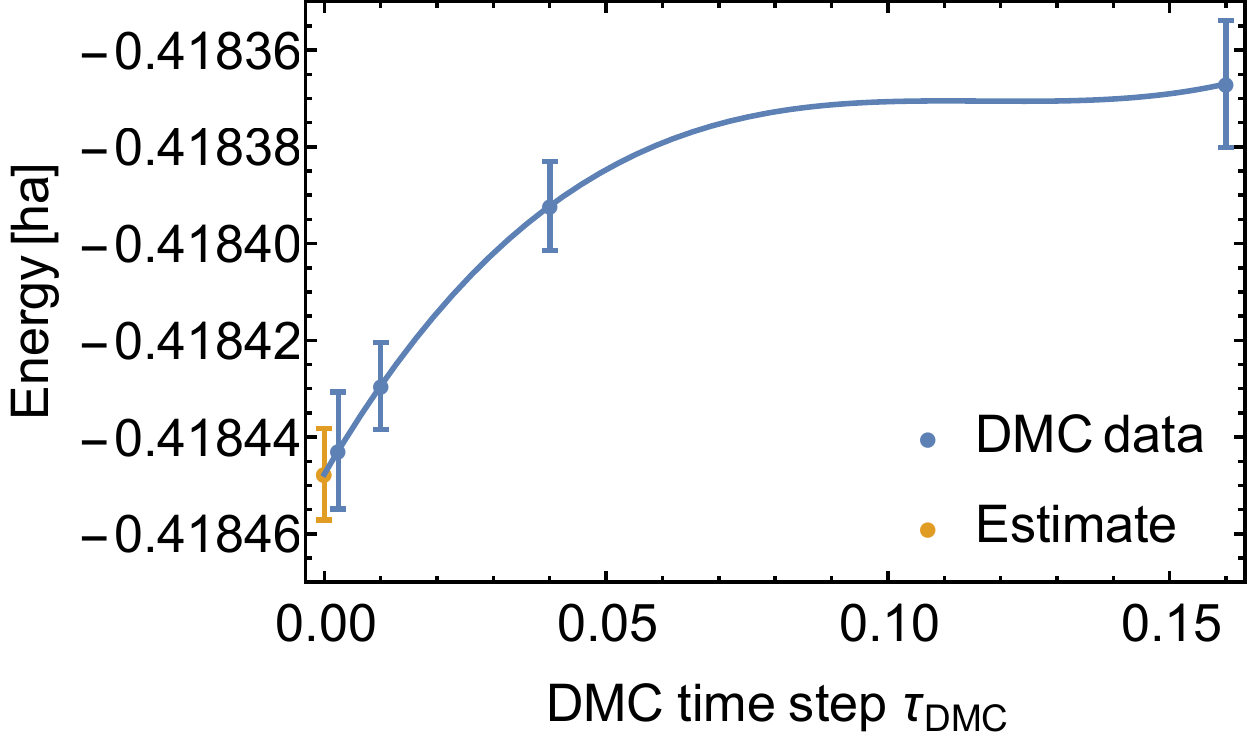}
    \captionof{figure}{
      Determination of the linear region in energy vs. time step scaling for 
      $m_e / m_h = 1,$ $r_{\ast} / a_{\text{B}} = 1$ for a negative trion. The 
      population for every time step is changed as $\sim 1 / \tau_{\tmop{DMC}}$.
    }
    \label{fig:2-dtdmcscaling}
  \end{minipage}\hspace{0.03\linewidth}
  \begin{minipage}{.48\textwidth}
    \centering
    \includegraphics[width=.9\linewidth]{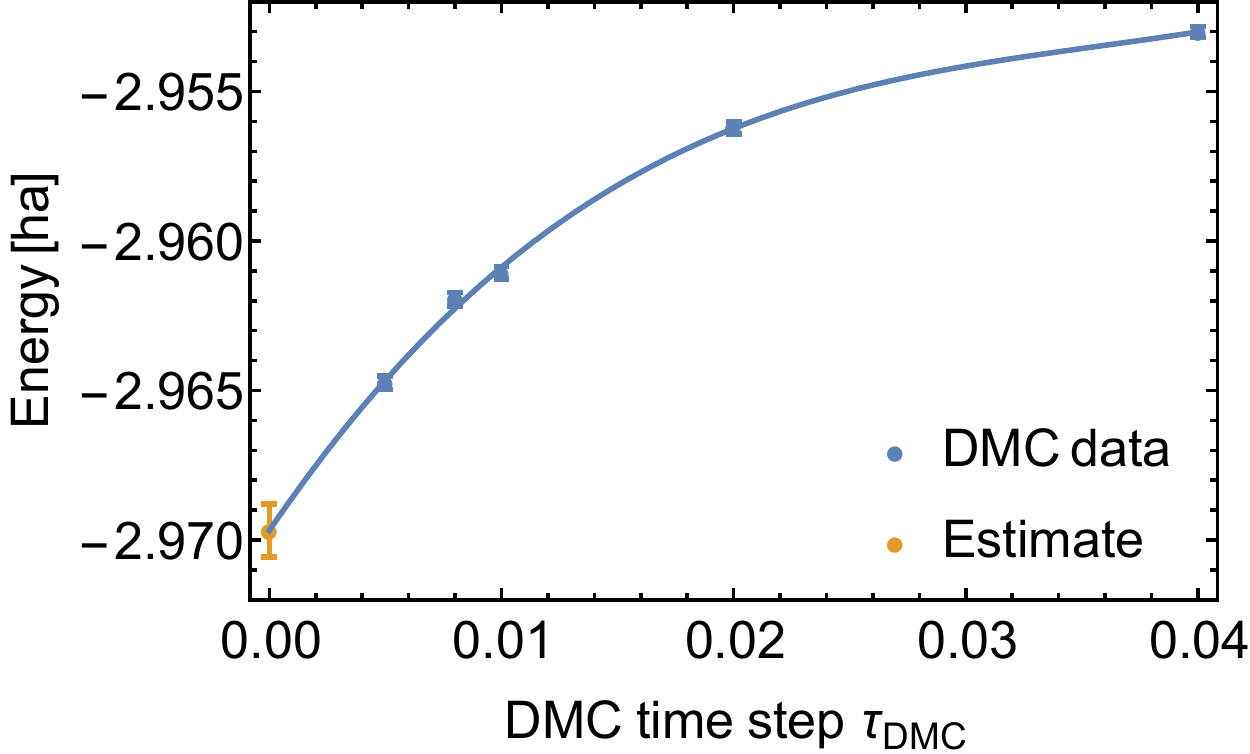}
    \captionof{figure}{
      Similarly to Fig.~\ref{fig:2-dtdmcscaling}, but for a donor-bound 
      biexciton with $m_e / m_h = 0.3, r_{\ast} / a_{\text{B}} = 0.1$. D$^+$XX 
      results for $0 < r_{\ast} < \infty$ and $0 < m_e / m_h < \infty$ were 
      calculated by Ryo Maezono.
    }
    \label{fig:2-dtdmcscaling2}
  \end{minipage}
\end{figure}

\begin{figure}[h!]
  \resizebox{14.5cm}{!}{\includegraphics{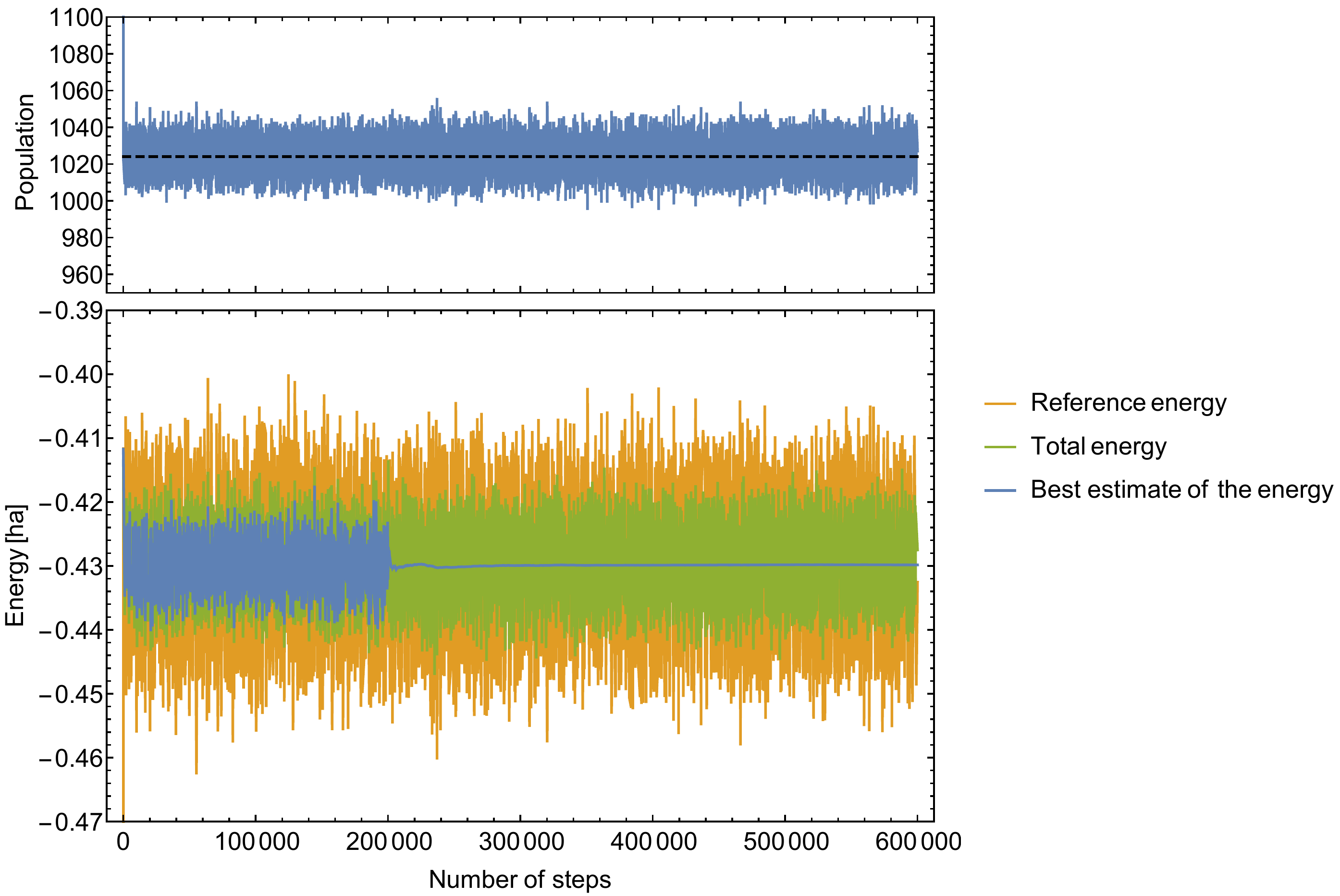}}
  \caption{Population and energy during the DMC calculation. One can see that
  the population control is keeping the population near the fixed value of
  1024 (black dashed line). We show reference energy, total energy and the
  best estimate of the DMC energy. The system is a negative trion with $m_e /
  m_h = 0.9, r_{\ast} / a_{\text{B}} = 1.$ The equilibration phase ends after
  200{\hspace{0.15em}}000 steps.\label{fig:2-dmcenergypopulation}}
\end{figure}

Figure \ref{fig:2-dmcenergypopulation} shows an example of the reference
energy, total energy of a system (defined as the average over the
configuration population of the local energy at any given iteration), and the
best estimate of the DMC energy during the DMC simulation. After
$N_{\tmop{equil}}$ steps, the system is well equilibrated and we start the
accumulation stage. The average population in this example was set to 1024 and
we can see that the population control mechanism is working correctly, keeping
the population near this number.

Finally, in Fig.~\ref{fig:2-reblock} we show an example of the reblocking
method. The reblocked standard error in the energy reaches a plateau as
explained in Chapter \ref{ch:2-dmcerrors}, and this value is used for the
standard error estimate.

\begin{figure}[h!]
  \resizebox{11cm}{!}{\includegraphics{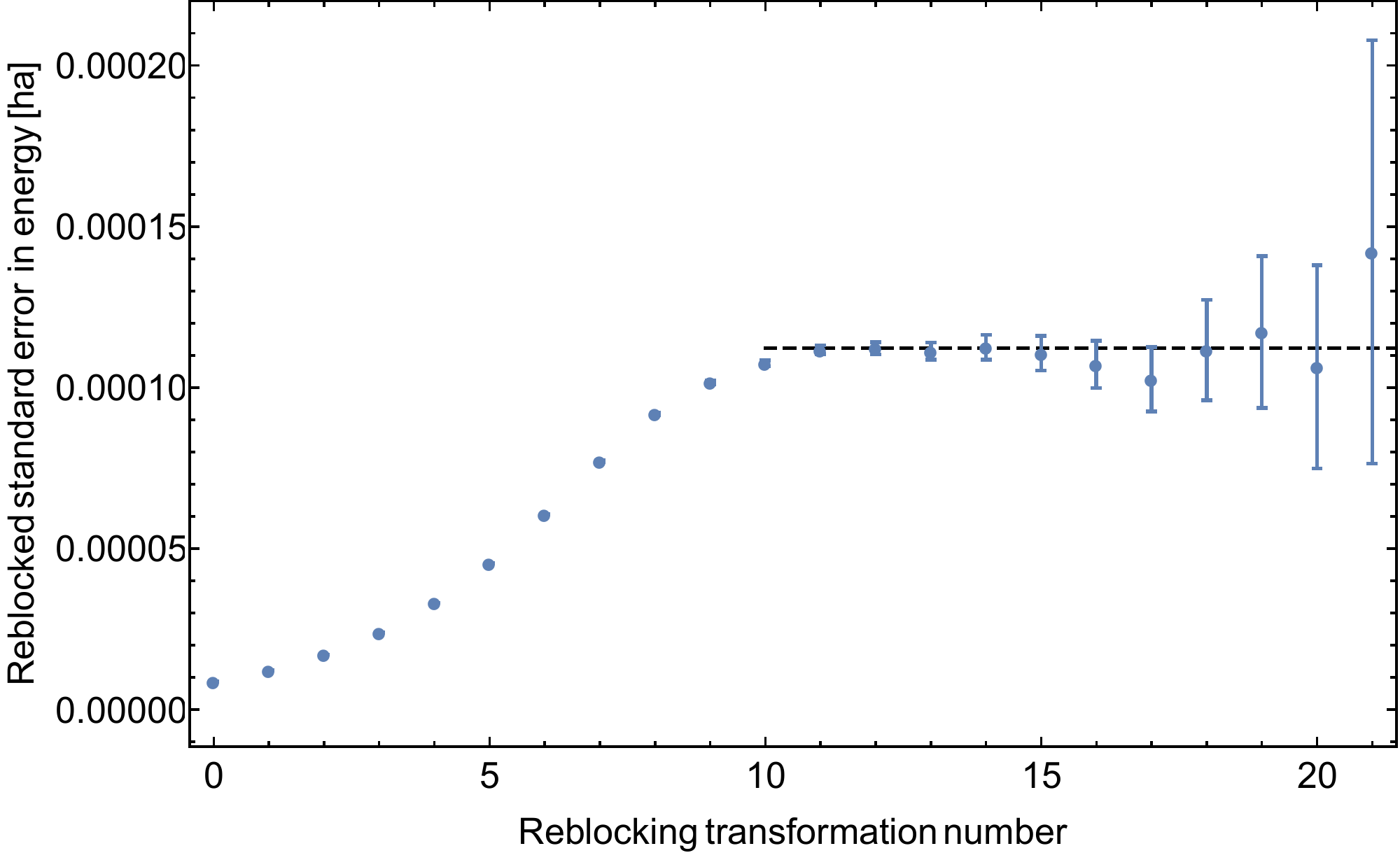}}
  \caption{Reblocking method used on the donor-bound biexciton system with
  $m_e / m_h = 0.1, r_{\ast} / a_{\text{B}} = 1$. The reblocking
  transformation number is the binary logarithm of the block size. One can see
  the plateau in the reblocked error after reaching a block size of
  $2^{11}$.\label{fig:2-reblock}}
\end{figure}

\section{Binding energy}

\subsection{Exciton}

\begin{figure}[h!]
  {\hspace{3em}}\resizebox{13.5cm}{!}{\includegraphics{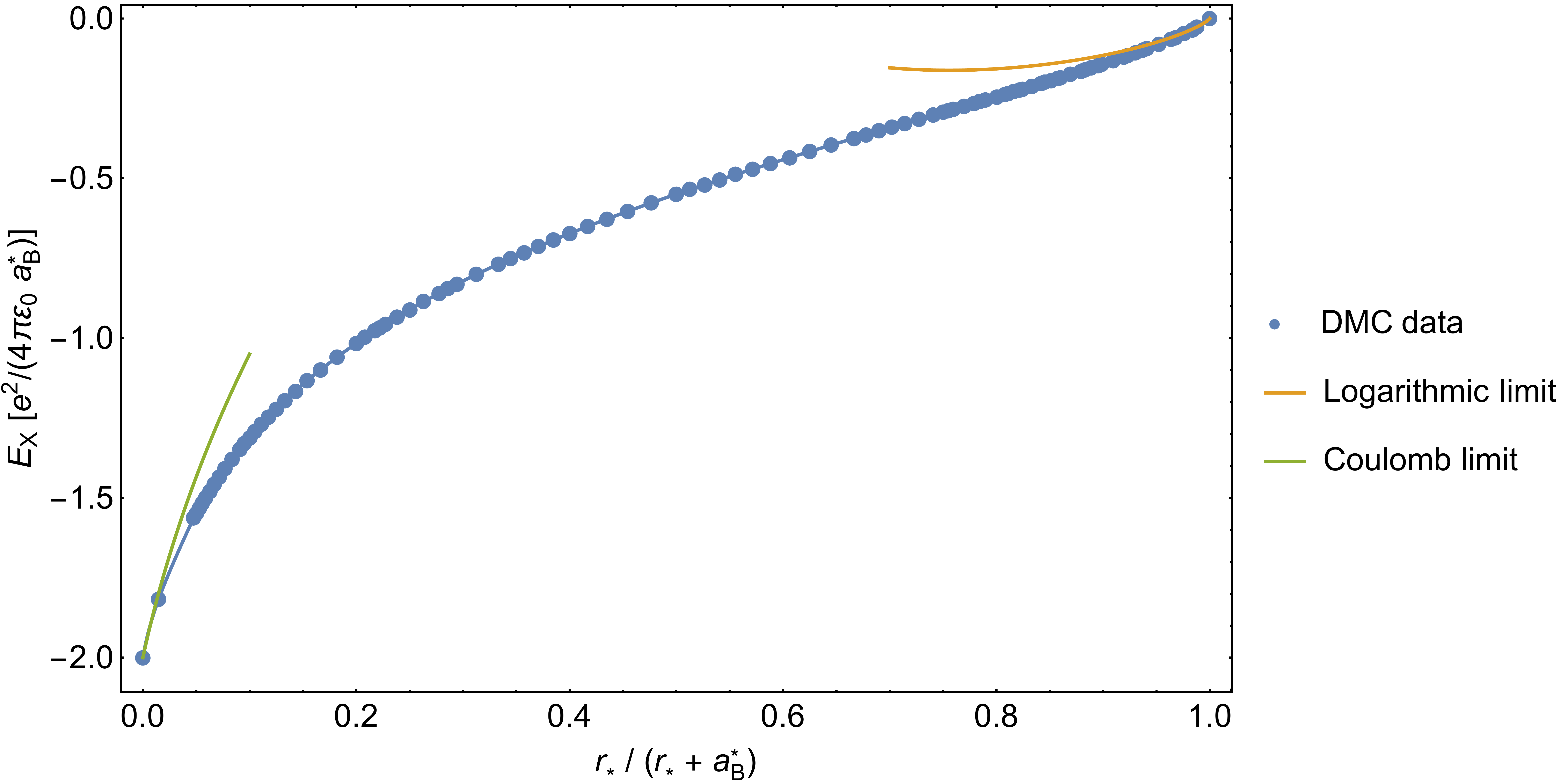}}
  \caption{Exciton binding energy, $E_{\text{X}}$. The plot also shows the
  first order corrections to the energy in the logarithmic (orange) and
  Coulomb (green) limits.\label{fig:2-excitonbinding}}
\end{figure}

The binding energy of an exciton is presented in Figure
\ref{fig:2-excitonbinding} and is indeed independent of the mass ratio if one
uses excitonic units of energy as explained in Chapter \ref{ch:2-exciton}. In
the Coulomb limit $(r_{\ast} \rightarrow 0)$, we recover the well-known
excitonic energy of $- 4 \tmop{Ry}^{\ast}$. In order to determine the
behaviour of the energy near the Coulomb limit, we evaluate the first order
correction as:
\begin{equation}
  \langle \Delta v \rangle = \frac{\langle \psi | \Delta v | \psi
  \rangle}{\langle \psi | \psi \rangle} = \frac{2 \pi \int_0^{\infty} \psi^2
  \Delta vr \mathd r}{2 \pi \int_0^{\infty} \psi^2 r \mathd r} =
  \frac{\int_0^{\infty} \psi^2 v_{\tmop{Keldysh}} r \mathd r}{\int_0^{\infty}
  \psi^2 r \mathd r} - \frac{\int_0^{\infty} \psi^2 v_{\tmop{Coulomb}} r
  \mathd r}{\int_0^{\infty} \psi^2 r \mathd r},
\end{equation}
where $\Delta v = v_{\tmop{Keldysh}} - v_{\tmop{Coulomb}}$ is the difference
between the full Keldysh potential from Eq.~(\ref{eq:2-effinteraction}) and
its zeroth\mbox{-}order expansion for $r \rightarrow 0$, \tmtextit{i.e.} the
Coulomb potential, Eq.~(\ref{eq:2-Coulomb}). The excitonic wave function from
Eq.~(\ref{eq:2-Xwavefunction}) was used and the correction was found to be:
\begin{equation}
  \langle \Delta v \rangle  \underset{r_{\ast} \rightarrow 0}{\longequal} 
  \frac{e^2}{4 \pi \varepsilon_0 a_{\text{B}}^{\ast}}  \frac{16
  r_{\ast}}{a_{\text{B}}^{\ast}} + O (r_{\ast}^2),
\end{equation}
which is linear in $r_{\ast}$. The correction is shown in
Fig.~\ref{fig:2-excitonbinding} (green line).

The logarithmic limit constant was determined to be
\begin{equation}
  C_{\text{X}} = 0.41057748 (10) .
\end{equation}
The logarithmic limit behaviour from Eq.~(\ref{eq:2-logexcitonbehaviour}) is
also shown in Fig.~\ref{fig:2-excitonbinding} (orange line) and matches the
DMC data near $r_{\ast} \rightarrow \infty$. The correction to the energy near
the logarithmic limit, if the energy is measured in the units of $e^2 / (4 \pi
\varepsilon_0 r_{\ast})$, was also evaluated. The correction was evaluated
numerically using VMC: firstly the wave function was optimised using the pure
logarithmic interaction from Eq.~(\ref{eq:2-logarithmic}), and then the wave
function was used to evaluate $\langle \Delta v \rangle = \langle
v_{\tmop{Keldysh}} - v_{\tmop{logarithmic}} \rangle$. The results are
presented in Fig.~\ref{fig:2-xcorr} and the correction was found to have a
square root dependence in $1 / r_{\ast}$.

\begin{figure}[h]
  \resizebox{9cm}{!}{\includegraphics{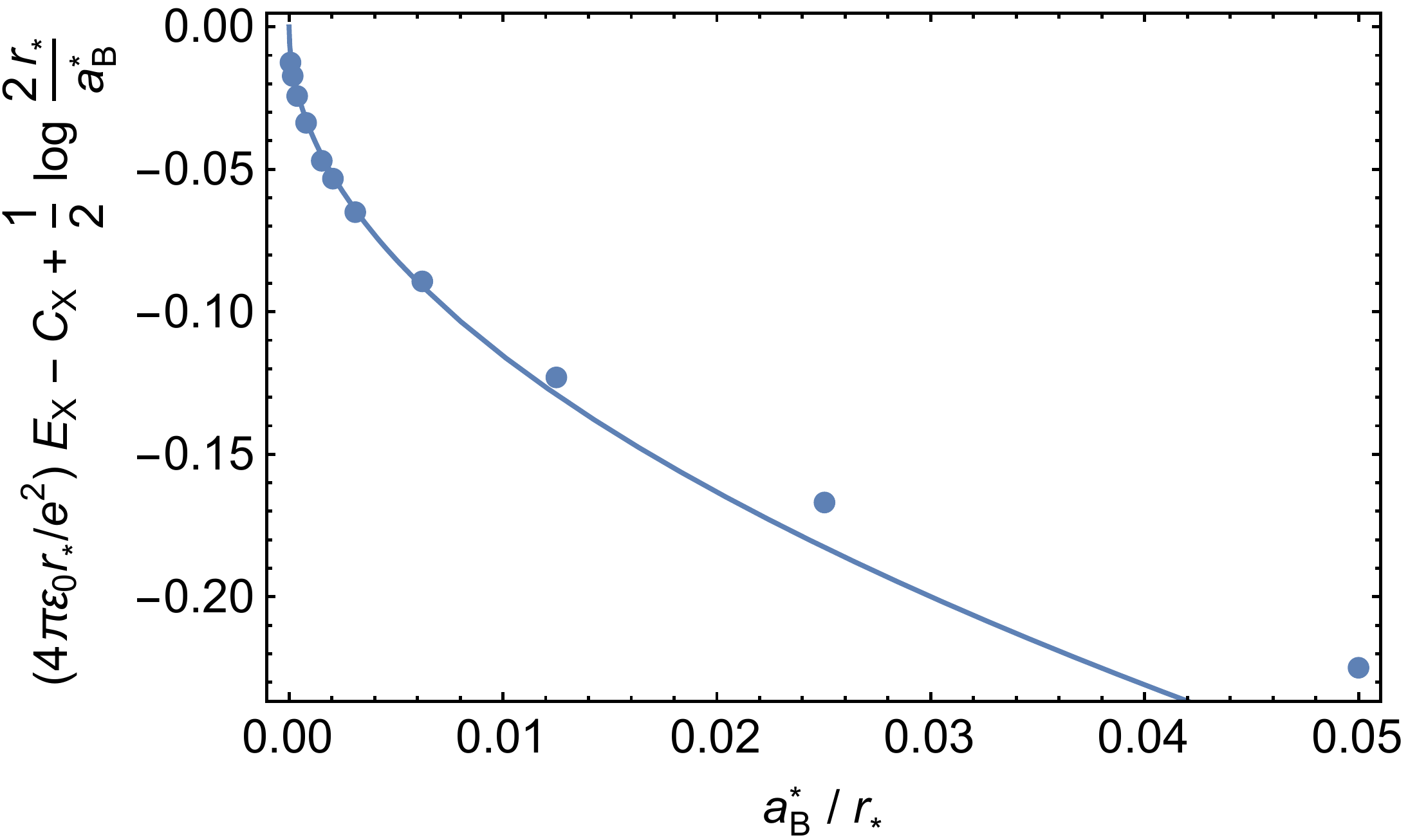}}
  \caption{Correction to the energy near the logarithmic limit. The solid line
  is a fit to the square root behaviour $- a \sqrt{1 / r_{\ast}}$, with $a =
  1.154 (23)$.\label{fig:2-xcorr}}
\end{figure}

Since the energy is independent of the mass ratio, one can easily take the
limits of heavy electron mass and light electron mass, in which the complex
will look like a donor-bound electron and an acceptor-bound hole respectively.
The numerical results remain the same, however in the units, we need to
remember to exchange $\mu \rightarrow m_h$ for a heavy electron and $\mu
\rightarrow m_e$ for a light electron.

Two interpolation formulas were devised for the exciton binding energy. The
simpler formula is accurate to 5\%:
\begin{equation}
  E_{\text{X}} = \frac{e^2}{4 \pi \varepsilon_0 a_{\text{B}}^{\ast}}  (1 -
  \nu)  \frac{- 2 + 0.5 \nu \log (1 - \nu)}{1 + 1.31 \sqrt{\nu}} .
\end{equation}
The second, more complicated but also more accurate formula was devised:
\begin{equation}
  E_{\text{X}} = \frac{e^2}{4 \pi \varepsilon_0 a_{\text{B}}^{\ast}}  (1 -
  \nu)  \frac{- 2 + 16 \nu + a_1 \nu^{3 / 2} + a_2 \nu^2 + a_3 \nu^{5 / 2} +
  a_4 \nu^4}{1 + b_1 \nu^2 + b_2 \nu^{5 / 2}} . \label{eq:2-xbindfit}
\end{equation}
Fitting parameters are given in Appendix \ref{ch:appendixQMCfittingX}. The
formula has a relative error of 0.05\%, which can be seen on the histogram in
Fig.~\ref{fig:2-xbindhisto}. In the Coulomb limit, the fitting formula
recovers both the zeroth and first order perturbation to the energy.

\begin{figure}[h]
  \resizebox{6.5cm}{!}{\includegraphics{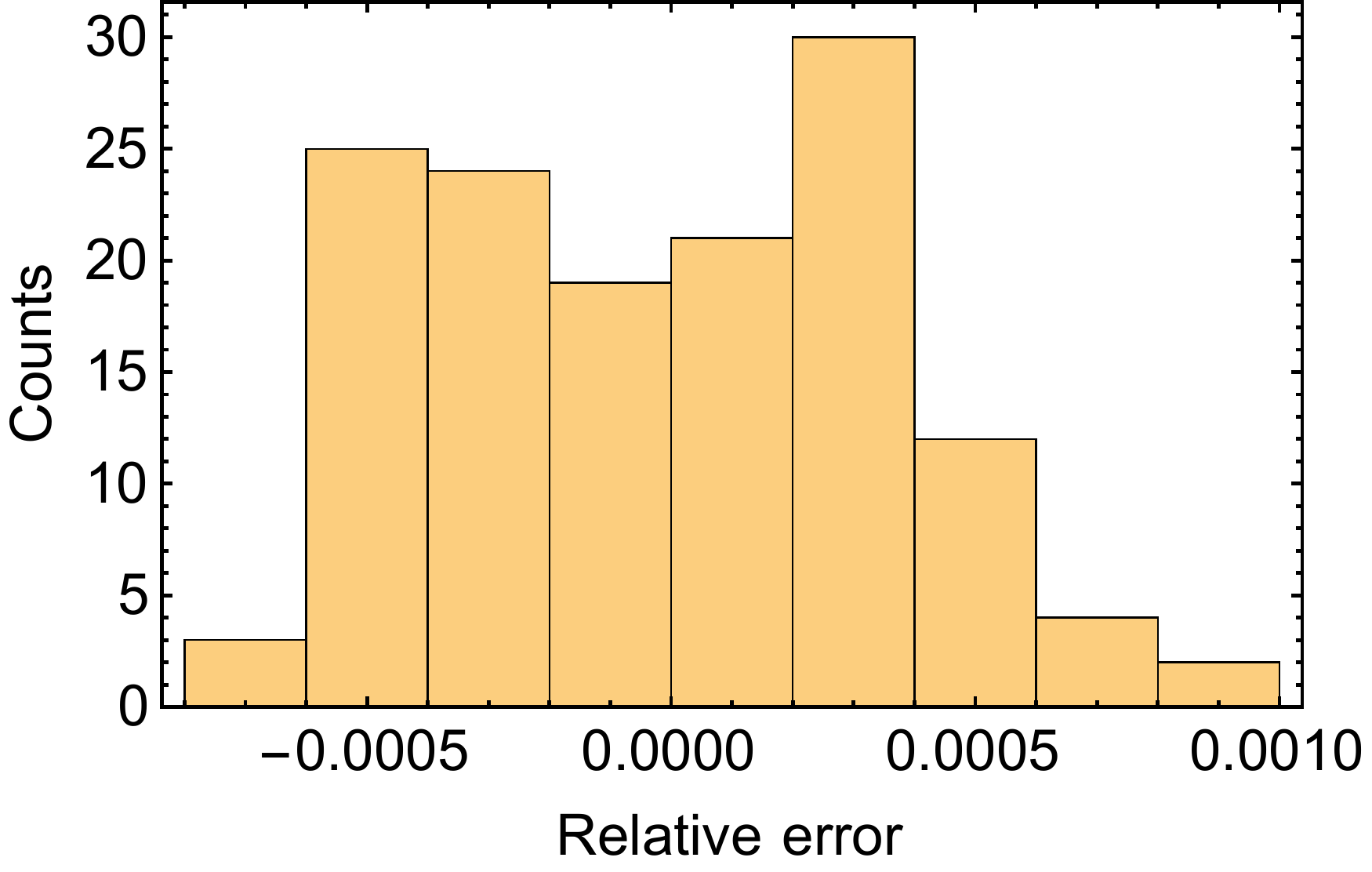}}
  \caption{Histogram for the relative error in the exciton binding energy
  fitting formula from Eq.~(\ref{eq:2-xbindfit}).\label{fig:2-xbindhisto}}
\end{figure}

\subsection{Trion}

Figure \ref{fig:2-trionbinding} shows the negative trion binding energy as a
3D surface plot that includes dependence on both the mass ratio and the
susceptibility. The main purpose of this plot is to observe the general
behaviour of the energy. On the other hand, Fig.~\ref{fig:2-trionbinding2D}
shows only the dependence on the rescaled susceptibility for lines of constant
mass ratio, but can be used to extract numerical values of the trion binding
energy.

\begin{figure}[h]
  \resizebox{12cm}{!}{\includegraphics{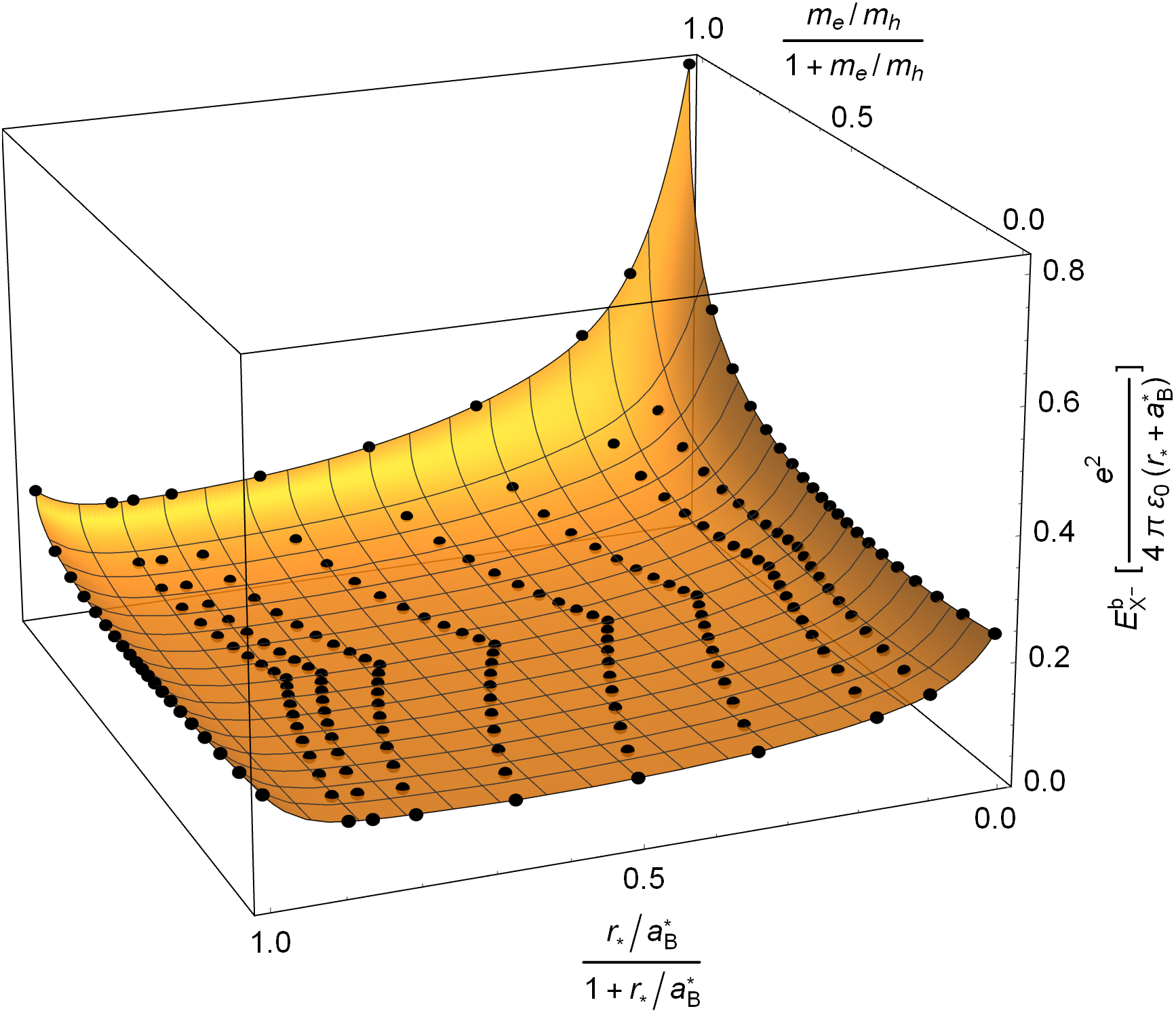}}
  \caption{Binding energy of a negative trion,
  $E_{\text{X}^-}^{\text{b}}$.\label{fig:2-trionbinding}}
\end{figure}

\begin{figure}[h]
  \resizebox{12cm}{!}{\includegraphics{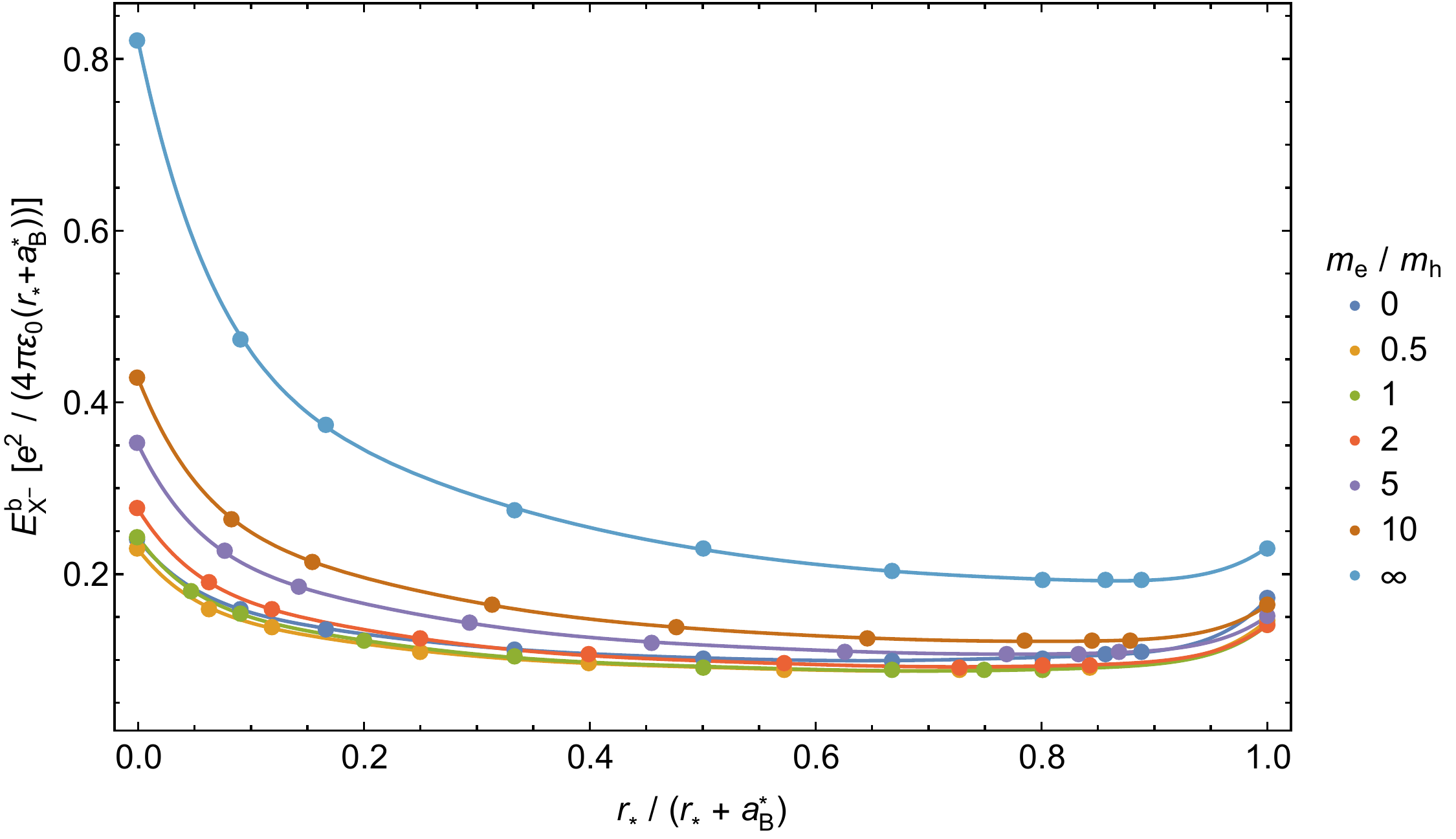}}
  \caption{Negative trion binding energy, $E_{\text{X}^-}^{\text{b}}$, as a
  function of rescaled susceptibility.\label{fig:2-trionbinding2D}}
\end{figure}

The Coulomb limit results agree up to the error bar with the results of
Ref.~{\cite{Spink2015}}. The logarithmic limit results are taken from
Ref.~{\cite{Bogdan}}, where the DMC calculations were found to be in agreement
with the analytical results obtained using the shooting method. The $m_e =
m_h$ values were also compared with the results of
Ref.~{\cite{Velizhanin2015}} and agreement was found.

In the heavy electron limit $(m_e / m_h \rightarrow \infty)$, the complex
resembles an H$_2^+$ molecule, and the energy has the expected square-root
behaviour in the mass ratio for a constant susceptibility, as predicted by the
Born-Oppenheimer approximation in Chapter \ref{ch:2-extrememass}. Since the
negative trion in this limit behaves like a complex with two acceptors and a
hole, in order to determine the total energy of this complex, it was necessary
to vary the separation of the two acceptors to find the most energetically
favourable position. See Fig.~\ref{fig:2-sep1} for an example of how the
minimum of energy as a~function of the separation is found and
Fig.~\ref{fig:2-sep2} to see the optimal separation of acceptors for each
value of susceptibility.

\begin{figure}
  \centering
  \begin{minipage}{.48\textwidth}
    \centering
    \includegraphics[width=.9\linewidth]{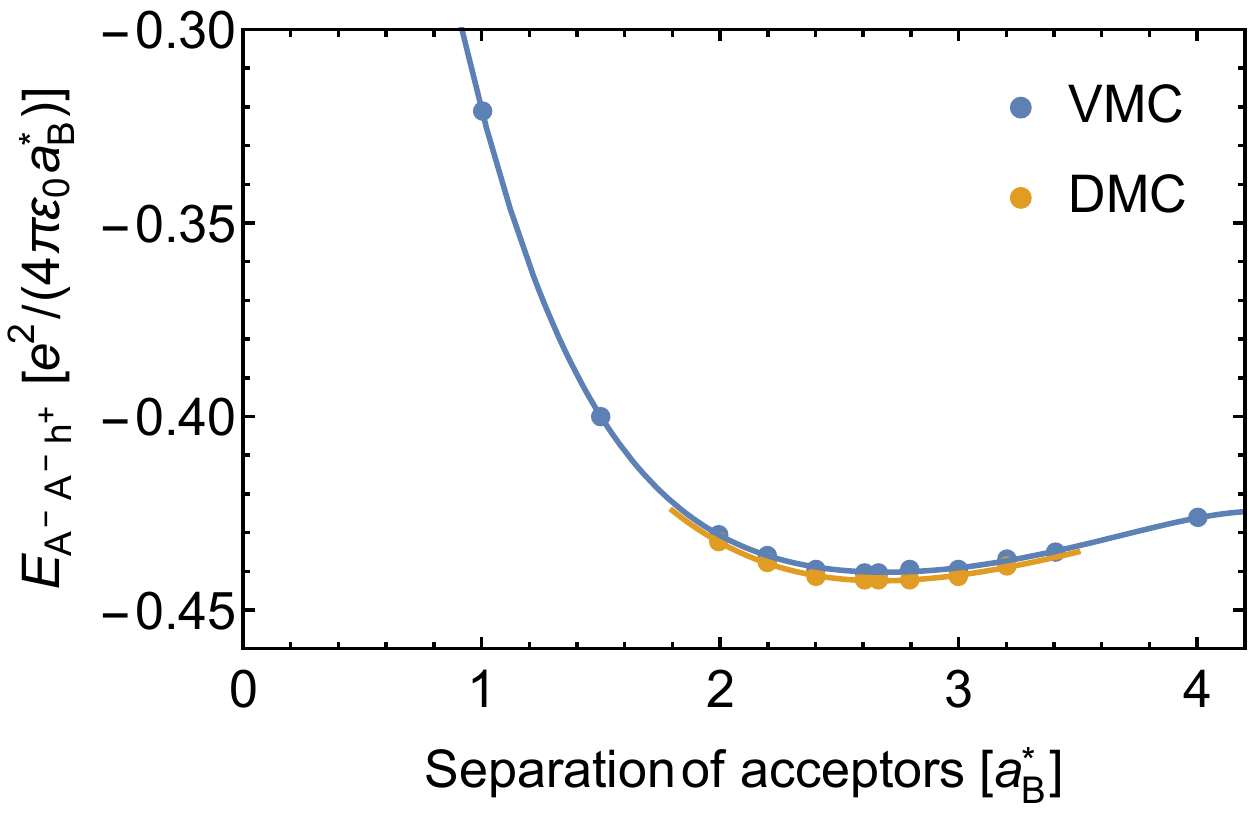}
    \captionof{figure}{
      The distance between two acceptors is varied, and the minimum of energy 
      is found for each value of $r_{\ast} / a_{\text{B}}^{\ast}$. Example for 
      $r_{\ast} / a_{\text{B}}^{\ast} = 2$.
    }
    \label{fig:2-sep1}
  \end{minipage}\hspace{0.03\linewidth}
  \begin{minipage}{.48\textwidth}
    \centering
    \includegraphics[width=.9\linewidth]{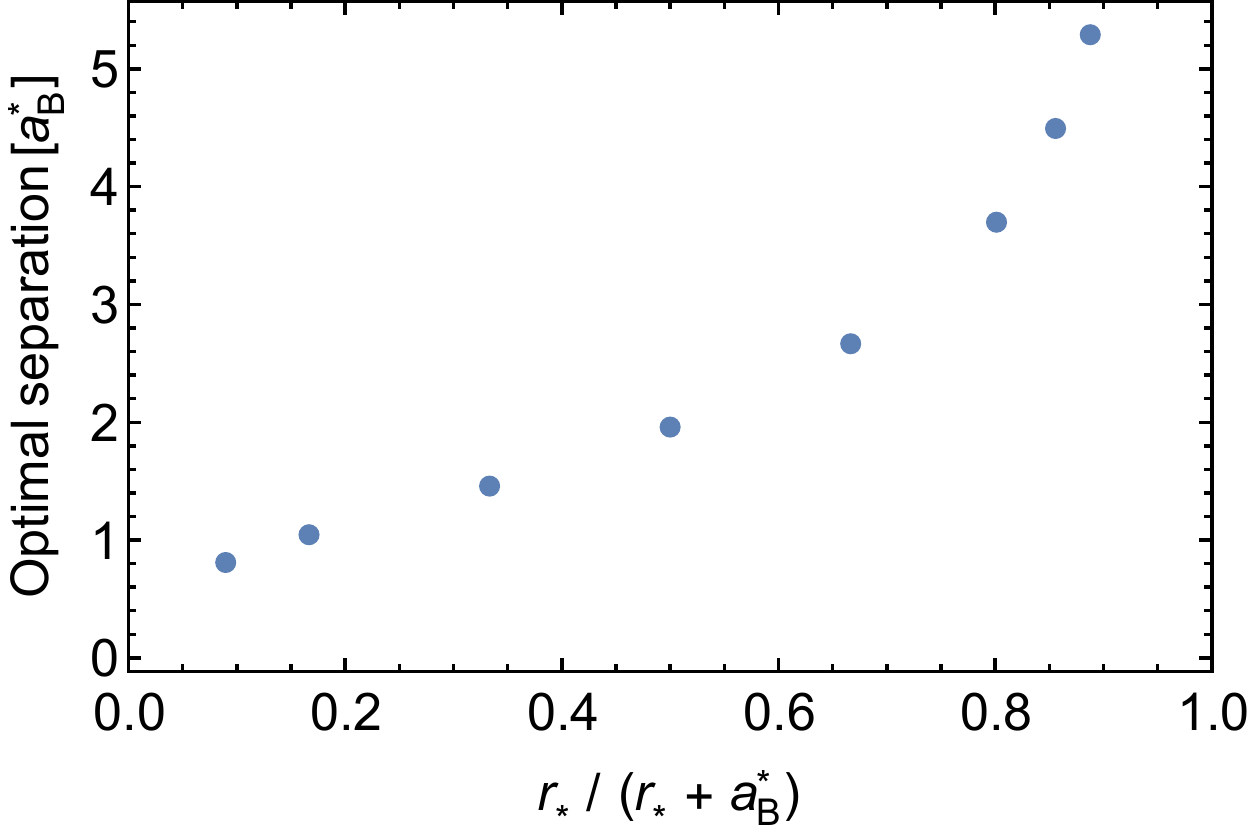}
    \captionof{figure}{
      Separation of the acceptors which
      minimises \ the energy \ in the A$^-$A$^-$h$^+$ complex,
      for each value of the rescaled susceptibility.
    }
    \label{fig:2-sep2}
  \end{minipage}
\end{figure}

Finally, one can see that there is a steep decrease of values near the
logarithmic limit. In order to understand this decrease, we recall that for an
exciton, there was a square root dependence present near the logarithmic limit
if one uses the energy units of $e^2 / (4 \pi \varepsilon_0 r_{\ast})$. We
conclude that a~similar behaviour must be present in case of a trion, and this
is why we observe the steep decrease. This also suggests that previous studies
that used only the logarithmic limit to determine the binding energy of
complexes in 2D semiconductors greatly overestimated the values.

The DMC data was fitted (to an accuracy within 5\%) using the interpolation
formula,
\begin{eqnarray}
  E_{\text{X}^-}^{\text{b}} & = & \frac{e^2}{4 \pi \varepsilon_0
  a_{\text{B}}^{\ast}}  \left( 1 - \sqrt{\nu} \right) \\
  &  & \times \left[ \left( 0.73 - 0.58 \sqrt{\nu} + 0.22 \nu^2 \right) (2 -
  \eta) \underset{}{} - \left( 1.2 - 1. \sqrt{\nu} + 0.32 \nu^2 \right)
  \sqrt{1 - \eta} \right] . \nonumber
\end{eqnarray}
Alternatively, one can use a general fitting ansatz:
\begin{equation}
  E_{\text{X}^-}^{\text{b}} = \frac{e^2}{4 \pi \varepsilon_0  \left( r_{\ast}
  + a_{\text{B}}^{\ast} \right)} \left( \sum_i \sum_j a_{ij} \nu^i \left(
  \sqrt{1 - \eta} \right)^j \right) \label{eq:2-trionbindfit},
\end{equation}
with fitting parameters given in Appendix \ref{ch:appendixQMCfittingXminus},
that gives results with relative error of 0.5\% (see the histogram in
Fig.~\ref{fig:2-trionbindhisto}). The fitting has the correct square root
behaviour in the Born-Oppenheimer limit.

\begin{figure}[h]
  \resizebox{6.5cm}{!}{\includegraphics{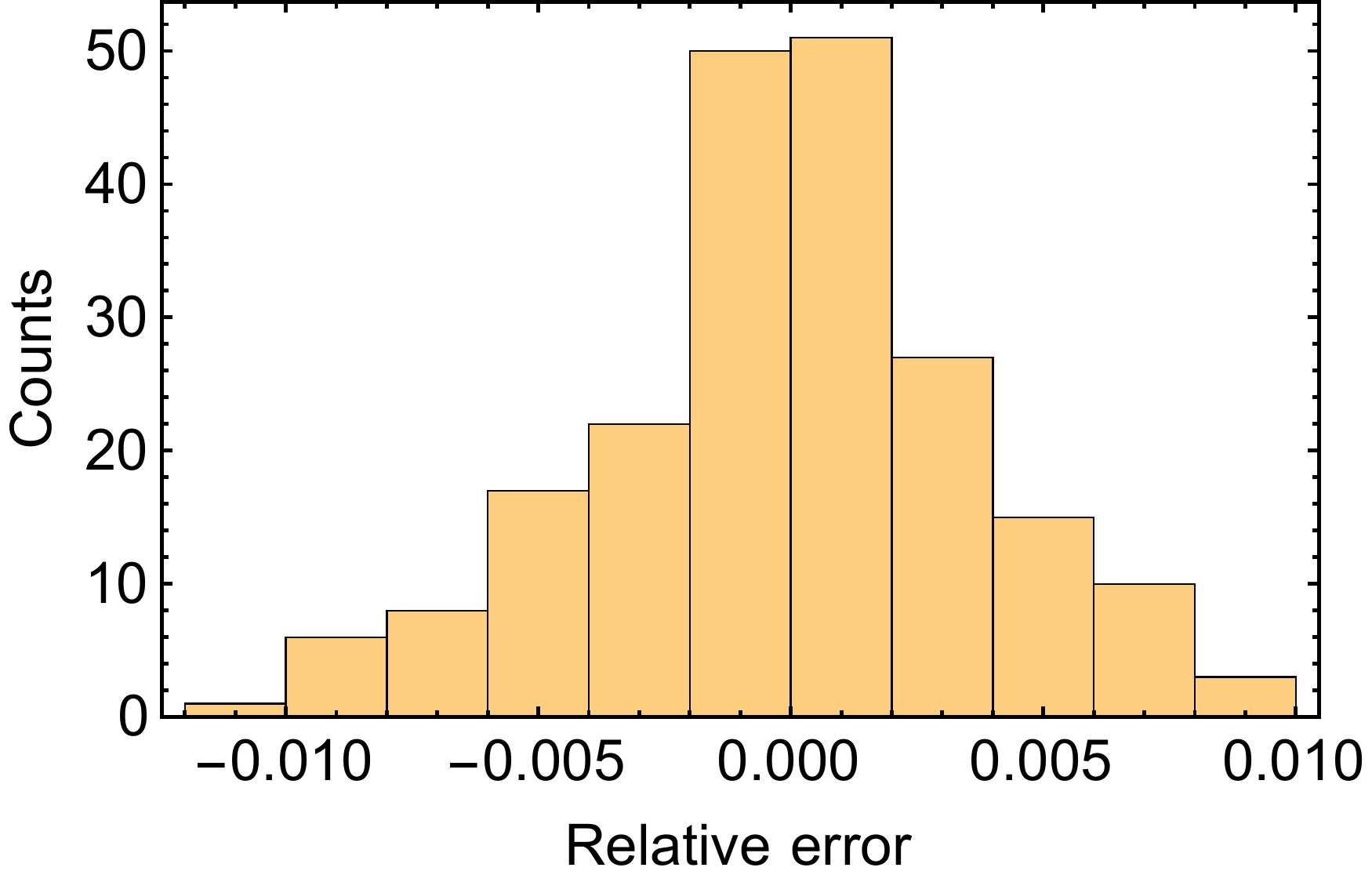}}
  \caption{Histogram of the relative error of the trion binding energy fitting
  formula from Eq.~(\ref{eq:2-trionbindfit}).\label{fig:2-trionbindhisto}}
\end{figure}

\subsection{Donor-bound exciton}

We present the binding energy of a donor-bound exciton complex in
Figs.~\ref{fig:2-DXbinding}--\ref{fig:2-DXbinding2}. One can see that for $m_e
/ m_h \gtrsim 1$, the binding energy reaches values close to zero. In this
region, the calculations were especially difficult, since the complex would
tend to unbind very easily thus producing not a~ground state energy, but a
local minimum. During the wave function optimisation, the cutoff lengths (see
Ch. \ref{ch:2-trialWF}) were set to small values, in order to keep the complex
bound. The region where we expect the complex to be either very weakly bound
or completely unbound is shown in Fig.~\ref{fig:2-DXbinding-unbound}.

\begin{figure}[h!]
  \resizebox{12cm}{!}{\includegraphics{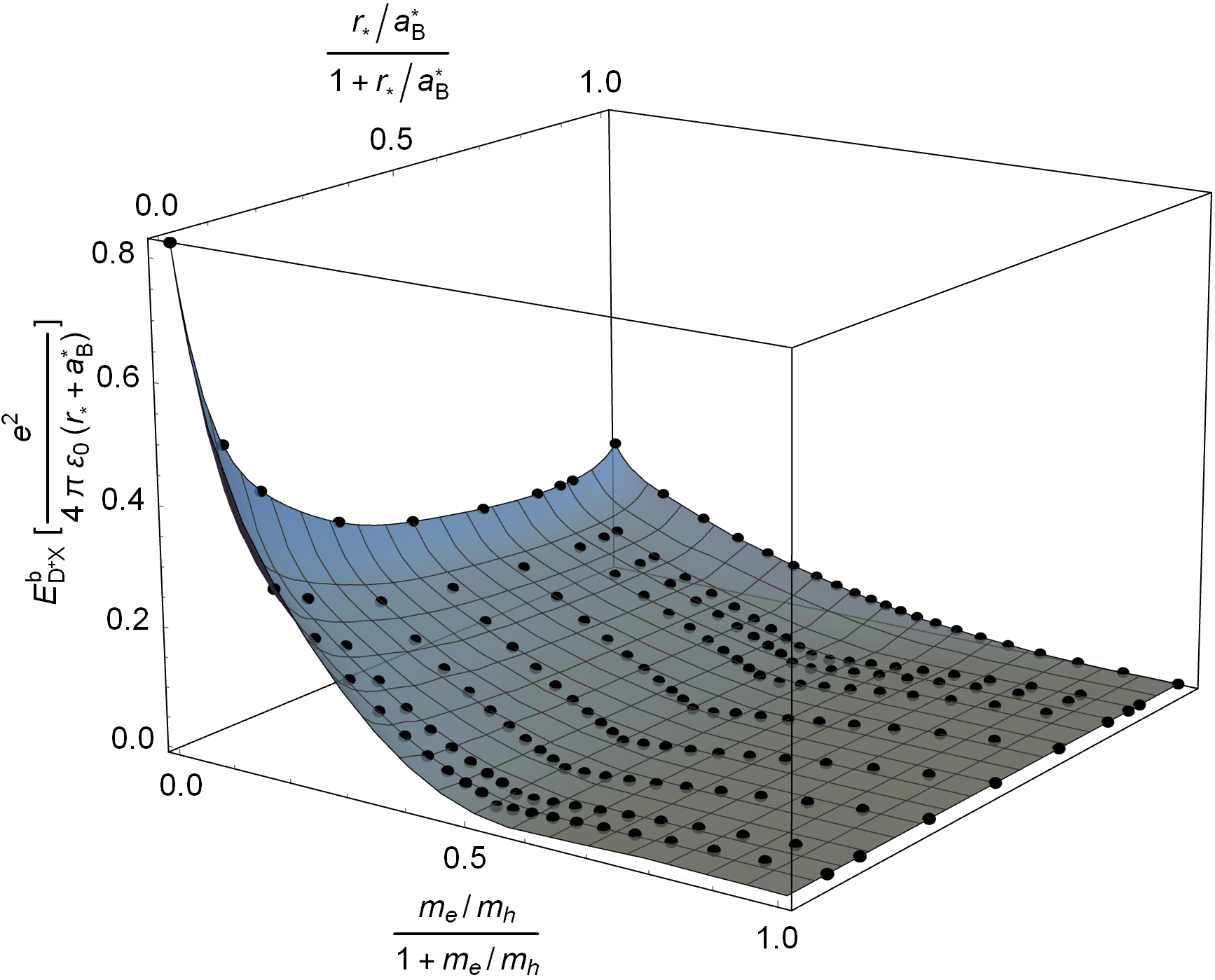}}
  \caption{Binding energy of the donor-bound exciton complex,
  $E_{\tmop{DX}}^{\text{b}}$. Logarithmic limit results are taken from
  Ref.~{\cite{Bogdan}}.\label{fig:2-DXbinding}}
\end{figure}

\begin{figure}[h!]
  \centering
  \begin{minipage}{.48\textwidth}
    \centering
    \includegraphics[width=.99\linewidth]{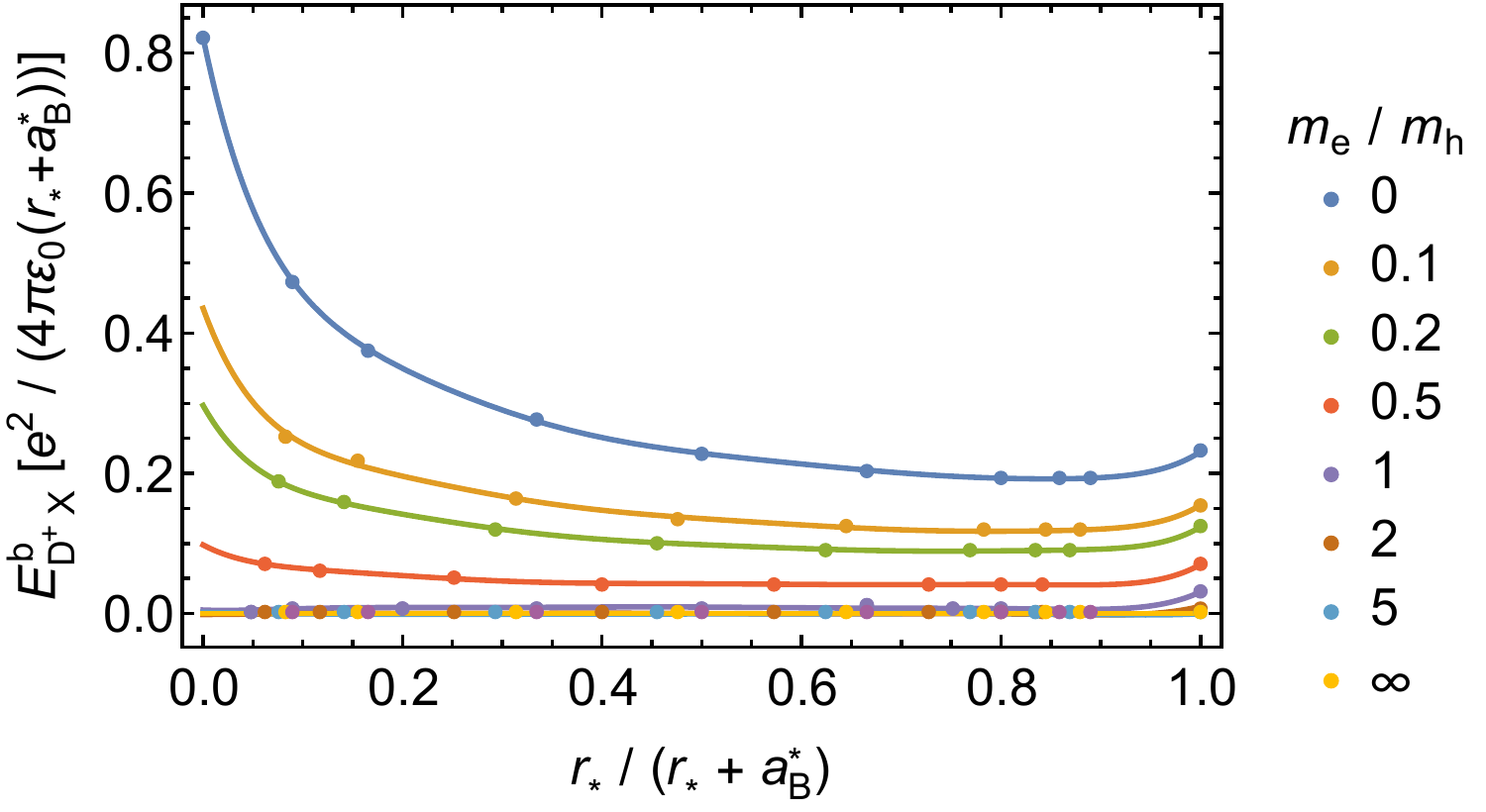}
    \captionof{figure}{
      Binding energy of the donor-bound exciton complex, 
      $E_{\tmop{DX}}^{\text{b}}$, as a function of the rescaled susceptibility 
      for different values of the mass ratio.
    }
    \label{fig:2-DXbinding2}
  \end{minipage}\hspace{0.03\linewidth}
  \begin{minipage}{.48\textwidth}
    \centering
    \includegraphics[width=.9\linewidth]{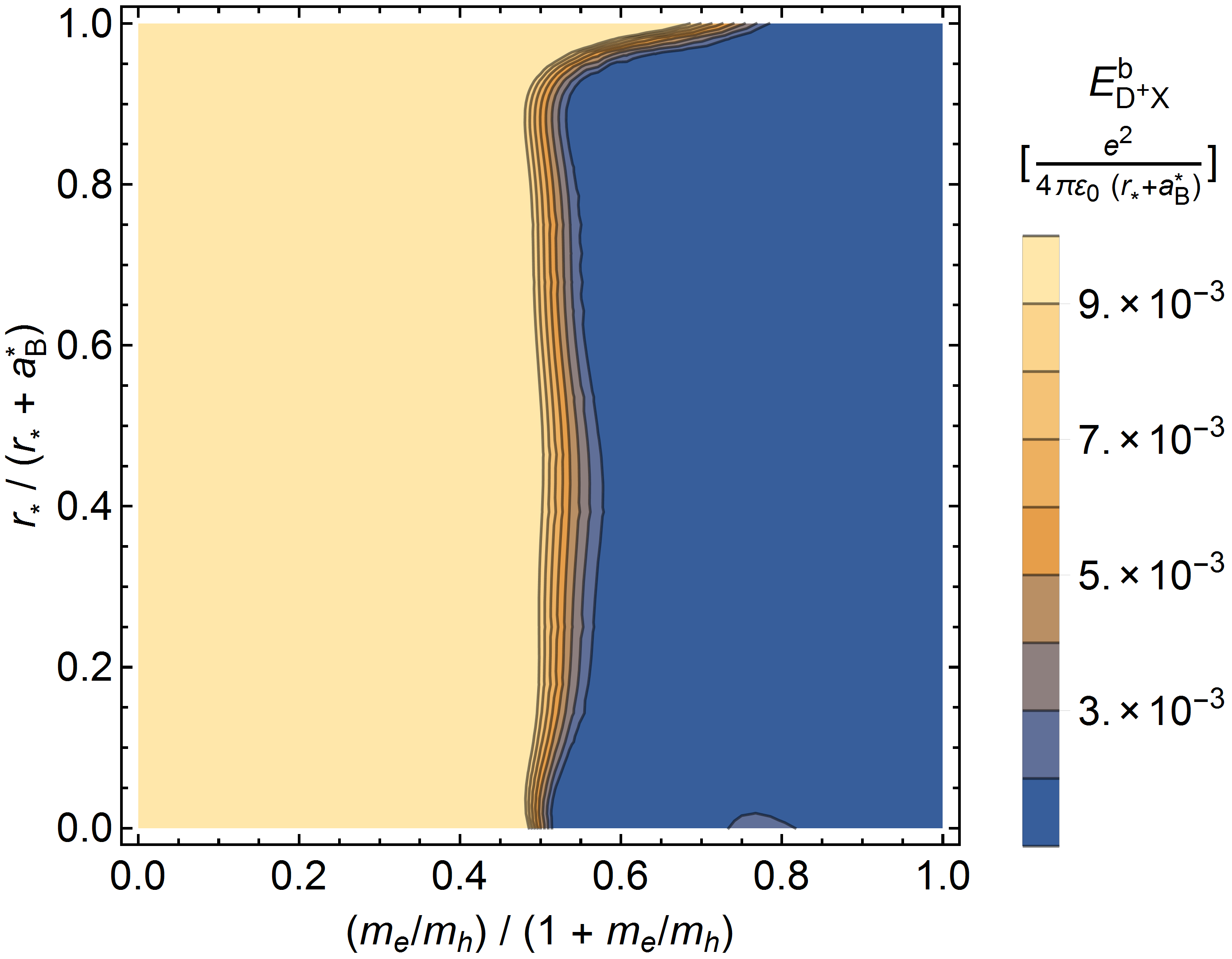}
    \captionof{figure}{
      Contour plot of the binding energy $E_{\tmop{DX}}^{\text{b}}$, showing 
      the region that is probably unbound (in blue).
    }
    \label{fig:2-DXbinding-unbound}
  \end{minipage}
\end{figure}

In the extreme mass ratio limit of a very heavy hole ($m_e / m_h \rightarrow
0$), the complex consists of two positive donors and an electron, and it has
the same binding energy values as the negative trion in the limit of infinite
electron mass (\tmtextit{cf.}
Figs.~\ref{fig:2-trionbinding}--\ref{fig:2-trionbinding2D}). On the other
hand, in the limit of $m_e / m_h \rightarrow \infty$, the complex is unbound,
as explained in Ch.~\ref{ch:2-extrememass}.

\subsection{Biexciton}

\begin{figure}[h]
  \resizebox{12cm}{!}{\includegraphics{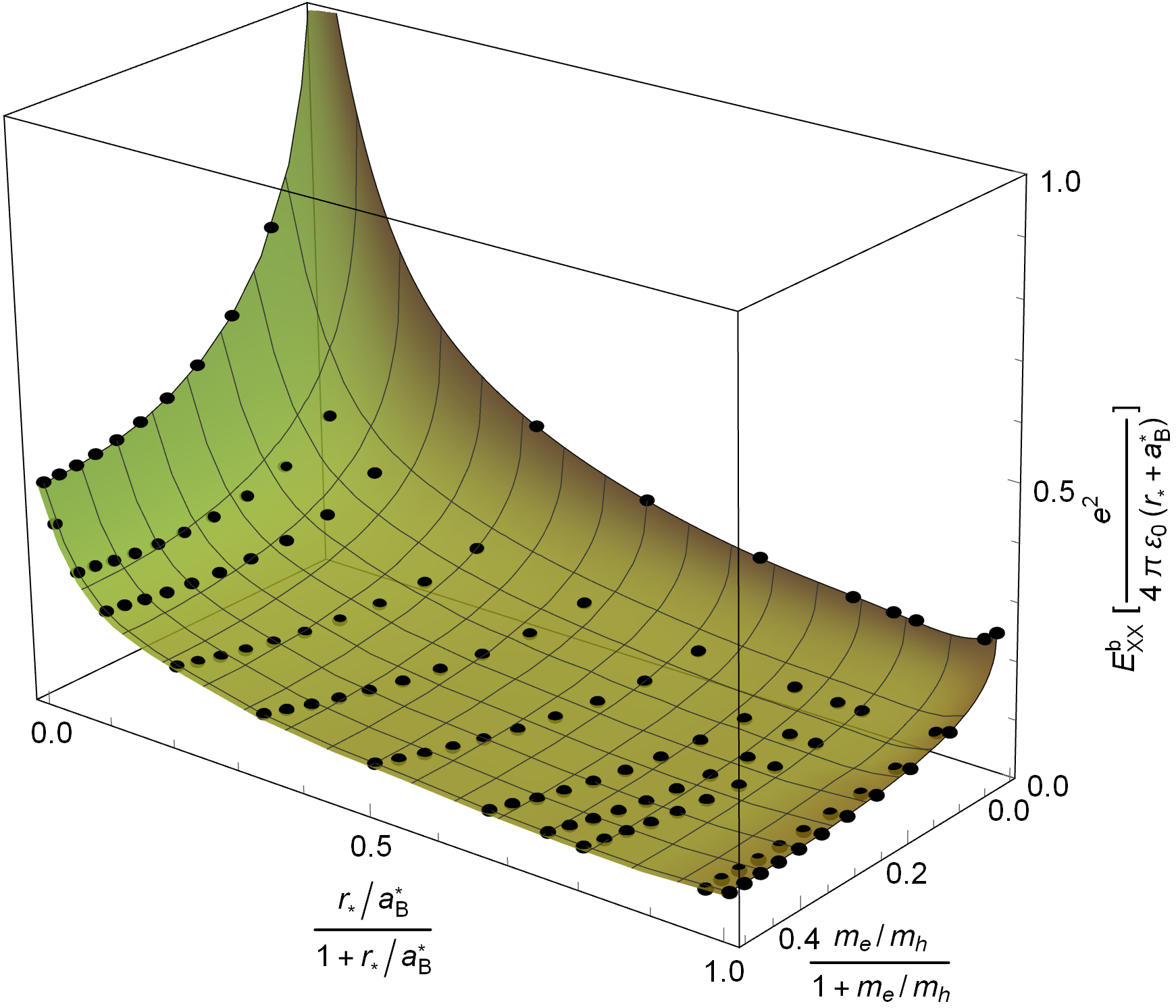}}
  \caption{Biexciton binding energy, $E_{\tmop{XX}}^{\text{b}}$, as a function
  of the rescaled mass ratio and rescaled susceptibility. The surface is the
  interpolation formula from Eq.~(\ref{eq:2-XXbinding-fit}). The Coulomb
  limit, the logarithmic limit, extreme mass ratio limit of $m_e / m_h
  \rightarrow 0$, equal mass ratio results and results for $r_{\ast} /
  a_{\text{B}} = \{ 0.03, 60 \}$ were calculated by
  E.~Mostaani.\label{fig:2-XXbinding}}
\end{figure}

The biexciton binding energies are presented in
Figs.~\ref{fig:2-XXbinding}--\ref{fig:2-XXbinding2}. Since
$E_{\tmop{XX}}^{\text{b}}$ in our units of choice should be invariant under
electron--hole exchange, the plot is symmetric in the plane of $m_e / m_h =
1$, and only results with $m_e / m_h \leqslant 1$ need to be shown.

\begin{figure}[h!]
  \centering
  \begin{minipage}{.48\textwidth}
    \centering
    \includegraphics[width=.95\linewidth]{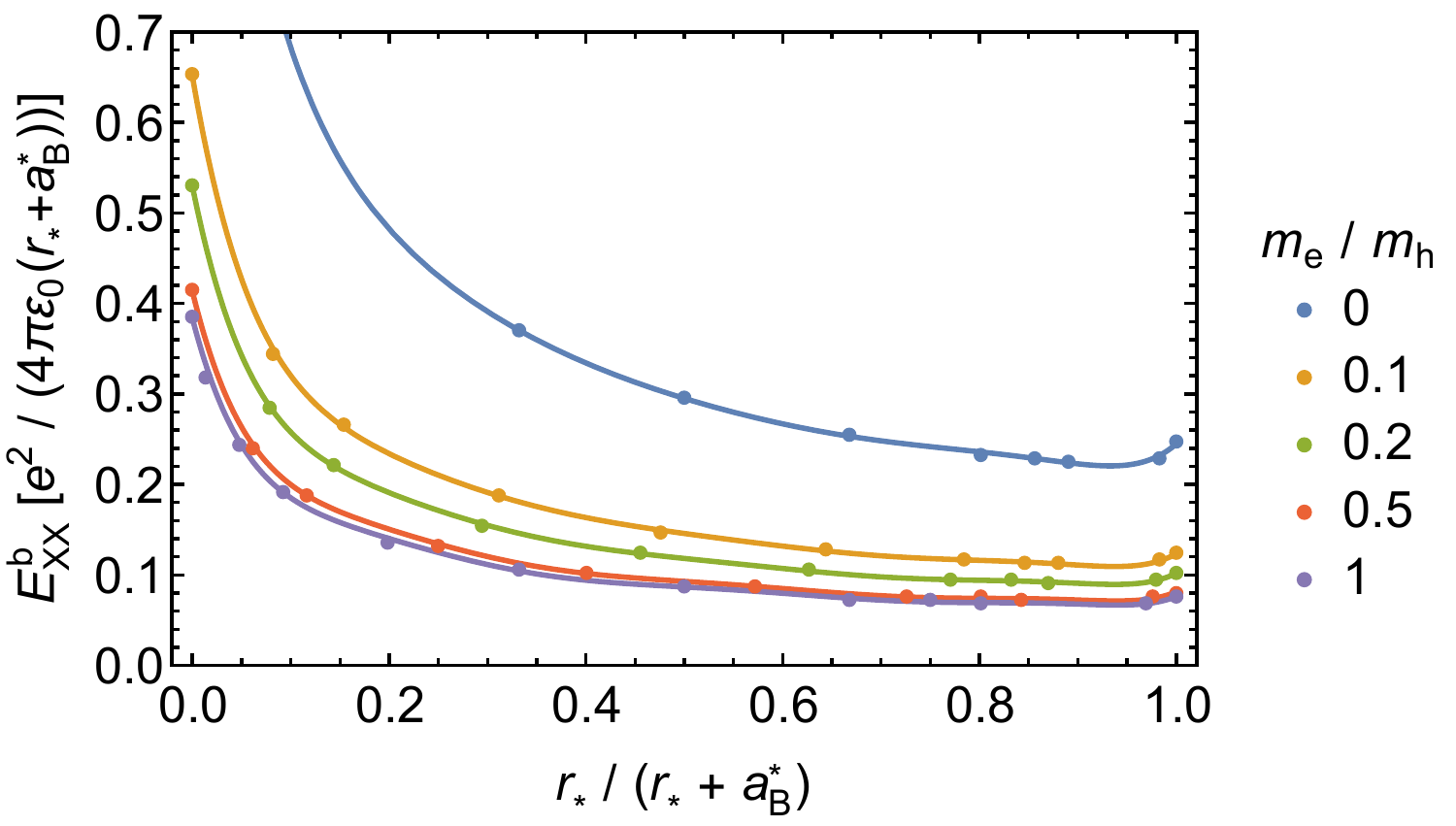}
    \captionof{figure}{
      Biexciton binding energy as a function of rescaled susceptibility for set 
      values of the mass ratio.
    }
    \label{fig:2-XXbinding2}
  \end{minipage}\hspace{0.03\linewidth}
  \begin{minipage}{.48\textwidth}
    \centering
    \includegraphics[width=.8\linewidth]{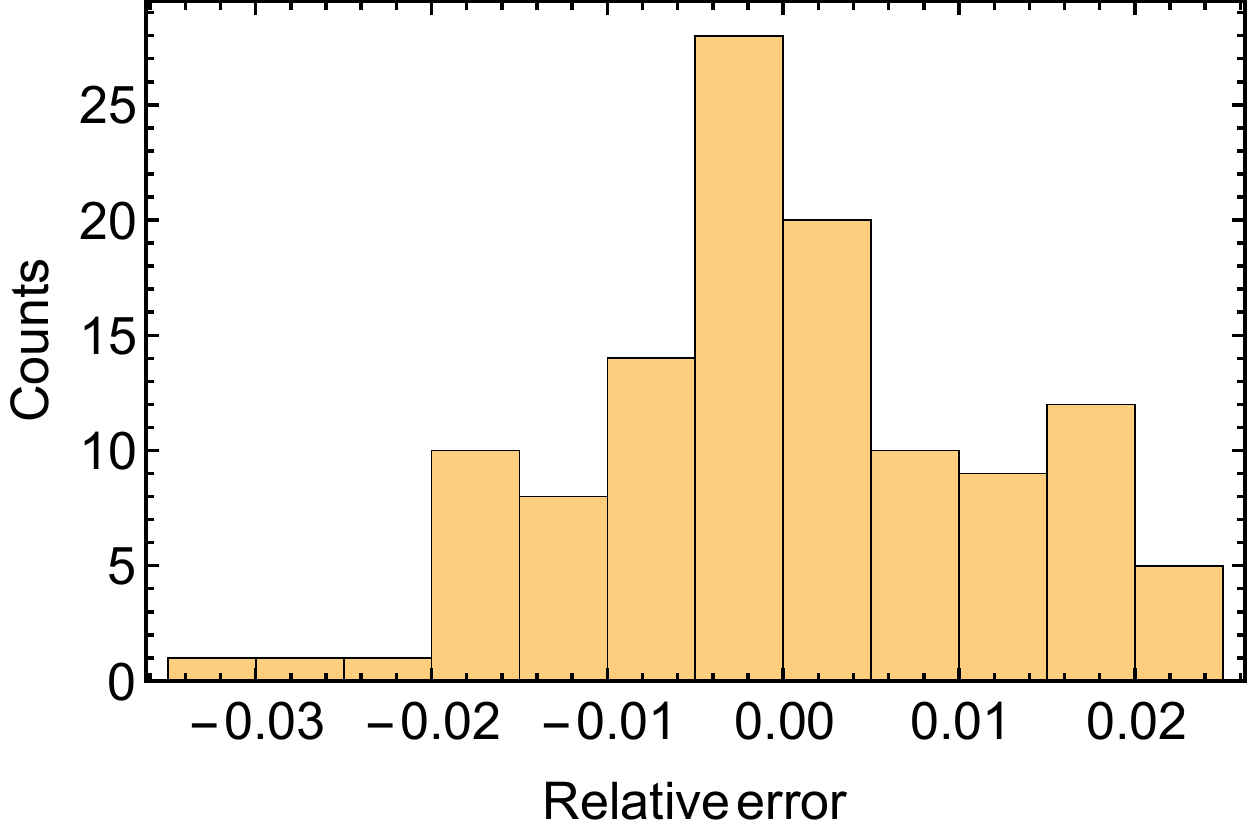}
    \captionof{figure}{
      Histogram of the relative error of the interpolation formula from
      Eq.~(\ref{eq:2-XXbinding-fit}).
    }
    \label{fig:2-XXhisto}
  \end{minipage}
\end{figure}

In the extreme mass ratio of $m_e / m_h \rightarrow 0$, the complex resembles
the H$_2$ molecule, and near the limit we can use the Born-Oppenheimer
approximation to see that the binding energy should have a square root
behaviour in the mass ratio (\tmtextit{cf.} Ch.~\ref{ch:2-extrememass}), which
is indeed observed on the plot in Fig.~\ref{fig:2-XXbinding}. Since we expect
the binding energy to be a smooth function of parameters $m_e / m_h$ and
$r_{\ast}$, and because it is symmetric at the line $m_e / m_h = 1$, we
conclude that at $m_e / m_h = 1$ the biexcitonic binding energy must have zero
slope.

A fitting formula was devised that incorporates the behaviour of the
biexcitonic binding energy described in the previous paragraph,
\begin{eqnarray}
  E_{\tmop{XX}}^{\text{b}} & = & \frac{e^2}{4 \pi \varepsilon_0
  a_{\text{B}}^{\ast}}  \left( 1 - \sqrt{\nu} \right)  \left( 1 - 1.2
  \sqrt{\eta (1 - \eta)} \right)  \\
  &  & \times \left[ \underset{}{} 2. - 17. \nu + 43. (\nu^{3 / 2} + \nu^2) +
  15.7 \nu^{5 / 2} \right] . \nonumber
\end{eqnarray}
The formula is accurate up to 5\% for $0.2 \leqslant m_e / m_h \leqslant 5$.

A more accurate formula was also devised, which correctly reproduces all the
limits,
\begin{equation}
  E_{\tmop{XX}}^{\text{b}} = \frac{e^2}{4 \pi \varepsilon_0  \left( r_{\ast} +
  a_{\text{B}}^{\ast} \right)}  \sum_i \sum_j a_{i j} [(1 - \eta)^{i / 2} +
  \eta^{i / 2}] \nu^j, \label{eq:2-XXbinding-fit}
\end{equation}
with the fitting coefficients given in Appendix
\ref{ch:appendixQMCfittingXXminus}. The relative error histogram from
Fig.~\ref{fig:2-XXhisto} shows that the estimated values have $\sim 2\%$
accuracy.

\subsection{Donor-bound trion and donor-bound biexciton}

In Fig.~\ref{fig:2-loglimit} we present the binding energies in the
logarithmic limit of all previously considered complexes, and the donor-bound
trion and donor-bound biexciton complexes.

The D$^+$X$^-$ complex in the limit of infinite hole mass resembles an H$_2$
molecule and we can see that there is a square root behaviour in the mass
ratio near this limit, in agreement with the Born\mbox{-}Oppenheimer
approximation. We can also see that there seems to be a~square root behaviour
in the limit of a very heavy electron.

\begin{figure}[h]
  \resizebox{12cm}{!}{\includegraphics{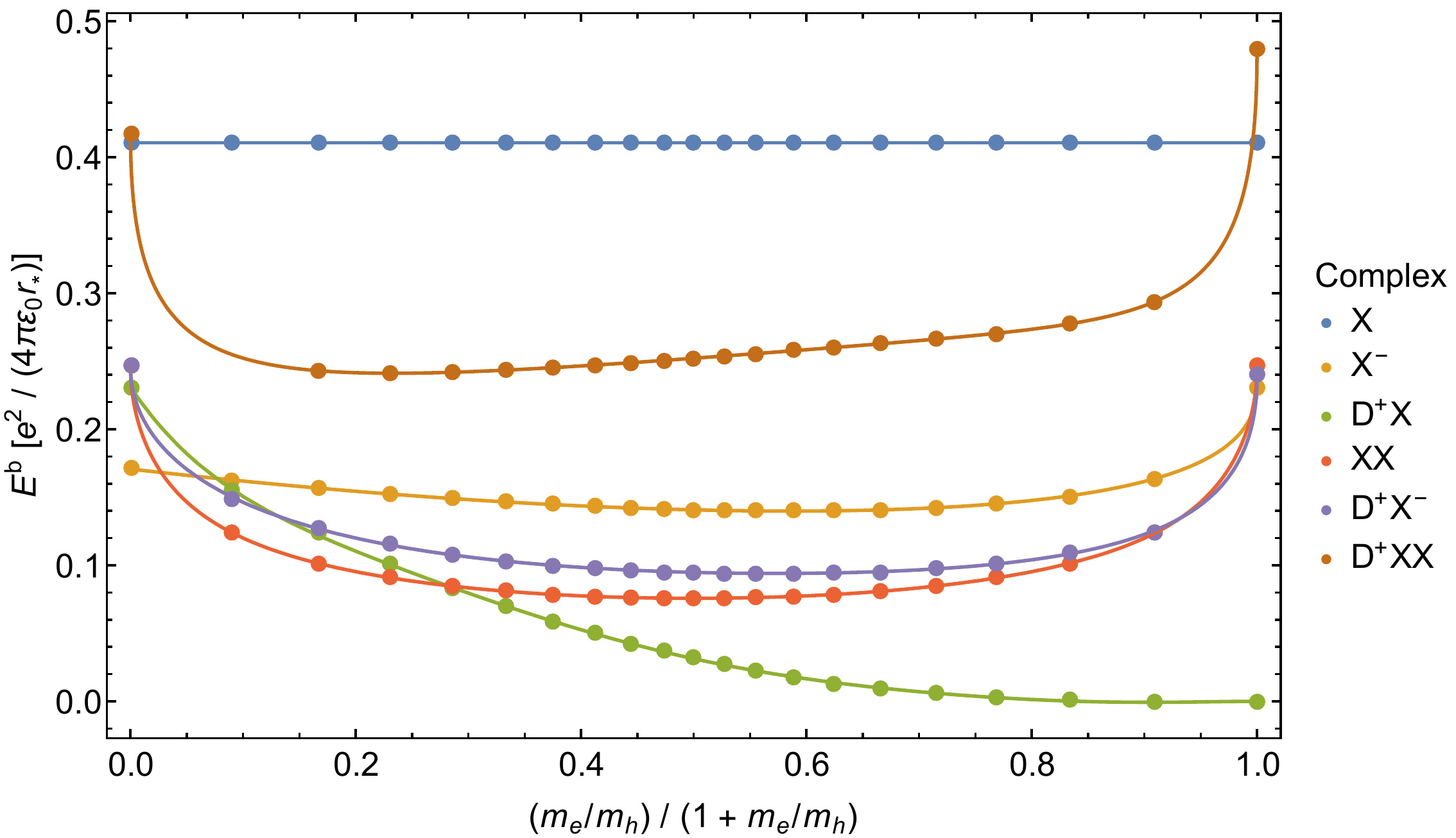}}
  \caption{Binding energy of charge carrier complexes in the logarithmic
  limit. Excitonic energy does not include the logarithmic divergence. Trion
  and donor-bound exciton results are taken from Ref.~{\cite{Bogdan}},
  biexciton results were calculated by Elaheh Mostaani, and the donor-bound
  trion results were obtained by Cameron Price.\label{fig:2-loglimit}}
\end{figure}

The donor-bound biexciton in the limit of heavy electrons has already been
discussed in Ch.~\ref{ch:2-extrememass}. In the limit of heavy holes ($m_e /
m_h \rightarrow 0$), this complex consists of three fixed donors and two light
electrons and there is a question of how the three donors can be positioned
with respect to each other. The most natural position that three particles of
the same sign would assume is an equilateral triangle. To check if this
assumption is correct we first determined how the total energy changes if we
distribute the three donors in the corners of equilateral triangle and then
vary the triangle side. Figure \ref{fig:2-DDDee-side} shows an example case of
$r_{\ast} / a_{\text{B}}^{\ast} = 1$.

\begin{figure}[h]
  \resizebox{8.5cm}{!}{\includegraphics{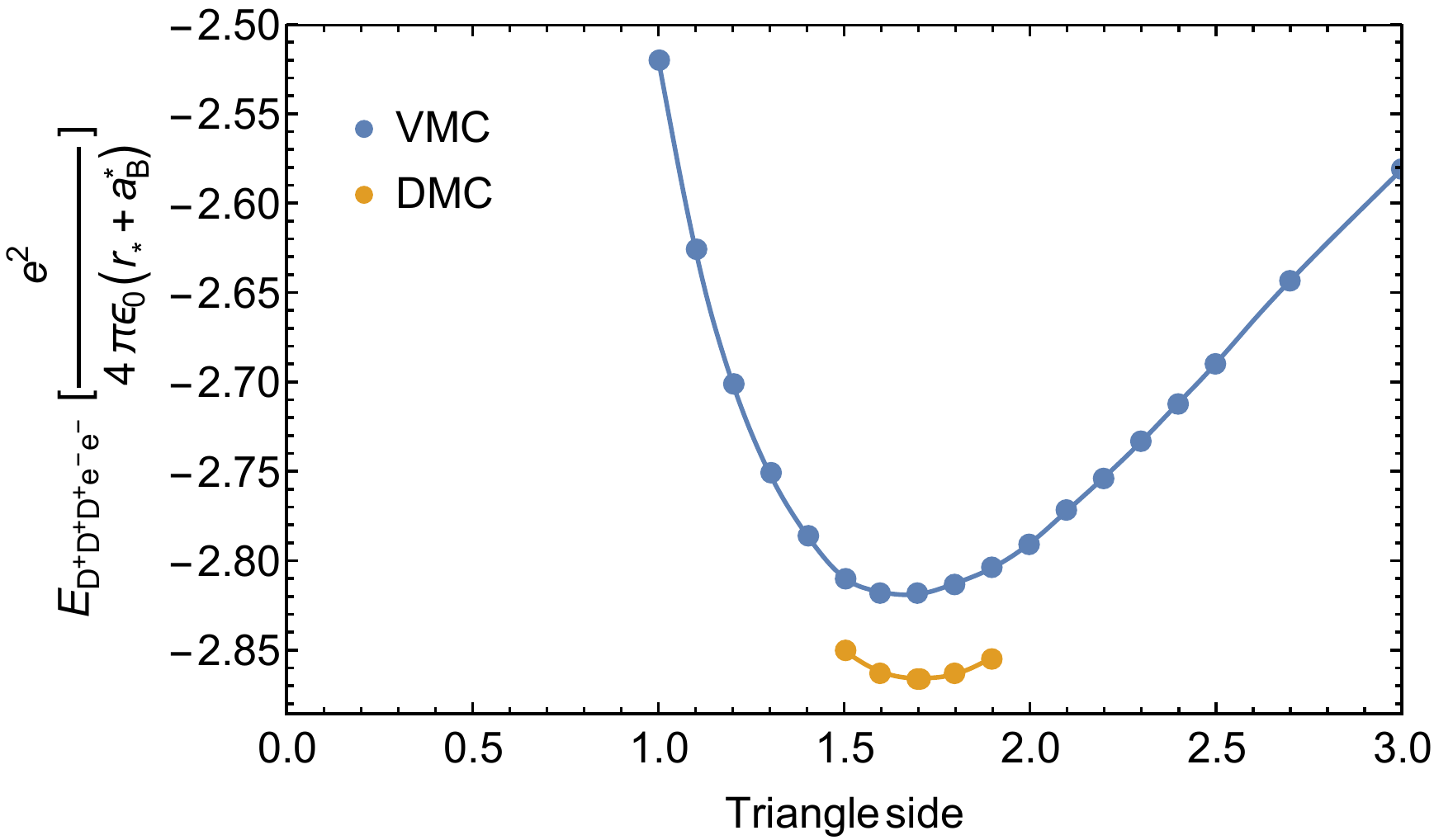}}
  \caption{Total energy of the complex of three donors and two electrons, with
  the donors placed in the corners of an equilateral triangle. Example for
  $r_{\ast} / a_{\text{B}}^{\ast} = 1$.\label{fig:2-DDDee-side}}
\end{figure}

\begin{figure}[h]
  \resizebox{!}{6cm}{\includegraphics{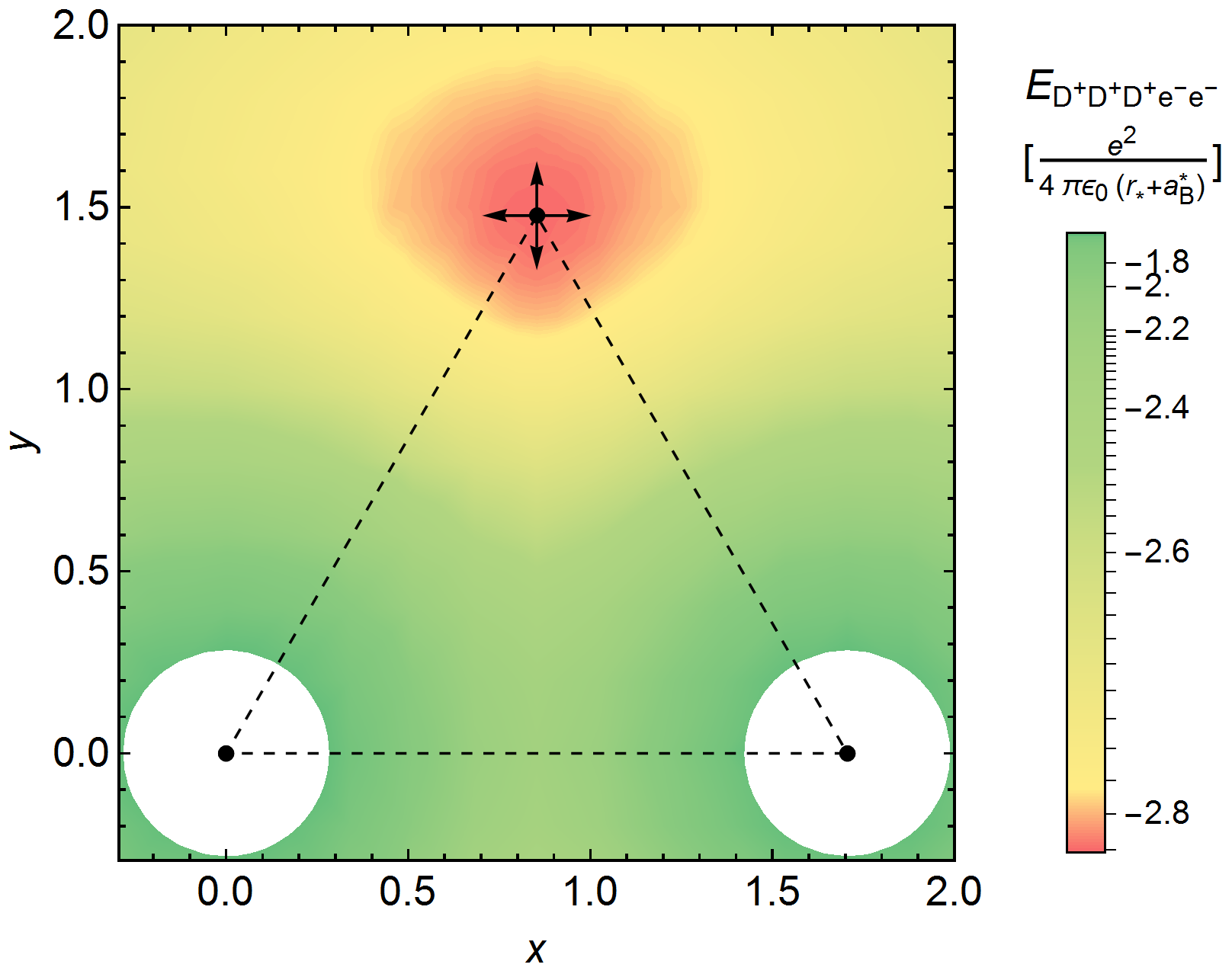}}
  \
  \resizebox{!}{6cm}{\includegraphics{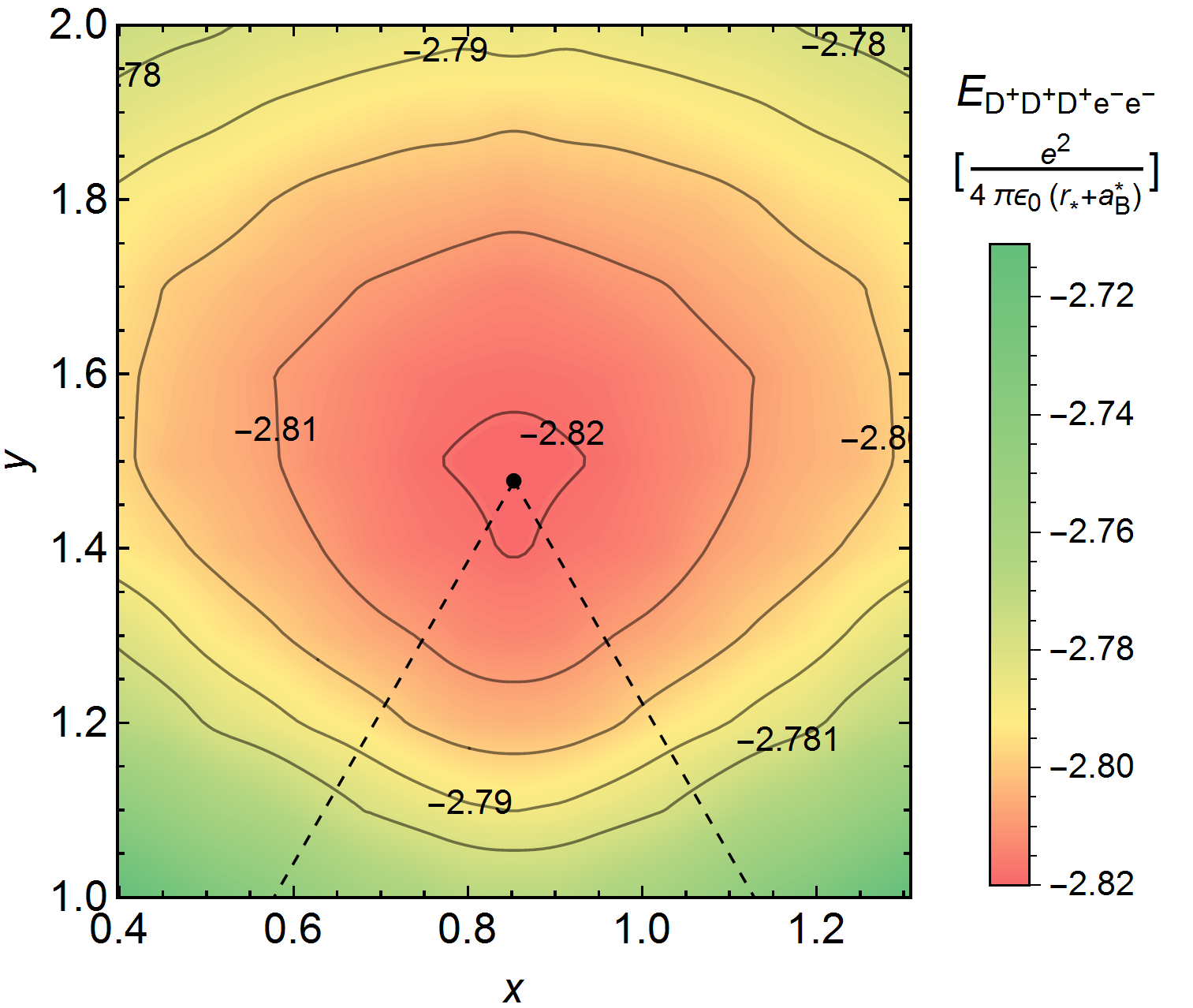}}
  \caption{Total energy of the complex of three donors and two electrons. Left
  plot shows the zoom. We fix two of the donor atoms and change the position
  of the third one.\label{fig:2-DDDee-triangle}}
\end{figure}

After finding the side length $a_{\min}$ that minimises the total energy of
the system, we then change the position of one of the donor atoms (fix the
remaining two to $(0, 0)$ and $(a_{\min}, 0)$ coordinates) and again observe
the effect on the total energy. Figure \ref{fig:2-DDDee-triangle} presents the
results, which clearly show that the preferred structure is indeed an
equilateral triangle.

\section{Contact pair correlation density}\label{ch:contactPDF}

The Mott--Wannier model of charge carrier complexes is valid provided the
complexes extend over many unit cells of the underlying crystal. However, when
charge carriers are present at the same point of space there is likely to be
an energy penalty due to exchange effects and the distortion of the same set
of unit cells. We may represent this effect by introducing additional contact
interactions $A_{\mathrm{ee}} \delta (\vec{r}_{\mathrm{e{\uparrow}}} -
{\nobreak} \vec{r}_{\mathrm{e{\downarrow}}})$ and $A_{\mathrm{eh}}  [\delta
(\vec{r}_{\mathrm{e{\uparrow}}} - \vec{r}_{\mathrm{h}}) + \delta
(\vec{r}_{\mathrm{e{\downarrow}}} - \vec{r}_{\mathrm{h}})]$, where
$A_{\mathrm{ee}}$ and $A_{\mathrm{eh}}$ are constants and
$\vec{r}_{\mathrm{e{\uparrow}}}$, $\vec{r}_{\mathrm{e{\downarrow}}}$, and
$\vec{r}_{\mathrm{h}}$ are the two electron positions and the hole position in
a negative trion. Determining $A_{\mathrm{ee}}$ and $A_{\mathrm{eh}}$ by
\tmtextit{ab initio} calculations would be extremely challenging, and so we
leave them as free parameters to be determined in experiments. If we evaluate
the additional, small contact interaction within first-order perturbation
theory then we find that the correction due to the contact interaction in
a~negative trion can be written as $A_{\mathrm{eh}}
\rho_{\mathrm{eh}}^{\text{X}^-} (0) + A_{\mathrm{ee}}
\rho_{\mathrm{ee}}^{\text{X}^-} (0)$, where the electron--hole pair density is
\begin{equation}
  \rho_{\text{eh}}^{\text{X}^-} (\vec{r}) = \left\langle \delta \left( \vec{r}
  - \vec{r}_{\text{e}{\uparrow}} + \vec{r}_{\text{h}} \right) + \delta \left(
  \vec{r} - \vec{r}_{\text{e}{\downarrow}} + \vec{r}_{\text{h}} \right)
  \right\rangle,
\end{equation}
and similarly for the electron--electron pair density. The pair density can be
evaluated by binning the interparticle distances sampled in VMC and DMC
calculations and using extrapolated estim{\nobreak}ation (see
Ch.~\ref{ch:2-extrapolatedestimation}).

\begin{figure}[h]
  \resizebox{7.5cm}{!}{\includegraphics{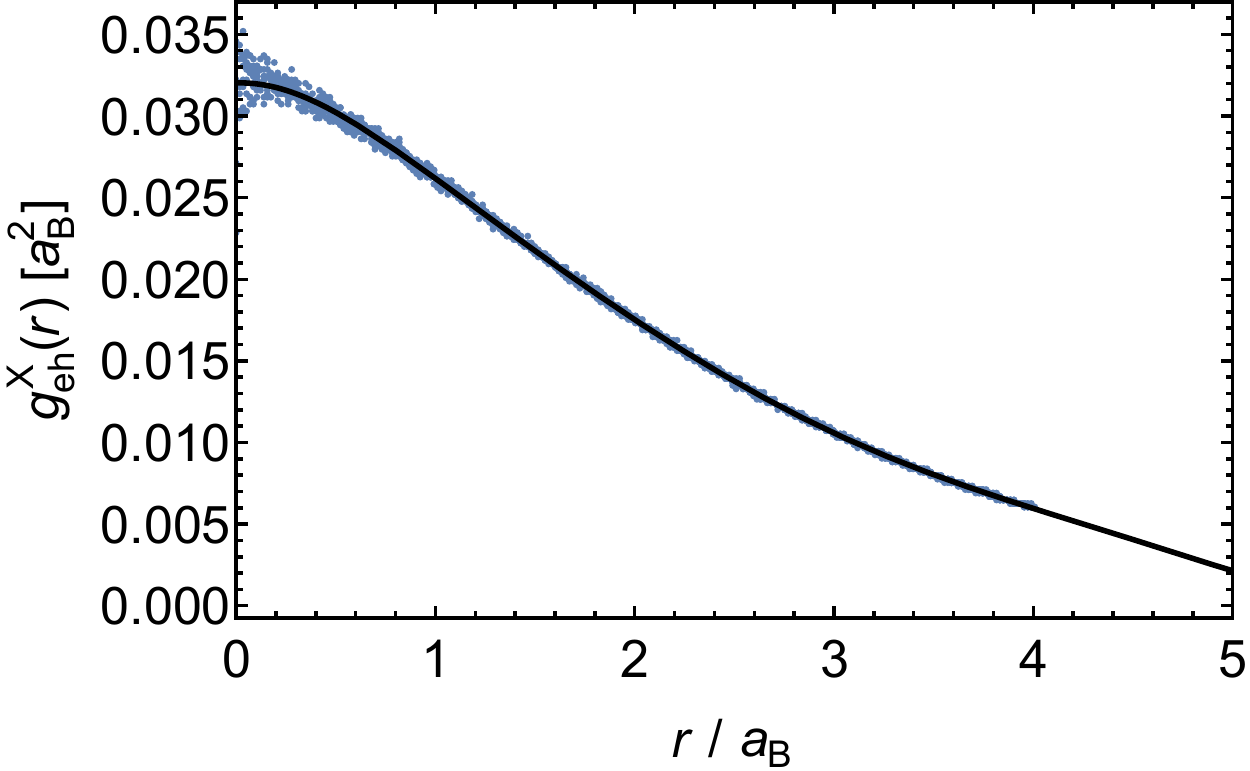}}
  \caption{Example of extrapolation of the exciton pair density
  $g_{\tmop{eh}}^{\text{X}} (r)$ to zero separation using
  Eq.~(\ref{eq:2-kimball1}) (black line). Example for $r_{\ast} = 8
  a_{\text{B}}, m_e / m_h = 0.3$.\label{fig:2-contactexample}}
\end{figure}

Contact pair-density data have been calculated by extrapolating the
electron--hole and electron--electron pair densities to zero separation for
each $m_e / m_h$ and $r_{\ast}$ value considered. The model functions from
Eqs.~(\ref{eq:2-kimball1}) and (\ref{eq:2-kimball2}) were fitted to our
pair-density data with small $r$ (less than $0.1 a_{\text{B}}$--$1
a_{\text{B}}$, depending on the noise in the data), with the data being
weighted by $2 \pi r$ -- see Fig.~\ref{fig:2-contactexample} for an example.
Figure \ref{fig:2-contactPDFtrion} presents the calculated values of contact
electron--hole pair-density for a negative trion.

\begin{figure}[h]
  \resizebox{10cm}{!}{\includegraphics{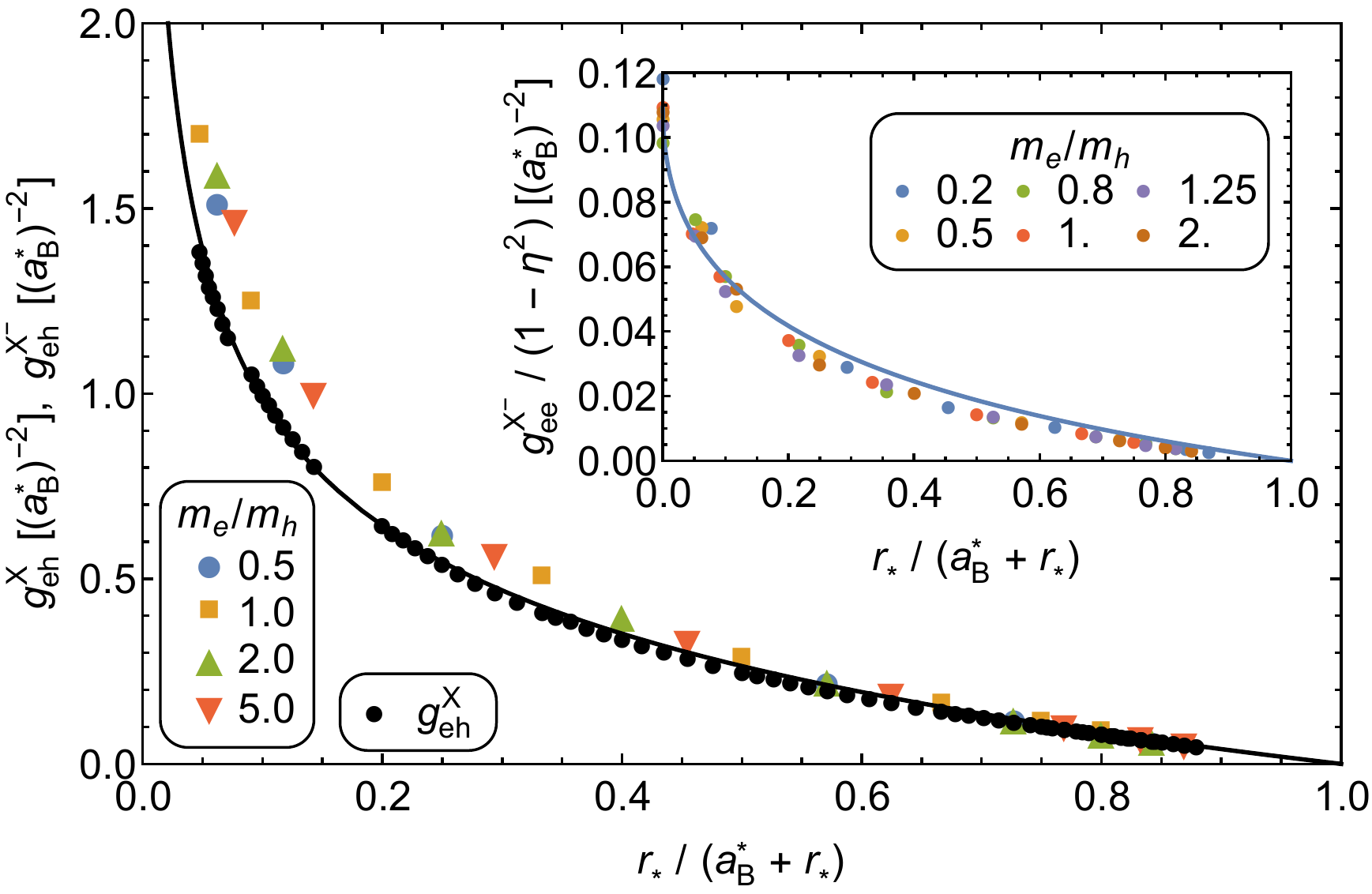}}
  \caption{Contact exciton pair density $g_{\tmop{eh}}^{\text{X}}$, contact
  electron--hole pair density $g_{\tmop{eh}}^{\text{X}^-}$ for a negative
  trion. The black line is the interpolation formula from
  Eq.~(\ref{eq:2-contactint}). The contact electron--electron pair density
  $g_{\tmop{ee}}^{\text{X}^-}$ for a negative trion is shown in an inset. The
  blue line is the interpolation formula for $g_{\tmop{ee}}^{\text{X}^-}$ from
  Eq.~(\ref{eq:2-contactint2}).\label{fig:2-contactPDFtrion}}
\end{figure}

We devised the following interpolation formulas for the contact pair
densities,
\begin{eqnarray}
  g_{\tmop{eh}}^{\text{X}} & \approx & \frac{8.}{\left( a_{\text{B}}^{\ast}
  \right)^2}  \frac{1 - \nu}{1 + 20. \sqrt{\nu}},  \label{eq:2-contactint}\\
  g_{\tmop{eh}}^{\text{X}^-} & \approx & g_{\tmop{eh}}^{\text{X}}, 
  \label{eq:2-contactinttrioneh}\\
  g_{\tmop{ee}}^{\text{X}^-} & \approx & \frac{0.11}{\left(
  a_{\text{B}}^{\ast} \right)^2}  \frac{1 - \sqrt{\nu}}{1 + \sqrt{\nu}}  (1 -
  \eta^2),  \label{eq:2-contactint2}
\end{eqnarray}
which can be used to extract values with up to 5\% error. Equation
(\ref{eq:2-contactinttrioneh}) is valid if the trion wave function can be
approximated as a product of spatially separated exciton and electron wave
functions (see Appendix \ref{ch:AppendixCPDF} for details).

\begin{figure}[h]
  \centering
  \begin{minipage}{.48\textwidth}
    \centering
    \includegraphics[width=.99\linewidth]{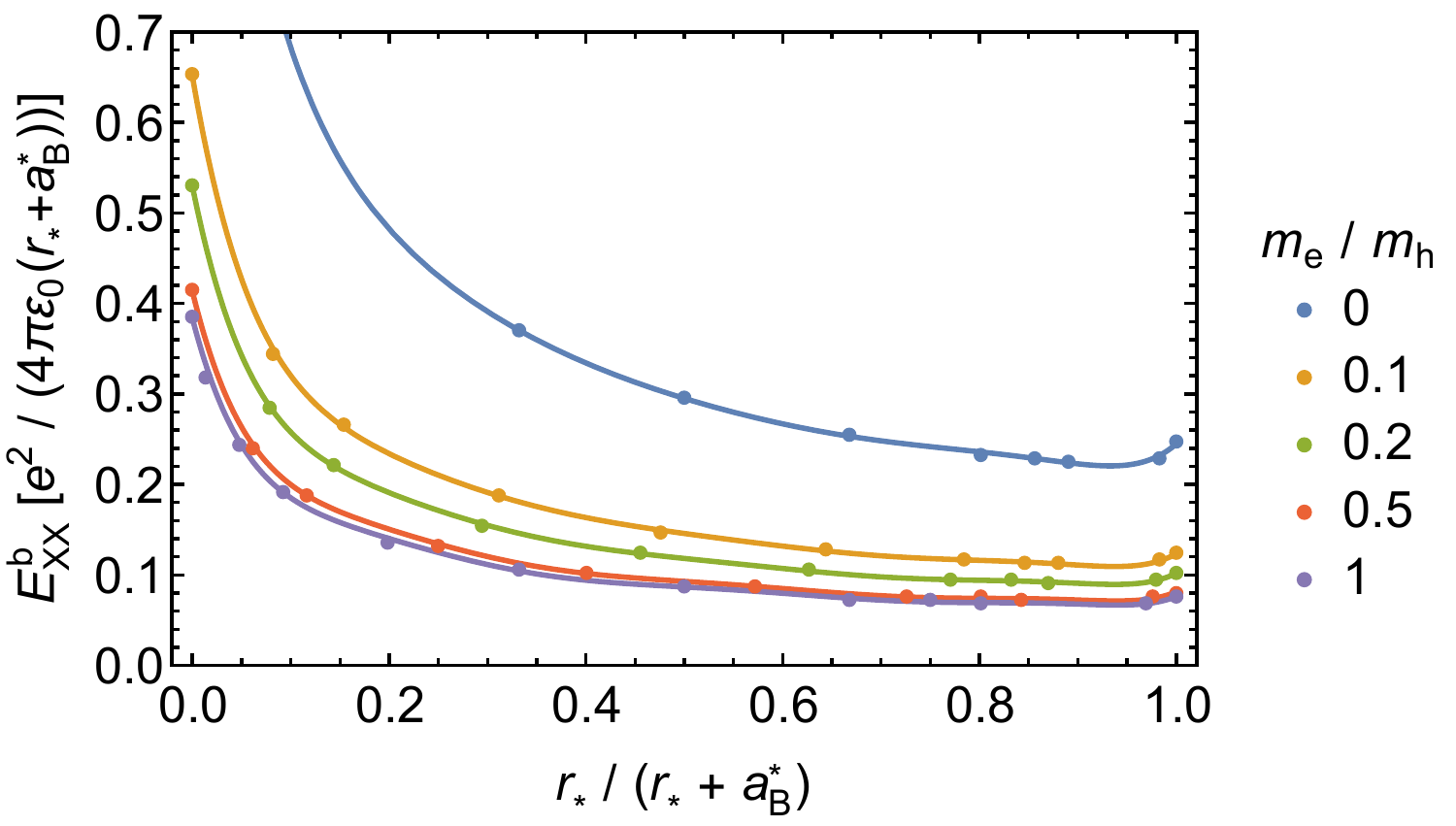}
    \captionof{figure}{
      Electron--hole contact pair density for the biexciton complex. The black 
      line is $2 g_{\tmop{eh}}^{\text{X}}$ (Eq.~\ref{eq:2-contactintxxeh}). 
      Results for $g^{\tmop{XX}}$ were obtained by E.~Mostaani.
    }
    \label{fig:2-contactPDFbiex1}
  \end{minipage}\hspace{0.03\linewidth}
  \begin{minipage}{.48\textwidth}
    \centering
    \includegraphics[width=.8\linewidth]{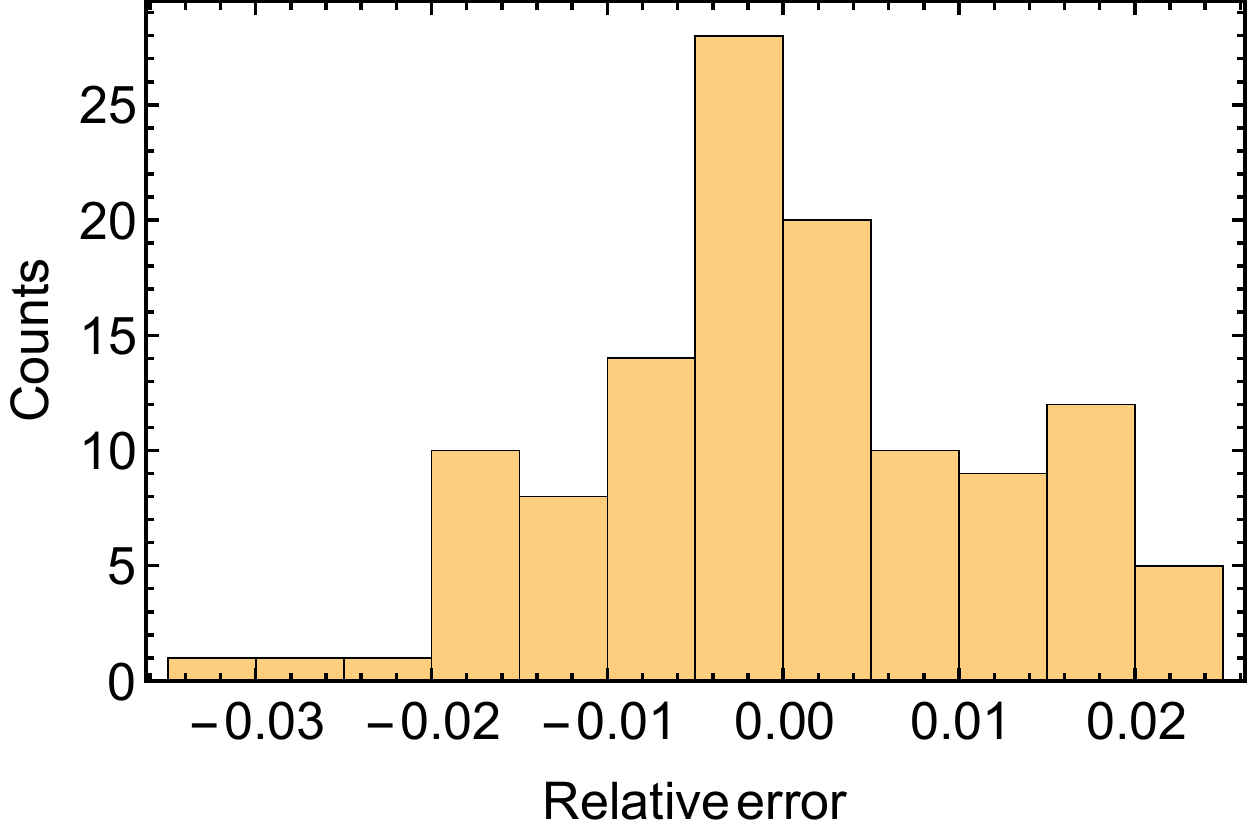}
    \captionof{figure}{
      Electron--electron (ee), and hole--hole (hh) contact pair densities for
      the biexciton complex.
    }
    \label{fig:2-contactPDFbiex2}
  \end{minipage}
\end{figure}

A similar procedure can be used to extract contact pair densities for
biexcitons. The results are presented in Figs.~\ref{fig:2-contactPDFbiex1} and
\ref{fig:2-contactPDFbiex2}. The electron--hole contact pair density can be
approximated as (for explanation see Appendix \ref{ch:AppendixCPDF}):
\begin{eqnarray}
  g_{\tmop{eh}}^{\tmop{XX}} & \approx & 2 g_{\tmop{eh}}^{\text{X}} . 
  \label{eq:2-contactintxxeh}
\end{eqnarray}

\section{Conclusions and comparison with experiments}

Using the DMC values for binding energies of excitons, trions, donor-bound
excitons, and biex{\nobreak}citons, we compiled Table
\ref{tab:2-BEcomparison}, which compares our results with the previous
theoretical and experimental work.

We can clearly see that the logarithmic limit results overestimate the values
of the binding energies: near $r_{\ast} \rightarrow \infty$, the energy has a
square root behaviour (see Fig.~\ref{fig:2-xcorr} and its discussion), so the
energy values for finite $r_{\ast}$ are much smaller that the logarithmic
limit values.

Exciton and trion binding energies match well with the experimental values.
Our results are also consistent with the subsequent DMC calculations
{\cite{Mayers2015}}, the results from the stochastic variational method
{\cite{Zhang2015excited}}, and the theoretical work that uses path-integral
Monte Carlo {\cite{Kylanpaa2015}}. However, values obtained using density
functional theory with the random phase approximation (DFT+RPA) seem to
underestimate the binding energies. We also notice that the positive trion
(X$^+$) energies are slightly smaller than the negative trion (X$^-$) energies
for a given material.

In addition to TMDCs, we have also calculated binding energies for
phosphorene, to see if the theory considered in this work goes beyond the
description of TMDCs only. We can see that although the results seem to match
the experimental values, the spread in the experimental data is high enough
that we cannot conclude yet if the match is real.
{\setlength{\tabcolsep}{0.25em}
\begin{table}[h!]
  \begin{tabular}{llcccccc}
    \hline \hline
    & \multirow{2}{*}{Material} & \multirow{2}{*}{$\frac{m_e}{m_h}$} & 
    \multirow{2}{*}{$r_{\ast}  [\text{\r A}]$} & 
    \multicolumn{4}{c}{$E^{\text{b}}$ 
    [meV] 
    } \\
    \cline{5-8}
    &  &  &  & This work &
    Theory & Log & Experiments\\
    \hline
    X & MoS$_2$ & 0.7 & 41.5 & 580 & 540 {\cite{Berkelbach2013}}, 550
    {\cite{Zhang2015excited,Mayers2015}} & -- & 570 {\cite{Klots2014}}, 500
    {\cite{Mai2014}}\\
    & MoSe$_2$ & 0.7 & 51.7 & 500 & 470 {\cite{Berkelbach2013}}, 480
    {\cite{Zhang2015excited,Mayers2015}} & -- & 550 {\cite{Ugeda2014}}\\
    & MoTe$_2$ & 0.8 & 60.0 & 450 &  & -- & \\
    & WS$_2$ & 0.6 & 37.9 & 560 & 500 {\cite{Berkelbach2013}}, 520
    {\cite{Zhang2015excited,Mayers2015}} & -- & 320 {\cite{Chernikov2014}},
    700 {\cite{Ye2014}}\\
    & WSe$_2$ & 0.6 & 45.1 & 500 & 450 {\cite{Berkelbach2013}}, 470
    {\cite{Zhang2015excited,Mayers2015}} & -- & 370 {\cite{He2014}}\\
    & WTe$_2$ & 0.4 & 53.9 & 440 &  & -- & \\
    & Phosph. & 1.1 & 3.66 & 1460 &  & -- & \\
    \hline
    X$^-$ & MoS$_2$ &  &  & 37 & 26 {\cite{Berkelbach2013}}, 33.8
    {\cite{Mayers2015}} & 50 & 34 {\cite{Zhang2014}}, 35{\cite{Zhang2013}},
    40{\cite{Lin2014,Lui2014}}\\
    & MoSe$_2$ &  &  & 32 & 21 {\cite{Berkelbach2013}}, 28.4
    {\cite{Mayers2015}} & 40 & 30 {\cite{Ross2013,Singh2014}}\\
    & MoTe$_2$ &  &  & 28 &  & 34 & 25 {\cite{Lezama2015}}\\
    & WS$_2$ &  &  & 37 & 26 {\cite{Berkelbach2013}}, 34.0
    {\cite{Mayers2015}} & 55 & 20--40\cite{Mitioglu2013},
    34{\cite{Zhu2015}}, 36{\cite{Chernikov2014}}\\
    & WSe$_2$ &  &  & 32 & 22 {\cite{Berkelbach2013}}, 29.5
    {\cite{Mayers2015}} & 46 & 30 {\cite{Srivastava2015,Jones2013,Wang2014}}\\
    & WTe$_2$ &  &  & 29 &  & 40 & \\
    & Phosph. &  &  & 132 &  & 553 & 100 {\cite{Yang2015}}, 90--190
    {\cite{Zhang2014phosph}}\\
    \hline
    X$^+$ & MoS$_2$ &  &  & 37 &  & 49 & \\
    & MoSe$_2$ &  &  & 31 &  & 39 & 30 {\cite{Ross2013}}\\
    & MoTe$_2$ &  &  & 28 &  & 34 & \\
    & WS$_2$ &  &  & 37 &  & 53 & \\
    & WSe$_2$ &  &  & 32 &  & 45 & 30 {\cite{Srivastava2015}}, 24
    {\cite{Jones2013}}\\
    & WTe$_2$ &  &  & 29 &  & 38 & \\
    & Phosph. &  &  & 130 &  & 556 & \\
    \hline
    D$^+$X & MoS$_2$ &  &  & 10.1 &  & 18 & \\
    & MoSe$_2$ &  &  & 9.1 &  & 14 & \\
    & MoTe$_2$ &  &  & 6.7 &  & 10 & \\
    & WS$_2$ &  &  & 12.0 &  & 23 & \\
    & WSe$_2$ &  &  & 10.9 &  & 19 & \\
    & WTe$_2$ &  &  & 15.3 &  & 22 & \\
    & Phosph. &  &  & 7.3 &  & 105 & \\
    \hline
    XX & MoS$_2$ &  &  & 23 & 22\tmtextit{.}7 {\cite{Mayers2015}},
    22\tmtextit{.}7(5) {\cite{Kylanpaa2015}} & 27 & 60 {\cite{Sie2015}}, 70
    {\cite{Mai2015}}\\
    & MoSe$_2$ &  &  & 19 & 17\tmtextit{.}7 {\cite{Mayers2015}},
    19\tmtextit{.}3(5) {\cite{Kylanpaa2015}} & 22 & \\
    & MoTe$_2$ &  &  & 16 & 14\tmtextit{.}4(4) {\cite{Kylanpaa2015}} & 18 &
    \\
    & WS$_2$ &  &  & 25 & 23\tmtextit{.}3 {\cite{Mayers2015}},
    23\tmtextit{.}9(5) {\cite{Kylanpaa2015}} & 30 & 65
    {\cite{Plechinger2015}}\\
    & WSe$_2$ &  &  & 21 & 20\tmtextit{.}2 {\cite{Mayers2015}},
    20\tmtextit{.}7(5) {\cite{Kylanpaa2015}} & 25 & 52 {\cite{You2015}}\\
    & WTe$_2$ &  &  & 19 &  & 23 & \\
    & Phosph. &  &  & 135 &  & 300 & \\
    \hline \hline
  \end{tabular}
  \caption{Exciton (X), negative and positive trion (X$^-$ and X$^+$),
  donor-bound exciton (D$^+$X) and biexciton (XX) binding energies for
  selected TMDCs and phosphorene (``Phosph.''). We show our numerical DMC 
  results, previous
  numerical calculations using the DFT+RPA method {\cite{Berkelbach2013}}, the
  stochastic variational method {\cite{Zhang2015excited}}, subsequent (to our
  work) DMC results {\cite{Mayers2015}}, and results from the
  path\mbox{-}integral Monte Carlo method {\cite{Kylanpaa2015}}. We also show
  logarithmic limit results (``Log'', valid for $r_{\ast} \rightarrow
  \infty$), and various experimental values. For the trion, donor-bound
  exciton and biexciton, the appropriate values of $m_e / m_h$ and $r_{\ast}$
  are the same as for an exciton and thus are left blank. Parameters for TMDCs
  are taken from
  Refs.~\cite{Bogdan,Berkelbach2013,Kumar2012,Dawson1987,Liu2011,Lv2015}.
  Phosphorene parameters can be found in
  Ref.~{\cite{Seixas2015}}.\label{tab:2-BEcomparison}}
\end{table}}

Although the results for biexciton (XX) binding energy seem to agree between
all theoretical works, there is a huge discrepancy with the experiments. The
experimental values seem to be three times higher than the theoretical ones
(compare for example the expected spectrum in Fig.~\ref{spectrum} with the one
from Ref.~{\cite{You2015}}). We provide a few possible explanations of this
disagreement. Since the DMC method is statistically exact for these systems,
either the model is inappropriate or there is an issue with the experimental
results.

Firstly, the determination of which absorption/luminescence peak matches a
specific charge carrier complex may be impossible or very hard in some types
of experiment, and therefore the mismatch in the binding energies may be
simply due to misclassification of the peaks. For example, the biexciton peak
may actually be a trion peak and vice versa. If this is the case, then the
theory underestimates the binding energies roughly twice, which could be due
to the contact interaction not being included in our calculations.

Secondly, it may be that the contact interaction for a biexciton complex has
very high values of the constants $A_{i j}$ (from
Chapter~\ref{ch:contactPDF}). However, there is no physical reason why the
contact interaction should be much different for excitons, trions and
biexcitons, and therefore we would rather expect to see similar discrepancy
for all the complexes.

If the Mott-Wannier model breaks down for the system of charge carrier
complexes, then we should not be able to describe accurately binding energies
of excitons and trions. However, the sizes of the complexes are much bigger
than the lattice constant, and therefore the Mott-Wannier model should be
valid in our considerations.

Lastly, there is one prominent difference between Fig.~\ref{spectrum} and any
experimental spectrum: due to finite temperature in the experiment, the peaks
will always have a non-zero thickness (spread), and therefore if two peaks are
close to each other, it is very hard to experimentally distinguish the two
unless high precision is obtained. Hence, the XX peak may be obscured in the
data by the background noise and the trion peaks. The peak that is now
identified as the XX peak may then be another complex, for example a biexciton
bound to an impurity (as Fig.~\ref{fig:2-loglimit} suggests that the
donor\mbox{-}bound biexciton has a higher binding energy than a trion).
Ref.~{\cite{Zhang2015excited}} also shows that an excited biexciton state may
be responsible for this peak.

The contact interaction may be determined in the future using our contact pair
density results. Either one has to calculate the values of the contact
interaction constants \tmtextit{ab initio}, or they could be determined
through the experiment. Ref.~{\cite{Qiu2016}} also suggests that instead of
approximating the effective interaction as a combination of the Keldysh
interaction and a~Dirac delta contact interaction, one can treat the in-plane
dielectric constant~$\varepsilon$ as function to be determined \tmtextit{ab
initio}. However, that approach requires determination of the dielectric
function $\varepsilon$ for each 2D material separately. The method used in
this work is much more general, and our results can be used for any known or
as yet unknown 2D semiconductor for which the Mott-Wannier approximation with
the Keldysh interaction is valid.

To summarise, using diffusion Monte Carlo we have studied charge carrier
complexes in 2D semiconductors, in particular in transition-metal
dichalcogenides. The binding energies obtained are statistically correct and
were calculated for a full range of mass ratios and in\mbox{-}plane
sus{\nobreak}cept{\nobreak}ibil{\nobreak}ities. Excitonic and trionic energies
match the experiments very well, but biexcitonic binding energies are greatly
underestimated. We provided possible explanations of this behaviour. A
classification of trions and biexcitons is also presented. Finally, we have
given results for the contact pair densities, which may be used in the future
to determine the contact interaction between charge carriers.

\chapter*{Epilogue}

Although nanotechnology provides a prospect of expanding our technological
capabilities, the theoretical understanding of quantum systems is a crucial
first step on this road. In this work, we have studied low-dimensional quantum
systems with properties that could be used to advance current technology
beyond its present limits.

Firstly, a one-dimensional generalised $t$-$V$ model of fermions on a lattice
was investigated. The fermions have finite-range interactions that cause the
existence of Mott insulating densities in the system. Otherwise, the model
behaves as a Luttinger liquid. We have succeeded in showing how to extend
previous analytical analysis past nearest-neighbours interactions. We have
also adapted a~new method, the strong coupling expansion, usually used in
investigations of lattice field theories, and used it to determine
higher\mbox{-}order corrections to the ground-state energy and critical
parameters of the extended $t$\mbox{-}$V$ model near the insulating phase. The
method is insensitive to the presence or absence of integrability and goes
beyond perturbation theory. It works best for systems with low degeneracy of
the unperturbed ground state, and was shown in this work to be very versatile,
as it can be used both analytically and numerically. We have also summarised
the strong and weak points of all the methods that were used on the
generalised $t$\mbox{-}$V$ model, in order to provide guidance in choosing the
correct methodology for future investigations of models with long-range
interactions.

Phase diagrams that include possible charge density waves of the system were
also studied. We have shown how to determine analytically all possible phases
and their energies in low Mott insulating densities for any value of the
interaction range. Higher densities were investigated using brute\mbox{-}force
analysis and example systems were used to show that the number of possible
insulating phases grows quickly with the range of interactions. At finite
temperature, this may indicate the loss of insulating properties of the
system.

Due to the generality of this model (\tmtextit{i.e.} the potentials considered
have no given values), it may describe an experimental one\mbox{-}dimensional
system of fermions in an optical lattice. Otherwise, it~provides a theoretical
framework for how to deal with and what to expect from systems with
finite-range interactions.

The second model investigated was a system of charge carrier complexes in
two-dimensional semiconductors. In transition-metal dichalcogenides, complexes
of two (excitons), three (trions) and four (biexcitons) charge carriers were
found experimentally to have large binding energies that are prominently
visible on the photoluminescence and photoabsorption spectra. These complexes
are crucial in the understanding of the optoelectronic properties of the 2D
semiconductors. We have provided a~classification of trions and biexcitons in
transition-metal dichalcogenides that incorporates the difference in spin
polarisation for molybdenum- and tungsten-based materials, and that can be
used to explain the fine structure in the spectra of those materials. Using
diffusion Monte Carlo, a numerical method that is statistically exact for the
charge carrier complexes, we have calculated the binding energies of complexes
with distinguishable particles. Our investigations also include a case where a
complex is bound to a charged impurity. The results were found to be
consistent with other theoretical and experimental work. Our results are
however much more complete: we provide a full range of results that are
calculated using the Mott-Wannier model with the Keldysh interaction.

There is however a~disagreement between the theory and experiments on the
biexciton binding energy. We suggest some resolutions of this issue: it may be
either an effect of misclassification of the peak in the experiment,
underestimation of the contact interaction in the theory, or the combination
of both. We have also extracted contact pair densities, which in the future
may be used to determine the strength of the contact interaction.

Our results have one major advantage: due to the full spectrum of input
parameters (effective masses and the in-plane susceptibility of the material)
that were investigated, they can be used to determine properties of charge
carriers in a wide range of systems, rather than being focused on just a
number of existing materials. We have also provided interpolating formulas
that can be utilised to easily extract binding energies for any
two\mbox{-}dimensional semiconductor, for which the Mott-Wannier model is
applicable.

In conclusion, our work provides a major theoretical advancement in the
understanding of one- and two-dimensional quantum systems that have possible
applications in electronic devices and it is our hope that it will be used in
the near future to advance our technological progress.

\appendix
\chapter{Strong coupling expansion}

\section{Truncated Hamiltonians}\label{ch:appendixsceham}

Truncated Hamiltonians from Chapter \ref{ch:1-SCEcritical}. For the sake of
simplicity, the zeros in the truncated Hamiltonians are represented as dots.
Off\mbox{-}diagonal elements should be multiplied by $(- t)$.

\paragraph{$Q = 1 / 2$ (half-filling), $p = 1$ (integrable), SCE step 3:}
\begin{equation}
  \resizebox{\textwidth}{!}{$
  \scriptsize
  \left(\begin{array}{cccccccc}
    \cdot & \sqrt{L} & \cdot & \cdot & \cdot & \cdot & \cdot & \cdot\\
    \sqrt{L} & U & 2 & \sqrt{2 L - 10} & \cdot & \cdot & \cdot & \cdot\\
    \cdot & 2 & U & \cdot & 2 & \sqrt{L - 6} & \cdot & \cdot\\
    \cdot & \sqrt{2 L - 10} & \cdot & 2 U & \sqrt{\frac{2}{L - 5}} & (L - 7)
    \sqrt{\frac{8}{(L - 6) (L - 5)}} & \sqrt{\frac{8 (L - 7)}{(L - 6) (L -
    5)}} & \sqrt{\frac{3 (L - 7) (L - 8)}{L - 5}}\\
    \cdot & \cdot & 2 & \sqrt{\frac{2}{L - 5}} & U & \cdot & \cdot & \cdot\\
    \cdot & \cdot & \sqrt{L - 6} & (L - 7) \sqrt{\frac{8}{(L - 6) (L - 5)}} &
    \cdot & 2 U & \cdot & \cdot\\
    \cdot & \cdot & \cdot & \sqrt{\frac{8 (L - 7)}{(L - 6) (L - 5)}} & \cdot &
    \cdot & 2 U & \cdot\\
    \cdot & \cdot & \cdot & \sqrt{\frac{3 (L - 7) (L - 8)}{L - 5}} & \cdot &
    \cdot & \cdot & 3 U
  \end{array}\right)
  $}
\end{equation}

\paragraph{$Q = 1 / 3$, $p = 2$ (non-integrable), SCE step 3:}
\begin{equation}
  \resizebox{\textwidth}{!}{$
  \left(\begin{array}{cccccccccccc}
    \cdot & \sqrt{\frac{2 L}{3}} & \cdot & \cdot & \cdot & \cdot & \cdot &
    \cdot & \cdot & \cdot & \cdot & \cdot\\
    \sqrt{\frac{2 L}{3}} & U_2 & \sqrt{3} & 2 & \sqrt{\frac{4 L}{3} - 10} &
    \cdot & \cdot & \cdot & \cdot & \cdot & \cdot & \cdot\\
    \cdot & \sqrt{3} & U_1 & \cdot & \cdot & \sqrt{\frac{1}{3}} &
    \sqrt{\frac{2 L - 17}{3}} & \sqrt{\frac{5}{3}} & \cdot & \cdot & \cdot &
    \cdot\\
    \cdot & 2 & \cdot & U_2 & \cdot & 1 & \cdot & \sqrt{\frac{9}{5}} & 2 &
    \sqrt{\frac{2 L - 21}{3}} & \sqrt{\frac{1}{5}} & \cdot\\
    \cdot & \sqrt{\frac{4 L}{3} - 10} & \cdot & \cdot & 2 U_2 & \cdot &
    \sqrt{\frac{6 (2 L - 17)}{2 L - 15}} & \sqrt{\frac{24}{5 (2 L - 15)}} &
    \sqrt{\frac{6}{2 L - 15}} & \sqrt{\frac{8 (2 L - 21)}{2 L - 15}} & -
    \sqrt{\frac{6}{5 (2 L - 15)}} & \sqrt{\frac{2 (L - 12) (2 L - 21)}{2 L -
    15}}\\
    \cdot & \cdot & \sqrt{\frac{1}{3}} & 1 & \cdot & 2 U_2 & \cdot & \cdot &
    \cdot & \cdot & \cdot & \cdot\\
    \cdot & \cdot & \sqrt{\frac{2 L - 17}{3}} & \cdot & \sqrt{\frac{6 (2 L -
    17)}{2 L - 15}} & \cdot & U_1 + U_2 & \cdot & \cdot & \cdot & \cdot &
    \cdot\\
    \cdot & \cdot & \sqrt{\frac{5}{3}} & \sqrt{\frac{9}{5}} &
    \sqrt{\frac{24}{5 (2 L - 15)}} & \cdot & \cdot & U_1 & \cdot & \cdot &
    \cdot & \cdot\\
    \cdot & \cdot & \cdot & 2 & \sqrt{\frac{6}{2 L - 15}} & \cdot & \cdot &
    \cdot & U_2 & \cdot & \cdot & \cdot\\
    \cdot & \cdot & \cdot & \sqrt{\frac{2 L - 21}{3}} & \sqrt{\frac{8 (2 L -
    21)}{2 L - 15}} & \cdot & \cdot & \cdot & \cdot & 2 U_2 & \cdot & \cdot\\
    \cdot & \cdot & \cdot & \sqrt{\frac{1}{5}} & - \sqrt{\frac{6}{5 (2 L -
    15)}} & \cdot & \cdot & \cdot & \cdot & \cdot & U_1 & \cdot\\
    \cdot & \cdot & \cdot & \cdot & \sqrt{\frac{2 (L - 12) (2 L - 21)}{2 L -
    15}} & \cdot & \cdot & \cdot & \cdot & \cdot & \cdot & 3 U_2
  \end{array}\right)
  $}
\end{equation}

\paragraph{$Q = 1 / 4$, $p = 3$ (non-integrable), SCE step 3:}
\begin{equation}
  \resizebox{\textwidth}{!}{$
  \left( \begin{array}{ccccccccccccc}
    \cdot & \sqrt{\frac{L}{2}} & \cdot & \cdot & \cdot & \cdot & \cdot & \cdot
    & \cdot & \cdot & \cdot & \cdot & \cdot\\
    \sqrt{\frac{L}{2}} & U_3 & \sqrt{3} & 2 & \sqrt{L - 10} & \cdot & \cdot &
    \cdot & \cdot & \cdot & \cdot & \cdot & \cdot\\
    \cdot & \sqrt{3} & U_2 & \cdot & \cdot & \sqrt{\frac{10}{3}} &
    \sqrt{\frac{1}{3}} & \sqrt{\frac{1}{6} (3 L - 34)} & \sqrt{\frac{5}{3}} &
    \cdot & \cdot & \cdot & \cdot\\
    \cdot & 2 & \cdot & U_3 & \cdot & \cdot & 1 & \cdot & \sqrt{\frac{9}{5}} &
    2 & \sqrt{\frac{L - 14}{2}} & \sqrt{\frac{1}{5}} & \cdot\\
    \cdot & \sqrt{L - 10} & \cdot & \cdot & 2 U_3 & \cdot & \cdot &
    \sqrt{\frac{2 (3 L - 34)}{L - 10}} & \sqrt{\frac{16}{5 (L - 10)}} &
    \sqrt{\frac{4}{L - 10}} & \sqrt{\frac{8 (L - 14)}{L - 10}} & -
    \sqrt{\frac{4}{5 (L - 10)}} & \sqrt{\frac{3 (L - 16) (L - 14)}{2 (L -
    10)}}\\
    \cdot & \cdot & \sqrt{\frac{10}{3}} & \cdot & \cdot & U_1 & \cdot & \cdot
    & \cdot & \cdot & \cdot & \cdot & \cdot\\
    \cdot & \cdot & \sqrt{\frac{1}{3}} & 1 & \cdot & \cdot & 2 U_3 & \cdot &
    \cdot & \cdot & \cdot & \cdot & \cdot\\
    \cdot & \cdot & \sqrt{\frac{1}{6} (3 L - 34)} & \cdot & \sqrt{\frac{2 (3 L
    - 34)}{L - 10}} & \cdot & \cdot & U_2 + U_3 & \cdot & \cdot & \cdot &
    \cdot & \cdot\\
    \cdot & \cdot & \sqrt{\frac{5}{3}} & \sqrt{\frac{9}{5}} &
    \sqrt{\frac{16}{5 (L - 10)}} & \cdot & \cdot & \cdot & U_2 & \cdot & \cdot
    & \cdot & \cdot\\
    \cdot & \cdot & \cdot & 2 & \sqrt{\frac{4}{L - 10}} & \cdot & \cdot &
    \cdot & \cdot & U_3 & \cdot & \cdot & \cdot\\
    \cdot & \cdot & \cdot & \sqrt{\frac{L - 14}{2}} & \sqrt{\frac{8 (L -
    14)}{L - 10}} & \cdot & \cdot & \cdot & \cdot & \cdot & 2 U_3 & \cdot &
    \cdot\\
    \cdot & \cdot & \cdot & \sqrt{\frac{1}{5}} & - \sqrt{\frac{4}{5 (L - 10)}}
    & \cdot & \cdot & \cdot & \cdot & \cdot & \cdot & U_2 & \cdot\\
    \cdot & \cdot & \cdot & \cdot & \sqrt{\frac{3 (L - 16) (L - 14)}{2 (L -
    10)}} & \cdot & \cdot & \cdot & \cdot & \cdot & \cdot & \cdot & 3 U_3
  \end{array} \right)
  $}
\end{equation}
\section{Ground states formulas}\label{ch:appendixscegs}

Here we present the ground states that diagonalise the truncated Hamiltonians
from the previous section.

\paragraph{$Q = 1 / 2$ (half-filling), $p = 1$ (integrable), SCE step 3:}
\begin{equation}
  | \tmop{GS} \rangle = \left(\begin{array}{c}
    1 - \frac{1}{2} L \frac{t^2}{U_1^2} + \frac{1}{8} (L^2 + 2 L)
    \frac{t^4}{U_1^4} - \frac{1}{48} L (L + 10) (L - 4) \frac{t^6}{U_1^6} + O
    \left( \frac{t^8}{U_1^8} \right)\\
    \sqrt{L} \frac{t}{U_1} - \frac{1}{2} \sqrt{L} (L + 2) \frac{t^3}{U_1^3} +
    \frac{1}{8} \sqrt{L} L (L + 6) \frac{t^5}{U_1^5} + O \left(
    \frac{t^7}{U_1^7} \right)\\
    2 \sqrt{L} \frac{t^2}{U_1^2} - \sqrt{L} (L + 5) \frac{t^4}{U_1^4} + O
    \left( \frac{t^6}{U_1^6} \right)\\
    \sqrt{\frac{1}{2} L (L - 5)} \frac{t^2}{U_1^2} - \sqrt{\frac{1}{8} L (L -
    5)} (L + 4) \frac{t^4}{U_1^4} + O \left( \frac{t^6}{U_1^6} \right)\\
    5 \sqrt{L} \frac{t^3}{U_1^3} + O \left( \frac{t^5}{U_1^5} \right)\\
    (2 L - 13) \sqrt{\frac{L}{L - 6}} \frac{t^3}{U_1^3} + O \left(
    \frac{t^5}{U_1^5} \right)\\
    \sqrt{\frac{L (L - 7)}{L - 6}} \frac{t^3}{U_1^3} + O \left(
    \frac{t^5}{U_1^5} \right)\\
    \sqrt{\frac{1}{6} L (L - 7) (L - 8)} \frac{t^3}{U_1^3} + O \left(
    \frac{t^5}{U_1^5} \right)
  \end{array}\right)
\end{equation}

\paragraph{$Q = 1 / 3$, $p = 2$ (non-integrable), SCE step 3:}
\begin{equation}
  \resizebox{\textwidth}{!}{$
  | \tmop{GS} \rangle = \left(\begin{array}{c}
    1 - \frac{L}{3 U_2^2} t^2 + \left( \frac{L (L + 3)}{18 U_2^4} - \frac{2
    L}{U_1 U_2^3} - \frac{L}{U_1^2 U_2^2} \right) t^4 + \left( - \frac{L (L -
    6) (L + 15)}{162 U_2^6} + \frac{2 L (L + 2)}{3 U_1 U_2^5} + \frac{L (4 L -
    129)}{12 U_1^2 U_2^4} - \frac{9 L}{U_1^3 U_2^3} - \frac{5 L}{U_1^4 U_2^2}
    \right) t^6 + O (t^8)\\
    \sqrt{\frac{2 L}{3}} t + \sqrt{\frac{2 L}{3}} \left( - \frac{3 + L}{3
    U_2^3} + \frac{3}{U_1 U_2^2} \right) t^3 + \sqrt{\frac{2 L}{3}} \left(
    \frac{L (L + 9)}{18 U_2^5} - \frac{3 L + 8}{U_1 U_2^4} - \frac{2 L - 17}{2
    U_1^2 U_2^3} + \frac{5}{U_1^3 U_2^2} \right) t^5 + O (t^7)\\
    \sqrt{2 L}  \frac{1}{U_1 U_2} t^2 + \sqrt{2 L} \left( - \frac{L + 2}{3 U_1
    U_2^3} + \frac{1}{6 U_1^2 U_2^2} + \frac{5}{3 U_1^3 U_2} \right) t^4 + O
    (t^6)\\
    \sqrt{\frac{2 L}{3}} \frac{2}{U_2^2} t^2 + \sqrt{\frac{2 L}{3}} \left( -
    \frac{2 L + 15}{3 U_2^4} + \frac{23}{2 U_1 U_2^3} + \frac{3}{U_1^2 U_2^2}
    \right) t^4 + O (t^6)\\
    \sqrt{L (2 L - 15)} \frac{1}{3 U_2^2} t^2 + \sqrt{\frac{L}{2 L - 15}}
    \left( - \frac{(L + 6) (2 L - 15)}{9 U_2^4} + \frac{4 L - 29}{U_1 U_2^3} +
    \frac{2}{U_1^2 U_2^2} \right) t^4 + O (t^6)\\
    \sqrt{\frac{2 L}{3}} \left( \frac{1}{U_2^3} + \frac{1}{2 U_1 U_2^2}
    \right) t^3 + O (t^5)\\
    \sqrt{\frac{2 L (2 L - 17)}{3}}  \frac{1}{U_1 U_2^2} t^3 + O (t^5)\\
    \sqrt{\frac{2 L}{15}}  \left( \frac{8}{U_1 U_2^2} + \frac{5}{U_1^2 U_2}
    \right) t^3 + O (t^5)\\
    \sqrt{\frac{2 L}{3}}  \frac{5}{U_2^3} t^3 + O (t^5)\\
    \sqrt{L (4 L - 42)}  \frac{2}{3 U_2^3} t^3 + O (t^5)\\
    \sqrt{\frac{2 L}{15}}  \frac{1}{U_1 U_2^2} t^3 + O (t^5)\\
    \sqrt{2 L (L - 12) (2 L - 21)}  \frac{1}{9 U_2^3} t^3 + O (t^5)
  \end{array}\right)
  $}
\end{equation}

\paragraph{$Q = 1 / 4$, $p = 3$ (non-integrable), SCE step 3:}
\begin{equation}
  \resizebox{\textwidth}{!}{$
  | \tmop{GS} \rangle = \left( \begin{array}{c}
    1 - \frac{Lt^2}{4 U_3^2} + \left( \frac{L (L^{} + 4)}{32 U_3^4} - \frac{3
    L}{2 U_2 U_3^3} - \frac{3 L}{4 U^2_2 U_3^2} \right) t^4 + L \left( -
    \frac{(L - 8) (L + 20)}{384 U_3^6} + \frac{3 (L - 43)}{16 U_2^2 U_3^4} +
    \frac{3 L + 8}{8 U_2 U_3^5} - \frac{5}{2 U_1^2 U_2^2 U_3^2} - \frac{5}{U_1
    U_2^3 U_3^2} - \frac{5}{U_1 U_2^2 U_3^3} - \frac{15}{4 U_2^4 U_3^2} -
    \frac{27}{4 U_2^3 U_3^3} \right) t^6 + O (t^8)\\
    \sqrt{\frac{L}{2}} \frac{1}{U_3} t + \sqrt{\frac{L}{2}} \left( - \frac{L +
    4}{4 U_3^3 } + \frac{3}{U_2 U_3^2} \right) t^3 + \sqrt{\frac{L}{2}} \left(
    \frac{L (L + 12)}{32 U_3^5} - \frac{3 L - 34}{4 U_2^2 U_3^3} - \frac{9 L +
    32}{4 U_2 U_3^4} + \frac{5}{U_2^3 U_3^2} + \frac{10}{U_1 \text{} U_2^2
    U_3^2} \right) t^5 + O (t^7)\\
    \sqrt{\frac{3 L}{2}} \frac{1}{U_2 U_3} t^2 + \sqrt{\frac{3 L}{2}} \left(
    \frac{10}{3 U_1 U_2^2 U_3} - \frac{(3 L + 8)}{12 U_2 U_3^3} + \frac{5}{3
    U_2^3 U_3} + \frac{1}{6 U_2^2 U_3^2} \right) t^4 + O (t^6)\\
    \sqrt{2 L} \frac{1}{U_3^2} t^2 + \sqrt{2 L} \left( - \frac{L + 10}{4
    U_3^4} + \frac{3}{2 U_2^2 U_3^2} + \frac{23}{4 U_2 U_3^3} \right) t^4 + O
    (t^6)\\
    \sqrt{\frac{L (L - 10)}{8}} \frac{1}{U_3^2} t^2 + \sqrt{\frac{L}{2 (L -
    10)}} \left( \frac{3 L - 29}{U_2 U_3^3} - \frac{(L - 10) (L + 8)}{8 U_3^4}
    + \frac{2}{U_2^2 U_3^2} \right) t^4 + O (t^6)\\
    \frac{\sqrt{5 L}}{U_1 U_2 U_3} t^3 + O (t^5)\\
    \sqrt{\frac{L}{2}} \left( \frac{1}{U_3^3} + \frac{1}{2 U_2 U_3^2} \right)
    t^3 + O (t^5)\\
    \frac{\sqrt{L (3 L - 34)}}{2 U_2 U_3^2} t^3 + O (t^5)\\
    \sqrt{\frac{2 L}{5}} \left( \frac{5}{2 U_2^2 U_3} + \frac{4}{U_2 U_3^2}
    \right) t^3 + O (t^5)\\
    \sqrt{\frac{L}{2}} \frac{5}{U_3^3} t^3 + O (t^5)\\
    \sqrt{L (L - 14)} \frac{1}{U_3^3} t^3 + O (t^5)\\
    \sqrt{\frac{L}{10}} \frac{1}{U_2 U_3^2} t^3 + O (t^5)\\
    \sqrt{\frac{L (L - 14) (L - 16)}{3}} \frac{1}{4 U_3^3} t^3 + O (t^5)
  \end{array} \right)
  $}
\end{equation}

\newpage
\section{Correlators}\label{ch:appendixscecorr}

Here we present the fermion--fermion correlators for critical densities, see
Chapter \ref{ch:1-SCEcritical}.

\paragraph{$Q = 1 / 3$, $p = 2$ (non-integrable), SCE step 3:}
\begin{eqnarray}
  \hat{N}_1 & = & \tmop{diag} (0, 0, 1, 0, 0, 0, 1, 1, 0, 0, 1, 0) \\
  \hat{N}_2 & = & \tmop{diag} (0, 1, 0, 1, 2, 2, 1, 0, 1, 2, 0, 3) \\
  \hat{N}_3 & = & \tmop{diag} \left( \frac{L}{3}, \frac{L - 6}{3}, \frac{L -
  8}{3}, \frac{L - 6}{3}, \frac{(L - 9) (2 L - 21)}{3 (2 L - 15)}, \frac{L -
  9}{3}, \right. \\
  &  & \left. \frac{\frac{2}{3} L^2 - 15 L + 86}{2 L - 17}, \frac{L - 9}{3},
  \frac{L - 6}{3}, \frac{\frac{2}{3} L^2 - 15 L + 87}{2 L - 21}, \frac{L -
  9}{3}, \frac{(L - 15) (2 L - 27)}{3 (2 L - 21)} \right) \nonumber\\
  \hat{N}_4 & = & \tmop{diag} \left( 0, 1, \frac{5}{3}, 1, \frac{4 L - 33}{2 L
  - 15}, 1, \frac{16 L - 171}{3 (2 L - 17)}, 3, 1, \frac{4 L - 45}{2 L - 21},
  3, \frac{6 (L - 12)}{2 L - 21} \right) \\
  \hat{N}_5 & = & \tmop{diag} \left( 0, 1, \frac{5}{3}, 2, \frac{4 L - 39}{2 L
  - 15}, 2, \right. \\
  &  & \left. \frac{16 L - 156}{3 (2 L - 17)}, 1, 2, \frac{6 (L - 12)}{2 L -
  21}, 1, \frac{3 (2 L^2 - 54 L + 369)}{(2 L - 21) (L - 12)} \right) \nonumber
\end{eqnarray}

\paragraph{$Q = 1 / 4$, $p = 3$ (non-integrable), SCE step 3:}
\begin{eqnarray}
  \hat{N}_1 & = & \tmop{diag} (0, 0, 0, 0, 0, 1, 0, 0, 0, 0, 0, 0, 0) \\
  \hat{N}_2 & = & \tmop{diag} (0, 0, 1, 0, 0, 0, 0, 1, 1, 0, 0, 1, 0) \\
  \hat{N}_3 & = & \tmop{diag} (0, 1, 0, 1, 2, 0, 2, 1, 0, 1, 2, 0, 3) \\
  \hat{N}_4 & = & \tmop{diag} \left( \frac{L}{4}, \frac{L}{4} - 2, \frac{L}{4}
  - \frac{8}{3}, \frac{L}{4} - 2, \frac{(L - 14) (L - 12)}{4 (L - 10)},
  \frac{L}{4} - \frac{29}{10}, \frac{L}{4} - 3, \right. \\
  &  & \left. \frac{3 L^2 - 90 L + 680}{4 (3 L - 34)}, \frac{L}{4} - 3,
  \frac{L}{4} - 2, \frac{L^2 - 30 L + 232}{4 (L - 14)}, \frac{L}{4} - 3,
  \frac{(L - 20) (L - 18)}{4 (L - 14)} \right) \nonumber\\
  \hat{N}_5 & = & \tmop{diag} \left( 0, 1, \frac{4}{3}, 1, \frac{2 (L - 12)}{L
  - 10}, 1, 0, \frac{7 (L - 14)}{3 L - 34}, 2, 1, \frac{2 (L - 16)}{L - 14},
  2, \frac{3 (L - 18)}{L - 14} \right) 
\end{eqnarray}
\section{Near-critical ground state energies for even number of
particles}\label{ch:appendixsceeven}

For even $N$, the energies will have an additional correction, due to
additional phase factor that is acquired during hopping of the particle in its
bosonic spin-half formulation (\tmtextit{cf.} footnote in Chapter
\ref{ch:1-jordanwigner}). The ground state energies were calculated to be:
\begin{eqnarray}
  E_{0 \circ} & = & - 2 N \frac{t^2}{U_p} + O (t^4), \\
  E_{1 \circ} & = & - 2 \left( \cos \left( \frac{\pi}{L} - | \phi | \right)
  \right) t - 2 \left( N - 2 \sin^2  \left( \frac{\pi}{L} - | \phi | \right)
  \right)  \frac{t^2}{U_p} + 2 F' (L) \frac{t^3}{U_p^2} + O (t^4), \\
  E_{2 \circ} & = & - 4 \left( \cos \frac{\pi}{N + 2} \cos \frac{\pi}{L}
  \right) t + A (N) A' (N, L)  \frac{t^2}{U_p} + B (N) B' (N, L) 
  \frac{t^3}{U_p^2} + O (t^4), 
\end{eqnarray}
where functions $A (N)$ and $B (N)$ are taken from Eq.~(\ref{eq:1-corr2}) and
other approximations are
\begin{equation}
  A' (N, L) \approx 1 - \frac{1}{(0.139639 L^2 - 0.0358 L + 0.8911) (0.1683 N
  + 0.354)},
\end{equation}
\begin{equation}
  B' (N, L) \approx 1 + \frac{2.44 N + 1}{0.76 L^2 - 5. L + 68.},
\end{equation}
\begin{equation}
  F' (L, N, \phi) \approx 2 + 2.81105 \left( 1 - \cos \frac{10.2633}{L}
  \right) + F'_1 (L, N) | \phi | + F'_2 (L, N) | \phi |^2 .
\end{equation}

\chapter{Charge density waves}

\section{Higher critical densities: Mathematica
code}\label{ch:appendixCDWMath}

We start by generating a partial basis for a specific density \tmverbatim{Q}
and interaction range \tmverbatim{p}. Particle number \tmverbatim{n} is set to
one, and will be increased after every iteration. Basis is generated according
to Theorem \ref{1-CDWth1}.

{\vspace{0.4cm}\setstretch{1.}
\begin{tmindent}
  \begin{tmcode}
{\color[HTML]{707070}(* System setup *)}
\tmcolor{blue}{Q} = 1/3; {\color[HTML]{707070}(* Particle density *)}
\tmcolor{blue}{p} = 4; {\color[HTML]{707070}(* Maximum interaction range *)}
\tmcolor{blue}{U} = \{\tmcolor{blue}{U1}, \tmcolor{blue}{U2}, \tmcolor{blue}{U3}, \tmcolor{blue}{U4}\}; {\color[HTML]{707070}(* All non-zero potential energies *)}
\tmcolor{blue}{n} = 12; {\color[HTML]{707070}(* Number of particles in the system *)}

{\color[HTML]{707070}(* Generate partial basis *)}
\tmcolor{blue}{f}[\tmtextit{{\color[HTML]{008000}n\_}}, \tmtextit{{\color[HTML]{008000}k\_}}] := Permutations[Join[ConstantArray[1, \tmtextit{{\color[HTML]{008000}k}}], ConstantArray[0, \tmtextit{{\color[HTML]{008000}n}}-\tmtextit{{\color[HTML]{008000}k}}]]]
\tmcolor{blue}{Possibilities} = Map[Join[\{1\}, Table[0, \{{\color[HTML]{008080}m}, 2, 1/\tmcolor{blue}{Q}\}], \tmtextit{{\color[HTML]{008000}\#}}]\&, \tmcolor{blue}{f}[(\tmcolor{blue}{n}-1)/\tmcolor{blue}{Q}, \tmcolor{blue}{n}-1]];
  \end{tmcode}
\end{tmindent}}

\begin{figure}[h]
  \centering
  \resizebox{10cm}{!}{\includegraphics{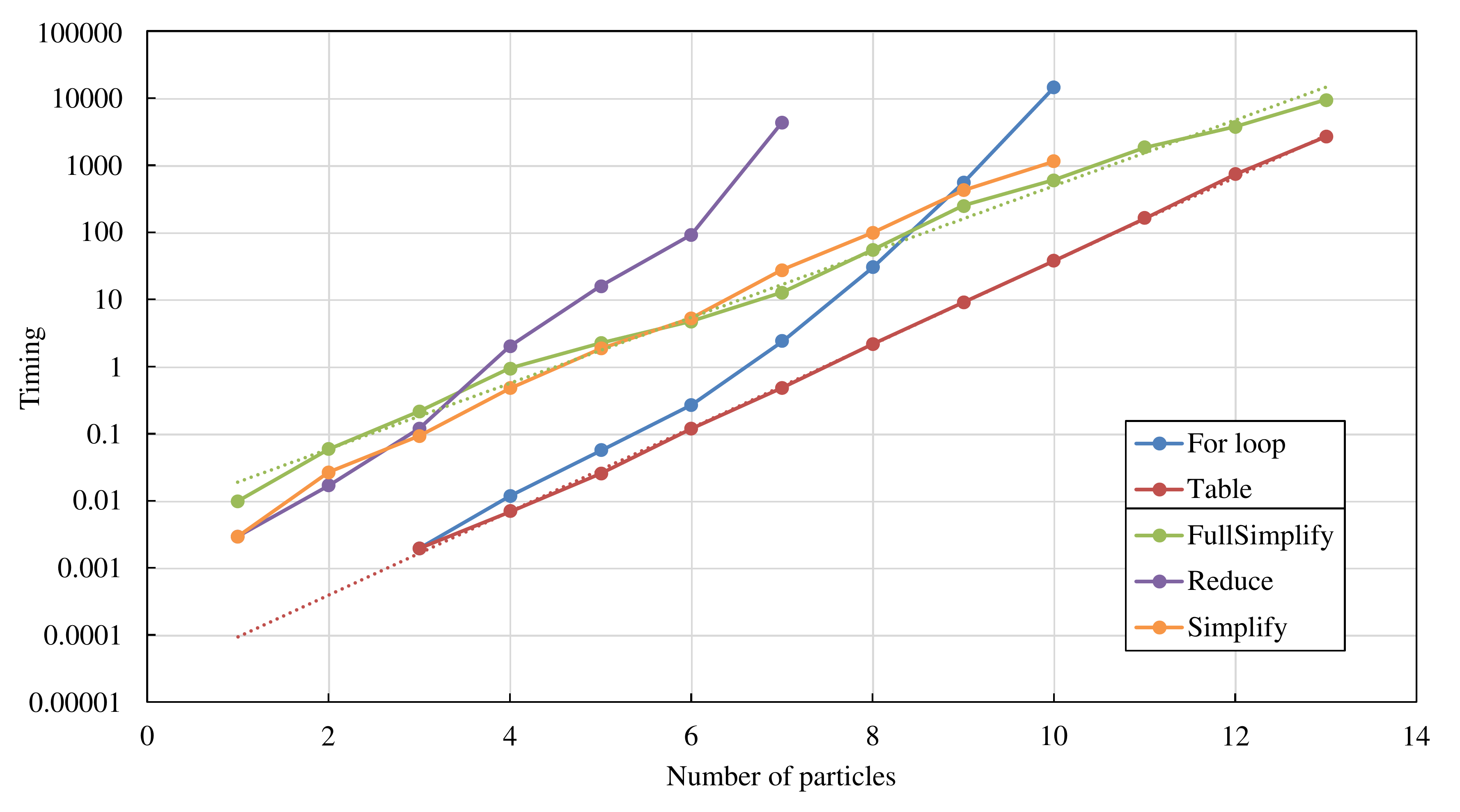}}
  \caption{Efficiency of setting up tables (\tmverbatim{For[]} loop and
  \tmverbatim{Table[]} function) and running simplification
  (\tmverbatim{Simplify[]}, \tmverbatim{FullSimplify[]} and
  \tmverbatim{Reduce[]} functions) for an example system $p = 4, Q = 1 / 2$.
  Dotted lines show exponential scaling.\label{fig:appendixCDWefficiency}}
\end{figure}

We calculate energy density for every state and then the list of energies is
simplified by removing any duplicates. For \tmverbatim{n}~$>$~1, we also need
to make sure that the list consists of energies calculated in previous
iterations. Instead of using loop statements to fill in tables, it was found
that the Mathematica's \tmverbatim{Table[]} function was more efficient and
the time consumption scaled exactly exponentially with the system size (see
Fig.~\ref{fig:appendixCDWefficiency}).

{\vspace{0.4cm}\setstretch{1.}
\begin{tmindent}
\begin{tmcode}
\tmcolor{blue}{Energies} = Table[0,\{Length[\tmcolor{blue}{Possibilities}]\}];
(\tmcolor{blue}{Energies} = Table[
\ \ \ \tmcolor{blue}{ene} = 0;
\ \ \ Do[
\ \ \ \ \ \ \ If[\tmcolor{blue}{Possibilities}[[{\color[HTML]{008B8B}i}, {\color[HTML]{008B8B}k}]] == 1 \&\&
\ \ \ \ \ \ \ \ \ \ \ \tmcolor{blue}{Possibilities}[[{\color[HTML]{008B8B}i}, Mod[{\color[HTML]{008B8B}k} + {\color[HTML]{008B8B}j}, Length[\tmcolor{blue}{Possibilities}[[1]]], 1]]] == 1,
\ \ \ \ \ \ \ \ \ \ \ \tmcolor{blue}{ene} = \tmcolor{blue}{ene} + \tmcolor{blue}{U}[[{\color[HTML]{008B8B}j}]]
\ \ \ \ \ \ \ ]
\ \ \ \ \ \ \ , \{{\color[HTML]{008B8B}j},1,\tmcolor{blue}{p}\}, \{{\color[HTML]{008B8B}k},1,Length[\tmcolor{blue}{Possibilities}[[1]]]\}
\ \ \ ];
\ \ \ \tmcolor{blue}{ene}
\ \ \ , \{{\color[HTML]{008B8B}i},1,Length[\tmcolor{blue}{Energies}]\}
]) //AbsoluteTiming {\color[HTML]{707070}(* Print timing for efficiency check *)}

\tmcolor{blue}{Energies} = DeleteDuplicates[
\ \ \ FullSimplify[Join[\tmcolor{blue}{PreviousEnergies}, \tmcolor{blue}{Energies}/Length[\tmcolor{blue}{Possibilities}[[1]]]]]
];
\end{tmcode}
\end{tmindent}
}

Next step is to assess whether the condition $\forall_{\beta \neq \alpha}
E_{\alpha} < E_{\beta}$ is false (\tmtextit{c.f.} Chapter
\ref{ch:1-CDWdetailscalc}) for a given $E_{\alpha}$. The list of energies
$E_{\alpha}$ that do not render this condition false must define the phases of
the system. Mathematica provides three simplifying statements that can be used
in this case:
\begin{itemize}
  \item \tmtexttt{Simplify[]}, which uses algebraic and other simple
  transformations to find the simplest form possible;
  
  \item \tmtexttt{FullSimplify[]}, which uses much more advanced
  transformations, that could involve elementary and special functions; the
  final form is at least as simple as the one returned by
  \tmtexttt{Simplify[]};
  
  \item \tmtexttt{Reduce[]}, which solves equations or inequalities and
  eliminates quantifiers in the statement provided.
\end{itemize}
Mathematica's \tmverbatim{Reduce[]} function was found to give the most
reliable simplification results, however it also needed much higher
computational resources. \tmverbatim{Simplify[]} and
\tmverbatim{FullSimplify[]} were found to be quite similar in resource
consumption (see Fig. \ref{fig:appendixCDWefficiency}) and the latter was
chosen due to higher reliability for simplifying complicated conditions.

{\vspace{0.4cm}\setstretch{1.}
\begin{tmindent}
  \begin{tmcode}
\tmcolor{blue}{a} = \tmcolor{blue}{U1}>0 \&\& \tmcolor{blue}{U2}>0 \&\& \tmcolor{blue}{U3}>0 \&\& \tmcolor{blue}{U4}>0;
\tmcolor{blue}{len} = Length[\tmcolor{blue}{Energies}];
Row[\{ProgressIndicator[Dynamic[\tmcolor{blue}{len}], \{0, \tmcolor{blue}{len}\}], " ", Dynamic[\tmcolor{blue}{len}]\}]
\ \ \ {\color[HTML]{707070}(* Show progress *)}
For[{\color[HTML]{008B8B}i} = 1, {\color[HTML]{008B8B}i} <= Length[\tmcolor{blue}{Energies}], {\color[HTML]{008B8B}i}++, \ \ \ \tmcolor{blue}{b} = \tmcolor{blue}{a};
\ \ \ For[{\color[HTML]{008B8B}j} = 1, {\color[HTML]{008B8B}j} <= Length[\tmcolor{blue}{Energies}], {\color[HTML]{008B8B}j}++,
\ \ \ \ \ \ \ If[{\color[HTML]{008B8B}i} != {\color[HTML]{008B8B}j},
\ \ \ \ \ \ \ \ \ \ \ \tmcolor{blue}{b} = \tmcolor{blue}{b} \&\& \tmcolor{blue}{Energies}[[{\color[HTML]{008B8B}i}]] < \tmcolor{blue}{Energies}[[{\color[HTML]{008B8B}j}]]
\ \ \ \ \ \ \ ]
\ \ \ ];
\ \ \ If[Not[FullSimplify[\tmcolor{blue}{b}]],
\ \ \ \ \ \ \ \tmcolor{blue}{Energies} = Drop[\tmcolor{blue}{Energies}, \{{\color[HTML]{008B8B}i}\}];
\ \ \ \ \ \ \ {\color[HTML]{008B8B}i}--
\ \ \ ];
\ \ \ \tmcolor{blue}{len} = Length[\tmcolor{blue}{Energies}]
] //AbsoluteTiming {\color[HTML]{707070}(* Print timing for efficiency check *)}
\tmcolor{blue}{PreviousEnergies} = \tmcolor{blue}{Energies}
  \end{tmcode}
\end{tmindent}}

Another simplification followed, this time using \tmverbatim{Reduce[]}
function. Final step consists of checking the generated states against Theorem
\ref{1-CDWth2}.

\section{Matrix product states}\label{ch:appendixMPS}

\subsection{Brief introduction}

In the matrix product states approach, the idea is that the ground state of
the desired system can be represented as a tensor product of matrices (matrix
product state, MPS) residing separately on each site, \tmtextit{i.e.}:
\begin{equation}
  | \psi \rangle = \sum_{\{ \sigma_i \}} A^{\sigma_1} A^{\sigma_2} \cdots
  A^{\sigma_L} | \sigma_1 \sigma_2 \cdots \sigma_L \rangle,
\end{equation}
or in a diagrammatic form presented in Figure \ref{fig:appendixmps-tensor}a.
$| \sigma_1 \sigma_2 \cdots \sigma_L \rangle$ are states of the computational
basis (can be a Fock basis or a spin basis). $A$ are matrices that reside on
sites, and may be rectangular. Their dimensions are such that by contracting
them, $A^{\sigma_1} A^{\sigma_2} \cdots A^{\sigma_L}$, we recover a~scalar.

\begin{figure}[h]
  \resizebox{12cm}{!}{\includegraphics{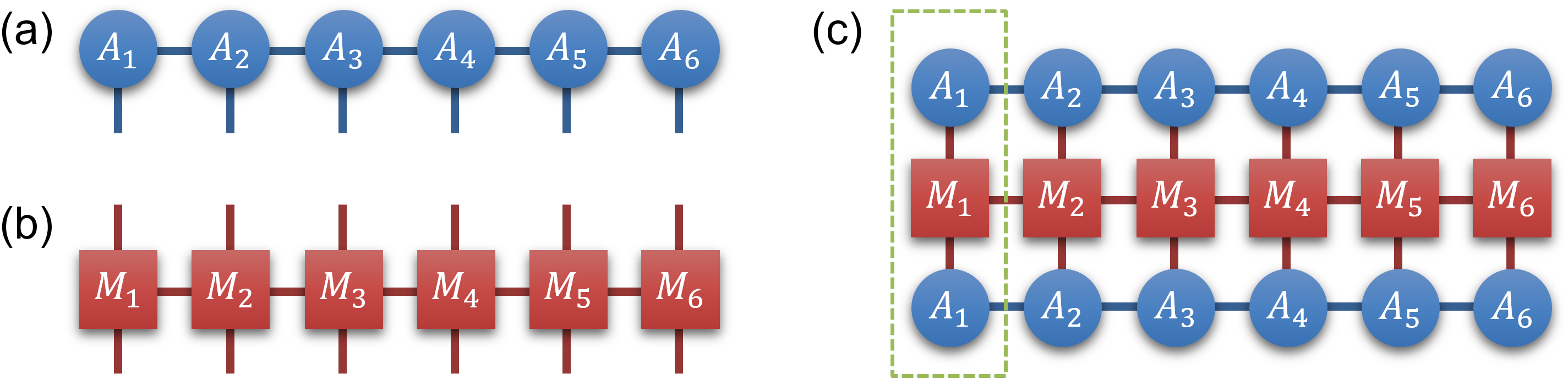}}
  \caption{(a) State $| \psi \rangle$ as a matrix product state. Every leg
  denotes tensor index, line connecting two tensors is a contraction. (b) The
  Hamiltonian or any other operator in a matrix product operator form.
  (c)~Calculation of the energy in the system, $\langle \psi | H | \psi
  \rangle$, as an application of MPS on both sides of an MPO. Notice that one
  can locally contract the tensors (green box).\label{fig:appendixmps-tensor}}
\end{figure}

Similarly, the Hamiltonian of the model can be represented as a matrix product
operator (MPO),
\begin{equation}
  H = \sum_{\{ \sigma_i \}} \sum_{\{ \sigma'_j \}} M^{\sigma_1 \sigma'_1}
  M^{\sigma_2 \sigma'_2} \cdots M^{\sigma_L \sigma'_L} | \sigma_1 \sigma_2
  \cdots \sigma_L \rangle \langle \sigma_1' \sigma_2' \cdots \sigma_L' |,
\end{equation}
or diagrammatically on Figure \ref{fig:appendixmps-tensor}b. The ingenuity of
the matrix product states approach lies in the fact that one can locally
contract an application of the Hamiltonian on a state -- see
Fig.~\ref{fig:appendixmps-tensor}c. This makes computational algorithms, such
as density matrix product state (DMRG) approach much more efficient.

\subsection{MPO representation of the Hamiltonian}

We will now devise an MPO representation of the Hamiltonian of the generalised
$t$-$V$ model. Firstly, we start with the spin equivalent of the Hamiltonian
in a $p = 2$ case, in which there is now the following potential term:
\begin{equation}
  U_1 \sum_{i = 1}^L \mathbbm{P}^{\uparrow}_i \mathbbm{P}^{\uparrow}_{i + 1} +
  U_2 \sum_{i = 1}^L \mathbbm{P}^{\uparrow}_i \mathbbm{P}^{\uparrow}_{i + 2} .
\end{equation}
To represent the Hamiltonian as an MPO, we shall consider the action of the
automaton, presented in Fig.~\ref{fig:appendixmps-auto1}.

\begin{figure}[h]
  \centering
  \begin{minipage}{.48\textwidth}
    \centering
    \includegraphics[width=.55\linewidth]{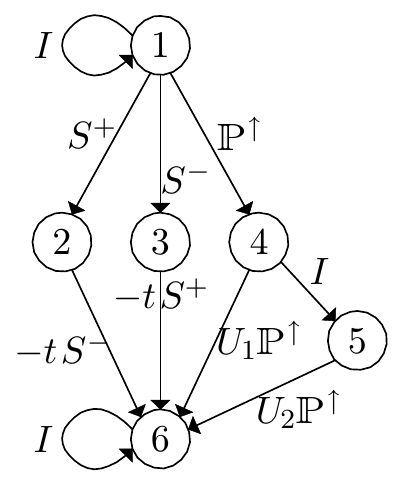}
    \captionof{figure}{
      Automaton for the $p = 2$ generalised $t$-$V$ model.
    }
    \label{fig:appendixmps-auto1}
  \end{minipage}\hspace{0.03\textwidth}
  \begin{minipage}{.48\textwidth}
    \centering
    \includegraphics[width=.9\linewidth]{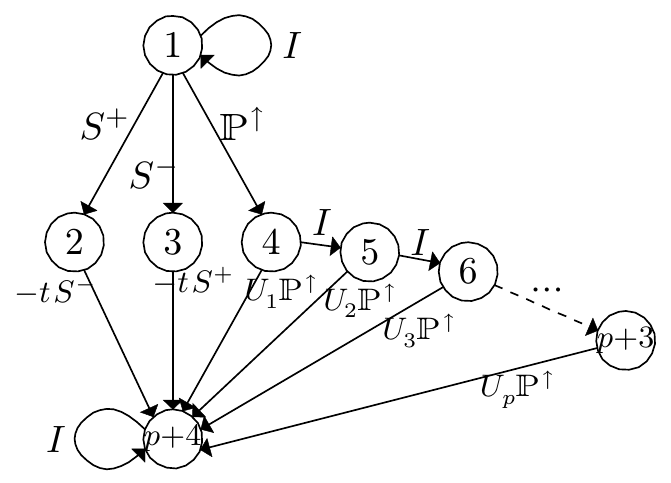}
    \captionof{figure}{
      Automaton for the generalised $t$-$V$ model with any interaction range 
      $p$.
    }
    \label{fig:appendixmps-auto2}
  \end{minipage}
\end{figure}

Such an automaton can be written as the following one-site matrix:
\begin{equation}
  M^{[i]} = \left(\begin{array}{cccccc}
    I & 0 & 0 & 0 & 0 & 0\\
    S^+ & 0 & 0 & 0 & 0 & 0\\
    S^- & 0 & 0 & 0 & 0 & 0\\
    \mathbbm{P}^{\uparrow} & 0 & 0 & 0 & 0 & 0\\
    0 & 0 & 0 & I & 0 & 0\\
    0 & - t S^- & - t S^+ & U_1 \mathbbm{P}^{\uparrow} & U_2
    \mathbbm{P}^{\uparrow} & I
  \end{array}\right),
\end{equation}
where each row/column is a transition between specific automaton state.
Additionally, for open-boundary conditions, we have:
\begin{equation}
  M^{[1]} = \left(\begin{array}{cccccc}
    0 & - t S^- & - t S^+ & U_1 \mathbbm{P}^{\uparrow} & U_2
    \mathbbm{P}^{\uparrow} & I
  \end{array}\right), \qquad M^{[1]} = \left(\begin{array}{c}
    I\\
    S^+\\
    S^-\\
    \mathbbm{P}^{\uparrow}\\
    0\\
    0
  \end{array}\right) .
\end{equation}
A similar automaton (see Fig.~\ref{fig:appendixmps-auto2}) can be devised for
the Hamiltonian with any interaction range $p$. The corresponding one-site
matrix is:
\begin{equation}
  M^{[i]} = \left(\begin{array}{cccccccccc}
    I & 0 & 0 & 0 & 0 & 0 & \cdots & 0 & 0 & 0\\
    S^+ & 0 & 0 & 0 & 0 & 0 & \cdots & 0 & 0 & 0\\
    S^- & 0 & 0 & 0 & 0 & 0 & \cdots & 0 & 0 & 0\\
    \mathbbm{P}^{\uparrow} & 0 & 0 & 0 & 0 & 0 & \cdots & 0 & 0 & 0\\
    0 & 0 & 0 & I & 0 & 0 & \cdots & 0 & 0 & 0\\
    0 & 0 & 0 & 0 & I & 0 & \cdots & 0 & 0 & 0\\
    \vdots & \vdots & \vdots & \vdots & \ddots & \ddots & \ddots & \vdots &
    \vdots & \vdots\\
    0 & 0 & 0 & 0 & 0 & 0 & \ddots & 0 & 0 & 0\\
    0 & 0 & 0 & 0 & 0 & 0 & \ddots & I & 0 & 0\\
    0 & - t S^- & - t S^+ & U_1 \mathbbm{P}^{\uparrow} & U_2
    \mathbbm{P}^{\uparrow} & U_3 \mathbbm{P}^{\uparrow} & \cdots & U_{p \um 1}
    \mathbbm{P}^{\uparrow} & U_p \mathbbm{P}^{\uparrow} & I
  \end{array}\right) .
\end{equation}
For a periodic system, the problem is much more complex. For example, an
automaton of a periodic XY model is presented in
Fig.~\ref{fig:appendixmps-auto3}. The corresponding one-site representation
is:

\begin{figure}
  \centering
  \begin{minipage}{.48\textwidth}
    \centering
    \includegraphics[width=0.95\linewidth]{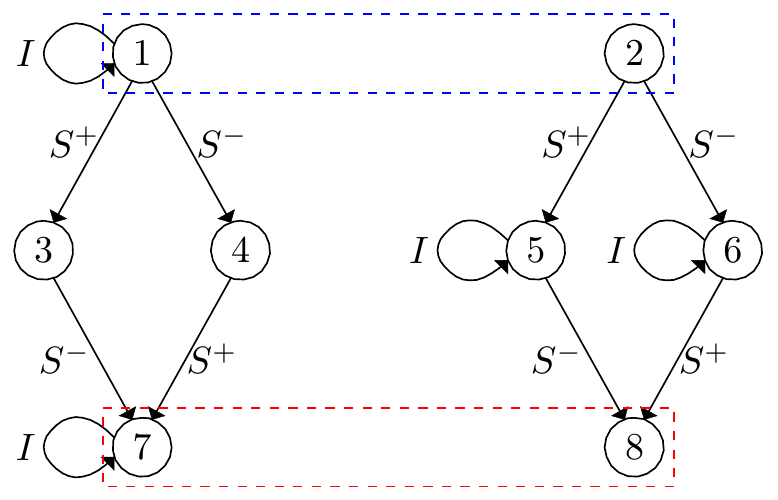}
    \captionof{figure}{
      Automaton for a simple periodic XY model.
    }
    \label{fig:appendixmps-auto3}
  \end{minipage}\hspace{0.03\textwidth}
  \begin{minipage}{.48\textwidth}
    \centering\vspace{0.5cm}
    \includegraphics[width=.3\linewidth]{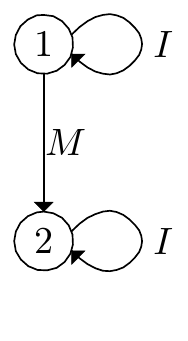}
    \captionof{figure}{
      Automaton for a one-site operator $M$.
    }
    \label{fig:appendixmps-auto4}
  \end{minipage}
\end{figure}

\begin{equation}
  M^{[i]} = \left(\begin{array}{cccccccc}
    I & \cdot & \cdot & \cdot & \cdot & \cdot & \cdot & \cdot\\
    \cdot & \cdot & \cdot & \cdot & \cdot & \cdot & \cdot & \cdot\\
    S^+ & \cdot & \cdot & \cdot & \cdot & \cdot & \cdot & \cdot\\
    S^- & \cdot & \cdot & \cdot & \cdot & \cdot & \cdot & \cdot\\
    \cdot & S^+ & \cdot & \cdot & I & \cdot & \cdot & \cdot\\
    \cdot & S^- & \cdot & \cdot & \cdot & I & \cdot & \cdot\\
    \cdot & \cdot & S^- & S^+ & \cdot & \cdot & I & \cdot\\
    \cdot & \cdot & \cdot & \cdot & S^- & S^+ & \cdot & \cdot
  \end{array}\right),
\end{equation}
\begin{equation}
  M^{[1]} = \left(\begin{array}{cccccccc}
    \cdot & \cdot & S^- & S^+ & S^- & S^+ & I & \cdot
  \end{array}\right), \qquad M^{[L]} = \left(\begin{array}{c}
    I\\
    \cdot\\
    S^+\\
    S^-\\
    S^+\\
    S^-\\
    \cdot\\
    \cdot
  \end{array}\right),
\end{equation}
where zeroes are represented as dots for clarity.

One-site operators, such as the particle number operator $\sum_i \hat{n}_i$,
can be always represented as an MPO using the automaton from
Fig.~\ref{fig:appendixmps-auto4} and the following matrix representation:
\begin{equation}
  M^{[i]} = \left(\begin{array}{cc}
    I & 0\\
    M & I
  \end{array}\right) .
\end{equation}

\chapter{Quantum Monte Carlo}

\section{Fitting formulas}\label{ch:appendixQMCfitting}

\subsection{Exciton}\label{ch:appendixQMCfittingX}

Fitting parameters from Eq.~(\ref{eq:2-xbindfit}) are given in
Table~\ref{tab:2-xbindfit}.

\begin{table}[h!]
  \centering
  \begin{tabular}{ccc}
    \hline\hline
    Parameter & Estimate & Standard error\\
    \hline
    $a_1$ & $- 55.566$ & 0.024\\
    $a_2$ & $102.45$ & 0.14\\
    $a_3$ & $- 99.57$ & 0.25\\
    $a_4$ & $43.06$ & 0.16\\
    $a_5$ & $- 4.380$ & 0.029\\
    $b_1$ & $- 4.718$ & 0.013\\
    $b_2$ & $3.718$ & 0.013\\
    \hline\hline
  \end{tabular}
  \caption{Fitting parameters for the exciton binding
  energy.\label{tab:2-xbindfit}}
\end{table}

\newpage

\subsection{Trion}\label{ch:appendixQMCfittingXminus}

Fitting parameters from Eq.~(\ref{eq:2-trionbindfit}) are given in
Table~\ref{tab:2-trionbindfit}.

\begin{table}[h!]
  \begin{tabular}{ccc}
    \hline\hline
    Parameter & Estimate & Standard error\\
    \hline
    $a_{00}$ & 0.8210 & 0.0015\\
    $a_{01}$ & $- 1.76$ & 0.16\\
    $a_{02}$ & 1.5 & 1.5\\
    $a_{03}$ & 0.4 & 5.8\\
    $a_{04}$ & 0.2 & 11.\\
    $a_{05}$ & $- 4.$ & 13.\\
    $a_{06}$ & 4. & 8.\\
    $a_{07}$ & $- 1.5$ & 1.9\\
    $a_{10}$ & $- 6.31$ & 0.05\\
    $a_{11}$ & 13.2 & 0.6\\
    $a_{12}$ & $- 10.$ & 4.\\
    $a_{13}$ & $- 11.$ & 13.\\
    $a_{14}$ & 26. & 19.\\
    $a_{15}$ & $- 19.$ & 14.\\
    $a_{16}$ & 5. & 4.\\
    $a_{20}$ & 38.6 & 0.6\\
    $a_{21}$ & $- 66.5$ & 2.1\\
    $a_{22}$ & 58. & 7.\\
    $a_{23}$ & $- 8.$ & 17.\\
    $a_{24}$ & $- 22.$ & 21.\\
    $a_{25}$ & 17. & 14.\\
    $a_{26}$ & $- 3.3$ & 4.0\\
    $a_{30}$ & $- 148.6$ & 3.6\\
    $a_{31}$ & 187. & 7.\\
    $a_{32}$ & $- 140.$ & 10.\\
    $a_{33}$ & 40. & 11.\\
    $a_{34}$ & 0. & 8.\\
    $a_{35}$ & $- 2.8$ & 2.5\\
    $a_{40}$ & 360. & 11.\\
    $a_{41}$ & $- 302.$ & 15.\\
    $a_{42}$ & 172. & 11.\\
    $a_{43}$ & $- 35.$ & 6.\\
    $a_{44}$ & 2.8 & 1.8\\
    $a_{50}$ & $- 550.$ & 19.\\
    $a_{51}$ & 279. & 18.\\
    $a_{52}$ & $- 107.$ & 7.\\
    $a_{53}$ & 9.9 & 1.6\\
    $a_{60}$ & 518. & 20.\\
    $a_{61}$ & $- 137.$ & 12.\\
    $a_{62}$ & 27.3 & 2.1\\
    $a_{70}$ & $- 273.$ & 12.\\
    $a_{71}$ & 27.6 & 3.1\\
    $a_{80}$ & 62.2 & 3.0\\
    \hline\hline
  \end{tabular}
  \caption{Fitting parameters for the trion binding
  energy.\label{tab:2-trionbindfit}}
\end{table}

\newpage

\subsection{Biexciton}\label{ch:appendixQMCfittingXXminus}

Fitting parameters from Eq.~(\ref{eq:2-XXbinding-fit}) are given in Table
\ref{tab:2-biexbindfit}. Notice that $a_{2 j}$ parameters are missing, since
they can be incorporated into $a_{0 j}$.

\begin{table}[h!]
  \begin{tabular}{ccc}
    \hline\hline
    Parameter & Estimate & Standard error\\
    \hline
    $a_{00}$ & 6.9 & {\hspace{2.5em}}2.3\\
    $a_{01}$ & $- 40.$ & 11.\\
    $a_{02}$ & 161. & 21.\\
    $a_{03}$ & $- 460.$ & 32.\\
    $a_{04}$ & 900. & {\underline{5}}0.\\
    $a_{05}$ & $- 1170.$ & {\underline{7}}0.\\
    $a_{06}$ & 1010. & {\underline{6}}0.\\
    $a_{07}$ & $- 514.$ & 37.\\
    $a_{08}$ & 118. & 9.\\
    $a_{10}$ & $- 3.3$ & 0.7\\
    $a_{11}$ & 20.2 & 3.9\\
    $a_{12}$ & $- 82.$ & 10.\\
    $a_{13}$ & 195. & 20.\\
    $a_{14}$ & $- 265.$ & 31.\\
    $a_{15}$ & 189. & 26.\\
    $a_{16}$ & $- 54.$ & 8.\\
    $a_{30}$ & $- 3.7$ & 2.9\\
    $a_{31}$ & 12. & 11.\\
    $a_{32}$ & $- 17.$ & 13.\\
    $a_{33}$ & 12. & 8.\\
    $a_{34}$ & $- 4.5$ & 3.1\\
    $a_{40}$ & 1.3 & 1.3\\
    $a_{41}$ & $- 4.$ & 4.\\
    $a_{42}$ & 2.7 & 3.5\\
    \hline\hline
  \end{tabular}
  \caption{Fitting parameters for the biexciton binding
  energy.\label{tab:2-biexbindfit}}
\end{table}

\section{Contact pair density}\label{ch:AppendixCPDF}

Here we present calculations that justify approximations of the contact pair
density for a~negative trion and a biexciton from
Eqs.~(\ref{eq:2-contactinttrioneh}) and (\ref{eq:2-contactintxxeh}).

The contact pair density for a trion is defined as:
\begin{equation}
  \rho_{\tmop{eh}}^{\text{X}^-} (0) = \left\langle \delta \left(
  \vec{r}_{\text{e}_1} - \vec{r}_{\text{h}_1} \right) + \delta \left(
  \vec{r}_{\text{e}_2} - \vec{r}_{\text{h}_1} \right) \right\rangle .
\end{equation}
Let us assume that the trion wave function $\psi$ can be separated into
exciton wave function $\phi$ and a free electron wave function $\varphi$, that
are spatially separated:
\begin{equation}
  \psi \left( \vec{r}_{\text{e}_1}, \vec{r}_{\text{e}_2}, \vec{r}_{\text{h}_1}
  \right) = \phi \left( \vec{r}_{\text{e}_1}, \vec{r}_{\text{h}_1} \right)
  \varphi \left( \vec{r}_{\text{e}_2} \right) .
\end{equation}
We can write the normalisation constant of the wave function $\psi$ as:
\begin{eqnarray}
  C & = & \int | \psi |^2 \mathd \vec{r}_{\text{e}_1} \mathd
  \vec{r}_{\text{e}_2} \mathd \vec{r}_{\text{h}_1} \\
  & = & \underbrace{\int \left| \phi \left( \vec{r}_{\text{e}_1},
  \vec{r}_{\text{h}_1} \right) \right|^2 \mathd \vec{r}_{\text{e}_1} \mathd
  \vec{r}_{\text{h}_1}}_{C_{\phi}}  \underbrace{\int \left| \varphi \left(
  \vec{r}_{\text{e}_2} \right) \right|^2 \mathd
  \vec{r}_{\text{e}_2}}_{C_{\varphi}} \nonumber\\
  & = & C_{\phi} C_{\varphi} . \nonumber
\end{eqnarray}
Now, we can evaluate the trion contact pair density:
\begin{eqnarray}
  \rho_{\tmop{eh}}^{\text{X}^-} & = & \frac{1}{C} \int \left[ \delta \left(
  \vec{r}_{\text{e}_1} - \vec{r}_{\text{h}_1} \right) \underset{}{} + \delta
  \left( \vec{r}_{\text{e}_2} - \vec{r}_{\text{h}_1} \right) \right] | \psi
  |^2 \mathd \vec{r}_{\text{e}_1} \mathd \vec{r}_{\text{e}_2} \mathd
  \vec{r}_{\text{h}_1} \\
  & = & \frac{1}{C} \int \left| \phi \left( \vec{r}_{\text{e}_1},
  \vec{r}_{\text{e}_1} \right) \right|^2  \left| \varphi \left(
  \vec{r}_{\text{e}_2} \right) \right|^2 \mathd \vec{r}_{\text{e}_1} \mathd
  \vec{r}_{\text{e}_2} \nonumber\\
  &  & + \frac{1}{C} \underbrace{\int \left| \phi \left(
  \vec{r}_{\text{e}_1}, \vec{r}_{\text{e}_2} \right) \right|^2  \left| \varphi
  \left( \vec{r}_{\text{e}_2} \right) \right|^2 \mathd \vec{r}_{\text{e}_1}
  \mathd \vec{r}_{\text{e}_2}}_{\substack{\approx\ 0 \text{, since the wave
  functions } \phi \text{ and } \varphi \text{ are spatially separated}\\
  \text{(exciton and free electron are distant from each other)}
  }} \nonumber\\
  & \approx & \frac{1}{C_{\phi}} \int \left| \phi \left(
  \vec{r}_{\text{e}_1}, \vec{r}_{\text{e}_1} \right) \right|^2 \mathd
  \vec{r}_{\text{e}_1} \nonumber\\
  & = & \rho_{\tmop{eh}}^{\text{X}}, \nonumber
\end{eqnarray}
which explains Eq.~(\ref{eq:2-contactinttrioneh}).

Similarly, the biexciton contact pair density is:
\begin{equation}
  \rho_{\tmop{eh}}^{\tmop{XX}} (0) = \left\langle \delta \left(
  \vec{r}_{\text{e}_1} - \vec{r}_{\text{h}_1} \right) + \delta \left(
  \vec{r}_{\text{e}_1} - \vec{r}_{\text{h}_2} \right) + \delta \left(
  \vec{r}_{\text{e}_2} - \vec{r}_{\text{h}_1} \right) + \delta \left(
  \vec{r}_{\text{e}_2} - \vec{r}_{\text{h}_2} \right) \right\rangle .
\end{equation}
We separate the biexciton wave function $\psi$ into two exciton wave functions
$\phi$ and $\varphi$:
\begin{equation}
  \psi \left( \vec{r}_{\text{e}_1}, \vec{r}_{\text{e}_2},
  \vec{r}_{\text{h}_1}, \vec{r}_{\text{h}_2} \right) = \phi \left(
  \vec{r}_{\text{e}_1}, \vec{r}_{\text{h}_1} \right) \varphi \left(
  \vec{r}_{\text{e}_2}, \vec{r}_{\text{h}_2} \right) .
\end{equation}
The normalisation of $\psi$ is:
\begin{eqnarray}
  C & = & \int | \psi |^2 \mathd \vec{r}_{\text{e}_1} \mathd
  \vec{r}_{\text{e}_2} \mathd \vec{r}_{\text{h}_1} \mathd \vec{r}_{\text{h}_2}
  \\
  & = & \underbrace{\int \left| \phi \left( \vec{r}_{\text{e}_1},
  \vec{r}_{\text{h}_1} \right) \right|^2 \mathd \vec{r}_{\text{e}_1} \mathd
  \vec{r}_{\text{h}_1}}_{C_{\phi}}  \underbrace{\int \left| \varphi \left(
  \vec{r}_{\text{e}_2}, \vec{r}_{\text{h}_2} \right) \right|^2 \mathd
  \vec{r}_{\text{e}_2} \mathd \vec{r}_{\text{h}_2}}_{C_{\varphi}} \nonumber\\
  & = & C_{\phi} C_{\varphi} . \nonumber
\end{eqnarray}
Finally, we evaluate the biexciton contact pair density:
\begin{eqnarray}
  \rho_{\tmop{eh}}^{\tmop{XX}} & = & \frac{1}{C} \int \left[ \delta \left(
  \vec{r}_{\text{e}_1} - \vec{r}_{\text{h}_1} \right) \underset{}{} + \delta
  \left( \vec{r}_{\text{e}_1} - \vec{r}_{\text{h}_2} \right) \right. \\
  &  & \left. + \underset{}{} \delta \left( \vec{r}_{\text{e}_2} -
  \vec{r}_{\text{h}_1} \right) + \delta \left( \vec{r}_{\text{e}_2} -
  \vec{r}_{\text{h}_2} \right) \right] | \psi |^2 \mathd \vec{r}_{\text{e}_1}
  \mathd \vec{r}_{\text{e}_2} \mathd \vec{r}_{\text{h}_1} \mathd
  \vec{r}_{\text{h}_2} \nonumber\\
  & = & \frac{1}{C} \int \left| \phi \left( \vec{r}_{\text{e}_1},
  \vec{r}_{\text{e}_1} \right) \right|^2  \left| \varphi \left(
  \vec{r}_{\text{e}_2}, \vec{r}_{\text{h}_2} \right) \right|^2 \mathd
  \vec{r}_{\text{e}_1} \mathd \vec{r}_{\text{e}_2} \mathd \vec{r}_{\text{h}_2}
  \nonumber\\
  &  & + \frac{1}{C} \underbrace{\int \left| \phi \left(
  \vec{r}_{\text{e}_1}, \vec{r}_{\text{h}_1} \right) \right|^2  \left| \varphi
  \left( \vec{r}_{\text{e}_2}, \vec{r}_{\text{e}_1} \right) \right|^2 \mathd
  \vec{r}_{\text{e}_1} \mathd \vec{r}_{\text{e}_2} \mathd
  \vec{r}_{\text{h}_1}}_{\substack{
    \approx\ 0 \text{, since the wave functions } \phi \text{ and } \varphi
    \text{ are spatially separated}\\
    \text{(the two excitons do not occupy the same space)}
  }} \nonumber\\
  &  & + \frac{1}{C} \underbrace{\int \left| \phi \left(
  \vec{r}_{\text{e}_1}, \vec{r}_{\text{e}_2} \right) \right|^2  \left| \varphi
  \left( \vec{r}_{\text{e}_2}, \vec{r}_{\text{h}_2} \right) \right|^2 \mathd
  \vec{r}_{\text{e}_1} \mathd \vec{r}_{\text{e}_2} \mathd
  \vec{r}_{\text{h}_2}}_{\approx\ 0 \text{, as above}} \nonumber\\
  &  & + \frac{1}{C} \int \left| \phi \left( \vec{r}_{\text{e}_1},
  \vec{r}_{\text{h}_1} \right) \right|^2  \left| \varphi \left(
  \vec{r}_{\text{e}_2}, \vec{r}_{\text{e}_2} \right) \right|^2 \mathd
  \vec{r}_{\text{e}_1} \mathd \vec{r}_{\text{e}_2} \mathd \vec{r}_{\text{h}_1}
  \nonumber\\
  & \approx & \frac{1}{C_{\phi}} \int \left| \phi \left(
  \vec{r}_{\text{e}_1}, \vec{r}_{\text{e}_1} \right) \right|^2 \mathd
  \vec{r}_{\text{e}_1} + \frac{1}{C_{\varphi}} \int \left| \varphi \left(
  \vec{r}_{\text{e}_2}, \vec{r}_{\text{e}_2} \right) \right|^2 \mathd
  \vec{r}_{\text{e}_2} \nonumber\\
  & = & \rho_{\tmop{eh}}^{\text{X}_1} + \rho_{\tmop{eh}}^{\text{X}_2}
  \nonumber\\
  & = & 2 \rho_{\tmop{eh}}^{\text{X}}, \nonumber
\end{eqnarray}
which justifies Eq.~(\ref{eq:2-contactintxxeh}).

\end{document}